\documentclass[reprint, amsmath,amssymb,showpacs,floatfix,longbibliography, aps, twocolumn, superscriptaddress,nofootinbib]{revtex4-1}
\usepackage{CJKutf8}

\usepackage{amsthm}
\newtheorem{theorem}{Theorem}

\usepackage[utf8]{inputenc}
\usepackage{amsmath,amssymb,amsfonts,stmaryrd}
\DeclareMathOperator*{\argmax}{arg\,max}

\usepackage{physics}
\usepackage{dsfont}
\usepackage{graphicx}
\usepackage{xcolor}
\usepackage{graphicx}
\usepackage{cancel}
\usepackage{natbib}
\usepackage{color}
\usepackage{tikz}
\usepackage{mathrsfs}
\usepackage{relsize}
\usepackage{hyperref}
\hypersetup{colorlinks=true,citecolor=blue,linkcolor=blue, urlcolor=blue, breaklinks=true}
\hypersetup{linktocpage}
\usepackage[all]{xy}
\usepackage{yhmath}
\usepackage{blkarray}

\usepackage{diagbox}
\usepackage{algorithm}
\usepackage{algpseudocode}
\usepackage{lipsum}
\algrenewcommand\algorithmicrequire{\textbf{Input:}}
\algrenewcommand\algorithmicensure{\textbf{Output:}}

\usepackage{graphics}

\usepackage{bm} 
\usepackage{subfigure} 
\usepackage{environ}
\usepackage{url}
\usepackage[margin=0.75in]{geometry}
\usepackage{scalerel}
\usepackage{stackengine,wasysym}
\newcommand\wwide[1]{\ThisStyle{%
  \setbox0=\hbox{$\SavedStyle#1$}%
  \stackengine{-.1\LMpt}{$\SavedStyle#1$}{%
    \stretchto{\scaleto{\SavedStyle\mkern.2mu\AC}{.5150\wd0}}{.6\ht0}%
  }{O}{c}{F}{T}{S}%
}}

\newcommand*\circled[1]{\tikz[baseline=(char.base)]{\node[shape=circle,draw,inner sep=0.25pt] (char) {#1};}}

\allowdisplaybreaks

\usepackage{mathtools}

\newcommand{\FF}{{\mathbb F}}

\newcommand{\ZZ}{{\mathbb Z}}

\newcommand{\eps}{\epsilon}

\newcommand{\bx}{{\overline{x}}}
\newcommand{\by}{{\overline{y}}}
\newcommand{\mX}{{\mathcal{X}}}
\newcommand{\mZ}{{\mathcal{Z}}}
\newcommand{\mS}{{\mathcal{S}}}
\newcommand{\mP}{{\mathcal{P}}}
\newcommand{\mO}{{\mathcal{O}}}

\NewEnviron{eqs}{%
\begin{equation}\begin{split}
    \BODY
\end{split}\end{equation}}

\newcommand{\change}{\color{black}}

\makeatletter
\newenvironment{breakablealgorithm}
{
    \begin{center}
        \refstepcounter{algorithm}
        \renewcommand{\caption}[1]
        {
            \addcontentsline{loa}{algorithm}{\protect\numberline{\thealgorithm}##1}
            \parbox{0.5\textwidth}
            {
                \hrule height.8pt depth0pt \kern2pt
                {\raggedright\textbf{\fname@algorithm~\thealgorithm} ##1\par}
                \kern2pt\hrule\kern2pt
            }
        }
}
{
        \kern2pt\hrule\relax
    \end{center}
}
\makeatother

\begin{document}

\begin{CJK*}{UTF8}{gbsn}

\title{Extracting topological orders of generalized Pauli stabilizer codes in two dimensions}

\author{Zijian Liang (梁子健)}
\thanks{equal contribution.}
\affiliation{International Center for Quantum Materials, School of Physics, Peking University, Beijing 100871, China}
\affiliation{School of Physics, Shandong University, Jinan 250100, China}

\author{Yijia Xu (许逸葭)}
\thanks{equal contribution.}
\affiliation{Joint Quantum Institute and Joint Center for Quantum Information and Computer Science, NIST/University of Maryland, College Park, Maryland 20742, USA}
\affiliation{Institute for Physical Science and Technology, University of Maryland, College Park, Maryland 20742, USA}

\author{Joseph T.~Iosue}
\affiliation{Joint Quantum Institute and Joint Center for Quantum Information and Computer Science, NIST/University of Maryland, College Park, Maryland 20742, USA}

\author{Yu-An Chen \CJKfamily{bsmi}(陳昱安)}
\email[E-mail: ]{yuanchen@pku.edu.cn}
\affiliation{International Center for Quantum Materials, School of Physics, Peking University, Beijing 100871, China}

\date{\today}

\begin{abstract}
In this paper, we introduce an algorithm for extracting topological data from translation invariant generalized Pauli stabilizer codes in two-dimensional systems, focusing on the analysis of anyon excitations and string operators. The algorithm applies to $\mathbb{Z}_d$ qudits, including instances where $d$ is a nonprime number.
This capability allows the identification of topological orders that differ from the $\mathbb{Z}_d$ toric codes.
{\color{black}It extends our understanding beyond the established theorem that Pauli stabilizer codes for $\mathbb{Z}_p$ qudits (with $p$ being a prime) are equivalent to finite copies of $\mathbb{Z}_p$ toric codes and trivial stabilizers.}
The algorithm is designed to determine all anyons and their string operators, enabling the computation of their fusion rules, topological spins, and braiding statistics. The method converts the identification of topological orders into computational tasks, including Gaussian elimination, the Hermite normal form, and the Smith normal form of truncated Laurent polynomials.
Furthermore, the algorithm provides a systematic approach for studying quantum error-correcting codes. We apply it to various codes, such as self-dual CSS quantum codes modified from the \textcolor{black}{2d honeycomb} color code and non-CSS quantum codes that contain the double semion topological order or the six-semion topological order.
\end{abstract}

\maketitle
\end{CJK*}

\tableofcontents

\section{Introduction}
The topological phases of matter in gapped systems are conceptualized as equivalence classes of physical systems. These systems are considered equivalent if they can be continuously transformed into each other while maintaining an energy gap above the ground state subspace. A defining characteristic of a topological phase is its stability against local perturbation. As a result, research in this field focuses on identifying topological invariants under continuous transformations and building models that exhibit various invariants. To this end, numerous intricate models have been studied and partially classified across different spatial dimensions, both with and without symmetries, in both condensed matter physics \cite{chamon2005quantum, Chen2011Complete, levin2012braiding, chen2012symmetry, gu2014effect,gu2014lattice,jian2014layer,wang2015topological,ye2016topological, yoshida2016topological, vijay2016fracton, Kapustin2017Higher, shirley2018fracton, Lan2018Classification, wang2018towards,cheng2018loop,nandkishore2019fractons, Lan2019Classification, pretko2020fracton,wang2020construction,qi2021fracton,Chen2021Disentangling, Johnson2022Classification, chen2023exactly} and quantum information \cite{bravyi1998quantum, dennis2002topological, kitaev2003fault, kitaev2006anyons, bombin2006topological,haah2011local, bombin2015gauge, bravyi2010majorana, vijay2015majorana, ellison2022pauli}.  Despite these advances, the complete classification and comprehensive characterization of topological phases remain areas of ongoing research. 

A notable subset of lattice models includes generalized Pauli stabilizer codes \cite{gottesman1997stabilizer}, characterized by Hamiltonian terms composed of commuting products of generalized Pauli operators for $\ZZ_d$ qudits.\footnote{A {\bf generalized Pauli stabilizer} Hamiltonian $H= - \sum_i S_i$ consists of stabilizers $S_i$ such that each $S_i$ is a product of generalized Pauli matrices and commutes with each other $[S_i, S_j] = 0$. The mathematical expression of generalized Pauli matrices will be shown explicitly in the next section. The group formed by $S_i$ (and their products) is the {\bf stabilizer group}.} Bombin's investigation into translation invariant topological stabilizer codes in two-dimensional lattices with qubit ($d=2$) degrees of freedom provided foundational insights \cite{bombin_Stabilizer_14}. He demonstrated that any non-chiral topological stabilizer code, under locality-preserving automorphisms of the operator algebra, can be expressed as a direct sum of toric code stabilizers and trivial stabilizers (corresponding to product states). A widely believed conjecture \cite{frohlich1990braid,rehren1989braid,kitaev2006anyons} states that all commuting projector Hamiltonians are non-chiral, suggesting that the classification of translation invariant Pauli stabilizer codes for two-dimensional qubit systems is essentially complete. Extending this framework, Haah applied these principles to prime-dimensional qudits\footnote{ A {\bf prime-dimensional qudit} refers to a $\ZZ_d$ qudit where the qudit dimension $d$ is a prime number.} \cite{haah_module_13, haah2016algebraic, haah_classification_21}. Given that any translation invariant topological general Pauli stabilizer code on such qudit systems must include a nontrivial ``boson'' \cite{haah_QCA_23}, it is inferred that all two-dimensional topological generalized Pauli stabilizer codes for prime-dimensional qudits can be transformed into a direct sum of finite copies of toric code stabilizers and trivial stabilizers using finite-depth Clifford circuits.
However, the classification of generalized Pauli stabilizer models with nonprime-dimensional qudits remains an unresolved challenge. It is conjectured that such systems may be equivalent to a direct sum of finite copies of trivial stabilizers and condensation descendants of the $\ZZ_d$ toric code stabilizers \cite{ellison2022pauli}.

The study of stabilizer codes is also crucial from the quantum code perspective, particularly in the context of error correction.
Self-correcting quantum memory and single-shot error-correcting codes can be constructed from topological phases \cite{bravyi2009no,bravyi2013quantum, terhal2015quantum, brown2016fault, duivenvoorden2018renormalization, campbell2019theory, kubica2022single, li2023phase}. The error correction capabilities and thresholds are also investigated \cite{andrist2015error, dauphinais2019quantum, varona2020determination, song2022optimal,zhu2022topological, schotte2022quantum, dua2023quantum, hastings2021dynamically, kesselring2022anyon, davydova2023floquet, davydova2023quantum, zhang2023xcube}.
In these codes, extended operators correspond to the logical operations within the code space. Recently, significant progress has been made in the experimental implementation of surface codes on various quantum computing platforms \cite{bluvstein2022quantum, bluvstein2023logical, google2023suppressing, krinner2022realizing, iqbal2023topological}.
The advancement of digital quantum computers and simulators has enabled the realization of various topological phases on quantum devices \cite{satzinger2021realizing, google2023non, semeghini2021probing, iqbal2023creation, foss2023experimental}. 
Notable examples of surface codes, including the toric code and the color code, have been sophisticatedly designed and analyzed. This leads to a crucial question: Can we systematically extract topological information from a wide range of quantum codes? For instance, inspired by the \textcolor{black}{2d honeycomb} color code, one can develop many new quantum codes, as shown in Fig.~\ref{fig: modified color codes}. A key challenge is the rapid determination of the properties of these codes.
Given the progress of experimental implementations, analyzing and characterizing different topological codes becomes critical. Many characterization protocols for ground state wavefunctions or parent Hamiltonians of different physical systems have been proposed \cite{kitaev2006anyons, Else2014Classifying, shiozaki2017many, shapourian2017partial, shiozaki2018many, shi2020fusion, Kawagoe2020Microscopic,cian2021many, dehghani2021extraction, Kawagoe2021Anomalies, cian2022extracting, kim2022chiral, kim2022modular, cong2019quantum, cong2022enhancing, ruba2022homological, fan2022generalized, tran2023measuring, xiao2023robust, kobayashi2023extracting, Zhang2023Topological}.
However, the systematic study of generalized Pauli stabilizer codes has been less investigated, and it would be fruitful to explore new models that could be implemented in near-term devices.
The close relationship between coding theory and the topological phases of matter suggests that coding theory techniques \cite{galindo2002information, tillich2013quantum,kovalev2012improved, bravyi2014homological, breuckmann2021quantum, breuckmann2021balanced, panteleev2021degenerate, tan2023fracton} are helpful for the study of topological orders.
This paper presents an algorithm that efficiently identifies the topological data of quantum codes via an algebraic method. This algorithm aims to streamline the study of the topological properties of generalized Pauli stabilizer codes with the assistance of classical computers.

\begin{figure}[htb]
    \centering
    \includegraphics[width=0.45\textwidth]{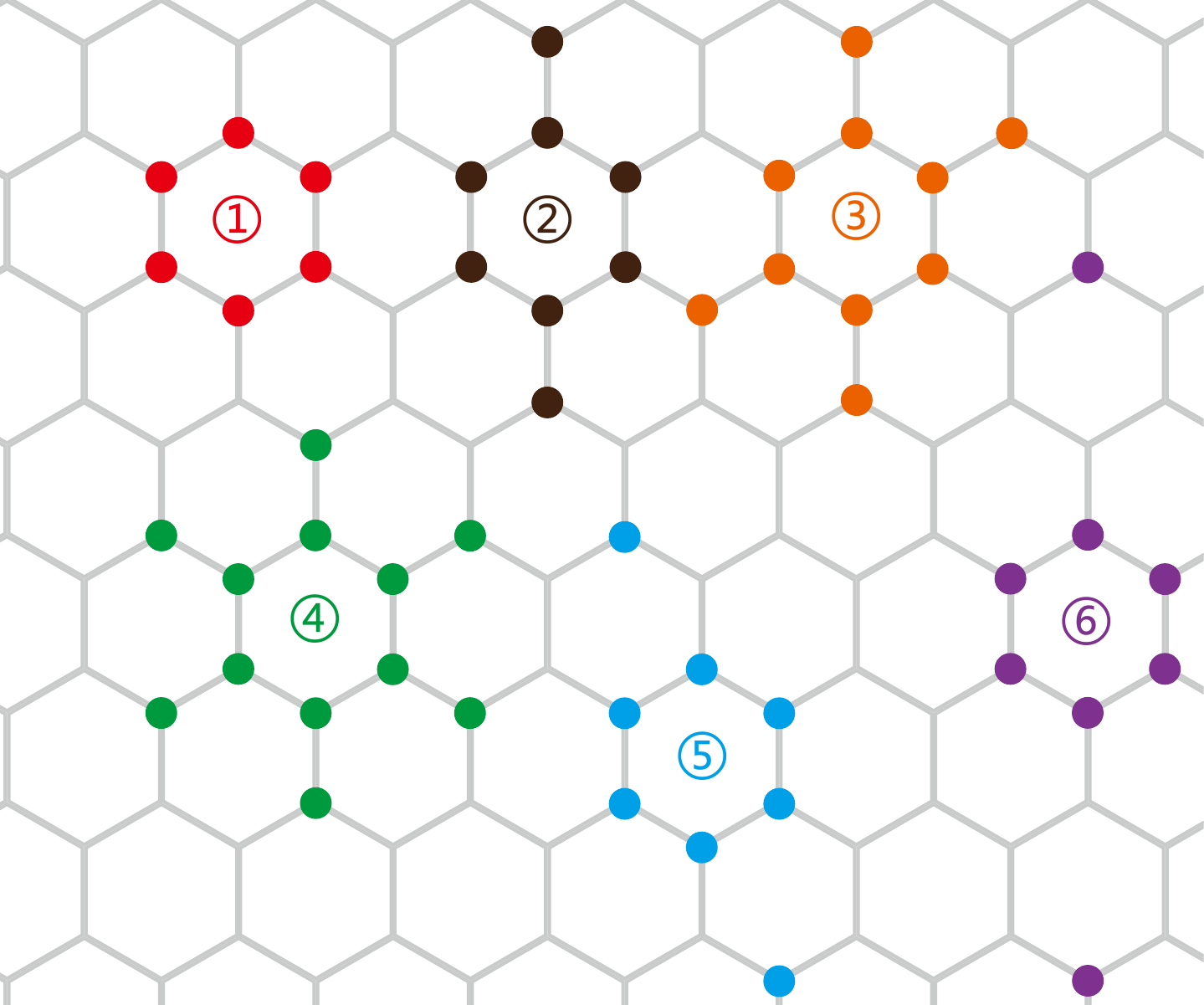}
\caption{Six examples of self-dual CSS codes in the 2d honeycomb lattice, where each vertex hosts one qubit. The first is the \textcolor{black}{2d honeycomb} color code \cite{bombin2006topological}, with stabilizers composed of all Pauli $X$ (or $Z$) operators on red vertices and their translations. The remaining five codes are similar but include additional vertices in each stabilizer term. While it is established that the \textcolor{black}{2d honeycomb} color code is equivalent to two copies of toric codes \cite{kubica2015unfolding}, the broader analysis of 2d translation invariant Pauli stabilizer models remains less clear. Our work introduces an algorithm capable of detecting topological orders in these models and extracting crucial information such as anyon string operators, fusion rules, topological spins, and braiding statistics.}
\label{fig: modified color codes}
\end{figure}

In summary, we provide an algorithm to determine the topological order of generalized Pauli stabilizer codes with $\ZZ_d$ qudits in two dimensions, which applies to both prime and non-prime $d$. Our method checks whether a generalized Pauli stabilizer code satisfies the topological order condition. Upon confirmation, the algorithm outputs the topological data of the input quantum code, the Abelian anyon theory (unitary modular tensor category). This includes identifying anyon types (simple objects), their fusion rules, explicit string operators, self-statistics (topological spins), and braiding statistics. The functions of our algorithm are shown in Fig.~\ref{fig:cartoon}.

The paper is organized as follows: In Sec.~\ref{sec: physical intuitions}, we review the Pauli stabilizer codes of qudits and the theory of Abelian anyons to motivate the development of our work. In Sec.~\ref{sec: laurent_polynomial}, we review the Laurent polynomial formulation of translation invariant Pauli stabilizer codes with examples and the definition of topological data in the Laurent polynomial framework. In Sect.~\ref{sec: computational_method}, we provide details of individual techniques used in the algorithm for extracting topological order, including the procedure to check the topological order condition, solve the anyon equation, and obtain topological spins for given anyons. In Sec.~\ref{sec: algorithm}, we discussed the workflow of the algorithm. In Sec.~\ref{sec: application_to_codes}, we apply our algorithm to various stabilizer codes defined on qubits or nonprime-dimensional qudits and show the effectiveness of this algorithm on extracting anyon string operators, fusion rules, topological spins, and braiding statistics.
{\color{black}
Sec.~\ref{sec: time complexity} analyzes the time complexity of this algorithm and demonstrates its scaling through numerical computation. Finally, the discussion regarding connections with other works and potential future directions is presented in Sec.~\ref{sec: discussion}.
}

\begin{figure*}[t]
    \centering
    \includegraphics[width=0.9\textwidth]{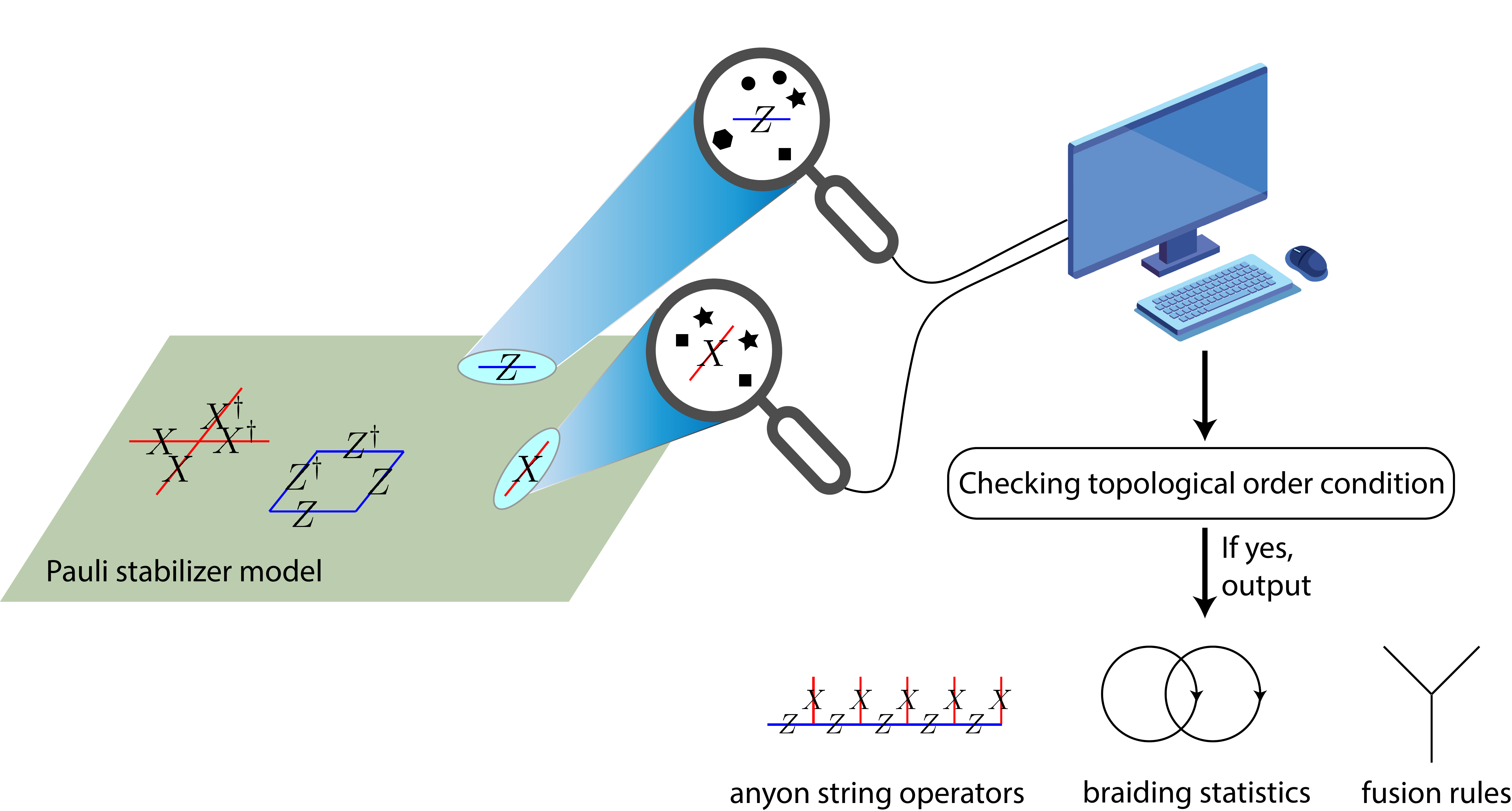}
    \caption{Illustration of the algorithm. We want to determine whether a given Pauli stabilizer model presents a topological order. If it is a topological order, what are anyon string operators, braiding statistics, and fusion rules? In the first step, we record the syndrome pattern of a single Pauli operator, which is indicated by circles, stars, squares, or hexagons representing the violations of different stabilizers. Next, we convert syndrome patterns into the Laurent polynomial rings. Then, we use our algorithm to check whether the Pauli stabilizer code satisfies the topological order condition. If satisfied, this algorithm can obtain the topological data of the given model, including anyon string operators, fusion rules, topological spins, and braiding statistics.}
    \label{fig:cartoon}
\end{figure*}

\section{Physical intuition}\label{sec: physical intuitions}

This section provides a pedagogical overview of generalized Pauli operators and Abelian anyon theories in the stabilizer formalism with the microscopic picture. Then, we outline the workflow of our algorithm.

First, recall standard definitions of $d \times d$ {\bf generalized Pauli matrices} for a $\ZZ_d$ qudit:
\begin{eqs}
    X=\sum_{j \in \mathbb{Z}_d}|j+1\rangle\langle j| ,\quad
    Z=\sum_{j \in \mathbb{Z}_d}  \omega^j |j\rangle\langle j| \text {. }
\end{eqs}
where $\omega$ is defined as $\omega := \exp (\frac{2 \pi i}{d})$. More explicitly,
\begin{eqs}
    X = 
    \begin{bmatrix}
    0 & 0 & \cdots & 0 & 1 \\
    1 & 0 & \cdots & 0 & 0 \\
    0 & 1 & \cdots & 0 & 0 \\
    \vdots & \vdots & \ddots & \vdots & \vdots \\
    0 & 0 & \cdots & 1 & 0
    \end{bmatrix},~
    Z = 
    \begin{bmatrix}
    1 & 0 & 0 & \cdots & 0 \\
    0 & \omega & 0 & \cdots & 0 \\
    0 & 0 & \omega^2 & \cdots & 0 \\
    \vdots & \vdots & \vdots & \ddots & \vdots \\
    0 & 0 & 0 & \cdots & \omega^{d-1}
    \end{bmatrix},
\end{eqs}
$X$ and $Z$ satisfy the commutation relation
\begin{eqs}
    Z X = \omega X Z.
\end{eqs}
For the sake of simplicity, ``Pauli" will be used as a shorthand for ``generalized Pauli".

We begin by considering a local\footnote{``Local" refers to each operator, e.g., stabilizer $S_i$, having Pauli matrices supported on a finite area in the lattice.} Pauli stabilizer Hamiltonian on a two-dimensional lattice. Our initial step is to check whether this Hamiltonian fulfills the {\bf topological order (TO) condition} \cite{haah_module_13, haah2016algebraic, haah_classification_21}\footnote{\color{black} More precisely, this condition stems from the TQO-1 condition as defined in Refs.~\cite{bravyi2010topological, bravyi2011short} for a broader range of commuting projector Hamiltonians.}. The detailed mathematical formulation of this condition will be addressed in subsequent sections. Briefly, the TO condition requires that any local operator $\mathcal{O}$, which commutes with all stabilizers, must be a product of certain stabilizers, denoted as $\mathcal{O} = \prod_{i \in A} S_i$ for a set $A$.
In other words, the TO condition implies that the stabilizer group is ``saturated,'' i.e., no more local operator can be added to {\color{black} it} while commuting with all existing stabilizers.
Fundamentally, this condition indicates the local indistinguishability of the ground state(s) in a local Pauli stabilizer code that satisfies this criterion.
If the Hamiltonian exhibits a degeneracy in its ground state, these states cannot be differentiated by any local operator. An example is the toric code \cite{kitaev2003fault}, which illustrates topological order with a 4-fold degeneracy on a torus. Thus, a local Pauli stabilizer code that meets the TO condition is referred to as {\bf topological Pauli stabilizer code}, which suggests the presence of topological order in this code (which may be a trivial order). 

Furthermore, given a topological Pauli stabilizer code, our interest then lies in identifying the specific type of ``topological order" it exhibits. {\bf Topological orders} in two dimensions are categorized by ``unitary modular tensor categories" (UMTC) \cite{rowell2006quantum, rowell2009classification, wang2010topological, Wang2022in, Plavnik2023Modular}, which describe the properties of local excitations in the low-energy spectrum of the Hamiltonian. It is important to note that stabilizer models are restricted to generating Abelian anyon theories, a subset within the broader scope of UMTCs. While UMTCs have a formal definition for general cases, this paper will focus only on Abelian anyon theories pertaining to stabilizer formalism and provide a more streamlined description.

\begin{figure}[thb]
    \centering
    \includegraphics[width=0.5\textwidth]{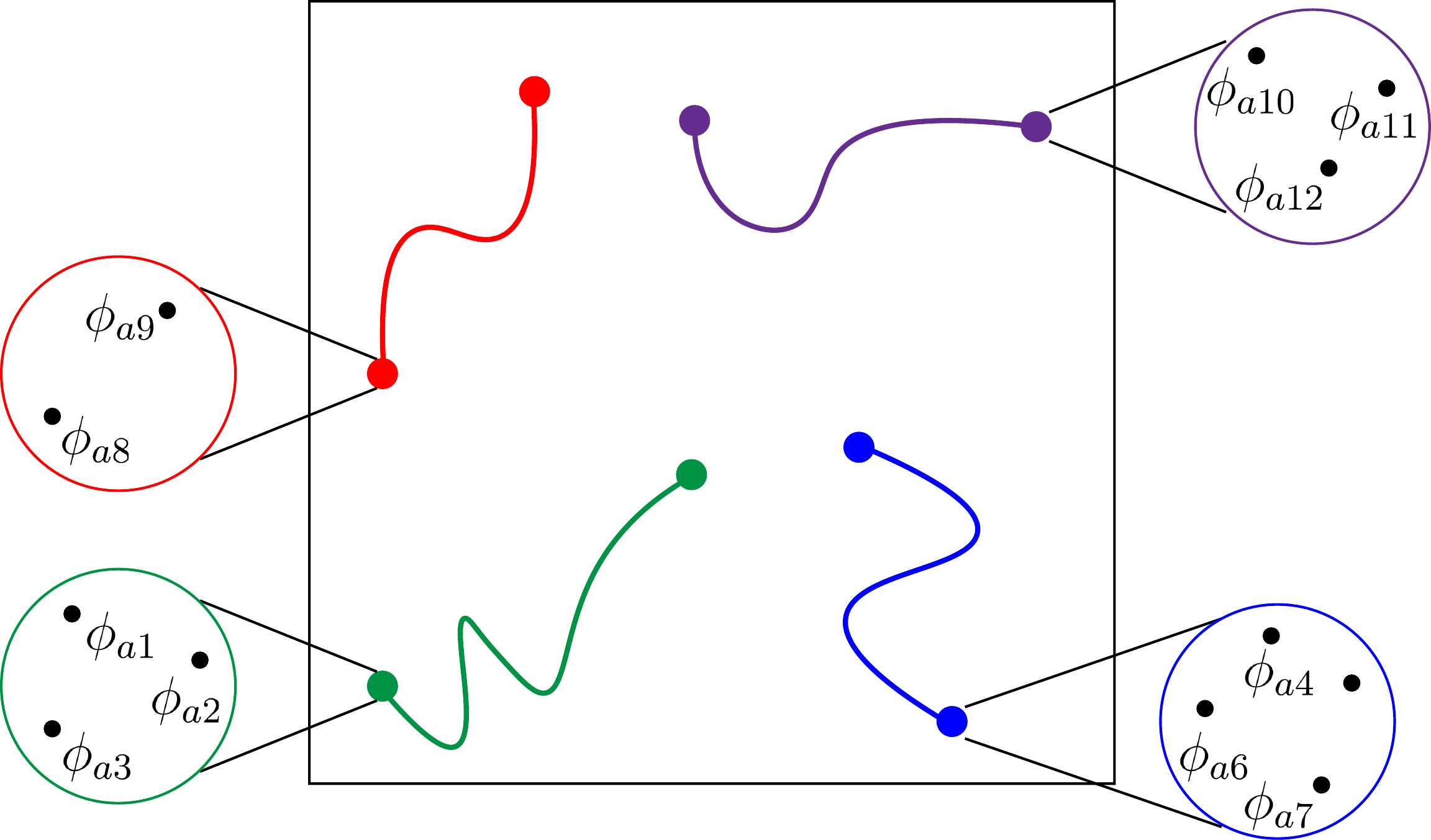}
    \caption{The black dots label the locations where stabilizers act as $e^{i \phi}$ (with $\phi \neq 0$) on the state. These stabilizers are labeled as $S_{a_1}$, $S_{a_2}$, and so on... If these violated stabilizers are spatially far apart, we group $\{ S_{a_i} \}$ into different patches and treat each patch as a local anyon. In topological stabilizer codes, any anyon is generated by a string operator, i.e., a product of Pauli matrices along the string that only fails to commute with a finite number of stabilizers near its endpoints.}
\label{fig: local anyons}
\end{figure}

An {\bf anyon} is a (local) violation of stabilizers on the lattice. Given a ground state $| \Psi_{\text{gs}} \rangle$, which is in the $+1$ eigenstate of all the stabilizers: $S_i | \Psi_{\text{gs}} \rangle = | \Psi_{\text{gs}} \rangle$. However, if a Pauli operator $M$ is applied on the ground state, the perturbed state might not be in the +1 eigenspace of stabilizers: $S_i (M | \Psi_{\text{gs}} \rangle) =  e^{i \phi_i} (M | \Psi_{\text{gs}} \rangle)$ with $\phi_i \in [0, 2 \pi)$, where $\phi_i$ depends on the commutation relation between $S_i$ and $M$. \textcolor{black}{This perturbed state contains an anyon labeled as a list $\{ \phi_1,...,\phi_N\}$ consisting $U(1)$ angles\footnote{\textcolor{black}{Note that we record the angles $\phi$ for all the stabilizers and their translations.}}.}  So far, an anyon is defined for a state, which is a global description.
Next, the locality plays a crucial role. If the violated stabilizers ($\phi_i \neq 0$) are spatially far apart, e.g., created by a long string operator $M$ that violates only a finite number of stabilizers near its two endpoints, we divide the stabilizers with $\phi_i \neq 0$ into local patches, as the physical picture shown in Fig.~\ref{fig: local anyons}. Each patch can be viewed as a local anyon, which $\{ \phi_i \}$ is referred to as the {\bf syndrome pattern}. Another way to visualize this local anyon is to extend $M$ into a semi-infinite\footnote{There are subtleties about semi-infinite operators in the infinite plane since the algebra is only well-defined for finite operators. In practice, a clear way to define a local anyon $a$ labeled by local ${\phi_i}$ is that for any arbitrary large $R$, there always exists some Pauli operator $M$ such that it violates $S_i$ the same as ${\phi_i}$, while commutes with all other stabilizers within range $R$. Refs.~\cite{haah_module_13, haah_classification_21} demonstrate that $M$ in a topological Pauli stabilizer code can always be chosen as a string operator that is extendable.} string, which only violates stabilizers around one endpoint, where the anyon is located. In two-dimensional topological Pauli stabilizer codes, each instance of local violations is attributable to string operators. Nonetheless, this property ceases to be applicable in higher-dimensional models, primarily due to the fracton phases of matter \cite{chamon2005quantum,haah2011local,vijay2016fracton,shirley2018fracton,pretko2020fracton}. Our method is adaptable to higher-dimensional topological Pauli stabilizer codes with string operators; however, our focus is predominantly on two-dimensional applications for the remainder of this study.

Now, anyon types (or superselection sectors) can be defined as equivalence classes under the equivalence relations between anyons $v$ and $v'$:
\begin{eqs}
    v := \{ \phi_{i}\} \sim v' := \{ \phi'_{i}\},
\label{eq: definition of anyon equivalence}
\end{eqs}
if and only if $\{ \phi_{i}\}$ and $\{ \phi'_{i}\}$ are differed by local Pauli operators. In other words, if the syndrome pattern $v'$ can be achieved by applying local Pauli operators on the state with the syndrome pattern $v$, two syndrome patterns $v$ and $v'$ are identified as the same type. With the concept of anyon types, we can now discuss the fusion rules. The {\bf fusion rules} of (Abelian) anyons describe the process of bringing two anyons $a$ and $b$ closer to each other (via their string operators) and identifying the composite of them as a third anyon $c$ under the equivalence relation~\eqref{eq: definition of anyon equivalence}. We express the fusion rule as
\begin{eqs}
    a \times b = c.
\end{eqs}
Also, the {\bf topological spin} $\theta(a)$ can be computed for each anyon $a$, which determines the exchange statistics (by spin-statistics theorem), e.g., boson, fermion, and semion. The T-junction process that exchanges the positions of two particles can detect the topological spin \cite{LW06, alicea2011non, KL20, fidkowski2022gravitational, haah_QCA_23}. Let $\bar{\gamma}_1$, $\bar{\gamma}_2$, and $\bar{\gamma}_3$ be paths sharing a common endpoint $p$ and ordered counter-clockwise around $p$, as in Fig.~\ref{fig: T junction 3 paths}. Then, the topological spin $\theta(a)$ of anyon $a$ is computed by the expression:
\begin{eqs} 
    W^a_{3} (W^a_{2})^\dagger W^a_{1} = \theta(a) W^a_{1}(W^a_{2})^\dagger W^a_{3},
\label{eq: statistics formula}
\end{eqs}
where $W^a_{i}$ is the string operator moving anyon $a$ along the path $\bar{\gamma}_i$.
\begin{figure}
\centering
    \includegraphics[width=.32\textwidth]{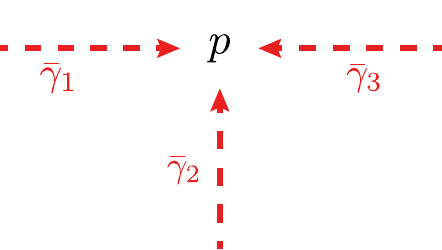}
     \caption{The exchange statistics (topological spin) of anyon $a$ can be computed using the formula in Eq.~\eqref{eq: statistics formula}. $\bar{\gamma}_1$, $\bar{\gamma}_2$, and $\bar{\gamma}_3$ are oriented paths in the lattice incident at the same position $p$. The string operators $W^a_{1}$, $(W^a_{2})^\dagger$ and $W^a_{3}$ of anyon $a$ along paths $\bar{\gamma}_1$, $-\bar{\gamma}_2$, and $\bar{\gamma}_3$ might not commute, giving the exchange statistics $\theta(a)$.}
     \label{fig: T junction 3 paths}
\end{figure}
Eq.~\eqref{eq: statistics formula} represents the exchange of two anyons\footnote{Anyon $a_1$ and $a_2$ are in the same anyon type. The subscript is a label to keep track of them.} $a_1,a_2$ can be seen as follows. Suppose that $a_1$ is initialized at the start $\bar \gamma_1(0)$ of $\bar \gamma_1$ and $a_2$ at the start $\bar \gamma_3(0)$ of $\bar \gamma_3$. Then $(W^a_{3})^\dagger W^a_{2} (W^a_{1})^\dagger W^a_{3}(W^a_{2})^\dagger W^a_{1}$ takes $a_1$ from $\bar\gamma_1(0) \to p \to \bar \gamma_2(0)$, then takes $a_2$ from $\bar\gamma_3(0) \to p \to \bar\gamma_1(0)$, and then takes $a_1$ from $\bar\gamma_2(0) \to p \to \bar\gamma_3(0)$.
Thus, in the end, we have exchanged the positions of $a_1$ and $a_2$, while keeping all dynamical phases cancel out. By Eq.~\eqref{eq: statistics formula}, this results in the phase $\theta(a)$. In Abelian anyon theories, the {\bf braiding statistics} (e.g., the phase factors arising from the Aharonov-Bohm effect) are completely determined by topological spins. More precisely, let $B_\theta(a,b)$ be the $U(1)$ phase of braiding anyon $a$ around anyon $b$ counterclockwise. It is related to the topological spins of $a$, $b$ and $a \times b$ by the relation
\begin{eqs}
    \label{eq:braiding statistics}
    B_\theta (a, b) = \frac{ \theta(a \times b)}{ \theta(a) \theta(b)}.
\end{eqs}
In Appendix~\ref{appendix: braiding_statistics}, we derive Eq.~\eqref{eq:braiding statistics} with the definition $W^{a\times b}_{i} := W^{b}_{i} W^{a}_{i}$ and the fact that $W^{a, b}_{{1,2,3}}$ are Pauli operators that anti-commute by a $U(1)$ phase.
Therefore, the complete data for Abelian anyon theories are anyon types, fusion rules, and topological spins.

\begin{figure*}[t]
    \centering
    \includegraphics[width=0.8\textwidth]{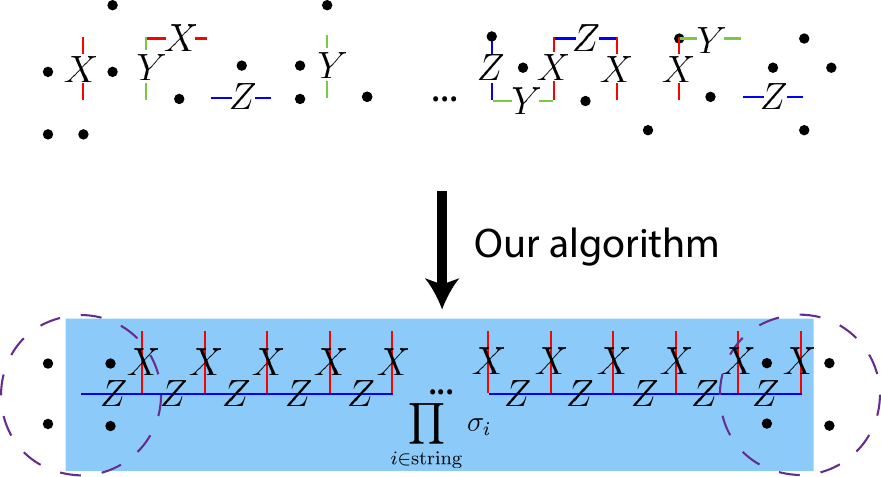}
    \caption{\textcolor{black}{Demonstration of the procedure for finding string operators (details in Sec.~\ref{sec: Solving Anyon Equations Zp}). For each single-qudit Pauli operator $\mathcal{P}$, it will locally violate stabilizer terms at different positions denoted as black dots. Note for the general case, a single-qudit Pauli operator $X/Z$ may violate stabilizer terms in multiple (can be greater than 2) positions in contrast to the toric code where a single-qudit Pauli $X/Z$ error only violates two stabilizer terms. For example, in the $\mathbb{Z}_2$ cases, black dots are the stabilizer terms giving the -1 signs. If we combine the two single-qubit error where some of the syndromes (black dots) are overlapped, the overlapped syndrome will be eliminated. 
    The goal of this procedure is to combine single-qudit Pauli operators at different positions to form a string operator such that only the stabilizers near the endpoints are violated.}}
    \label{fig: finding string operator}
\end{figure*}

In summary, the algorithm for extracting the topological order from a topological Pauli stabilizer code includes the following steps:
\begin{enumerate}
    \item Enumerate string operators that create anyons around its endpoints.
    \item Classify these anyons into different types by the equivalence relation \eqref{eq: definition of anyon equivalence}.
    \item Obtain the fusion rules of anyon types.
    \item Compute the topological spin (and braiding statistics) for each anyon type from its string operator.
\end{enumerate}
The primary aim is to identify string operators for distinct anyon types. This task is inherently challenging, but locality and translation invariance can simplify this problem. We focus on a finite region, designated as $A$, where, due to the finite amount of anyon types in two dimensions and the assumption of translation invariance, all topological data are contained inside $A$ (if $A$ is large enough). Anyon string operators, fusion rules, topological spins, and braiding statistics can be analyzed in this region $A$.
{\color{black}
As illustrated in Refs.~\cite{haah_module_13, watanabe2023ground}, the required size $l_A$ of the region $A$ should be upper bounded by $d \cdot 2^r$, where $d$ is the qudit dimension and $r$ is the range of stabilizers. Our method successfully retrieves all topological information for typical stabilizer codes with reasonable qudit dimensions and ranges of stabilizers.\footnote{\color{black} Our algorithm has the time complexity $O(\log d \cdot l_A^4 ) = O(\log d \cdot d^4  \cdot 2^{4r}$, which is from the modified Gaussian elimination process discussed in Secs.~\ref{sec: nonprime-dimensional qudit} and~\ref{sec: time complexity}. For the exotic $\ZZ_p$ toric code example constructed in Ref.~\cite{watanabe2023ground}, especially with a large prime $p$, such as $p=1009$, the computational power required surpasses what a personal computer can provide.} 
}
We obtain the syndrome pattern for each single Pauli error and combine them to form a string-like operator, shown in Fig.~\ref{fig: finding string operator}. The naive approach of listing all combinations of Pauli matrices will take an exponential time. The approach described later in this paper can achieve the discovery of string operators for all anyon types within polynomial time.

Note that our way of classifying topological Pauli stabilizer codes based on topological data is different from Refs.~\cite{bombin_Stabilizer_14, haah_classification_21}, which utilize Clifford circuits to transform a topological Pauli stabilizer code into finite copies of the standard toric code stabilizers and trivial stabilizers. The equivalence of Pauli stabilizer codes by the Clifford circuit transformation is a stronger condition than the equivalence of their topological data. For example, consider $\ZZ_4$ Pauli stabilizer codes:
\begin{eqs}
    H_1 := -\sum_{i=1}^N Z_i, \quad H_2 := - \sum_{i=1}^N (X_i^2 + Z_i^2).
\end{eqs}
Both codes have trivial topological data since each one has the ground state as a product state $\ket{0}^{\otimes N}$ or $(\frac{\ket{0}+\ket{2}}{\sqrt{2}})^{\otimes N}$, but there is no Clifford circuit that transforms from $H_1$ to $H_2$. Therefore, instead of using the equivalence of parent Hamiltonians under Clifford circuit transformations to classify Pauli stabilizer codes, our classification is based on topological data, which are properties of its low-energy excitations.

\section{Formulating translation invariant Pauli stabilizer models}\label{sec: laurent_polynomial}

In this section, we first review the polynomial method to formulate translation invariant Pauli stabilizer models and demonstrate how to utilize it to derive the topological data in Sec.~\ref{sec: review_polynomial}. In Sec.~\ref{sec: laurent_polynomial_example}, we provide the polynomial representations of various known models.

\subsection{Review of the Laurent Polynomial method on a square lattice}\label{sec: review_polynomial}
\begin{figure}[htb]
\centering
\resizebox{5cm}{!}{%
\begin{tikzpicture}
\draw[thick] (-3,0) -- (3,0);\draw[thick] (-3,-2) -- (3,-2);\draw[thick] (-3,2) -- (3,2);
\draw[thick] (0,-3) -- (0,3);\draw[thick] (-2,-3) -- (-2,3);\draw[thick] (2,-3) -- (2,3);
\draw[->] [thick](0,0) -- (1,0);\draw[->][thick] (0,2) -- (1,2);\draw[->][thick] (0,-2) -- (1,-2);
\draw[->][thick] (0,0) -- (0,1);\draw[->][thick] (2,0) -- (2,1);\draw[->][thick](-2,0) -- (-2,1);
\draw[->][thick] (-2,0) -- (-1,0);\draw[->][thick] (-2,2) -- (-1,2);\draw[->][thick](-2,-2) -- (-1,-2);
\draw[->][thick] (-2,-2) -- (-2,-1);\draw[->][thick] (0,-2) -- (0,-1);\draw[->] [thick](2,-2) -- (2,-1);
\filldraw [black] (-2,-2) circle (1.5pt) node[anchor=north east] {\large 1};
\filldraw [black] (0,-2) circle (1.5pt) node[anchor=north east] {\large 2};
\filldraw [black] (2,-2) circle (1.5pt) node[anchor=north east] {\large 3};
\filldraw [black] (-2,-0) circle (1.5pt) node[anchor=north east] {\large 4};
\filldraw [black] (0,0) circle (1.5pt) node[anchor=north east] {\large 5};
\filldraw [black] (2,0) circle (1.5pt) node[anchor=north east] {\large 6};
\filldraw [black] (-2,2) circle (1.5pt) node[anchor=north east] {\large 7};
\filldraw [black] (0,2) circle (1.5pt) node[anchor=north east] {\large 8};
\filldraw [black] (2,2) circle (1.5pt) node[anchor=north east] {\large 9};
\draw (-1,-1) node{\Large a};
\draw (1,-1) node{\Large b};
\draw (1,1) node{\Large c};
\draw (-1,1) node{\Large d};
\end{tikzpicture}
}
\caption{We put Pauli matrices $X_e$, $Y_e$, and $Z_e$ on each edge.}
\label{fig:square}
\end{figure}
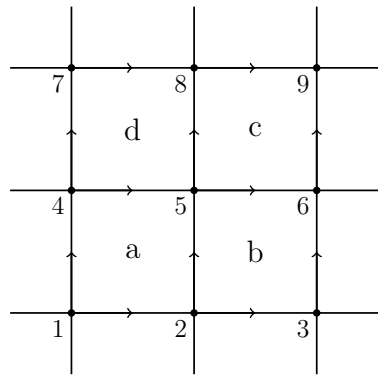

The Laurent polynomial has been demonstrated to be useful for searching 3d fracton phases \cite{dua2019sorting} and error-correcting codes \cite{chen2022error,chien2022optimizing}, which
serves as a basis for our algorithm. In this subsection, we will review the Laurent polynomial and how to use it to represent translation invariant stabilizer codes.

This section examines a scenario involving two $\ZZ_d$ qudits in each unit cell, exemplified by a qudit at each edge of a square lattice. This setting is extendable to cases with \textcolor{black}{$w$} qudits per unit cell.

Our initial step is to demonstrate that any Pauli operator, constituted by a finite tensor product of Pauli matrices at different sites, can be represented (up to an overall factor) as a column vector over the polynomial ring $R = \ZZ_d [x, y, x^{-1}, y^{-1}]$\footnote{This ring includes all polynomials in $x$, $x^{-1}$, $y$, $y^{-1}$, with coefficients in $\ZZ_d$.}, as established in Ref.~\cite{haah_module_13}.
We assign column vectors over $\ZZ_d$ to the (generalized) Pauli matrices $X_{12}$, $Z_{12}$, $X_{14}$, and $Z_{14}$, depicted in Fig.~\ref{fig:square}:
\begin{eqs}
    \mX_{12}=
    \left[\begin{array}{c}
        1 \\
        0 \\
        \hline
        0 \\
        0
    \end{array}\right],~
    \mZ_{12}=
    \left[\begin{array}{c}
        0 \\
        0 \\
        \hline
        1 \\
        0
    \end{array}\right],~
    \mX_{14}=
    \left[\begin{array}{c}
        0 \\
        1 \\
        \hline
        0 \\
        0
    \end{array}\right],~
    \mZ_{14}=
    \left[\begin{array}{c}
        0 \\
        0 \\
        \hline
        0 \\
        1
    \end{array}\right].
\end{eqs}
{\color{black}
In this paper, column vector representations of Pauli operators are denoted using curly letters.
}
The coefficients in these vectors correspond to their powers:
\begin{eqs}
    \mathcal{P} =
    \left[\begin{array}{c}
        i \\
        j \\
        \hline
        k \\
        l
    \end{array}\right]
    ~\Rightarrow~
    \mathcal{P}^m =
    \left[\begin{array}{c}
        m i \\
        m j \\
        \hline
        m k \\
        m l
    \end{array}\right],
    \forall m \in \ZZ_d.
\end{eqs}
The translation of operators is achieved using polynomials of $x$ and $y$ to denote translations in the $x$ and $y$ directions, respectively. To illustrate, translating the operator on edge $e_{12}$ to edge $e_{78}$ (using the vector $(0, 2)$) or to edge $e_{58}$ (using the vector $(1, 1)$) involves multiplying the column vector of the operator by $y^2$ or $xy$, respectively:
\begin{eqs}
    \mZ_{78}= y^2
    \mZ_{12}
    =
    \left[\begin{array}{c}
        0 \\
        0 \\
        \hline
        y^2 \\
        0
    \end{array}\right],~
    \mX_{58}= x y
    \mX_{14}
    =
    \left[\begin{array}{c}
        0 \\
        xy \\
        \hline
        0 \\
        0
    \end{array}\right].
\end{eqs}
In conclusion, any Pauli operator can be decomposed as follows:
\begin{eqs}
    P = \eta X^{a_1}_{e_1} X^{a_2}_{e_2} \cdots X^{a_n}_{e_n} Z^{b_1}_{e'_1} Z^{b_2}_{e'_2} \cdots Z^{b_m}_{e'_m},
\end{eqs}
where $\eta$ represents a root of unity of order $2d$. After dropping the overall phase $\eta$, the corresponding column vector for this operator is a linear combination of individual Pauli matrices, expressed as
\begin{eqs}
    \mathcal{P} = & a_1 \mX_{e_1} + a_2 \mX_{e_2} + \cdots + a_n \mX_{e_n} \\
    & + b_1 \mZ_{e'_1} + b_2 \mZ_{e'_2} + \cdots + b_m \mZ_{e'_m}.
\end{eqs}
More examples are included in Fig.~\ref{fig:example_poly}.

\begin{figure}[htb]
    \centering
    \includegraphics[width=0.45\textwidth]{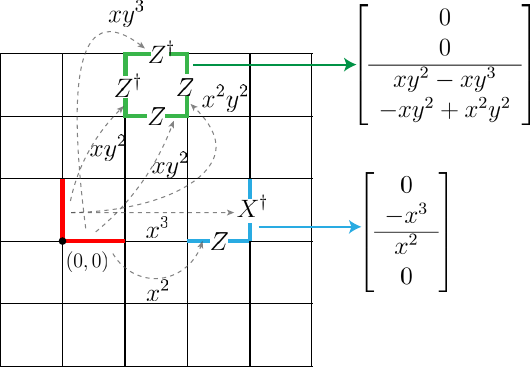}
    \caption{Examples of polynomial expressions for Pauli strings. The flux term on a plaquette and the $XZ$ term on edges are shown. The factors such as $x^2 y^2$ and $x^2$ represent the locations of the operators relative to the origin.}
\label{fig:example_poly}
\end{figure}

Next, we introduce the {\bf antipode map} that is a $\ZZ_d$-linear map from $R$ to $R$ defined by
\begin{eqs}
    x^a y^b \rightarrow \overline{x^a y^b}:=x^{-a} y^{-b}.
\end{eqs}
To determine whether two Pauli operators represented by vectors $v_1$ and $v_2$ commute or anti-commute, we define the dot product as
\begin{eqs}
    v_1 \cdot v_2 = \overline{v}_1^{T} \Lambda v_2,
\end{eqs}
where $T$ is the transpose operation on a matrix and
\begin{eqs}
    \Lambda=
    \left[\begin{array}{cc | cc}
        0 & 0 & 1 & 0 \\
        0 & 0 & 0 & 1 \\
        \hline
        -1 & 0 & 0 & 0 \\
        0 & -1 & 0 & 0 \\
    \end{array}\right]
\end{eqs}
is the matrix representation of the standard {\bf symplectic bilinear form}. For simplicity, we denote $\overline{(\cdots)}^T$ as $(\cdots)^\dagger$.

The two operators $v_1$ and $v_2$ commute if and only if the constant term of $v_1 \cdot v_2$ is zero. For example, we calculate the dot products
\begin{eqs}
    \mX_{12} \cdot \mZ_{12} = 1, \quad \mX_{58} \cdot \mZ_{14} = x^{-1}y^{-1},
\end{eqs}
and, therefore, $X_{12}$ and $Z_{12}$ anti-commute, whereas $X_{58}$ and $Z_{14}$ commute (their dot product only has a non-constant term $x^{-1} y^{-1}$). Furthermore, the physical interpretation of $\mX_{58} \cdot \mZ_{14} = x^{-1} y^{-1}$ is that shifting of $X_{58}$ in $-x$ and $-y$ directions by 1 step will anti-commute (by the factor of $\omega$) with $Z_{14}$. 

A translation invariant stabilizer code forms an $R$-submodule\footnote{The $R$-submodule is similar to a subspace of a vector space, but the entries of the vector are in the ring $R=\ZZ_d [x, y, x^{-1}, y^{-1}]$. In a ring, the inverse element may not exist. This is the distinction between a module and a vector space.} $\sigma$ such that
\begin{eqs}
    v_1 \cdot v_2 = v_1^\dagger \Lambda v_2 = 0, \quad \forall v_1, v_2 \in \sigma,
\label{eq: stabilizer condition}
\end{eqs}
i.e., a module of commuting Pauil operators. This $\sigma$ is named the {\bf stabilizer module}.
The Hamiltonian could have two (or more) terms per square to have a unique ground state on a simply-connected manifold, denoted as $H = -\sum_{\text{cells}} (S_1 + S_2)$ with corresponding column vectors
\begin{eqs}
    \mS_1 = \left[\begin{array}{c}
        f_1(x,y) \\
        f_2 (x,y) \\
        \hline
        g_1(x,y) \\
        g_2 (x,y)
    \end{array}\right], \quad
    \mS_2 = \left[\begin{array}{c}
        h_1(x,y) \\
        h_2 (x,y) \\
        \hline
        k_1(x,y) \\
        k_2 (x,y)
    \end{array}\right].
\label{eq: generic S1 and S2}
\end{eqs}
{\color{black} $S_1$ and $S_2$ constitute the \textbf{generators} of the stabilizer module, and will henceforth be referred to as \textbf{stabilizer generators}.}
For example, the trivial phase $H_0 = - \sum_e X_e$ is
\begin{eqs}
    \mS_1 = \left[\begin{array}{c}
        1 \\
        0 \\
        \hline
        0 \\
        0
    \end{array}\right], \quad
    \mS_2 = \left[\begin{array}{c}
        0 \\
        1 \\
        \hline
        0 \\
        0
    \end{array}\right],
\label{eq: trivial H0 SA SB}
\end{eqs}
and the standard $\ZZ_d$ toric code Hamiltonian 
\begin{eqs}
    H_{\text{TC}} = - \sum_v \vcenter{\hbox{\includegraphics[scale=.25]{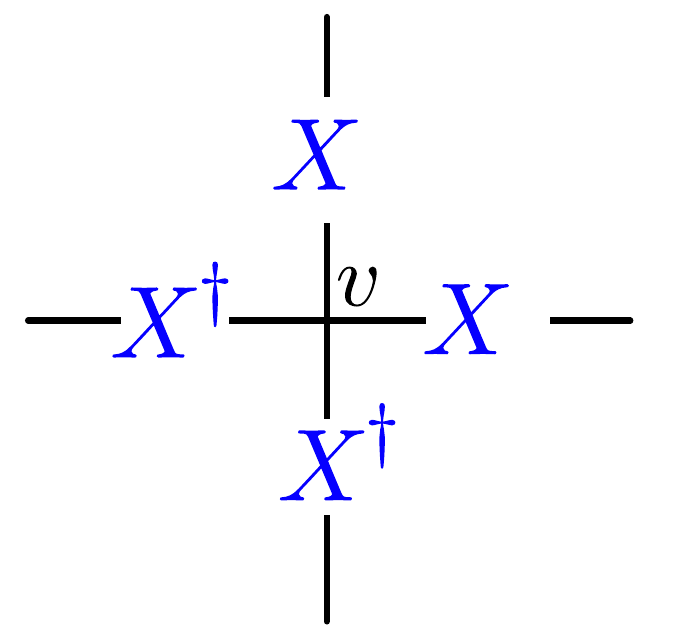}}} - \sum_p \vcenter{\hbox{\includegraphics[scale=.25]{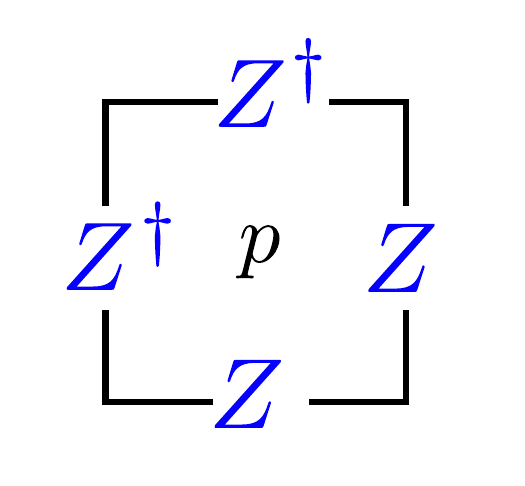}}},
\label{eq: toric code Hamiltonian}
\end{eqs}
corresponds to
\begin{eqs}
    \mS_1 = \left[\begin{array}{c}
        1-\bx \\
        1- \by \\
        \hline
        0 \\
        0
    \end{array}\right], \quad
    \mS_2 = \left[\begin{array}{c}
        0 \\
        0 \\
        \hline
        1-y \\
        -1+x
    \end{array}\right].
\label{eq: standard toric code SA SB}
\end{eqs}
For general stabilizers $\mathcal{S}_1$ and $\mathcal{S}_2$ in Eq.~\eqref{eq: generic S1 and S2}, the commutation condition in Eq.~\eqref{eq: stabilizer condition} implies
\begin{eqs}
    \begin{cases}
        -\overline{g}_1 f_1  -\overline{g}_2 f_2 + \overline{f}_1 g_1 + \overline{f}_2 g_2 &=0,\\
        -\overline{k}_1 h_1  -\overline{k}_2 h_2 + \overline{h}_1 k_1 + \overline{h}_2 k_2 &=0,\\
        -\overline{k}_1 f_1  -\overline{k}_2 f_2 + \overline{h}_1 g_1 + \overline{h}_2 g_2 &=0.\\
    \end{cases}
\end{eqs}
One simple solution is $g_1 = g_2 = h_1 = h_2 = 0$ and $f_1 = -\overline{k}_2$ and $f_2 = \overline{k}_1$:
\begin{eqs}
    \mS_1 = \left[
    \def\arraystretch{1.2}
    \begin{array}{c}
        f_1 \\
        f_2 \\
        \hline
         0 \\
        0
    \end{array}\right], \quad
    \mS_2 = \left[
    \def\arraystretch{1.2}
    \begin{array}{c}
        0 \\
        0 \\
        \hline
        \overline{f}_2 \\
        -\overline{f}_1
    \end{array}\right],
\end{eqs}
where $\mS_1$ only has $X$ part and $\mS_2$ only has $Z$ part, representing translation invariant CSS codes on a square lattice.

Next, we define the {\bf excitation map} for any Pauli operator $\mP$ (a column vector over $R$)
{\color{black}on a general Pauli stabilizer code with stabilizer generators $\mS_1,... ,\mS_t$ as
\begin{eqs}
    \eps (\mP):=[  \mS_1 \cdot \mP ,   \mS_2 \cdot \mP,..., \mS_t \cdot \mP],
\end{eqs}
which indicates how the Pauli operator violates stabilizers $\mS_1,...,,\mS_t$.}
The error syndromes of a single Pauli matrix for the generic Hamiltonian \eqref{eq: generic S1 and S2} \textcolor{black}{with stabilizers $\mS_1$ and $\mS_2$} can be written as
\begin{eqs}
    \eps(\mX_1) &= [  \mS_1 \cdot \mX_{12} ,   \mS_2 \cdot \mX_{12}] = [ -\overline{g}_1, -\overline{k}_1 ], \\
    \eps(\mX_2) &= [  \mS_1 \cdot \mX_{14},  \mS_2 \cdot \mX_{14}] = [ -\overline{g}_2, -\overline{k}_2 ], \\
    \eps(\mZ_1) &= [  \mS_1 \cdot \mZ_{12}, \mS_2 \cdot \mZ_{12}] = [ \overline{f}_1, \overline{h}_1 ], \\
    \eps(\mZ_2) &= [  \mS_1 \cdot \mZ_{14},  \mS_2 \cdot \mZ_{14}] = [ \overline{f}_2, \overline{h}_2 ].
    \label{eq: get error syndromes}
\end{eqs}
Each syndrome is a row vector with two entries in $\ZZ_d[x, y, x^{-1}, y^{-1}]$.
{\color{black}
For a general case where the stabilizer group is generated by $t$ stabilizers, the syndrome would be a row vector with $t$ entries in $\ZZ_d[x, y, x^{-1}, y^{-1}]$.
}

Given this excitation map $\eps$, the {\bf topological order (TO) condition} can be formulated as
\begin{eqs}
    \ker \eps = \sigma,
\label{eq: topological condition}
\end{eqs}
where $\sigma$ is the stabilizer module and the kernel of $\eps$ represents the space of local Pauli operators commuting with all stabilizers. This TO condition requires that if a local operator commutes with stabilizers, this operator must be a product of stabilizers (a column vector in the stabilizer module). If the TO condition is satisfied, the ground state space is indistinguishable by any local operator. The ground state degeneracy on different manifolds is a part of the {\bf topological order} of Pauli stabilizer codes. More precisely, the ground state degeneracy on a torus equals the number of anyon types, with each ground state labeled by an anyon string operator.

To obtain the possible anyons in this theory, for each $n \geq 1$, we solve the {\bf anyon equation}
\begin{eqs}
    &\eps \big(
    \alpha(x,y)  \mX_1 + \beta(x,y) \mX_2
    + \gamma(x,y) \mZ_1 + \delta(x,y) \mZ_2
    \big) \\
    &= (1-x^n) [\cdots, \cdots] := (1-x^n) v,
\label{eq: anyon equation}
\end{eqs}
where $v$ is a length-2 row vector, referred to as an anyon. This anyon equation is previously sketched in Fig.~\ref{fig: finding string operator}, and here is the precise mathematical definition in the Laurent polynomial formalism. The physical interpretation of this equation is that when we apply Pauli matrices $X_1$, $X_2$, $Z_1$, and $Z_2$ at locations $\alpha(x,y)$, $\beta(x,y)$, $\gamma(x,y)$, and $\delta(x,y)$, it violates the stabilizers around the origin $(0,0)$ and the point $(n, 0)$ with patterns $v$ and $-v$, respectively. This operator creates an anyon $v$ at $(0,0)$ and its antiparticle at $(n, 0)$.
Note that if $v$ is an anyon, $x^a y^b v$ is also an anyon for all $a,b \in \ZZ$. In addition, an anyon solved by a particular $n$ can be solved from any multiple of $n$. The anyon equation \eqref{eq: anyon equation} defines the string operator that moves an anyon in the $x$-direction.
{\color{black} In principle, the anyon equations should also be formulated in the $y$-direction. However, as noted in Refs.~\cite{haah_module_13, bombin_Stabilizer_14}, in two dimensions, an anyon inherently possesses both $x$-mover and $y$-mover operators, allowing it to move in both the $x$- and $y$-directions. Consequently, addressing the anyon equations solely in the $x$-direction suffices to determine all possible anyonic excitations.\footnote{This assertion does not hold in three dimensions, where fracton phases are present \cite{haah2011local, Vijay2015newTO, vijay2016fracton}. Here, particle excitations may exhibit restricted mobility.}
}

To check whether two anyons are of the same type, we rely on the following equivalence relation:
\begin{eqs}
    v'\sim v ~ (\text{$v'$ is equivalent to $v$}),
\end{eqs}
if and only if there exist finite-degree polynomials $p_1(x,y)$, $p_2(x,y)$, $p_3(x,y)$, $p_4(x,y)$ such that
\begin{eqs}
    v' = & ~v + p_1(x,y) \eps(\mX_1) + p_2(x,y) \eps(\mX_2) \\
    & + p_3(x,y) \eps(\mZ_1) + p_4(x,y) \eps(\mZ_2).
\label{eq: equivalence relation of anyons}
\end{eqs}
Physically, this can be understood as two anyons are of the same type if and only if they \textcolor{blue}{differ} by some local Pauli operators $X_1$, $X_2$, $Z_1$, $Z_2$ at locations specified by polynomials $p_1(x,y)$, $p_2(x,y)$, $p_3(x,y)$, $p_4(x,y)$. 
Note that even if anyon $v$ implies the existence of another anyon $x^a y^b v$, they do not need to be the same type. The Wen plaquette model is an example \cite{wen2003quantum}.

\subsection{Examples}\label{sec: laurent_polynomial_example}

\subsubsection{Trivial model $H_0 = - \sum_e X_e$}

We first consider the trivial Hamiltonian Eq.~\eqref{eq: trivial H0 SA SB}. We can easily compute
\begin{eqs}
    &\eps(\mX_1) = [0, 0], ~\eps(\mX_2) = [0, 0], \\
    &\eps(\mZ_1) = [1, 0], ~\eps(\mZ_2) = [0, 1].
\end{eqs}
Since linear combinations of $\{ \eps(\mX_1), \eps(\mX_2), \eps(\mZ_1), \eps(\mZ_2)\}$ already generate all possible polynomials in each entry, all anyons are equivalent to the trivial anyon $v_0 = [0, 0]$. There is only one trivial anyon in this model.

\subsubsection{Standard $\ZZ_d$ toric code}
We analyze the standard $\ZZ_d$ toric code Eq.~\eqref{eq: standard toric code SA SB}. \textcolor{black}{The syndromes of single-qubit Pauli errors} are
\begin{eqs}
    \eps(\mX_1) &= [0, -1+\by], ~\eps(\mX_2) = [0, 1-\bx], \\
    \eps(\mZ_1) &= [1-x, 0], ~\eps(\mZ_2) = [1-y, 0].
\end{eqs}
From the anyon equation \eqref{eq: anyon equation}, $\eps(\mX_2)$ can generate a solution $v_m = [0, 1]$ and general solutions $[0, x^a y^b] \quad \forall a, b \in \ZZ$ (by choosing $n=1$ and $\beta(x,y) = - x^{a+1} y^b$). Also, $\eps(\mZ_1)$ can generate a solution $v_e = [1, 0]$ and general solutions $[x^c y^d, 0] \quad \forall c, d \in \ZZ$. Therefore, a generic anyon can be expressed as $[x^{a_1} y^{b_1}+x^{a_2} y^{b_2}+\cdots,~ x^{c_1} y^{d_1}+x^{c_2} y^{d_2}+\cdots]$.

Next, we want to check whether these anyons are equivalent or not. From the identity
\begin{eqs}
    (1-x)(1+x+x^2+\cdots + x^{k-1}) = 1 -x^k,
\end{eqs}
we have $[ x^k p(y),\cdots] \sim [ p(y),\cdots]$ since they \textcolor{blue}{differ} by the polynomial factor $(1+x+x^2+\cdots + x^{k-1})p(y)$. Similarly, we have $\eps(\mZ_2)$ for the $y$ part. More precisely, we summarize the equivalence relation
\begin{eqs}
    [0, x^a y^b] &\sim [0, 1], \quad \forall a, b \in \ZZ,\\
    [x^c y^d, 0] &\sim [1, 0], \quad \forall c, d \in \ZZ.
\end{eqs}
Giving an anyon, adding polynomial factors of $\eps(\mX_1), \eps(\mX_2), \eps(\mZ_1), \eps(\mZ_2)$ can not change the parity of the sum of coefficients\footnote{The sum of coefficients can be computed by substituting $x \rightarrow 1$ and $y \rightarrow 1$.} in each entry. We conclude that the following four anyons are inequivalent:
\begin{eqs}
    v_0 = [0, 0], ~v_e= [1, 0], ~v_m= [0, 1], ~v_f= [1, 1].
\end{eqs}
Note that we can choose $v_e$ and $v_m$ as the {\bf basis anyons}, which generate the entire set of anyons. This choice is not canonical since another choice $\{ v_e, v_f\}$ is also valid.

\subsubsection{Two copies of $\ZZ_d$ toric codes}
Consider a new Hamiltonian modified from the toric code:
\begin{eqs}
    H_{\text{TC}} = &- \sum_v \vcenter{\hbox{\includegraphics[scale=.25]{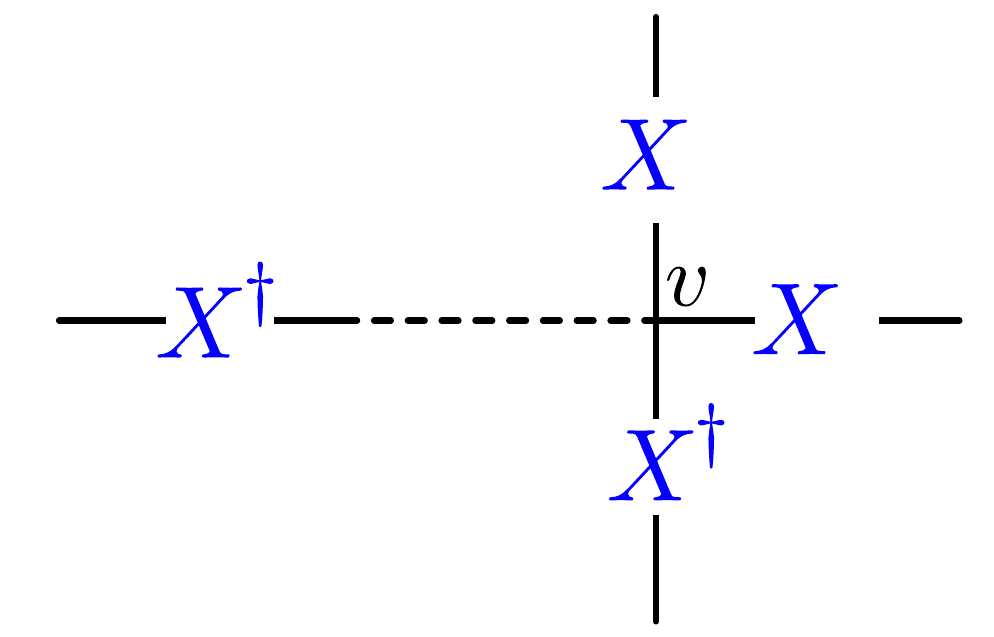}}}\\
    &- \sum_p \vcenter{\hbox{\includegraphics[scale=.25]{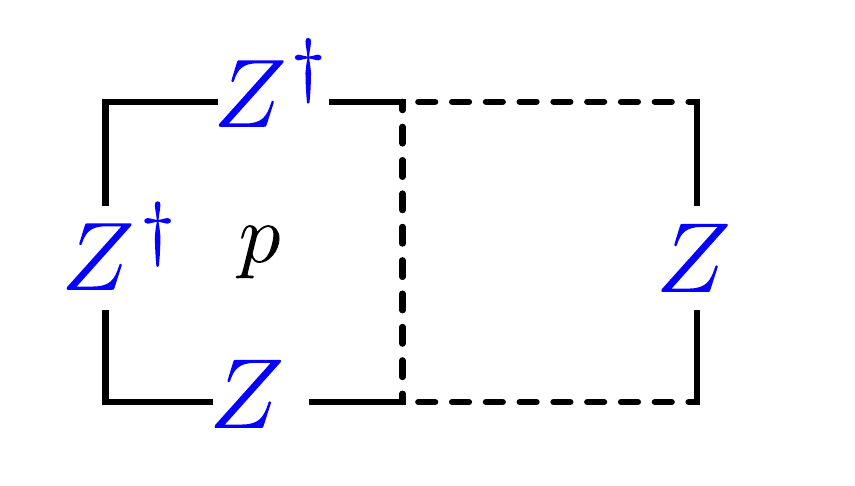}}},
\label{eq: 2 toric code Hamiltonian}
\end{eqs}
which stabilizers differ from Eq.~\eqref{eq: standard toric code SA SB}:
\begin{eqs}
    \mS_1 = \left[\begin{array}{c}
        1-\bx^2 \\
        1- \by \\
        \hline
        0 \\
        0
    \end{array}\right], \quad
    \mS_2 = \left[\begin{array}{c}
        0 \\
        0 \\
        \hline
        1-y \\
        -1+x^2
    \end{array}\right].
\end{eqs}
The square lattice can be partitioned so the Hamiltonian is exactly two decoupled $\ZZ_d$ toric codes. It has syndromes:
\begin{eqs}
    \eps(\mX_1) &= [0, -1+\by], ~\eps(\mX_2) = [0, 1-\bx^2],\\
    \eps(\mZ_1) &= [1-x^2, 0], ~\eps(\mZ_2) = [1-y, 0].
\end{eqs}
We focus on the $Z$ part first and try to solve 
\begin{eqs}
    \gamma(x,y)  ~ \eps(\mZ_1) + \delta(x,y) ~ \eps(\mZ_2)  = (1-x^n) v,
\end{eqs}
For $n=1$, we have
\begin{eqs}
    x^a y^b \eps(\mZ_1) = x^a y^b [1-x^2, 0] = (1-x) [x^a y^b + x^{a+1} y^b, 0].
\end{eqs}
Therefore, we have anyons $v^{n=1}_{a,b} := [x^a y^b + x^{a+1} y^b, 0]$ for arbitrary $a,b \in \ZZ$.
For $n=2$, we have
\begin{eqs}
    x^c y^d \eps(\mZ_1) = x^c y^d [1-x^2, 0] = (1-x^2) [x^c y^d , 0],
\end{eqs}
which gives anyons $v^{n=2}_{c,d} := [x^c y^d, 0]$ for arbitrary $c,d \in \ZZ$. Notice that all $v^{n=1}_{a,b}$ can be generated by $v^{n=2}_{c,d}$ since
\begin{eqs}
    v^{n=1}_{a,b} = v^{n=2}_{a,b} + v^{n=2}_{a+1,b}.
\end{eqs}
In fact, $v^{n=2}_{c,d}$ has already generated all possible polynomials in the first entry. We do not need to look for $n \geq 3$ cases.

Next, we use $\eps(\mZ_1)$ and $\eps(\mZ_2)$ to find the equivalence classes of those anyons. By some algebraic argument, there are four inequivalent classes:
\begin{eqs}
    1 &:= [0, 0], \\
    e_1 &:= [1, 0], \\
    e_2 &:= [x, 0], \\
    e_1e_2 &:= [1+x, 0]. \\
\end{eqs}
Similarly, for the $X$ part, we can solve
\begin{eqs}
    1 &:= [0, 0], \\
    m_1 &:= [0, 1], \\
    m_2 &:= [0, x], \\
    m_1m_2 &:= [0, 1+x]. \\
\end{eqs}
We can choose $\{ e_1, e_2, m_1, m_2\}$ as the basis anyons for the entire set of anyons.

\subsubsection{Double semion}

The double semion model \cite{levin2012braiding} is an example of topological order with a nontrivial ``twist,'' labeled by a 3-cocycle in the third cohomology $H^3(\ZZ_2, U(1))$ \cite{chen2012symmetry, dijkgraaf1990topological}. This double semion can be realized as a Pauli stabilizer code with $\ZZ_4$ qudits from condensing $e^2 m^2$ anyon in the standard $\ZZ_4$ toric code \cite{ellison2022pauli}. The stabilizers of this double semion code are
\begin{equation} 
    \begin{gathered}
    \mathcal{S}_1 = \vcenter{\hbox{\includegraphics[scale=.25]{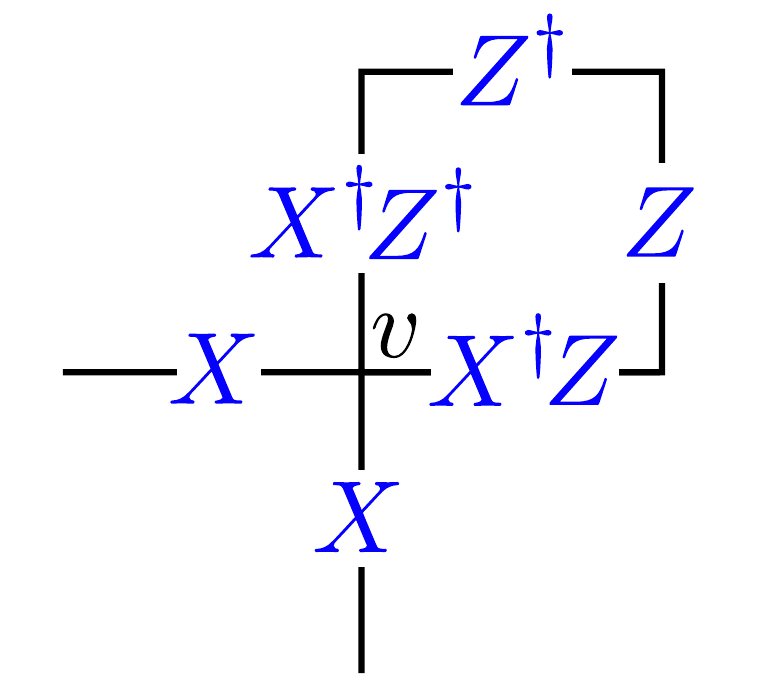}}}, \quad
    \mathcal{S}_2 = \vcenter{\hbox{\includegraphics[scale=.25]{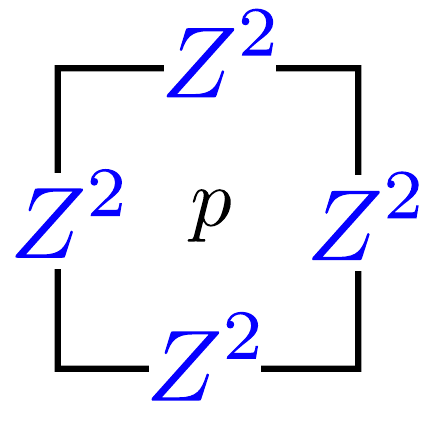}}}, \\
    \mathcal{S}_3 = \vcenter{\hbox{\includegraphics[scale=.25]{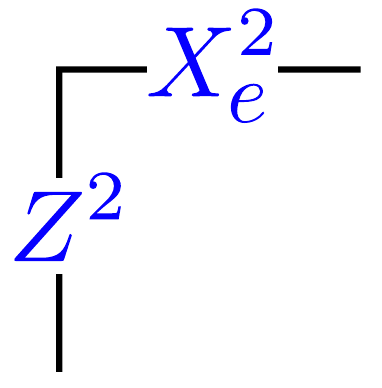}}}, \quad
    \mathcal{S}_4 = \vcenter{\hbox{\includegraphics[scale=.25]{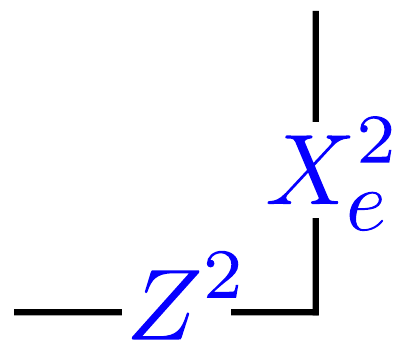}}}.
    \end{gathered}
\label{eq: DS terms}
\end{equation}
The stabilizers can be written as column vectors of $\ZZ_d[x, y ,x^{-1}, y^{-1}]$:
\begin{eqs}\label{eq:DS_stb_poly}
    \mathcal{S}_1&= \left[\begin{array}{c}
        \bx-1 \\
        \by-1\\
        \hline
        1-y \\
        -1+ x
    \end{array}\right], \quad
    \mathcal{S}_2 = \left[\begin{array}{c}
        0 \\
        0\\
        \hline
        2+2y \\
        2+2x 
    \end{array}\right], \\
    \mathcal{S}_3 &= \left[\begin{array}{c}
        2 \\
        0 \\
        \hline
        0 \\
        2\by
    \end{array}\right], \quad
    \mathcal{S}_4 = \left[\begin{array}{c}
        0 \\
        2\\
        \hline
        2\bx \\
        0 
    \end{array}\right].
\end{eqs}
It has error syndromes:
\begin{eqs}
    \eps(\mX_1) &= [\by-1,-2\by-2,0,-2x], ~\eps(\mZ_1) = [-1+x,0,2,0], \\
    \eps(\mX_2) &= [1-\bx,-2\bx-2,-2y,0], ~\eps(\mZ_2) = [-1+y,0,0,2].
\end{eqs}
We can find: 
\begin{eqs}\label{eq:DS_x_string}
    \eps(\mX_2)+y\eps(\mZ_1)&=(1-x)[-y-\bx,-2\bx,0,0],\\
    2\eps(\mZ_1)&=(1-x)[-2,0,0,0].
\end{eqs}
Therefore, we can get string operators in the $x$-direction and their corresponding anyons:
\begin{eqs}
    v_s&=[-y-\bx,-2\bx,0,0],\\
    v_b&=[-2,0,0,0].
\end{eqs}

With the anyons $v_s$ and $v_b$, we can find that they also satisfy:
\begin{eqs}\label{eq:DS_y_string}
    \eps(\mX_1)-x\eps(\mZ_2)&=(1-y)[x+\by,-2\by,0,0]=(1-y)(-\bx y)v_s,\\
    2\eps(\mZ_2)&=(1-y)[-2,0,0,0]=(1-y)v_b.
\end{eqs}
Therefore, we can get the string operators in the $y$-direction. And then, we can draw the string operators for $s$ and $b$ as
\begin{eqs}\label{eq:DS_string}
    W^s_e& :=
    \begin{gathered}
        \vcenter{\hbox{\includegraphics[scale=0.25]{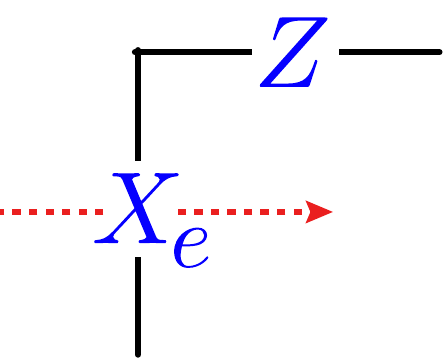}}},
        \vcenter{\hbox{\includegraphics[scale=0.25]{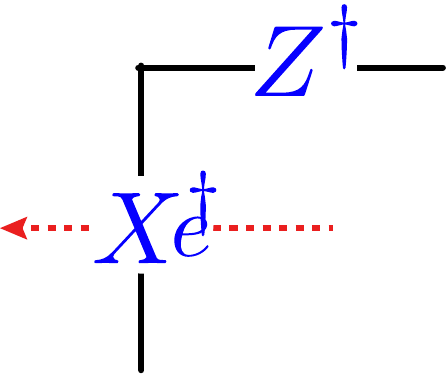}}},
        \vcenter{\hbox{\includegraphics[scale=0.25]{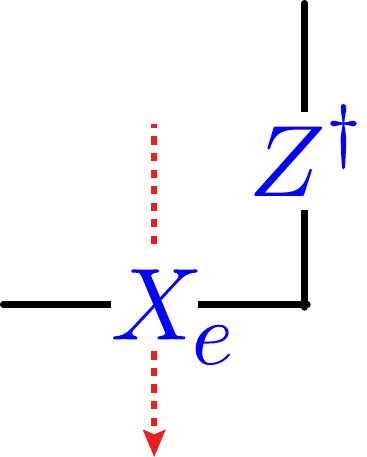}}},
        \vcenter{\hbox{\includegraphics[scale=0.25]{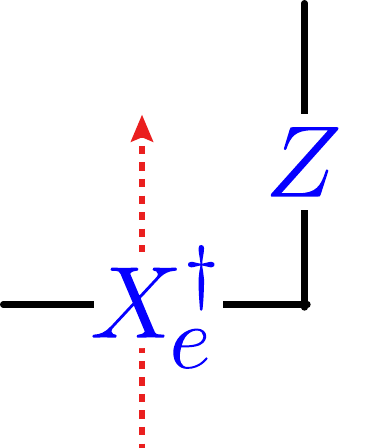}}},
    \end{gathered}\\
    W^b_e&:=
    \begin{gathered}
        \vcenter{\hbox{\includegraphics[scale=0.25]{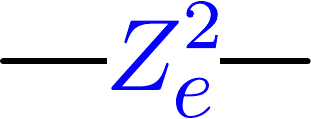}}}, \quad 
        \vcenter{\hbox{\includegraphics[scale=0.25]{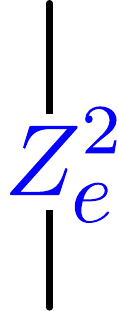}}}.
    \end{gathered}
\end{eqs}

\subsection{Shifted double semion}
Here, we also present a variant of the double semion code, named as the shifted double semion code, and how to obtain the string operators of it.
The stabilizers of the shifted double semion code are
\begin{eqs}\label{eq:shifted_DS_stabilizer}
    \mathcal{S}_1 & :=\begin{gathered}\includegraphics[scale=0.25]{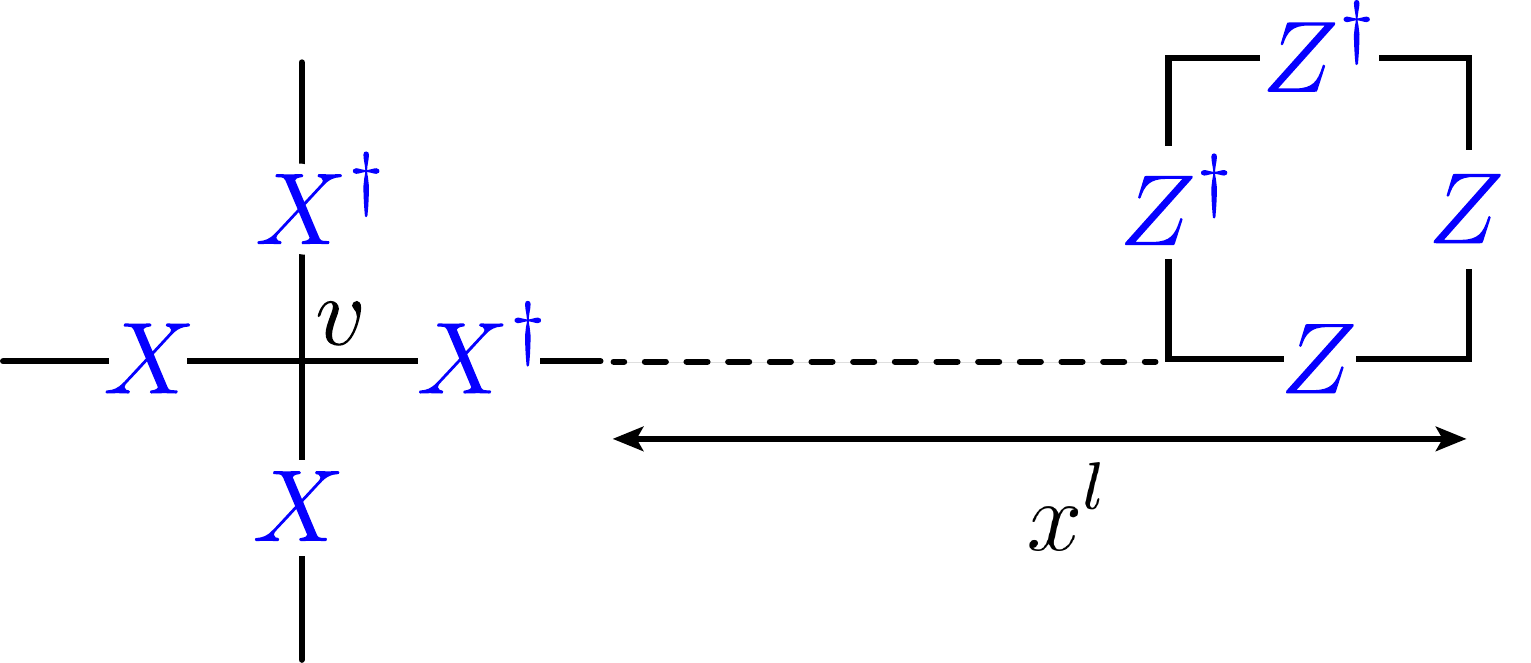}\end{gathered},\\
    \mathcal{S}_2 &:= \begin{gathered}\includegraphics[scale=0.25]{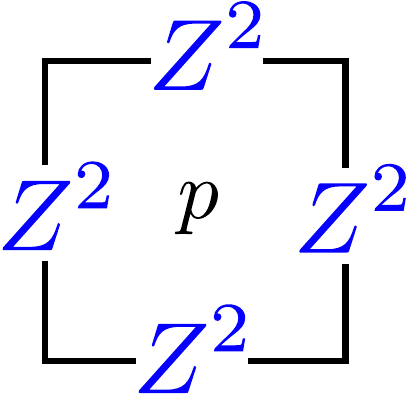}\end{gathered},\\
    \mathcal{S}_3 &:=\begin{gathered}\includegraphics[scale=0.25]{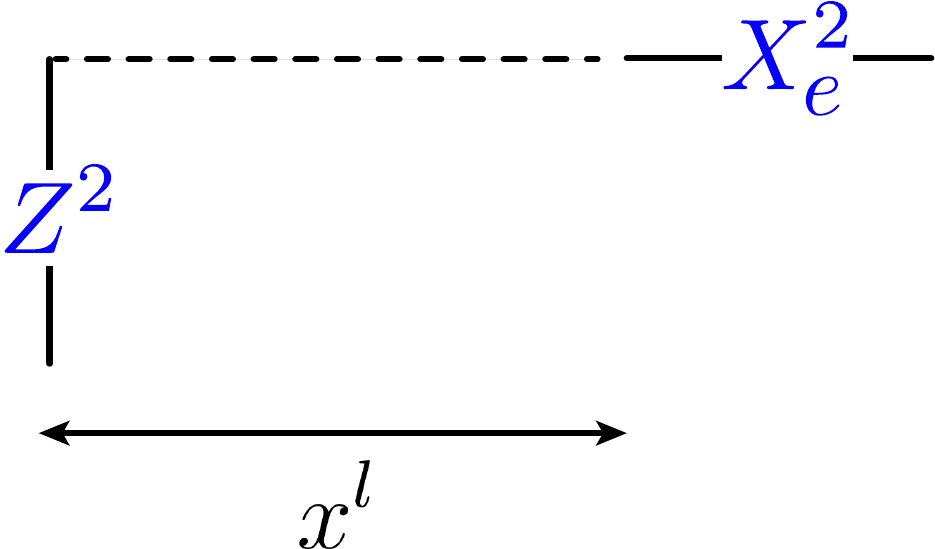}\end{gathered},\\
    \mathcal{S}_4 &:=\begin{gathered}\includegraphics[scale=0.25]{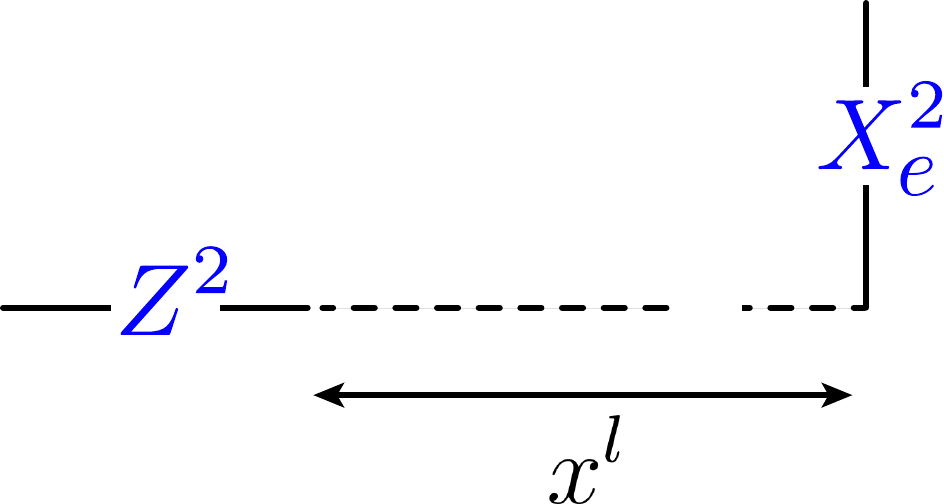}\end{gathered}.
\end{eqs}
For the double-semion code, we condense neighbored $e^2m^2$ anyons. For the shifted double-semion code, instead of condensing neighbored $e^2 m^2$ anyon, we condensed $e^2 m^2$ that are far away by $x^l$, shown as $\mathcal{S}_3, \mathcal{S}_4$ in Eq.~\eqref{eq:shifted_DS_stabilizer}. $\mathcal{S}_1$ and $\mathcal{S}_2$ generate a subgroup of the $\mathbb{Z}_4$ toric code stabilizer group, which commute with condensation terms $\mathcal{S}_3, \mathcal{S}_4$.
The stabilizers of shifted double semion code be written as polynomials
\begin{eqs}\label{eq:shifted_DS_stb_poly}
    \mathcal{S}_1 &=\begin{bmatrix}
        \overline{x}-1\\
        \overline{y}-1\\
        \hline 
        x^l(1-y)\\
        x^l (-1+x)
    \end{bmatrix},\quad \mathcal{S}_2 =\begin{bmatrix}
        0\\
        0\\
        \hline
        2+2y\\
        2+2x
    \end{bmatrix},\\
    \mathcal{S}_3&=\begin{bmatrix}
        2x^l\\
        0\\
        \hline
        0\\
        2\overline{y} 
    \end{bmatrix},\quad \mathcal{S}_4=\begin{bmatrix}
        0\\
        2x^l\\
        \hline 
        2\overline{x}\\
        0
    \end{bmatrix}.
\end{eqs}
If we set $l=0$, Eq.~\eqref{eq:shifted_DS_stb_poly} gives the stabilizers of the double semion code in Eq.~\eqref{eq:DS_stb_poly}.
Given the stabilizers, the error syndromes are 
\begin{eqs}
    \eps(\mX_1)&=[\overline{x}^l (\overline{y}-1), -2\overline{y}-2,0, -2x],\\
    \eps(\mX_2)&=[\overline{x}^l (1-\overline{x}),-2\overline{x}-2,-2 y,0],\\
    \eps(\mZ_1)&=[-1+x,0,2 \overline{x}^l,0], \quad \\
    \eps(\mZ_2)&= [-1+y,0,0,2 \overline{x}^l].
\end{eqs}
Similar to the calculation in Eq.~\eqref{eq:DS_x_string} and \eqref{eq:DS_y_string}, we have

\begin{eqs}
    \epsilon(\mX_2)+x^l y \epsilon(\mZ_1)&=[-\overline{x}^{l+1}-x^l y, -2\overline{x},0,0],\\
    \epsilon(\mX_1)-x^l \epsilon(\mZ_2)&= (1-y)[\overline{x}^l \overline{y}+x^l,-2\overline{y},0,0],\\
    2\epsilon(\mZ_1)&=(1-x)[-2,0,0,0],\\
    2\epsilon(\mZ_2)&=(1-y)[-2,0,0,0].
\end{eqs}
This gives us the string operators for anyon $s$ and $b$ along the $x$- and $y$-directions in Eq.~\eqref{eq:shifted_DS_xstring}. They generate the whole anyon theory of $\{e, s, b, \overline{s} \}$ where $\overline{s}=s\times b$.
\begin{eqs}\label{eq:shifted_DS_xstring}
    W_e^s:=& \begin{gathered}
        \includegraphics[scale=0.3]{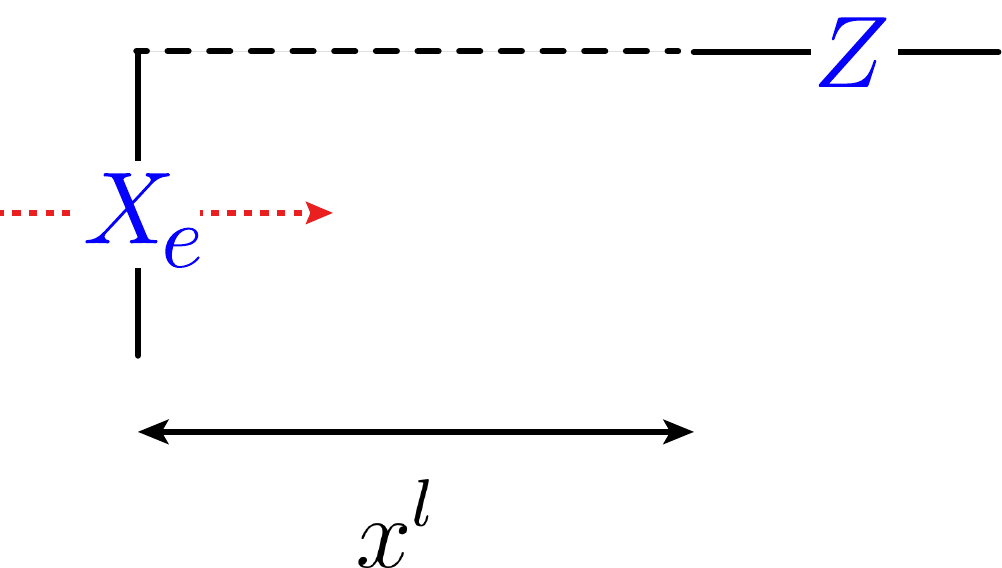}
    \end{gathered} ,\\
    &\begin{gathered}
        \includegraphics[scale=0.3]{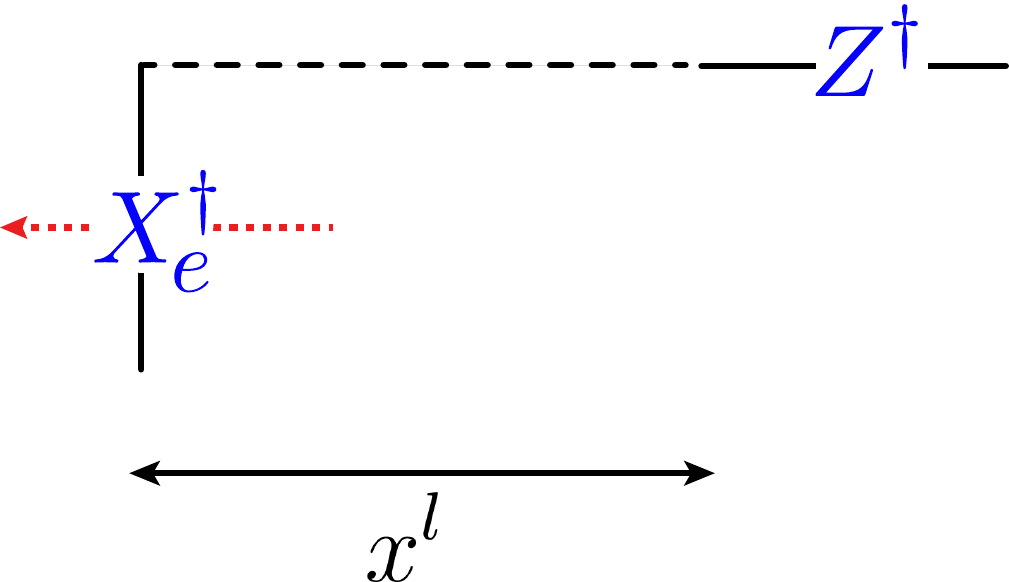}
    \end{gathered}, \\
    & \begin{gathered}
        \includegraphics[height=3cm]{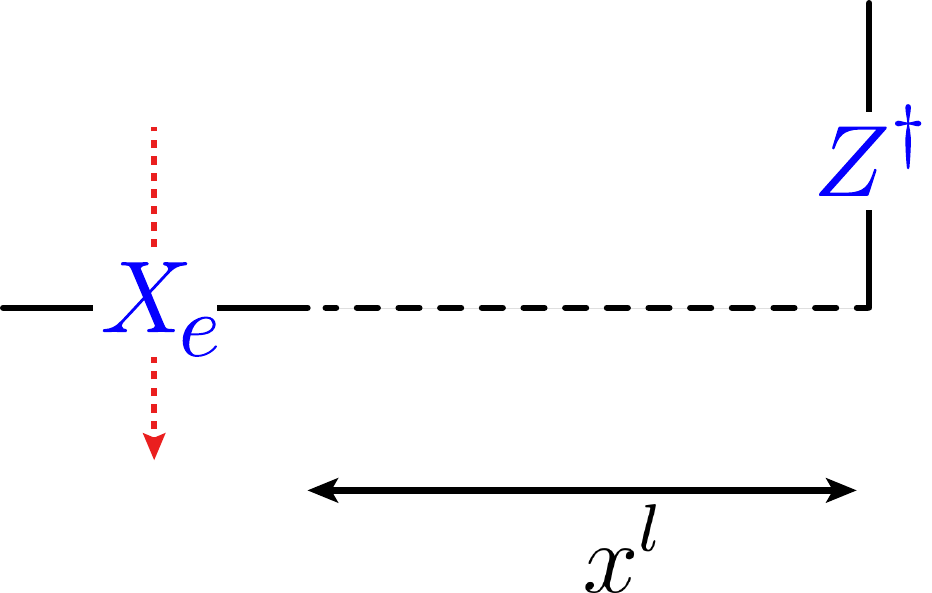}
    \end{gathered} ,\\
    &\begin{gathered}
        \includegraphics[height=3cm]{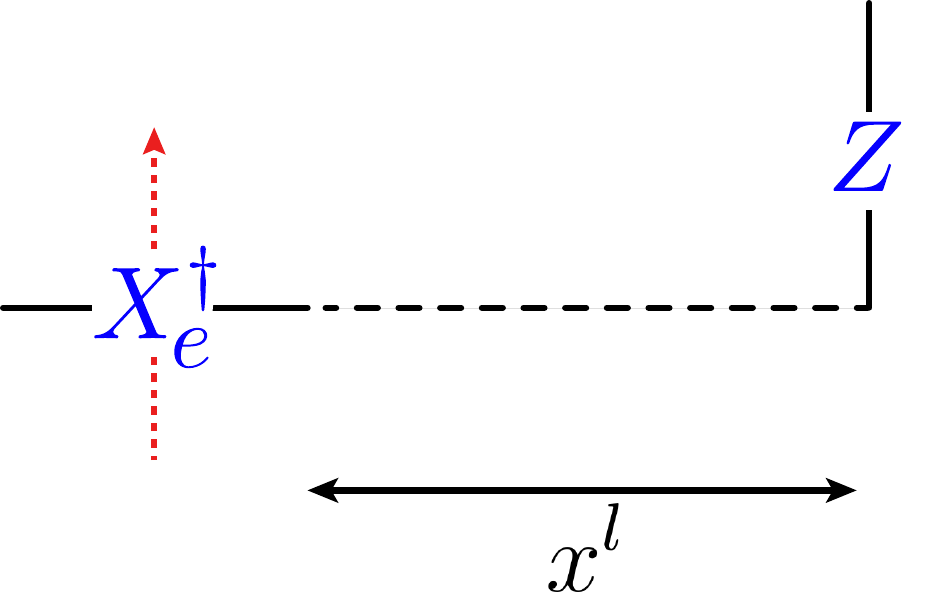}
    \end{gathered}, \\
    W_e^b:= &~~
    \begin{gathered}
        \includegraphics[scale=0.3]{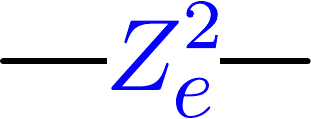}
    \end{gathered}, \quad
    \begin{gathered}
        \includegraphics[height=1.5cm]{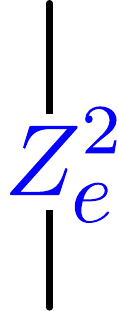}
    \end{gathered} \quad .
\end{eqs}

\section{Computational Method}\label{sec: computational_method}

In this section, we describe how to transform the problem of finding anyon string operators in the previous section into implementing (modified) Gaussian elimination, which is closely related to the Hermite normal form. Computers can perform these algorithms efficiently in a polynomial time. Later, we show that the fusion rules of anyons can be derived from the Smith normal form of anyon relations, and the topological spins and braiding statistics of anyons can be computed from the string operators directly.

To numerically solve the anyon equation~\eqref{eq: anyon equation} to obtain anyons in a stabilizer code, the module over a polynomial ring is not straightforward to work with. As standard techniques in linear algebra, \textbf{Gaussian elimination (GE)} only works for a field, and \textbf{Hermit normal form (HNF)} and \textbf{Smith normal form (SNF)} are only applicable to a \textbf{principal ideal domain (PID)}, which are reviewed in Appendix~\ref{sec: HSF and SNF}.
The polynomial ring $R=\ZZ_d [x, y, x^{-1}, y^{-1}]$ is neither a field nor a PID. Therefore, instead of treating $\eps(P)$ as a row vector in the module over $R$, we truncate the degree of polynomials and store coefficients in the polynomial as a vector over $\ZZ_d$. For instance, a polynomial such as $f(x,y)=1 + 2 x + xy + 3y^2 \in R$ can be expressed as a {\bf coefficient vector} over $\ZZ_d$
\begin{eqs}
    \begin{array}{ccccccccccccc}
                             &1 & x & y & \bx & \by & x^2 & xy & y^2 & x \by &  \bx y&\cdots &\\
        \wwide{f}=\big[  &1 & 2 & 0 & 0   & 0   & 0   & 1  & 3   & 0     &   0   &\cdots & \big],
    \end{array}
\end{eqs}
where each entry represents the coefficient of the corresponding monomial $x^a y^b$ in the polynomial $f(x, y)$.
In the rest of this paper, we denote $\wwide{f}$ as the coefficient vector of the polynomial $f(x,y) \in R$.
In practice, we choose the polynomial within $x^{\pm k}$ and $y^{\pm k}$,  and the monomial $x^a y^b$ corresponds to the $((a+k)+(b+k)(2k+1)+1)$-th entry in the coefficient vector for all $-k \leq a, b \leq k$. The length of the coefficient vector is $(2k+1)^2$. In other words, we restrict ourselves to a finite region $A$ that includes coordinates $(i, j)$ with $-k \leq i,j \leq k$ in a plane, as shown in Fig.~\ref{fig: the region A}.

\begin{figure*}[htb]
    \centering
    \includegraphics[width=.75\textwidth]{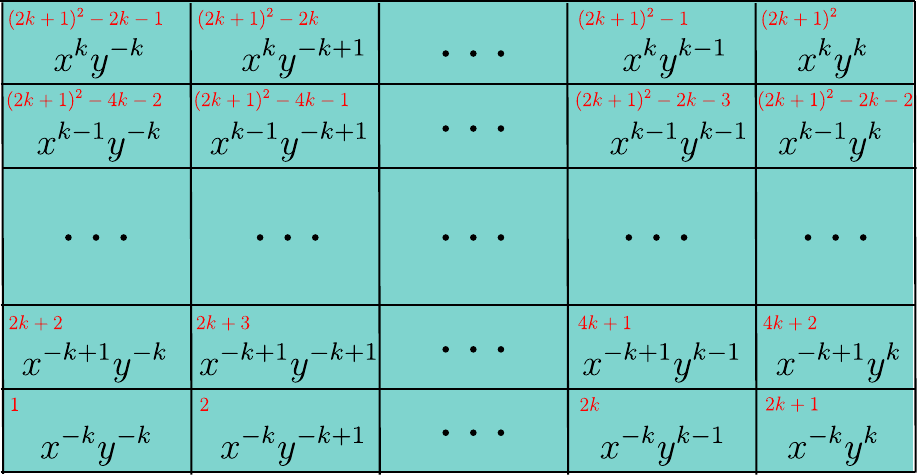}
    \caption{The region $A$ where the polynomials are truncated to. It contains all monomials with $x^{\pm k}$ and $y^{\pm k}$. The red numbers indicate the ordering in the coefficient vector. Any polynomial in $\ZZ_d[x, y, x^{-1}, y^{-1}]$ is truncated to a length-$(2k+1)^2$ row vector over $\ZZ_d$.
    }
    \label{fig: the region A}
\end{figure*}

Given a fixed $k$, we define the {\bf truncation map} as
\begin{eqs}
    f \in \ZZ_d [x, y, x^{-1}, y^{-1}] \to \wwide{f} \in \ZZ_d^{\otimes (2k+1)^2}.
\label{eq:polynomial_truncation}
\end{eqs}
{\color{black} The truncation size is governed by the parameter $k$, a large positive integer relative to the size of the stabilizers.}
We allow this truncation map to act on any {\color{black} $a \times b$} matrix $M$ over $\ZZ_d [x, y, x^{-1}, y^{-1}]$ by acting on each entry to expand to a row vector with length $(2k+1)^2$, and by joining these row vectors to form an {\color{black} $a \times b(2k+1)^2$} matrix $\wwide{M}$ over $\ZZ_d$. As argued in the previous section, sufficiently large $A$ allows us to recover all topological data of a Pauli stabilizer code. To optimize running time, we can choose a rectangular region specified by the polynomial within $x^{\pm k_x}$ and $y^{\pm k_y}$. For example, when focusing on anyon string operators in the $x$-direction, a greater value of $k_x$ compared to $k_y$ is preferable. For convenience, our discussion will focus on the square area $A$ and check the topological order condition, identify anyon string operators, and verify the fusion rules, topological spins, and braiding properties of anyons within $A$.

Also, we define the {\bf translational duplicate map $\mathrm{TD}_m$} that takes the input as a length-$l$ row vector $F=[f_1, f_2, \cdots, f_l]$ with each entry $f_i$ in $\ZZ_d [x, y, x^{-1}, y^{-1}]$ and returns a $(2m+1)^2 \times l$ matrix formed by its translations within $x^{\pm m} y^{\pm m}$:\footnote{We choose $m$ to be smaller then $k$. For example, we set $m = k -2r$, where $r$ is the maximum range of each stabilizer.}
\begin{eqs}
    F\to
    \mathrm{TD}_m(F):=
    \begin{bmatrix}
        x^{-m} y^{-m} F \\
        x^{-m+1} y^{-m} F \\
        \vdots \\
        x^{m-1} y^{-m} F \\
        x^{m} y^{-m} F \\
        \hline
        x^{-m} y^{-m+1} F \\
        x^{-m+1} y^{-m+1} F \\
        \vdots \\
        x^{m-1} y^{-m+1} F \\
        x^{m} y^{-m+1} F \\
        \hline
        \vdots \\
        \vdots \\
        \hline
        x^{-m} y^{m} F \\
        x^{-m+1} y^{m} F \\
        \vdots \\
        x^{m-1} y^{m} F \\
        x^{m} y^{m} F.
    \end{bmatrix}~,
\label{eq: TD definition}
\end{eqs}
where $x^a y^b F$ is the row vector multiplying $x^a y^b $ to each entry of $F$, i.e., $x^a y^b F = [x^a y^b f_1, x^a y^b f_2, \cdots, x^a y^b f_l]$. 

\subsection{Prime-dimensional qudit $\FF_p$}
As a warm-up, we assume that the qudit dimension $d$ is a prime number to illustrate the method, and will later (in Sec.~\ref{sec: nonprime-dimensional qudit}) generalize to an arbitrary integer $d$. When $d=p$ is a prime, $\ZZ_p = \FF_p$ is a finite field, greatly simplifying the computations.

\subsubsection{Topological order condition}\label{sec: checking TO condition Zp}
We check the TO condition~\eqref{eq: topological condition} in the region $A$. Given the syndrome for a single Pauli operator, we find possible local combinations of Pauli operators such that all syndromes are canceled. We verify that these combined operators are in the stabilizer group. 

First, we prepare a matrix consisting of error syndromes:
\begin{eqs}
    \wwide{M}_1 :=
    \begin{bmatrix}
        \wwide{\mathrm{TD}_m( \eps(\mX_1))} \\
        \wwide{\mathrm{TD}_m( \eps(\mX_2))} \\
        \wwide{\mathrm{TD}_m( \eps(\mZ_1))} \\
        \wwide{\mathrm{TD}_m( \eps(\mZ_2))} \\
    \end{bmatrix},
\label{eq: M1 matrix}
\end{eqs}
where the row space consists of the syndrome patterns of local Pauli operators.
Next, we perform the Gaussian elimination to obtain a new matrix, {\color{black} denoted as} $\mathrm{GE}(\wwide{M}_1)$. Each zero row in $\mathrm{GE}(\wwide{M}_1)$ represents a relation:
\begin{eqs}
    \wwide{\eps(p_1\mX_1 + p_2 \mX_2 + p_3 \mZ_1 + p_4 \mZ_2)} = \wwide{[0, 0]},
\end{eqs}
for some polynomials $p_1, p_2, p_3, p_4 \in \FF_p[x, y, x^{-1}, y^{-1}]$, restricted within $x^{\pm m} y^{\pm m}$. We denote the column vector $\mO := \eps(p_1\mX_1 + p_2 \mX_2 + p_3 \mZ_1 + p_4 \mZ_2)$ and compute the difference between the ranks of the following two matrices:
\begin{eqs}
    \wwide{M}_2 :=
    \begin{bmatrix}
        \wwide{\mathrm{TD}_{m'} (\mS_1^\dagger)} \\
        \wwide{\mathrm{TD}_{m'} (\mS_2^\dagger)} \\
    \end{bmatrix},
    ~\mathrm{and}~~
    \wwide{M}_2^\mO :=
    \begin{bmatrix}
        \wwide{\mO^\dagger} \\
        \wwide{\mathrm{TD}_{m'} (\mS_1^\dagger)} \\
        \wwide{\mathrm{TD}_{m'} (\mS_2^\dagger)} \\
    \end{bmatrix}.
\label{eq: M2 and M2O matrix Zp}
\end{eqs}
Note that we choose $m' > m$ (and still $m' < k$) such that stabilizers can cover the support of $\mO$.
The rank difference is $0$ if and only if the operator $\mO$ is spanned by $\mS_1$, $\mS_2$ and their translations. To check the TO condition, we verify that all zero rows in $\mathrm{GE}(\wwide{M}_1)$ correspond to a local operator $\mO$ with vanishing rank difference in Eq.~\eqref{eq: M2 and M2O matrix Zp}.

\subsubsection{Solving the anyon equation}\label{sec: Solving Anyon Equations Zp}

To solve the anyon equation
\begin{eqs}
    &\eps \big(
    \alpha(x,y)  \mX_1 + \beta(x,y) \mX_2
    + \gamma(x,y) \mZ_1 + \delta(x,y) \mZ_2
    \big) \\
    &= (1-x^n) v,
\label{eq: anyon equation 2}
\end{eqs}
we perform Gaussian elimination on the following $6(2m+1)^2 \times 2(2k+1)^2$ matrix
\begin{eqs}
    \wwide{M}_3 :=
    \begin{bmatrix}
        \wwide{\mathrm{TD}_m( \eps(\mX_1))} \\
        \wwide{\mathrm{TD}_m( \eps(\mX_2))} \\
        \wwide{\mathrm{TD}_m( \eps(\mZ_1))} \\
        \wwide{\mathrm{TD}_m( \eps(\mZ_2))} \\
        \wwide{\mathrm{TD}_m( [(1-x^n), 0]} \\
        \wwide{\mathrm{TD}_m( [0, (1-x^n)]} \\
    \end{bmatrix}.
\label{eq: M3 matrix}
\end{eqs}
After Gaussian elimination, each zero row in the matrix $\mathrm{GE} (\wwide{M}_3 )$ gives an anyon $v$ in this region:
\begin{eqs}
    \wwide{\eps(p_1\mX_1 + p_2 \mX_2 + p_3 \mZ_1 + p_4 \mZ_2)} = \wwide{[p_5 (1-x^n), p_6 (1-x^n)]},
\end{eqs}
where the anyon is $v:=[p_5, p_6]$.

\subsubsection{Equivalence relations between anyons}\label{sec: Equivalence relations between anyons Zp}

In this part, we provide an algorithm to categorize anyon types in a Pauli stabilizer model with $\ZZ_p$ qudits. Two anyons $v$ and $v'$ are equivalent if
\begin{eqs}
    v' = & ~v + p_1 \eps(\mX_1) + p_2 \eps(\mX_2) + p_3 \eps(\mZ_1) + p_4 \eps(\mZ_2).
\label{eq: anyon equivalence with restricted polynomials}
\end{eqs}
where $p_1$, $p_2$, $p_3$, and $p_4$ are polynomials within $x^{\pm m} y^{\pm m}$.
In other words, we compute the quotient space of anyons using row vectors $\{ \eps(\mX_1), \eps(\mX_2), \eps(\mZ_1), \eps(\mZ_2)\}$. To check this, we can compute the rank difference between the two matrices
\begin{eqs}\label{eq:syndrome_matrix}
    \wwide{M}_1 =
    \begin{bmatrix}
        \wwide{\mathrm{TD}_m( \eps(\mX_1))} \\
        \wwide{\mathrm{TD}_m( \eps(\mX_2))} \\
        \wwide{\mathrm{TD}_m( \eps(\mZ_1))} \\
        \wwide{\mathrm{TD}_m( \eps(\mZ_2))} \\
    \end{bmatrix},
    ~\mathrm{and}~~
    \wwide{M}_1^{v-v'}:=
    \begin{bmatrix}
        \wwide{v-v'}\\
        \wwide{\mathrm{TD}_m( \eps(\mX_1))} \\
        \wwide{\mathrm{TD}_m( \eps(\mX_2))} \\
        \wwide{\mathrm{TD}_m( \eps(\mZ_1))} \\
        \wwide{\mathrm{TD}_m( \eps(\mZ_2))} \\
    \end{bmatrix}.
\end{eqs}
Two matrices have equal rank if and only if the syndromes span $v-v'$, which is exactly the definition of equivalence $v\sim v'$ in Eq.~\eqref{eq: anyon equivalence with restricted polynomials}.

Moreover, since each qudit has a prime dimension $p$, any topological Pauli stabilizer code is equivalent to finite copies of $\ZZ_p$ toric codes. Therefore, we only need to count the number of basis anyons that generate all anyons, to obtain the topological order. Given the anyon set $V = \{ v_1, v_2, \cdots \}$ solved from the anyon equation~\eqref{eq: anyon equation}, to count the number of basis anyons, we compute the rank difference between
\begin{eqs}
    \wwide{M}_1 =
    \begin{bmatrix}
        \wwide{\mathrm{TD}_m( \eps(\mX_1))} \\
        \wwide{\mathrm{TD}_m( \eps(\mX_2))} \\
        \wwide{\mathrm{TD}_m( \eps(\mZ_1))} \\
        \wwide{\mathrm{TD}_m( \eps(\mZ_2))} \\
    \end{bmatrix},
    ~\mathrm{and}~~
    \wwide{M}_1^{V}:=
    \begin{bmatrix}
        \wwide{v_1}\\
        \wwide{v_2}\\
        \wwide{v_3}\\
        \vdots \\
        \wwide{\mathrm{TD}_m( \eps(\mX_1))} \\
        \wwide{\mathrm{TD}_m( \eps(\mX_2))} \\
        \wwide{\mathrm{TD}_m( \eps(\mZ_1))} \\
        \wwide{\mathrm{TD}_m( \eps(\mZ_2))} \\
    \end{bmatrix}.
\end{eqs}
If the rank difference is $2k$, the topological Pauli stabilizer code corresponds to $k$ copies of $\ZZ_p$ toric codes.

\subsection{Nonprime-dimensional qudit $\ZZ_d$}\label{sec: nonprime-dimensional qudit}
Over a ring $\ZZ_d$, Gaussian elimination does not work since the multiplication inverse of an element in $\ZZ_d$ might not exist, e.g., the element $2$ in $\ZZ_4$. Moreover, $\ZZ_d$ may contain zero divisors, i.e., a nonzero element $a$, such that there exists another nonzero element $x \in \ZZ_d$ satisfying $ax = 0$. The element $2$ in $\ZZ_4$ is again an example since $2 \times 2 = 0$ in $\ZZ_4$.
Therefore, we introduce the \textbf{modified Gaussian elimination (MGE)} algorithm over $\ZZ_d$, inspired by the Hermite normal form:
\begin{enumerate}
    \item Given a $n \times m$ matrix $A$ over $\ZZ_d$. We treat the entries in the first column $a_{i,1}$ $\forall i \in \{1, 2, \cdots ,n\}$ as integers $\{0, 1, 2, 3, \cdots, d-1\}$ in $\ZZ$. {\color{black} To restore the $\mathbb{Z}_d$ periodicity, we append a new row $[d, 0, 0, \cdots 0]$ to the bottom of matrix $A$, transforming it into a $(n+1) \times m$ matrix, denoted as $A'$.} Next, we find the greatest common divisor of its first column $\{a_{1,1}, a_{2,1}, \cdots, a_{n,1}, d \}$, denoted as $\gcd(a_{i,1})_d$.\footnote{For the convenience, we choose $\gcd(a_{i,1})_d$ to be in the set $\{0, 1, 2, 3, \cdots, d-1 \}$.} From the extended Euclidean algorithm, there exists a linear combination:
    \begin{eqs}
        &r_1 a_{1,1} + r_2 a_{2,1} + \cdots + r_n a_{n,1} +\textcolor{black}{r_0 d}\\
        &= \gcd(a_{i,1})_d.
    \end{eqs}
    Moreover, this linear combination can be obtained by subtracting one entry from another greater entry from the following list repetitively,
    \begin{eqs}
        [a_{1,1}, a_{2,1}, a_{3,1}, \cdots, a_{n,1}, d],
    \end{eqs}
     and the final entries are (up to reordering)
     \begin{eqs}
        [\gcd(a_{i,1})_d, 0, 0, \cdots, 0, 0].
    \end{eqs}
    Subtracting one entry from another and reordering would correspond to row operations in the matrix $A'$. Therefore, we apply corresponding row operations in the matrix $A'$ according to the extended Euclidean algorithm, which transforms the first column into
    \begin{eqs}
        [\gcd(a_{i,1})_d, 0, 0, \cdots, 0, 0]^T.
    \end{eqs}    
    \item If $\gcd(a_{i,1})_d$ is a zero divisor, i.e., $\gcd(a_{i,1})_d > 0$ and $$\gcd(a_{i,1})_d \times r^* \equiv 0 ~(\text{mod }d)$$ for some integer $r^*$ with $0 < r^* \leq d-1$ (WLOG, we assume that $r^*$ is the smallest one), we do not perform an operation in this row.\footnote{Denote $\gcd(a_{i,1})_d \times r^* = kd$. We notice that $r^*$ times this row has first entry $0$ (mod $d$) and might affect the space spanned by other rows. In principle, we should insert a new row at the end of the matrix, which is $r^*$ times the first row minus the row $[kd, 0, 0, \cdots, 0]$, before dealing with the second column. However, this inserting is redundant since it is equivalent to inserting a new row $[-kd, 0, 0, \cdots 0]$ that is already covered by the insertion of $[d,0,0, \cdots, 0]$ in step 1.}
    If $\gcd(a_{i,1})_d$ is not a zero divisor, we multiply a proper number in this row to make it $+1$.
    \item The first column and the first row are done. Repeat the above procedures on the submatrix without the first column and the first row.
\end{enumerate}
{\color{black}
In the original matrix $A$, linear \textbf{relations} exist between the row vectors; specifically, certain row vectors can be combined linearly to result in the zero row vector (mod $d$). These relations are crucial as they provide insights into the connections between the row vectors. For instance, suppose one row $r_1$ represents the syndrome pattern of an anyon $v$, another row $r_2$ represents the syndrome pattern of a different anyon $v'$, and a third row $r_3$ represents the syndrome of a local Pauli operator $\mathcal{P}$. If these rows satisfy the relation $r_1 - r_2 + r_3 = 0$, it implies that the anyons $v$ and $v'$ are related through the local operator $\mathcal{P}$, or more specifically,
\begin{eqs}
    v' = v + \eps (\mP),
\end{eqs}
indicating that $v$ and $v'$ are of the same anyon type. Thus, identifying all such relations between the row vectors is essential.

We provide a concrete example of implementing the modified Gaussian elimination algorithm on a selected matrix $A$ over $\ZZ_8$. The matrix $A$ is given by
\begin{equation}
    A=\begin{bmatrix}
    4 & 2 & 0  \\
    6 & 0 & 3  \\
    0 & 7 & 4  \\
\end{bmatrix}= \begin{bmatrix}
    - & v_1 & - \\
    - & v_2 & - \\
    - & v_3 & - 
\end{bmatrix},
\label{eq: example A matrix}
\end{equation}
where $v_1$, $v_2$, and $v_3$ represent the row vectors of $A$. Our goal is to derive the relations between the row vectors $v_1$, $v_2$, and $v_3$.

First, we embed this matrix over $\mathbb{Z}$ such that each entry is chosen to be $0, 1, 2, \cdots, 7$.
We insert a row $[8, 0, 0]$ on the bottom:
\begin{eqs}
    [A'|R]=\left[\begin{array}{c|c}
    \begin{matrix}
    4 & 2 & 0   \\
    6 & 0 & 3   \\
    0 & 7 & 4 \\
    \color{blue} 8 & \color{blue} 0 & \color{blue} 0  \end{matrix}&
    \begin{matrix}
         1& 0& 0& 0\\
         0& 1& 0& 0\\
         0& 0& 1& 0\\
        \color{blue} 0& \color{blue} 0& \color{blue} 0& \color{blue} 1 
    \end{matrix}
\end{array}\right],
\end{eqs}
where the matrix $R$ is used to track row operations during the following process, recording how each current row is derived from the rows in the original matrix $A$ as described in Eq.~\eqref{eq: example A matrix}, including those inserted rows.
The greatest common divisor of the first column is $2$, which can be obtained from $(-1) \times 6 + 1 \times 8$:
\begin{eqs}
    \left[\begin{array}{c|c}
    \begin{matrix}
        4 & 2 & 0  \\
        6 & 0 & 3  \\
        0 & 7 & 4  \\
        \color{blue}2 & \color{blue}0 & \color{blue}-3 
    \end{matrix}&
    \begin{matrix}
        1& 0& 0& 0\\
        0& 1& 0& 0\\
        0& 0& 1& 0\\
        \color{blue} 0& \color{blue} -1& \color{blue} 0& \color{blue} 1
    \end{matrix}
    \end{array}\right].
\end{eqs}
Subsequently, we position the last row at the top and utilize it to eliminate entries in the other rows:
\begin{eqs}
\left[\begin{array}{c|c}
    \begin{matrix}
    2 & 0 & -3  \\
    \color{blue}0 & \color{blue}2 & \color{blue} 6  \\
    \color{blue}0 & \color{blue}0 & \color{blue} 12 \\
    \color{blue}0 & \color{blue}7 & \color{blue} 4 
    \end{matrix}&
    \begin{matrix}
        0& -1& 0& 1\\
        \color{blue} 1& \color{blue} 2& \color{blue} 0& \color{blue} -2 \\
        \color{blue} 0& \color{blue} 4& \color{blue} 0& \color{blue} -3\\
        \color{blue} 0& \color{blue} 0& \color{blue} 1& \color{blue} 0\\
    \end{matrix}
    \end{array}\right]
\end{eqs}
The first row and column have been completed. From now on, the first row will not be involved in subsequent calculations. We will continue the process by initially inserting $[0, 8, 0]$:
\begin{eqs}
\left[\begin{array}{c|c}
    \begin{matrix}
        2 & 0 & -3  \\
        0 & 2 & 6  \\
        0 & 0 & 12  \\
        0 & 7 & 4  \\
        \color{blue}0 & \color{blue}8 & \color{blue}0 
    \end{matrix}& 
    \begin{matrix}
        0& -1& 0& 1& 0\\
        1&  2&  0&  -2& 0 \\
        0&  4&  0&  -3& 0\\
        0&  0&  1&  0& 0\\
        \color{blue}0& \color{blue}0& \color{blue}0& \color{blue}0& \color{blue} 1
    \end{matrix}
    \end{array}\right].
\end{eqs}
The gcd of the second column (except the entry in the first row) is $1$, obtained from $8-7$:
\begin{eqs}
\left[\begin{array}{c|c}
    \begin{matrix}
        2 & 0 & -3  \\
        0 & 2 & 6  \\
        0 & 0 & 12  \\
        0 & 7 & 4  \\
        \color{blue}0 & \color{blue} 1 & \color{blue} -4  
    \end{matrix}& 
    \begin{matrix}
        0& -1& 0& 1& 0\\
        1&  2&  0&  -2& 0 \\
        0&  4&  0&  -3& 0\\
        0&  0&  1&  0& 0\\
        \color{blue}0& \color{blue}0& \color{blue} -1& \color{blue}0& \color{blue}1
    \end{matrix}\end{array}\right].
\end{eqs}
Next, we place the last row in the second position and use it to cancel entries in the rows below:
\begin{eqs}
\left[\begin{array}{c|c}
    \begin{matrix}
        2 & 0 & -3  \\
        0 & 1 & -4  \\
        \color{blue}0 & \color{blue}0 & \color{blue}14  \\
        \color{blue}0 & \color{blue}0 & \color{blue}12 \\
        \color{blue}0 & \color{blue}0 & \color{blue}32  \\
    \end{matrix} & 
    \begin{matrix}
        1& 0& 0& -1& 0\\
        0& 1& 0& 1&  -1\\
        \color{blue} 1 &\color{blue} 2 &\color{blue} 2& \color{blue} -2 &\color{blue} -2\\
        \color{blue} 0& \color{blue} 4& \color{blue}  0& \color{blue} -3 &\color{blue} 0\\
        \color{blue} 0& \color{blue} 0& \color{blue} 8 & \color{blue} 0 &\color{blue} -7 \\        
    \end{matrix}
    \end{array}\right].
\end{eqs} 
Finally, we insert $[0,0,8]$
\begin{eqs}
\left[\begin{array}{c|c}
    \begin{matrix}
        2 & 0 & -3  \\
        0 & 1 & -4  \\
        0 & 0 & 14  \\
        0 & 0 & 12 \\
        0 & 0 & 32  \\
        \color{blue} 0& \color{blue} 0& \color{blue} 8
    \end{matrix} & 
    \begin{matrix}
        1& 0& 0& -1& 0& 0\\
        0& 1& 0& 1&  -1& 0\\
        1 & 2 & 2& -2 & -2& 0\\
        0&  4&   0&  -3 & 0& 0\\
        0&  0&  8 &  0 &-7 & 0\\ 
        \color{blue} 0& \color{blue} 0& \color{blue} 0&  \color{blue} 0& \color{blue} 0& \color{blue} 1
    \end{matrix}
    \end{array}\right],
\end{eqs}
and find the gcd of the third column (except the entries in the first and second rows) is $2$, which can be obtained from $14-12$:
\begin{eqs}
\left[\begin{array}{c|c}
    \begin{matrix}
        2 & 0 & -3  \\
        0 & 1 & -4  \\
        \color{blue}0 & \color{blue}0 & \color{blue}2  \\
        0 & 0 & 12 \\
        0 & 0 & 32  \\
        0& 0& 8
    \end{matrix} & 
    \begin{matrix}
        1& 0& 0& -1& 0& 0\\
        0& 1& 0& 1&  -1& 0\\
        \color{blue}1 & \color{blue}-2 & \color{blue}2& \color{blue}1 & \color{blue}-2& \color{blue}0\\
        0&  4&   0&  -3 & 0& 0\\
        0&  0&  8 &  0 &-7 & 0\\ 
        0& 0& 0&  0&  0& 1
    \end{matrix}
    \end{array}\right].
\end{eqs}
Finally, we use this $2$ to \textcolor{black}{cancel} all entries below:
\begin{eqs}
[A'|R]=\left[\begin{array}{c|c}
    \begin{matrix}
        2 & 0 & -3  \\
        0 & 1 & -4  \\
        0 & 0 & 2  \\
        \color{blue} 0 & \color{blue}0 & \color{blue}0 \\
        \color{blue}0 & \color{blue}0 & \color{blue}0  \\
        \color{blue}0 & \color{blue}0&  \color{blue}0
    \end{matrix} & 
    \begin{matrix}
        1& 0& 0& -1& 0& 0\\
        0& 1& 0& 1&  -1& 0\\
        1 & -2 & 2& 1 & -2& 0\\
        \color{blue}-6&  \color{blue}16&  \color{blue}-12&  \color{blue}-9 & \color{blue}12& \color{blue}0\\
        \color{blue}-16&  \color{blue}32&  \color{blue}-24&  \color{blue}-16 & \color{blue}25 & \color{blue}0\\ 
        \color{blue}-4& \color{blue}8& \color{blue}-8&  \color{blue}-4&  \color{blue}8& \color{blue}1
    \end{matrix}
    \end{array}\right].
    \label{eq: A'|R matrix}
\end{eqs}
We have achieved the row echelon form for the integer matrix $A$.\footnote{\textcolor{black}{To enhance the appearance of the final matrix, we can add the second row to the rows above it, adjusting entries in the second column to range from $0$ to $A'{2,2}-1$. This adjustment can be similarly applied to subsequent rows, ensuring that entries are maintained within $0$ to $A'{i,i}-1$ for the $i$-th column. The algorithm's efficiency could be slightly improved by using $[0, \ldots, d, \ldots, 0]$ to constrain all entries to the range between $0$ and $d$ during the computation. This strategy is beneficial because the Euclidean algorithm, used for computing the greatest common divisor (GCD) between two numbers $a$ and $b$, has a time complexity of $O(\log(\min(a,b)))$.
}} 
We then select the bottom-left $3\times3$ blocks of matrix $R$ to serve as the relation matrix resulting from the modified Gaussian elimination:
\begin{equation}
    \text{relation}:=\left[
    \begin{matrix}
        -6 & 16 & -12 \\
        -16 & 32 & -24 \\
        -4 & 8 & -8 \\
    \end{matrix}\right]
    =
    \left[
    \begin{matrix}
        2 & 0 & 4 \\
        0 & 0 & 0 \\
        4 & 0 & 0 \\
    \end{matrix}\right]
    \mod 8.
\end{equation}
The last three columns in $R$ will be modulo out by $\mathbb{Z}_8$, so they play no roles in the obtained relation.\footnote{During the computation, it is feasible to disregard the information in the last three columns to save computational resources.} The relation matrix traces the relations among $v_1$, $v_2$, and $v_3$ as derived from Eq.~\eqref{eq: example A matrix}:
\begin{equation}
    2v_1 + 4v_3 = [0, 0, 0], \quad 4v_1 = [0, 0, 0] \mod 8.
\end{equation}

}
We will prove the following four theorems about the modified Gaussian elimination.
\begin{theorem}
    \label{thm: reversible}
    The nonzero rows in the final matrix span the same space as the initial matrix.
\end{theorem}
\begin{proof}
    Each procedure is reversible, and the matrix does not lose any information.
\end{proof}
{\color{black}Next, we present a theorem that explains how the relations are preserved following the modified Gaussian elimination process:
}
\begin{theorem}
    \color{black} All relations between the original rows can be recovered from the zero rows resulting from the application of the modified Gaussian elimination algorithm.
\end{theorem}
\begin{proof}
    Consider a relation between the original rows $r = \sum_i (c_i a_{i,1}) = 0$. Because each row $a_{a,i}$ is spanned by the rows in the final rows $A'$ (according to Theorem~\ref{thm: reversible}), $r = \sum_i (c_i a_{i,1})$ is spanned by rows in $A'$. We need to show that the zero rows in $A'$ span it. This is trivial since $A'$ is in the row echelon form, and if the linear combination involves any nonzero row, the combined row will also be nonzero.
\end{proof}
{\color{black} Subsequently, to determine whether a row vector $v$ is spanned by the matrix $A$, we employ the modified Gaussian elimination algorithm to reduce $A$ to its row echelon form, denoted as $A'$.} Define $v^{(0)} = v$. Given $v^{(k-1)}$, if $a'_{k,k} | v^{(k-1)}_k$, where $a'_{k,k}$ is the $k$-th diagonal entry of $A'$ and $v^{(k-1)}_k$ is the $k$-th entry of $v^{(k-1)}$, define
\begin{eqs}
    \label{eq:update-rule}
    v^{(k)} = v^{(k-1)} - \frac{v^{(k-1)}_k}{a'_{k,k}} A'[k,-],
\end{eqs}
where $A'[k,-]$ is the $k$-th row of $A'$. Otherwise, we could not define $v^{(k)}$ and stop this process. Note that the first $k$ entries of $v^{(k)}$ are 0.
\begin{theorem}
    A row vector $v$ of length $n$ is spanned by row vectors of the matrix $A$ iff the process defined above achieves $v^{(n)}$.
\label{thm: check span}
\end{theorem}
\begin{proof}
Since, by Thm.~\ref{thm: reversible}, the row span of $A'$ is the same as the row span of $A$, we must prove that $v$ is spanned by the rows of $A'$ iff the process achieves $v^{(n)}$.
We begin with the $\Rightarrow$ direction. 
We assume that $v$ is spanned by the rows of $A'$, so that $v = v^{(0)} = \sum_{i=1}^n \alpha_i A'[i, -]$ for some coefficients $\alpha_i$. Since $A'$ is upper triangular, it must be that $v_{k} = \sum_{i=1}^k \alpha_i a'_{i,k}$.
Therefore,
\begin{equation}
    \label{eq:a-divides-v}
    \alpha_k a_{k,k}' = v_k - \sum_{i=1}^{k-1}\alpha_i a'_{i,k} .
\end{equation}
We claim that $v^{(k)} = v - \sum_{i=1}^k \alpha_i A'[i,-]$, so that the process indeed achieves $v^{(n)}$. We prove this by induction.
For the base case, we use that $a_{1,1}' \mid v_1$ by Eq.~\eqref{eq:a-divides-v}, so that by Eq.~\eqref{eq:update-rule} $v^{(1)} = v - \frac{v_1}{a_{1,1}'} A'[1,-] = v - \alpha_1 A'[1,-]$.
Now for the inductive step, we assume that $v^{(k-1)} = v - \sum_{i=1}^{k-1} \alpha_i A'[i,-]$. Hence, $v^{(k-1)}_k = \alpha_k a_{k,k}'$ by {\color{black} Eq.}~\eqref{eq:a-divides-v}.
Then by Eq.~\eqref{eq:update-rule} $v^{(k)} = v^{(k-1)} - \alpha_k A'[k,-]$, which proves the inductive step.

We now prove the $\Leftarrow$ direction. We assume that $a_{k,k}' \mid v_k^{(k-1)}$ for every $k$, and we thus define $\alpha_k$ by $v_k^{(k-1)} = \alpha_k a_{k,k}'$. 
Then, reversing Eq.~\eqref{eq:update-rule}, we find that $v^{(k-1)} = v^{(k)} + \alpha_k A'[k,-]$.
Noticing that $v^{(n)}$ must be $(0,\dots,0)$, we find that $v = v^{(0)} = \sum_{i=1}^n \alpha_i A'[i,-]$.
\end{proof}

\begin{theorem}
    The modified Gaussian elimination algorithm on the matrix $A$ over $\ZZ_d$ is equivalent to applying the Gaussian elimination algorithm on the following matrix over a PID $\ZZ$:
    \begin{eqs}
        \tilde A = \begin{bmatrix}
        A_{11} & A_{12} & \cdots & A_{1m} \\
        A_{21} & A_{22} & \cdots & A_{2m} \\
        \vdots & \vdots & \ddots & \vdots \\
        A_{n1} & A_{n2} & \cdots & A_{nm} \\
        d      & 0      & \cdots & 0      \\
        0      & d      & \cdots & 0      \\
        \vdots & \vdots & \ddots & \vdots \\
        0      & 0      & \cdots & d      
        \end{bmatrix}
    \end{eqs}
\label{thm: modified Gaussian as standard Gaussian}
\end{theorem}
\begin{proof}
    A linear system $xA = b \pmod d$, for $x$ a $1\times n$ row vector and $b$ a $1\times m$ row vector, is encoded by the linear equation $xA + dy = b$ over $\ZZ$, where $y$ is a $1\times m$ row vector of additional variables. The latter system is exactly $\tilde x \tilde A = b$ over $\ZZ$, where $\tilde x$ is the $1\times (n+m)$ row vector $(x, y)$. It follows that if $x$ is the solution to $xA = b \pmod d$ and $\tilde x$ is the solution to $\tilde x \tilde A = b$, then $\tilde x_i = x_i \pmod d$ for $i=1,\dots, n$.
\end{proof}

\subsubsection{Topological order condition}\label{sec: checking TO condition Zd}
We check the TO condition~\eqref{eq: topological condition} in the region $A$. Similar to Sec.~\ref{sec: checking TO condition Zp}, we prepare a matrix consisting of syndrome pattern of local Pauli operators:
\begin{eqs}
    \wwide{M}_1 :=
    \begin{bmatrix}
        \wwide{\mathrm{TD}_m( \eps(\mX_1))} \\
        \wwide{\mathrm{TD}_m( \eps(\mX_2))} \\
        \wwide{\mathrm{TD}_m( \eps(\mZ_1))} \\
        \wwide{\mathrm{TD}_m( \eps(\mZ_2))} \\
    \end{bmatrix}.
\label{eq: M1 matrix Zd}
\end{eqs}
Next, we perform the modified Gaussian elimination to obtain a new matrix $\mathrm{MGE}(\wwide{M}_1)$, and each zero row in $\mathrm{MGE}(\wwide{M}_1)$ corresponds to a relation:
\begin{eqs}
    \wwide{\eps(p_1\mX_1 + p_2 \mX_2 + p_3 \mZ_1 + p_4 \mZ_2)} = \wwide{[0, 0]},
\end{eqs}
for some polynomials $p_1$, $p_2$, $p_3$, and $p_4$ restricted within $x^{\pm m} y^{\pm m}$. We denote the column vector $\mO := \eps(p_1\mX_1 + p_2 \mX_2 + p_3 \mZ_1 + p_4 \mZ_2)$. However, the rank over a non-free module is not well-defined, so the previous method of computing rank difference in Sec.~\ref{sec: checking TO condition Zp} does not apply. Instead, we use Theorem~\ref{thm: check span} to check whether $\wwide{\mO^\dagger}$ is spanned by $\mathrm{MGE}(\wwide{M}_2)$, where $\wwide{M}_2$ is defined the same as previous Eq.~\eqref{eq: M2 and M2O matrix Zp}:
\begin{eqs}
    \wwide{M}_2 :=
    \begin{bmatrix}
        \wwide{\mathrm{TD}_{m'} (\mS_1^\dagger)} \\
        \wwide{\mathrm{TD}_{m'} (\mS_2^\dagger)} \\
    \end{bmatrix},
\label{eq: M2 and M2O matrix Zd}
\end{eqs}
where $m'$ is slightly larger than $m$ to make sure we have included all possible stabilizers to cover the operator $\mO$.
If all local operators $\mO$ obtained from the zero rows of $\mathrm{MGE}(\wwide{M}_1)$ are spanped by $\mathrm{MGE}(\wwide{M}_2)$, the TO condition is satisfied.

\subsubsection{Solving the anyon equation}\label{sec: Solving Anyon Equations Zd}
Gaussian elimination in Sec.~\ref{sec: Solving Anyon Equations Zp} is replaced by the modified Gaussian elimination in this section. To solve the anyon equation~\eqref{eq: anyon equation 2}, we perform the modified Gaussian elimination on $\wwide{M}_3$:
\begin{eqs}
    \wwide{M}_3 =
    \begin{bmatrix}
        \wwide{\mathrm{TD}_m( \eps(\mX_1))} \\
        \wwide{\mathrm{TD}_m( \eps(\mX_2))} \\
        \wwide{\mathrm{TD}_m( \eps(\mZ_1))} \\
        \wwide{\mathrm{TD}_m( \eps(\mZ_2))} \\
        \wwide{\mathrm{TD}_m( [(1-x^n), 0]} \\
        \wwide{\mathrm{TD}_m( [0, (1-x^n)]} \\
    \end{bmatrix}.
\label{eq: M3 matrix Zd}
\end{eqs}
Each zero row in the matrix $\mathrm{MGE}(\wwide{M}_3)$ gives an anyon $v$ in this region:
\begin{eqs}
    \wwide{\eps(p_1\mX_1 + p_2 \mX_2 + p_3 \mZ_1 + p_4 \mZ_2)} = \wwide{[p_5 (1-x^n), p_6 (1-x^n)]},
\end{eqs}
where the anyon is $v:=[p_5, p_6]$.

\subsubsection{Equivalence relations between anyons}\label{sec: Equivalence relations between anyons Zd}
In this part, we provide an algorithm to categorize anyon types in a Pauli stabilizer model with $\ZZ_d$ qudits. The equivalence relation of two anyons is defined in Eq.~\eqref{eq: anyon equivalence with restricted polynomials}. We can use Themrem~\ref{thm: check span} to check whether $\wwide{v'-v}$ can be spanned by $\mathrm{MGE}(\wwide{M}_1)$, with $\wwide{M}_1$ defined in Eq.~\eqref{eq: M1 matrix Zd}. This determines whether $v$ is equivalent to $v'$.

\subsection{Fusion rules from the Smith normal form}\label{sec: Smith normal form}

In this section, we discuss how to derive fusion rules of anyons. Assume that we have solved the anyon equation and obtain a list of anyon $V=\{v_1, v_2, v_3, \cdots \}$. Many are redundant, and we are looking for the basis anyons.

We consider the following procedure:
\begin{enumerate}
    \item Start with $M^{(0)}:=  MGE(\wwide{M}_1)$ and $V_{\mathrm{gen}}=\{ \}$.
    \item Given $M^{(i-1)}$, we use Theorem~\ref{thm: check span} to check whether $v_i$ is spanned by $M^{(i-1)}$. If $v_i$ is not spanned by $M^{(i-1)}$, we add $v_i$ into $V_{\mathrm{gen}}$ and define $M^{(i)}$ as $M^{(i-1)}$ with an extra row $v_i$ joined below. Otherwise, $M^{(i)}:=M^{(i-1)}$.
    \item Repeat the previous step until running over all anyons in $V$.
\end{enumerate}
From the construction, it is straightforward to prove that $V_{\mathrm{gen}}$ generates all anyons in $V$. However, $V_{\mathrm{gen}}$  can still be redundant. For example, consider the $\ZZ_{12}$ toric code stabilizer, and assume that we have solved anyons
\begin{eqs}
    V=\{v_{m^2}, v_{m^4}, v_{e^2m^4}, v_{e^2,m^3}, v_{e}, v_{m}, v_{e^3}, v_{m^6}\},
\end{eqs}
where $e$ and $m$ are the basis anyons (still unknown). According to the procedure above, the anyon set $V_{\mathrm{gen}}$ is
\begin{eqs}
    V_{\mathrm{gen}} = \{v_{m^2}, v_{e^2 m^4}, v_{m^3}, v_{e}\}.
\end{eqs}
These anyons are still redundant since we have
\begin{eqs}
    3 v_{m^2} + 2v_{m^3} \sim {\bf 0},
\end{eqs}
while both $3 v_{m^2}$ and $2v_{m^3}$ are not the trivial anyon ${\bf 0}$. In this section, for convenience, we use $a+b$ for anyon fusion to replace the standard $a\times b$ notation in the following steps of the Smith normal form technique. To completely remove this redundancy, we slightly modify the procedure above. When we add $v_i$ into $V_{\mathrm{gen}}$, we also compute the least multiple of $v_i$ such that it can be spanned by $M^{(i-1)}$ and record this relation.
In the example above, we obtain
\begin{eqs}
    &6 v_{m^2} \sim {\bf 0}, \\
    &0 v_{m^2} + 6 v_{e^2 m^4} \sim {\bf 0}, \\
    &3 v_{m^2} + 0 v_{e^2 m^4} + 2 v_{m^3} \sim {\bf 0}, \\
    &2 v_{m^2} -1 v_{e^2 m^4} + 0 v_{m^3} + 2v_{e} \sim {\bf 0},
\end{eqs}
which can be represented as the {\bf anyon relation matrix}
\begin{eqs}
    M =
    \begin{blockarray}{*{4}{c} l}
        \begin{block}{*{4}{>{\footnotesize}c<{}} l}
            v_{m^2}& v_{\change e^2 m^4}& v_{m^3}& v_e  \\
        \end{block}
        \begin{block}{[*{4}{c}]>{\footnotesize}l<{}}
            6&0&0&0 \\
            0&6&0&0 \\
        	3&0&2&0 \\
        	2&-1&0&2 \\
        \end{block}
  \end{blockarray}~~.
  \label{eq: anyon relation example}
\end{eqs}
Note that row operations on $M$ do not affect the labels of anyons in the anyon relation matrix, but the column operations on $M$ require the rearrangement of anyon labels on each column. More precisely, adding column $k$ to $j$ in the relation $M$ requires the redefinition of anyons as $v'_j = v_j$ and $v'_k = (- v_j + v_k)$:
\begin{eqs}
    M_{i,j} v_j + M_{i,k} v_k &= (M_{i,j} + M_{i,k}) v'_j + M_{i,k} v'_k ~\forall~i.
\label{eq: anyon rearrangement for M}
\end{eqs}
{\change
Then, we compute the Smith normal form of $M$, obtaining matrices $L$, $R$, and $A$ such that
\begin{equation}
    M=LAR,
\end{equation}
where $A$ is a diagonal integer matrix, and $L$ and $R$ are unimodular matrices (i.e., matrices with determinant $\pm 1$). For the matrix $M$ in Eq.~\eqref{eq: anyon relation example}, the explicit forms of $L$, $A$, and $R$ are:
\begin{eqs}
    L &= \begin{bmatrix}
		6&-66&21&-4 \\
		0&-6&2&-1 \\
		3&-31&10&-2 \\
		2&-19&6&-1
	\end{bmatrix}, \\
    A &= \begin{bmatrix}
		1&0&0&0 \\
		0&1&0&0 \\
		0&0&12&0 \\
		0&0&0&12
	\end{bmatrix}, \\
    R &= \begin{bmatrix}
		1&1&2&10 \\
		0&-3&4&-6 \\
		0&-1&1&-2 \\
		0&-1&0&-1
	\end{bmatrix}.
\end{eqs}
The matrix $R$ determines the new basis of anyons after the rearrangement, as described in Eq.~\eqref{eq: anyon rearrangement for M}. The orders of these basis anyons correspond to the diagonal entries of matrix $A$. For instance, the first row of $R$ indicates that
\begin{equation}
    1v_{m^2} + 1 v_{e^2m^4} + 2 v_{m^3} + 10 v_e = v_{e^{12}m^{12}},
\end{equation}
has order 1. This is consistent with the fact that $e^{12}m^{12}$ represents the trivial anyon in the $\ZZ_{12}$ toric code. In contrast, the third row of $R$ implies that
\begin{equation}
    0v_{m^2}-1v_{e^2m^4}+1v_{m^3}-2v_{e} = v_{e^{-4}m^{-1}},
\end{equation}
is an anyon of order 12. Similarly, the fourth row of $R$ shows that
\begin{equation}
    0v_{m^2}-1v_{e^2m^4}+0v_{m^3}-1v_{e} = v_{e^{-3}m^{-4}},
\end{equation}
is also an anyon of order 12.
We can verify that $e^{-4}m^{-1}$ and $e^{-3}m^{-4}$ generate the complete set of anyons in the $\ZZ_{12}$ toric code, which exhibits a $\ZZ_{12} \times \ZZ_{12}$ fusion rule.
}

\subsection{Topological spins from the T-junction process}\label{sec: topological_spin}

{\color{black}
Given the polynomial representation of stabilizers and anyons, the topological spin of anyons can also be deduced from the polynomial description of anyons \cite{haah_classification_21}. In this context, we reformulate the T-junction process \cite{LW06,alicea2011non, KL20, fidkowski2022gravitational, haah_QCA_23}, as depicted in Fig.~\ref{fig: T junction 3 paths} and described by Eq.~\eqref{eq: statistics formula}, using polynomial formalism.
Given the Pauli operators $P_x^v$ and $P_y^v$, represented in polynomial formalism, which facilitate the movement of anyon $v$ along the $x$- and $y$-directions by $n_x$ and $n_y$ steps respectively, we define extended string operators $u_1, u_2, u_3$ as follows:
\begin{eqs}\label{eq:polynomial_t_junction}
    U_1&=(x^{-q n_x}+ x^{-(q-1) n_x}+\cdots +x^{-n_x} )P_x^v,\\
    U_2&=(1+y^{-n_y}+y^{-2n_y}+\cdots+y^{-q n_y})P_y^v,\\
    U_3&=-(1+x^{n_x}+x^{2 n_x}+\cdots +x^{q n_x})P_x^v,
\end{eqs}
where $U_1$ and $U_3$ are the left- and right-moving string operators along the $x$-direction, respectively, and $U_2$ is the downward-moving string operator along the $y$-direction. Each string has been lengthened by a factor of $q$.
Hence, the topological spin of anyon $v$ can be written as
\begin{eqs}
    \theta(v)&=\lim_{q\rightarrow \infty } \exp \left( \frac{2 \pi i}{d} [U_1,U_2,U_3] \right) \\
    &:= \lim_{q\rightarrow \infty } \exp \left( \frac{2 \pi i}{d} [U_1, U_2]+ [U_2, U_3] + [U_3, U_1] \right),
\end{eqs}
where $[A, B]$ is the $\ZZ_d$ coefficient of $x^0y^0$ in $A^\dagger \Lambda B$. Note that while anyon $v$ is a point-like particle, it possesses a finite size. The limit $q \rightarrow \infty$ is implemented to ensure that the T-junction process yields the correct statistics. However, in practical applications, it is unnecessary to extend this limit to extremely high values of $q$. Instead, we only need to ensure that the lengths $qn_x$ and $q n_y$ are significantly larger than the sizes of the anyons. In this study, we begin our calculations of the braiding statistics at $q=2$ and gradually increase $q$. Once the resulting braiding statistics stabilize and become independent of $q$, we conclude that the selected $q$ value is sufficiently large.
}

\section{Algorithm for extracting topological data}\label{sec: algorithm}

In previous sections, we discussed how to obtain string operators by solving linear equations and how to obtain linearly-independent anyon set by the (modified) Gaussian elimination over fields or non-PID rings. Armed with those techniques, we can develop an algorithm that extracts topological data from a translation invariant Pauli stabilizer model.

In this section, we describe the algorithm that diagnoses and extracts anyonic statistics from the input translation invariant Pauli stabilizer model, which is written as vectors of polynomials over a $R=\ZZ_d[x,y,x^{-1},y^{-1}]$. The output of this algorithm includes the string operators of anyon excitations, their fusion rules, topological spins, and braiding statistics. Note that this algorithm is applicable to nonprime qudit dimensions such as $\mathbb{Z}_8$ or $\mathbb{Z}_9$ qudits.
This algorithm involves two parts:
\begin{enumerate}
    \item Solve string operators for all anyon types.
    \item Obtain fusion rules and topological spins of anyons from given string operators.
\end{enumerate}
It has four adjustable parameters: the truncation degree $k$ for polynomials (in Eq.~\eqref{eq:polynomial_truncation}), range of the translation duplicate map $m$ (in Eq.~\eqref{eq: TD definition}), range of the longest string $N_x, N_y$ (will be discussed in this section), the length of extended string operators $q$ (in Eq.~\eqref{eq:polynomial_t_junction}).

In Sec.~\ref{sec:extract_string}, we provide a detailed description of obtaining anyon string operators by solving linear equations. In Sec.~\ref{sec: algorithm_braiding}, we discuss the second part of this algorithm to extract topological spins and braiding statistics. 

\subsection{Extracting string operators from translation invariant Pauli stabilizers} \label{sec:extract_string}

A prior step of this algorithm is to check whether the translation invariant Pauli stabilizer code has topological order; if so, this algorithm calculates the anyon string operators. The following steps describe the routine of obtaining anyon string operators and their fusion rules.

\begin{itemize}
    \item Step 1: Given input $\mathcal{S}_i, i=1,..., t$, a generating set of stabilizers. Error syndromes $\epsilon(\mX), \epsilon(\mZ)$ can be obtained by Eq.~\eqref{eq: get error syndromes} for a given set of stabilizers $\mathcal{S}_i$.
    \item Step 2: Follow the procedure described in Sec.~\ref{sec: checking TO condition Zd}.
    \item Step 3: {\color{black} Repeat steps 3-1 to 3-3 for $n_x=1,..., N_x$, where $N_x$ are adjustable parameters that determine the longest string we search in the $x$-direction.
    After we sweep from $n_x=1$ to $n_x=N_x$, we will obtain the numbers of anyons for choices of $n_x$. We will choose the smallest value of $n_x$ that reaches the maximum number of anyons.}
    \begin{itemize}
        \item Step 3-1: Use the error syndromes $\epsilon(\mX), \epsilon(\mZ)$ and $(1-x^{n_x})$ to construct a matrix $\wwide{M}_3$ in the form of Eq.~\eqref{eq: M3 matrix}. 
        \item Step 3-2: Depending on prime or nonprime $d$, we perform Gaussian elimination or the modified Gaussian elimination in Sec.~\ref{sec: nonprime-dimensional qudit} on the matrix $\wwide{M}_3$. We obtain the relationship between its rows. Hence, we get the coefficients ($\alpha,\beta,\gamma,\delta$) and anyon $v$ in Eq.~\eqref{eq: anyon equation}.
        \item Step 3-3: We have a set of anyons with redundancy. For example, it includes an anyon spanned by others (up to syndromes of local Pauli operators.). To obtain the most compact form of anyons, we would like to keep the basis anyons and exclude others, following the procedure in Sec.~\ref{sec: Smith normal form}. We first use the matrix $\wwide{M}_1$ formed by the error syndromes. For the matrix $\wwide{M}_1$ over $\ZZ_d$, when $d$ is prime, we sequentially insert anyon row vectors into the matrix one by one and check the rank of the matrix after insertion (or use Theorem~\ref{thm: check span} for nonprime $d$). If the rank of the matrix is unchanged after adding the new anyon (or the anyon is spanned by the matrix), then this newly added anyon is generated by the existing anyons. \textcolor{black}{The new anyon that can be jointly generated by existed anyons and local syndromes will be neglect and we proceed with the rest of anyons such that only inequivalent anyons will remain.} Finally, we compute the Smith normal form of the anyon relation matrix to obtain the basis anyons and their fusion rules.\footnote{Since calculating the Smith normal form for a large matrix is challenging, in practice, we can first repeat the procedure for prime-dimensional qudits to refine the anyons, and further calculate the Smith normal form.}
    \end{itemize}
    \item Step 4: Using the coefficients of the error syndrome polynomials in Eq.~\eqref{eq: anyon equation}, we can obtain the string operators in the $x$-direction. More explicitly, from a solution of the equation
    \begin{eqs}
        &(1-x^{n_x})v \\
        =&~p_1(x,y)\eps(\mX_1) + p_2(x,y)\eps(\mX_2) \\
        &+ p_3(x,y)\eps(\mZ_1) + p_4(x,y)\eps(\mZ_2),
    \end{eqs}
    the string operator $P_x$ is obtained as
    \begin{eqs}
        P_x = \left[
        \def\arraystretch{1.2}
        \begin{array}{c}
            p_1(x,y) \\
            p_2(x,y) \\
            \hline
            p_3(x,y) \\
            p_4(x,y)
        \end{array}\right].
    \end{eqs}
    \item Step 5: From anyon $v$ solved in the $x$-direction, we utilize the matrix $\mathrm{MGE}(\wwide{M}_1)$ (modified Gaussian elimination of $\wwide{M}_1$ in Eq.~\eqref{eq: M1 matrix Zd}) and Theorem~\ref{thm: check span} to solve the anyon equation in the $-y$-direction:
    \begin{eqs}
        &(1-y^{-n_y})v \\
        =&~p_5(x,y)\eps(\mX_1) + p_6(x,y)\eps(\mX_2) \\
        & +p_7(x,y)\eps(\mZ_1) +p_8(x,y)\eps(\mZ_2).
    \end{eqs}
    From this solution, we obtain the string operator $P_y$ in the $-y$-direction
    \begin{eqs}
        P_y = \left[
        \def\arraystretch{1.2}
        \begin{array}{c}
            p_5(x,y) \\
            p_6(x,y) \\
            \hline
            p_7(x,y) \\
            p_8(x,y)
        \end{array}\right].
    \end{eqs}
    {\color{black} The $P_x$ and $P_y$ are shown in Fig.~\ref{fig:t_junction}(a).}
\end{itemize}
Now, we have obtained the string operators of the basis anyons.

\subsection{Extracting topological spins of anyons} \label{sec: algorithm_braiding}

In Sec.~\ref{sec:extract_string}, we have demonstrated how to extract anyon string operators by solving the anyon equation~\eqref{eq: anyon equation}, we can obtain the string operators of all anyons, labeled $P_x^i$ and $P_y^i$, which represent the string operator of the $i$-th anyon in the $x$-direction and $-y$-direction, respectively. To obtain the topological spin between all anyons, we perform the following steps:
\begin{enumerate}
    \item Obtain the topological spin $\theta$ for each basis anyon by using the T-junction process in Sec.~\ref{sec: topological_spin}.
    \item Construct the string operator for anyon $v_i \times v_j$ where $v_i$ and $v_j$ are basis anyons. For example, for anyon $v_1$ and anyon $v_2$, we first construct string operators that move $v_1 \times v_2$ in the $x$- and $-y$-directions of anyon $v_1 \times v_2$ as $P_x^{v_1 \times v_2} = P_x^{v_2} P_x^{v_1}$ and $P_y^{v_1 \times v_2} = P_y^{v_2} P_y^{v_1}$. Then, we substitute $P_x^{v_1\times  v_2}$ and $ P_y^{v_1\times  v_2}$ in Eq.~\eqref{eq:polynomial_t_junction} to obtain long string operators $U_1, U_2, U_3$.
    Because anyon $v_1\times v_2$ may have non-negligible size, we use the extended string operators $U_1, U_2, U_3$ as shown in Fig.~\ref{fig:t_junction} such that the braiding process is unaffected by the size of anyons. We have extended the length of string operators in $x$- and $y$-directions to $q$-times of the original length.
    \item {\color{black} Finally, we can calculate the topological spin of anyon $v_1\times v_2 $:
    \begin{equation*}
        \begin{split}
            &\quad\theta(v_1 \times v_2)= \exp \left( \frac{2 \pi i}{d} [U_1,U_2,U_3] \right) \\
            &= \exp \left( \frac{2 \pi i}{d} [U_1, U_2]+ [U_2, U_3] + [U_3, U_1] \right).
        \end{split}
    \end{equation*}
    Using the fact that the mutual braiding braiding $B(v_1, v_2)=\frac{\theta(v_1 \times v_2)}{\theta(v_1) \theta(v_2)}$ (see Appendix~\ref{appendix: braiding_statistics}), we can obtain the braiding between anyons $v_1$ and $v_2$.
    }
\end{enumerate}

\begin{figure*}[htb]
    \centering
    \includegraphics[width=0.95\textwidth]{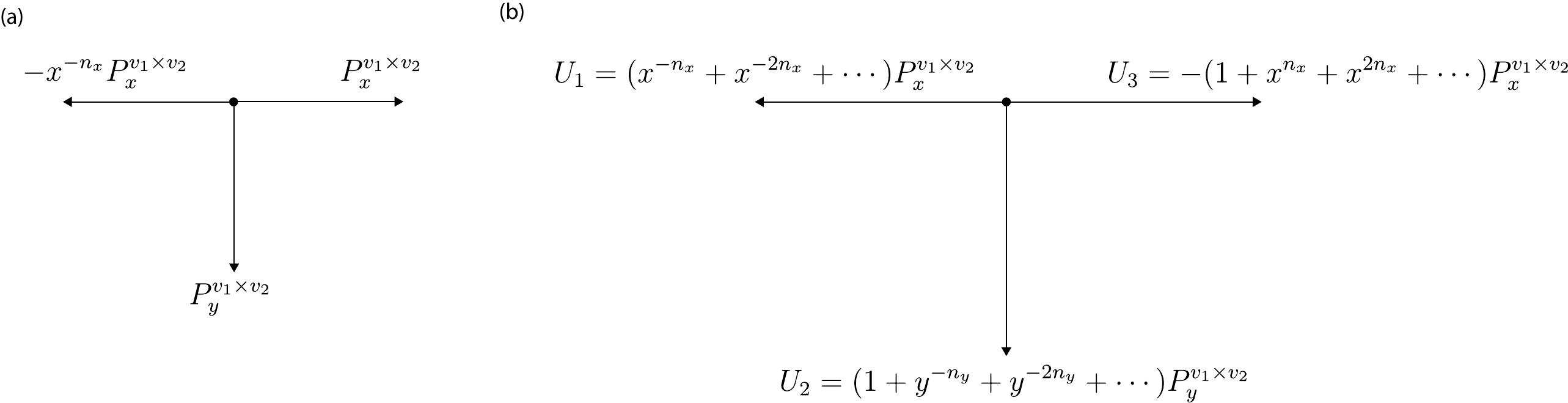}
    \caption{(a) The T-junction process for short string operators. Anyon $v_1\times v_2$ is the composite of two anyons $v_1$ and $v_2$, which might have a finite size larger than the short string, so the process might fail to detect the correct statistics of $v_1\times v_2$. (b) To avoid the size effect of anyon $v_1 \times v_2$, we extend $P_x$ and $P_y$ to the long string operators $U_1 ,U_2 ,U_3$, which is $q$ times longer.}
\label{fig:t_junction}
\end{figure*}

\begin{table}[htb]
    \centering
    \begin{tabular}{|c|c|c|c|c|c|}
    \hline
    &$v_1$& $v_2$& $v_3$& $v_4$ &$\cdots$\\
    \hline
    $v_1$& $\theta(v_1)$& $B(v_1,v_2)$& $B(v_1,v_3)$&$B(v_1,v_4)$&$\cdots$\\
    \hline
    $v_2$& $B(v_2,v_1)$& $\theta(v_2)$& $B(v_2,v_3)$&$B(v_2,v_4)$&$\cdots$\\
    \hline
    $v_3$& $B(v_3,v_1)$& $B(v_3,v_2)$& $\theta(v_3)$& $B(v_3,v_4)$&$\cdots$\\
    \hline
    $v_4$& $B(v_4,v_1)$& $B(v_4,v_2)$& $B(v_4,v_3)$&$\theta(v_4)$&$\cdots$\\
    \hline
    $\vdots$ & $\vdots$ & $\vdots$ & $\vdots$ & $\vdots$ & $\ddots$\\
    \hline
    \end{tabular}
    \caption{Topological spins and braiding statistics of basis anyons $\{v_1, v_2, v_3, v_4, \cdots\}$.}
    \label{tab: Braiding statistics of anyons 1}
\end{table}

\subsection{Rearranging the basis anyons}

This section demonstrates a method that rearranges the basis anyons of a $\ZZ_p$ (prime $p$) Pauli stabilizer code to decoupled $\{e,m\}$ pairs in finite copies of $\ZZ_p$ toric codes. For example, the \textcolor{black}{2d honeycomb} color code is two copies of $\mathbb{Z}_2$ toric codes, and we aim to find the decoupled pairs $\{e_1,m_1\}$ and $\{e_2,m_2\}$ from the basis anyons we found. The procedure is as follows
\begin{enumerate}
\item  Start from the entire anyon set $\mathcal{V}$ and compute all topological spins. {\color{black} We search for a nontrivial boson $b$, defined as a bosonic anyon that is not equivalent to the trivial anyon (the identity anyon) ${\bf 0} := [0, 0, \cdots 0]$. The existence of such a nontrivial boson is proven in Ref.~\cite{haah_QCA_23}.} The order of $b$ must be $p$.

\item We find another anyon $c$ such that 
\begin{eqs}
    \textcolor{black}{B(b,c)=\exp (\frac{2\pi i}{p})}.
\end{eqs}
The existence of such $c$ is proved in Ref.~\cite{Ellison2023paulitopological}.
{\color{black}
Without any loss of generality, we can assume $c$ to be a boson. If $c$ is not be a boson, i.e., $\theta(c) = \exp({2\pi i n_c}/{p})$, we can redefine it as $c' = c \times b^{-n_c}$. This ensures that $c'$ is a boson, with $\theta(c')=\theta(c) B(c,b)^{-n_c}=1$, while maintaining the same braiding properties with $b$.}

\item Start from $\{b, c\}$ and follow the procedure in Sec.~\ref{sec: Smith normal form} to obtain the basis anyons set:
\begin{eqs}
    V_{\mathrm{basis}}=\{ b, c, v_1, v_2, v_3, \cdots\}.
\end{eqs}
{\color{black}
We can further assume that all $v_j$ have no braiding with $b$. If a $v_j$ exhibits braiding with $b$, $B(v_j, b) = \exp({2\pi i n^{(b)}_j}/{p})$, it can be redefined as $v_j \times c^{-n^{(b)}_j}$. Similarly, if $B(v_j, c) = \exp({2\pi i n^{(c)}_j}/{p})$, $v_j$ can be redefined as $v_j \times b^{-n^{(c)}_j}$, ensuring that $v_j$ also has no braiding with $c$. Consequently, the pairs ${b, c}$ are effectively decoupled from the other basis anyons, denoted as $\{e_1, m_1\}$.
}

\item Consider a new set of anyons generated by $\{v_1, v_2, v_3, \cdots \}$.
Repeat the above steps to obtain $\{e_1, m_1\}$ decoupled from the remaining basis anyons.
Continue until all $\{e, m\}$ pairs are found.
\end{enumerate}
After the rearrangement of the basis anyons, the table of topological spins and braiding statistics, defined in Table~\ref{tab: Braiding statistics of anyons 1}, has the specific form:
\begin{equation}
    \begin{tabular}{|c|c|c|c|c|c|c|c|}
    \hline
    &$e_1$& $m_1$& $e_2$& $m_2$  & $e_3$& $m_3$ &$\cdots$\\
    \hline
    $e_1$& $1$ & \color{red}$-1$ & $1$&$1$& $1$&$1$& $\cdots$\\
    \hline
    $m_1$& \color{red}$-1$& $1$& $1$&$1$& $1$&$1$&$\cdots$\\
    \hline
    $e_2$& $1$& $1$& $1$& \color{red}$-1$&$1$&$1$& $\cdots$\\
    \hline
    $m_2$& $1$& $1$& \color{red}$-1$&$1$& $1$&$1$&$\cdots$\\
    \hline
    $e_3$& $1$& $1$& $1$& $1$&$1$&\color{red}$-1$& $\cdots$\\
    \hline
    $m_3$& $1$& $1$& $1$&$1$& \color{red}$-1$&$1$&$\cdots$\\
    \hline
    $\vdots$ & $\vdots$ & $\vdots$ & $\vdots$ & $\vdots$ & $\vdots$ &$\vdots$ &$\ddots$\\
    \hline
    \end{tabular}~.
\end{equation}

\section{Applications on various quantum codes}\label{sec: application_to_codes}

In this section, we apply the algorithm to various quantum codes to test the TO condition and extract their topological data (if there exist). We start from the modified color codes in Sec.~\ref{sec:modified_color_code}, which are self-dual CSS codes on $\mathbb{Z}_2$ qubits. According to Ref.~\cite{haah_classification_21}, they can be decomposed into finite copies of $\mathbb{Z}_2$ toric codes. Our algorithm confirms this and finds the $\{e,m\}$ pairs in the decoupled copies of the $\mathbb{Z}_2$ toric codes.
In Sec.~\ref{sec:css_double_semion}, we investigate anyons of CSS codes on $\mathbb{Z}_4$ qudits {\color{black}induced} from the double semion code \cite{liu2023subsystem}. We confirm the conjecture in Ref.~\cite{liu2023subsystem} that this code is two copies of $\mathbb{Z}_2$ toric code and find the decoupled $\{e,m\}$ pairs.
In Sec.~\ref{sec:six_semion}, we investigate a model called the six-semion code defined on $\mathbb{Z}_4 \times \mathbb{Z}_4$ qudits whose basis anyons are $v_1$ and $v_2$ with topological spin $-i$ that indicates that $v_1$ and $v_2$ are (anti-)semions. The mutual braiding between them gives the $i$ phase. 
The topological orders detected for these examples are summarized in Table.~\ref{tab:example_topological_order}.

\begin{table*}[t]
    \centering
    \begin{tabular}{|c|c|}
    \hline
         model & Topological order  \\
         \hline
        \textcolor{black}{2d honeycomb} color code (example $\circled{1}$)&  Two copies of $\mathbb{Z}_2$ toric codes\\
         \hline
         Example $\circled{2}$ &  Does not satisfy TO condition \\
         \hline
         Modified color code A (example $\circled{3}$) & Four copies of $\mathbb{Z}_2$ toric codes\\
         \hline
         Modified color code B (example $\circled{4}$)& Eight copies of $\mathbb{Z}_2$ toric codes\\
         \hline
         Modified color code C (example $\circled{5}$)& Four copies of $\mathbb{Z}_2$ toric codes\\
         \hline
         Modified color code D (example $\circled{6}$)& Six copies of $\mathbb{Z}_2$ toric codes\\
         \hline
         $\ZZ_4$ CSS code induced from double semion& Two copies of $\mathbb{Z}_2$ toric codes \\
         \hline
         Double semion code with $\ZZ_4$ qudits & Double semion topological order\\
         \hline
         Six-semion code with $\ZZ_4 \times \ZZ_4$ qudits& Six-semion topological order\\
         \hline
    \end{tabular}
    \caption{The topological orders of various examples in Fig.~\ref{fig: modified color codes} and, obtained from the algorithm described in Sec.~\ref{sec: algorithm}.Here, we check the \textcolor{black}{2d honeycomb} color code and double semion code with $\mathbb{Z}_4$ qudits (details in Appendix.~\ref{appendix: string_operators}), which serve as sanity checks. }
    \label{tab:example_topological_order}
\end{table*}

\subsection{Modified color code}\label{sec:modified_color_code}

We construct six examples of self-dual CSS codes on the 2d honeycomb lattice, denoted as modified color codes. Their stabilizer terms are shown in Fig.~\ref{fig: modified color codes}. Although Ref.~\cite{haah_classification_21} has shown that all the 2d translation invariant Pauli stabilizer models for prime-dimensional qudits can be decomposed into finite copies of toric codes, we do not have prior knowledge about the number of copies of toric codes that these modified color codes can be decomposed to. We use these modified color codes as a showcase to demonstrate the effectiveness of our algorithm in determining the topological data. For examples that satisfy the TO condition, the algorithm finds string operators, fusion rules, topological spins, and braiding statistics of anyons.

\begin{widetext}

The stabilizers and syndromes of those six examples are 
\begin{eqs}
    \mathcal{S}_1^{\circled{1}}=\left[\begin{array}{c}
        1+\overline{x}+y\\
        1+\overline{y}+x\\
        \hline
        0\\
        0
    \end{array}\right],\quad \mathcal{S}_2^{\circled{1}}=\left[\begin{array}{c}
         0  \\
         0\\
         \hline
         1+\overline{x}+y\\
         1+\overline{y}+x
    \end{array}\right],
\end{eqs}
\begin{eqs}
    &\epsilon(\mX_1)^{\circled{1}}=[0,1+x+\overline{y}],~ \epsilon(\mX_2)^{\circled{1}}=[0,1+ y  + \overline{x}],\\
     &\epsilon(\mZ_1)^{\circled{1}}=[1+x + \overline{y}, 0],~ \epsilon(\mZ_2)^{\circled{1}}=[1+ y + \overline{x}, 0],
\end{eqs}
\begin{eqs}
    \mathcal{S}_1^{\circled{2}}=\left[\begin{array}{c}
        1+\overline{x}+y+\overline{x} \overline{y}\\
        1+\overline{y}+x+xy\\
        \hline
        0\\
        0
    \end{array}\right], \mathcal{S}_2^{\circled{2}}=\left[\begin{array}{c}
         0  \\
         0\\
         \hline
         1+\overline{x}+y+\overline{x} \overline{y}\\
        1+\overline{y}+x+xy
    \end{array}\right],
\end{eqs}
\begin{eqs}
    &\epsilon(\mX_1)^{\circled{2}}=[0,1+x + \overline{y} + xy],~ \epsilon(\mX_2)^{\circled{2}}=[0,1+y  + \overline{x} + \overline{xy}],\\
    &\epsilon(\mZ_1)^{\circled{2}}=[1+x + \overline{y} + xy, 0],~ \epsilon(\mZ_2)^{\circled{2}}=[1+y  + \overline{x} + \overline{xy}, 0],
\end{eqs}
\begin{eqs}
    \mathcal{S}_1^{\circled{3}}=\left[\begin{array}{c}
        1+\overline{x}+y+\overline{x}\overline{y}+xy\\
        1+\overline{y}+x+\overline{xy}+xy\\
        \hline
        0\\
        0
    \end{array}\right]
    ,\quad \mathcal{S}_2^{\circled{3}}=\left[\begin{array}{c}
         0  \\
         0\\
         \hline
         1+\overline{x}+y+\overline{x}\overline{y}+xy\\
         1+\overline{y}+x+\overline{xy}+xy
    \end{array}\right],
\end{eqs}
\begin{eqs}
    &\epsilon(\mX_1)^{\circled{3}}=[0, 1 +x + \overline{y} + xy + \overline{xy}],
    ~ \epsilon(\mX_2)^{\circled{3}}=[0, 1 +y + \overline{x} + xy + \overline{xy}],\\
    &\epsilon(\mZ_1)^{\circled{3}}=[1+x  + \overline{y} + xy + \overline{xy}, 0],~ \epsilon(\mZ_2)^{\circled{3}}=[1+y + \overline{x} + xy + \overline{xy}, 0],
\end{eqs}
\begin{eqs}
    \mathcal{S}_1^{\circled{4}}=\left[\begin{array}{c}
        1+\overline{x}+y+\overline{xy}+xy+\overline{x}y\\
        1+\overline{y}+x+\overline{xy}+xy+x\overline{y}\\
        \hline
        0\\
        0
    \end{array}\right]
    ,\quad \mathcal{S}_2^{\circled{4}}=\left[\begin{array}{c}
         0  \\
         0\\
         \hline
         1+\overline{x}+y+\overline{xy}+xy+\overline{x}y\\
        1+\overline{y}+x+\overline{xy}+xy+x\overline{y}
    \end{array}\right],
\end{eqs}
\begin{eqs}
    &\epsilon(\mX_1)^{\circled{4}}=[0, 1 +x + \overline{y} + xy+ \overline{xy} + x\overline{y}],
    ~ \epsilon(\mX_2)^{\circled{4}}=[0,1 +y +  \overline{x} + xy + \overline{xy} + \overline{x}y],\\
    &\epsilon(\mZ_1)^{\circled{4}}=[ 1+x + \overline{y} + xy+ \overline{xy} + x\overline{y}, 0],
    ~ \epsilon(\mZ_2)^{\circled{4}}=[1+y  + \overline{x} + xy + \overline{xy} + \overline{x}y, 0],
\end{eqs}
\begin{eqs}
    \mathcal{S}_1^{\circled{5}}=\left[\begin{array}{c}
        1+\overline{x}+y+y^2\\
        1+\overline{y}+x+\overline{y}^2\\
        \hline
        0\\
        0
    \end{array}\right]
    ,\quad \mathcal{S}_2^{\circled{5}}=\left[\begin{array}{c}
         0  \\
         0\\
         \hline
         1+\overline{x}+y+y^2\\
        1+\overline{y}+x+\overline{y}^2\\
    \end{array}\right],
\end{eqs}
\begin{eqs}
    &\epsilon(\mX_1)^{\circled{5}}=[0,1+x + \overline{y} + \overline{y}^2],
    ~ \epsilon(\mX_2)^{\circled{5}}=[0,1+y  + \overline{x} + y^2],\\
    &\epsilon(\mZ_1)^{\circled{5}}=[ 1+x + \overline{y} + \overline{y}^2, 0],
    ~ \epsilon(\mZ_2)^{\circled{5}}=[1+y  + \overline{x} + y^2, 0],
\end{eqs}
\begin{eqs}
    \mathcal{S}_1^{\circled{6}}=\left[\begin{array}{c}
        1+\overline{x}+y+xy^3\\
        1+\overline{y}+x+\overline{x}\overline{y}^3\\
        \hline
        0\\
        0
    \end{array}\right]
    ,\quad \mathcal{S}_2^{\circled{6}}=\left[\begin{array}{c}
         0  \\
         0\\
         \hline
         1+\overline{x}+y+xy^3\\
        1+\overline{y}+x+\overline{x}\overline{y}^3\\
    \end{array}\right],
\end{eqs}
\begin{eqs}
    &\epsilon(\mX_1)^{\circled{6}}=[0,1+x + \overline{y} + \overline{x}\overline{y}^3],
    ~ \epsilon(\mX_2)^{\circled{6}}=[0,1+y  + \overline{x}+ xy^3],\\
    &\epsilon(\mZ_1)^{\circled{6}}=[ 1+x + \overline{y} + \overline{x}\overline{y}^3, 0],
    ~ \epsilon(\mZ_2)^{\circled{6}}=[1+y  + \overline{x}+ xy^3, 0].
\end{eqs}

\end{widetext}

{\color{black} Upon inputting the stabilizer and syndrome polynomials into the algorithm, we found that example \circled{2} does not exhibit topological order, whereas the others do. The operator described in Fig.~\ref{fig:modified_color_code_2_TO_condition} commutes with all stabilizers in example \circled{2}, yet it cannot be generated by those stabilizers:
}
\begin{figure}[htb]
    \centering
    \includegraphics[width=0.35\textwidth]{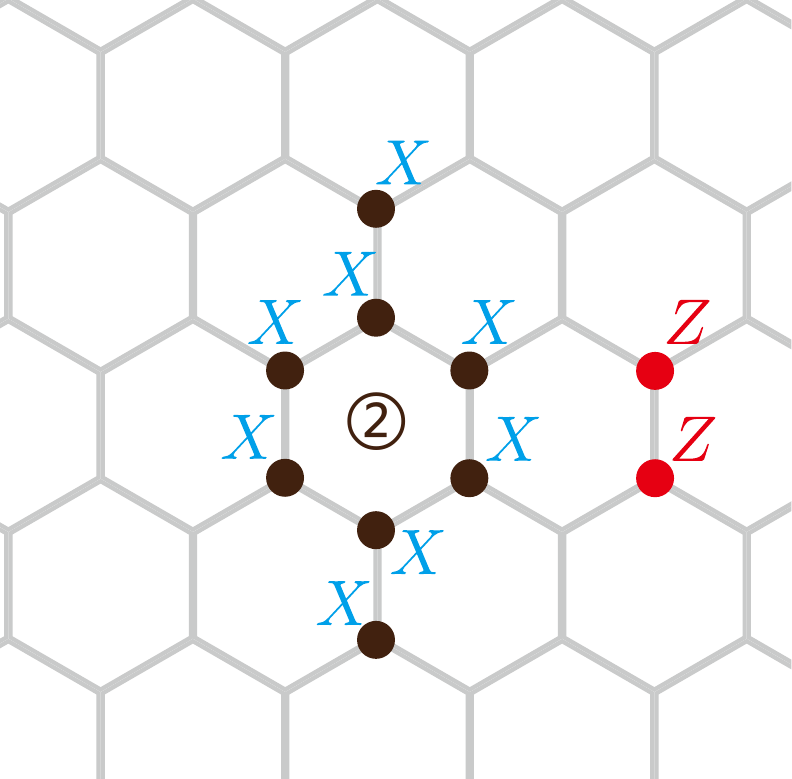}
    \caption{The local $Z$-operator, marked in red, commutes with all the $X$-stabilizers shown on the left. It has been verified that this local $Z$-operator cannot be generated from the stabilizers.}
\label{fig:modified_color_code_2_TO_condition}
\end{figure}

The algorithm provides the string operators of basis anyons for examples that satisfy topological order conditions, with pictorial descriptions shown in Appendix~\ref{appendix: string_operators}. The algorithm then performs the T-junction calculation and obtains the topological spins and braiding statistics for the examples \circled{1}, \circled{3}, \circled{4}, \circled{5}, \circled{6}. \textcolor{black}{We confirm those modified color codes with topological orders can always be written as copies of toric codes. Table~\ref{tab:Model 1} shows the anyon pairs we found match the analytical result of the 2d honeycomb color code. The braiding statistics of examples \circled{3}, \circled{4}, \circled{5}, \circled{6} are shown in Tables.~\ref{tab:Model 3}, \ref{tab:Model 4}, \ref{tab:Model 5}, \ref{tab:Model 6} in Appendix~\ref{appendix: braiding_tables}.}

\textcolor{black}{Here we present the correspondance between decoupled $\{e,m\}$ pairs in finite copies of toric code and anyons of color code examples.}
\begin{itemize}
    \item \textcolor{black}{For 2d honeycomb color code (example \textcircled{1}), we confirm that it can be decomposed into two copies of toric code where we find two pairs of anyons $\{v_1,v_2\},\{v_3,v_4\}$ correspond to two decoupled $\{e,m\} $ pairs of two-copies of toric code. 
    \item For modified color code (example \textcircled{3}), we confirm that it can be decomposed into four copies of toric code where we find four pairs of anyons $\{v_1, v_8\}, \{ v_2, v_4\}, \{v_3, v_6\}, \{v_5,v_7\}$ correspond to four decoupled $\{e,m\} $ pairs of four copies of toric code.
    \item For modified color code (example \textcircled{4}), we confirm that it can be decomposed into eight copies of toric code where we find eight pairs of anyons $\{v_1, v_9\}, \{ v_2, v_{11}\}, \{v_3, v_{13}\}$, $\{v_4,v_{12}\}, \{v_5,v_{14}\}, \{v_6,v_{10}\}, \{v_7,v_{15}\}, \{v_8,v_{16}\}$ correspond to eight decoupled $\{e,m\} $ pairs of eight-copies of toric code.
    \item For modified color code (example \textcircled{5}), we confirm that it can be decomposed into four copies of toric code where we find four pairs of anyons $\{v_1, v_4\}, \{ v_2, v_6\}, \{v_3, v_7\}, \{v_5,v_8\}$ correspond to four decoupled $\{e,m\} $ pairs of four-copies of toric code.
    \item For modified color code (example \textcircled{6}), we confirm that it can be decomposed into six copies of toric code where we find six pairs of anyons $\{v_1, v_4\}, \{ v_2, v_9\}, \{v_3, v_7\}, \{v_5,v_{10}\},  \{v_6,v_{11}\},  \{v_8,v_{12}\}$ correspond to six decoupled $\{e,m\} $ pairs of six-copies of toric code.}
\end{itemize}
This matches the theorem \cite{haah_classification_21} that 2d $\mathbb{Z}_2$ topological Pauli stabilizer codes can always be decomposed to copies of $\mathbb{Z}_2$ toric codes.

\begin{table}[htb]
    \centering
    \begin{tabular}{|c|c|c|c|c|}
    \hline
         &$v_1$& $v_2$& $v_3$& $v_4$\\
         \hline
         $v_1$& 1& \textcolor{red}{-1}& 1&1\\
         \hline
         $v_2$& \textcolor{red}{-1}& 1& 1&1\\
         \hline
         $v_3$& 1& 1& 1& \textcolor{red}{-1}\\
         \hline
         $v_4$& 1& 1& \textcolor{red}{-1}&1\\
         \hline
    \end{tabular}
    \caption{Topological spins and braiding statistics of anyons for the \textcolor{black}{2d honeycomb} color code (example \textcircled{1}). We can see this table is formed by two decoupled copies of toric code, where $\{v_1, v_2\}$ and $\{ v_3, v_4\}$ correspond to $\{e_1, m_1 \}$ and $\{e_2, m_2\}$ of two copies of toric code.}
    \label{tab:Model 1}
\end{table}

\subsection{CSS codes induced from double semion}\label{sec:css_double_semion}

In Sec.~\ref{sec:modified_color_code}, we have shown that our algorithm gives us the expected result for 2d $\mathbb{Z}_2$ topological codes whose local dimension is a prime integer. In this subsection, we apply our algorithm to a more challenging case, which is a topological Pauli stabilizer code defined on nonprime-dimensional qudits, and show its capability of finding the anyons for a topological Pauli stabilizer code with nonprime qudits.

There is a recent paper \cite{liu2023subsystem} discussing the mapping between \textcolor{black}{qudit} Pauli stabilizer codes to CSS codes. The authors map the double-semion Pauli stabilizer code that is a non-CSS code on $\ZZ_4$ qudits to a CSS code by doubling the number of qudits. This code was conjectured to be equivalent to two copies of $\mathbb{Z}_2$ toric codes. \textcolor{black}{Our algorithm confirms the conjecture.}

The stabilizers of this CSS code are
\begin{equation}\label{eq:CSS_double_semion_stb}
    \begin{gathered}
    \mathcal{S}_1 = \vcenter{\hbox{\includegraphics[scale=.25]{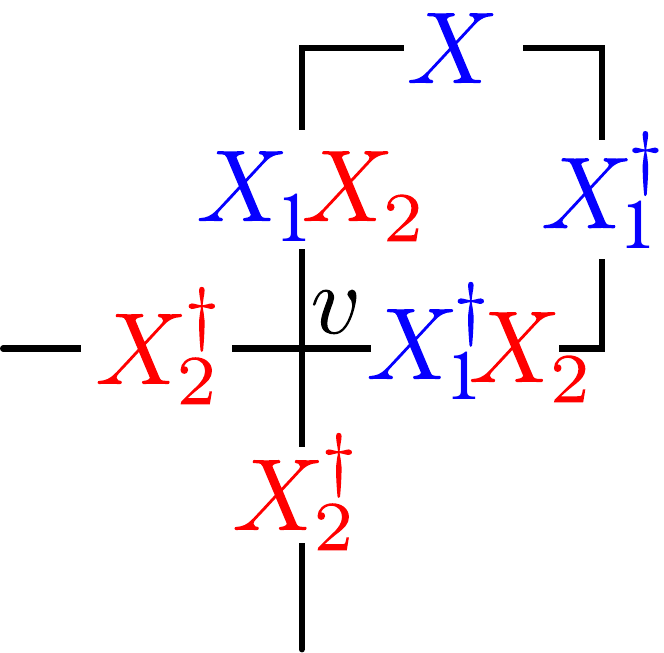}}}, \quad
    \mathcal{S}_2 = \vcenter{\hbox{\includegraphics[scale=.25]{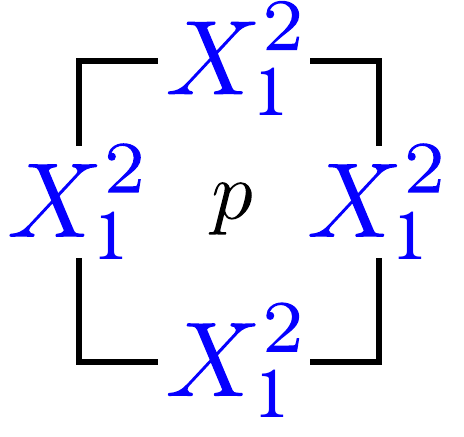}}}, \\
    \mathcal{S}_3 = \vcenter{\hbox{\includegraphics[scale=.25]{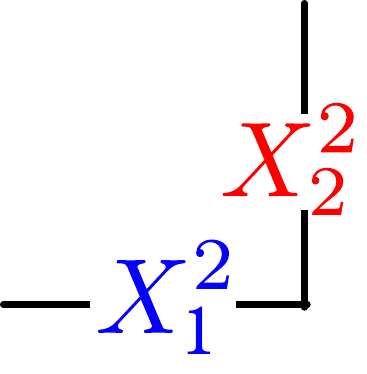}}}, \quad
    \mathcal{S}_4 = \vcenter{\hbox{\includegraphics[scale=.25]{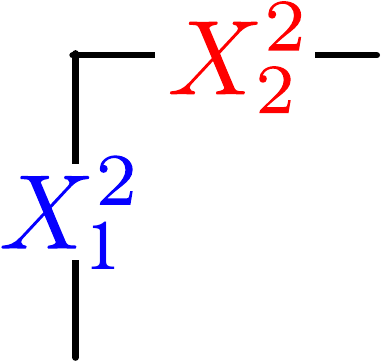}}}.\\
    \mathcal{S}_5 = \vcenter{\hbox{\includegraphics[scale=.25]{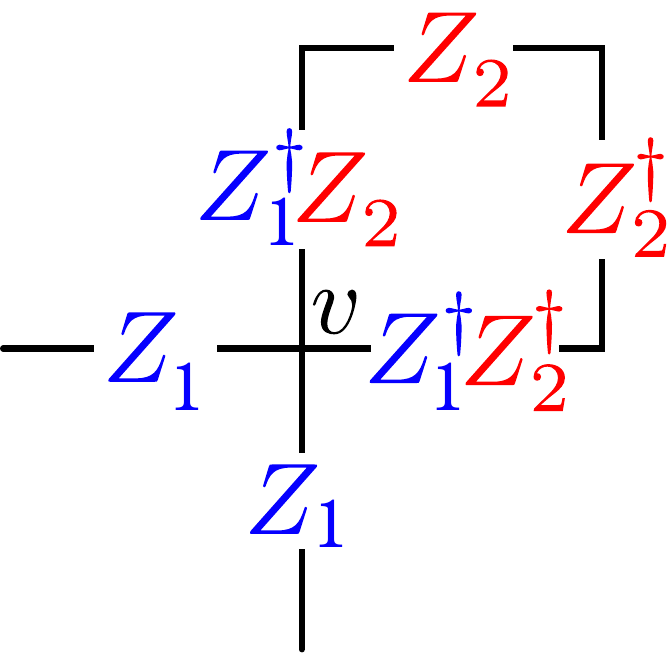}}}, \quad
    \mathcal{S}_6 = \vcenter{\hbox{\includegraphics[scale=.25]{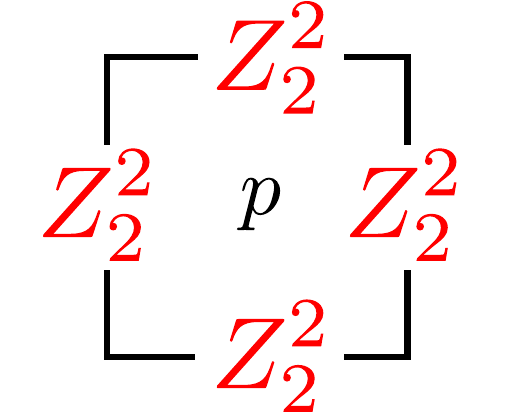}}}, \\
    \mathcal{S}_7 = \vcenter{\hbox{\includegraphics[scale=.25]{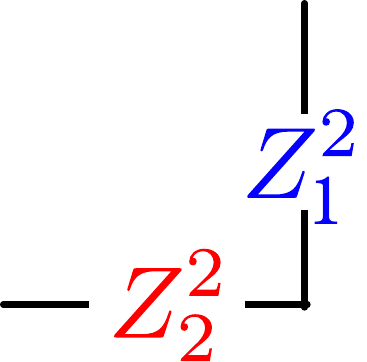}}}, \quad
    \mathcal{S}_8 = \vcenter{\hbox{\includegraphics[scale=.25]{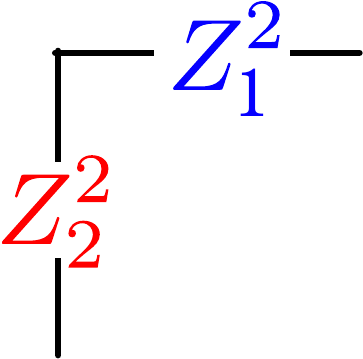}}}.
    \end{gathered}
\end{equation}
In Eq.~\eqref{eq:CSS_double_semion_stb}, we assign two $\ZZ_4$ qudits on each edge whose Pauli matrices are colored blue and red, respectively.

We utilize our algorithm to investigate the anyons in the CSS code and find anyon string operators, fusion rules, topological spins, and braiding statistics. The topological spins and braiding statistics of the basis anyons are shown in Table.~\ref{tab:CSS codes induced from double semion}, which shows there are two decoupled $\{e,m\}$ pairs: $\{v_1, v_3\}$ and $\{v_2, v_4\}$. Moreover, $2v_1 \sim 2 v_2 \sim 2 v_3 \sim 2 v_4 \sim \bf 0$ indicates that all basis anyons have order 2. Hence, this model is two copies of $\mathbb{Z}_2$ toric codes.

\begin{table}[htb]
    \centering
    \begin{tabular}{|c|c|c|c|c|}
    \hline
         &$v_1$& $v_2$& $v_3$& $v_4$\\
         \hline
         $v_1$& 1& 1& \textcolor{red}{-1}&1\\
         \hline
         $v_2$& 1& 1& 1&\textcolor{red}{-1}\\
         \hline
         $v_3$& \textcolor{red}{-1}& 1& 1& 1\\
         \hline
         $v_4$& 1& \textcolor{red}{-1}& 1&1\\
         \hline
    \end{tabular}
    \caption{Topological spins and braiding statistics of anyons for the CSS codes induced from double semion code. We can see this table is formed by two decoupled copies of toric code, where $\{v_1, v_3\}$ and $\{ v_2, v_4\}$ correspond to $\{e_1, m_1 \}$ and $\{e_2, m_2\}$ in the two copies of toric codes.}
    \label{tab:CSS codes induced from double semion}
\end{table}

\subsection{Six-semion stabilizer code}\label{sec:six_semion}

The six-semion stabilizer code has stabilizers
\begin{eqs}\label{eq:six_semion_stb}
    &\mathcal{S}_1=\begin{gathered}
\includegraphics[scale=0.25]{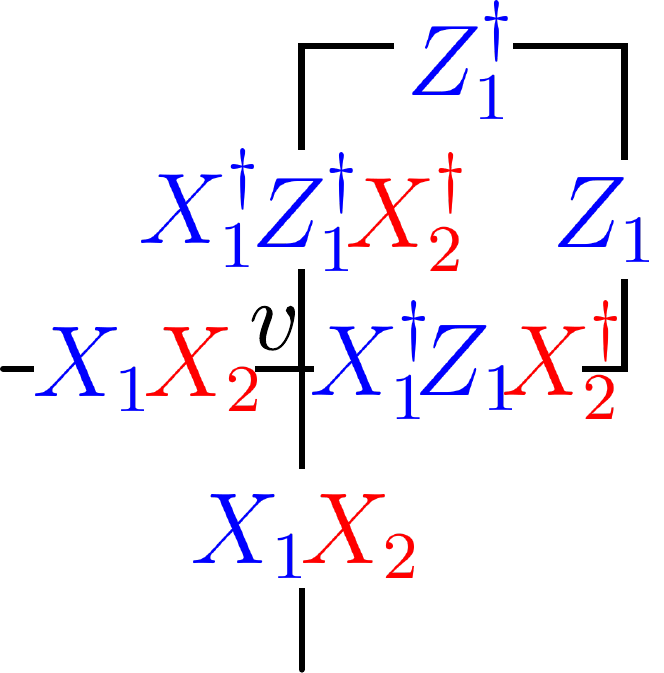}
    \end{gathered},~ \mathcal{S}_2=\begin{gathered}
\includegraphics[scale=0.25]{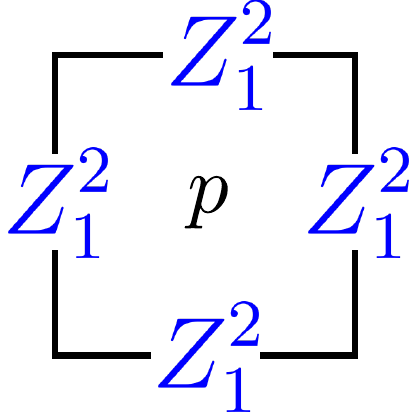}
    \end{gathered},\\
    &\mathcal{S}_3=\begin{gathered}
\includegraphics[scale=0.25]{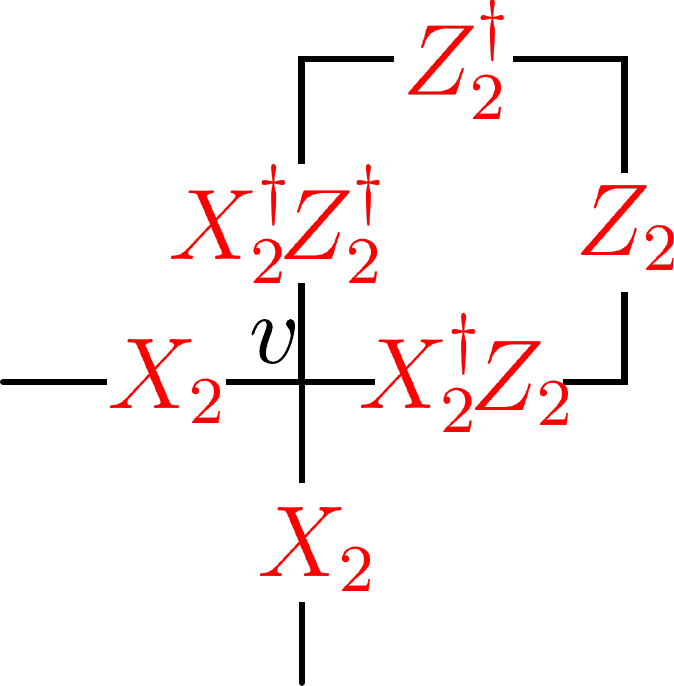}
    \end{gathered},~ \mathcal{S}_4=\begin{gathered}
\includegraphics[scale=0.25]{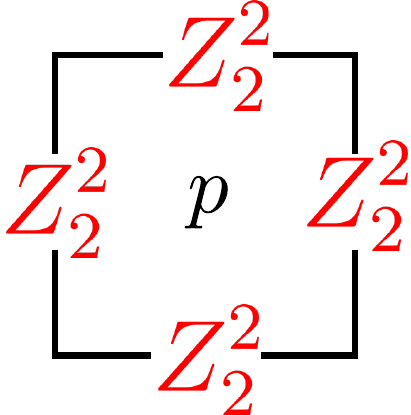}
    \end{gathered},\\
    &\mathcal{S}_5=\begin{gathered}
\includegraphics[scale=0.25]{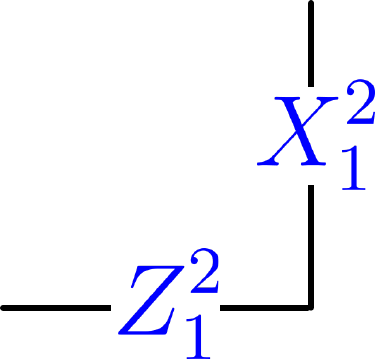}
    \end{gathered},~ \mathcal{S}_6=\begin{gathered}
\includegraphics[scale=0.25]{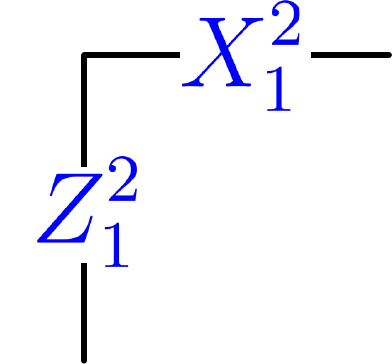}
    \end{gathered},\\
    &\mathcal{S}_7=\begin{gathered}
\includegraphics[scale=0.25]{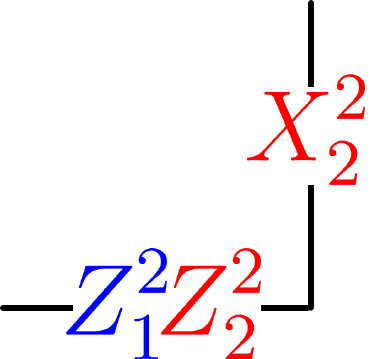}
    \end{gathered},~ \mathcal{S}_8=\begin{gathered}
\includegraphics[scale=0.25]{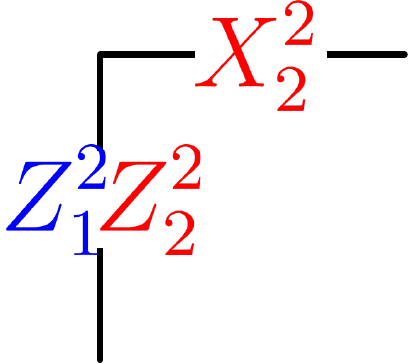}
    \end{gathered}.
\end{eqs}
This model can be regarded as condensing $e_1^2 m_1^2$ and $e_1^2 e_2^2 m_2^2$ anyons in two copies of $\mathbb{Z}_4$ toric codes \cite{ellison2022pauli}. The $\mathcal{S}_5, \mathcal{S}_6, \mathcal{S}_7, \mathcal{S}_8$ are the condensing terms. The first four terms of Eq.~\eqref{eq:six_semion_stb} come from the stabilizer group of two copies of $\mathbb{Z}_4$ toric code. The anyon theory of this model includes four bosons, six semions, and six anti-semions.
\begin{table}[h]
    \centering
    \begin{tabular}{|c|c|c|}
    \hline
         &$v_1$& $v_2$\\
         \hline
         $v_1$& $-i$& $i$\\
         \hline
         $v_2$& $i$& $-i$\\\hline
    \end{tabular}
    \caption{\color{black} Topological spins and braiding statistics of basis anyons in the six-semion code. The mutual braiding between $v_1$ and $v_2$ gives the $i$ phase, and the topological spins of $v_1$ and $v_2$ are $-i$, indicating that $v_1$ and $v_2$ are anti-semions.}
    \label{tab:six-semion}
\end{table}

We use the algorithm discussed in Sec.~\ref{sec: algorithm} to analyze this model and find the basis anyons $v_1$ and $v_2$. The topological spins and braiding statistics of them are shown in Table.~\ref{tab:six-semion}. The fusion rules are $4 v_1 \sim 4 v_2 \sim \bf 0$, indicating that both have order $4$.


\section{Time complexity of the algorithm}\label{sec: time complexity}

{\color{black}
In this section, we discuss the time complexity of our algorithm for extracting topological orders from Pauli stabilizer codes with $\mathbb{Z}_d$ qudits. The running time of the algorithm depends on several parameters, including the number of qudits per unit cell $w$, the number of stabilizer generators $t$, the geometric range of stabilizers $r$, the truncation range $k$ (as defined in Eq.~\eqref{eq:polynomial_truncation}), and the translational duplicate range $m$ and $m'$ (as described in Eqs.~\eqref{eq: TD definition},~\eqref{eq: M1 matrix Zd}, and~\eqref{eq: M2 and M2O matrix Zd}). For simplicity, we might use $m$ to refer to both $m$ and $m'$, as they are approximately equal. The pseudocode for our algorithm is provided in Appendix~\ref{appendix: pseudocode}. The algorithm consists of the following steps:
\begin{enumerate}
    \item Check the topological order condition of given stabilizer codes. The procedures outlined in Secs.~\ref{sec: checking TO condition Zp} and~\ref{sec: checking TO condition Zd} employ the (modified) Gaussian elimination algorithm twice to compute the matrices $\wwide{M}_1$ in Eq.~\eqref{eq: M1 matrix} (and Eq.~\eqref{eq: M1 matrix Zd} for non-prime qudits) and $\wwide{M}_2$ in Eq.\eqref{eq: M2 and M2O matrix Zp} (and Eq.~\eqref{eq: M2 and M2O matrix Zd} for non-prime qudits) where $\wwide{M}_1$ is a $[2w(2m+1)^2]\times [t(2k+1)^2]$ matrix and $\wwide{M}_2$ is a $[t(2m'+1)^2]\times [(2k+1)^2]$ matrix. The size of $\wwide{M}_2$ is smaller than $\wwide{M}_1$, so the time complexity of this step is dominated by $\wwide{M}_1$.
    \item Solve the anyon equations, as illustrated in Secs.~\ref{sec: Solving Anyon Equations Zp} and~\ref{sec: Solving Anyon Equations Zd}. This step involves applying the modified Gaussian elimination to matrix $\wwide{M}_3$ in Eq.~\eqref{eq: M3 matrix} (and Eq.~\eqref{eq: M3 matrix Zd} for non-prime qudits) for $n=1,..., N$ such that all anyons are found.
    This implies that the modified Gaussian elimination must be performed $N$ times on the $[(2w+t)(2m+1)^2] \times [t(2k+1)^2]$ matrix $\wwide{M}_3$. This step is the most time-consuming, dominating the running time in comparison to the previous steps where modified Gaussian elimination was applied to $\wwide{M}_1$ and $\wwide{M}_2$.
    \item Identify the equivalence classes of anyons. The details are discussed in Secs.~\ref{sec: Equivalence relations between anyons Zp} and~\ref{sec: Equivalence relations between anyons Zd}. For prime-dimensional qudits, equivalence can be readily verified by computing the rank, akin to the Gaussian elimination algorithm. For nonprime-dimensional qudits, anyons are examined individually. Whenever an anyon is encountered that cannot be generated by the existing set $V_{\mathrm{gen}}$, it is added to $V_{\mathrm{gen}}$, and the matrix $M^{(i)}$ is updated using the modified Gaussian elimination algorithm. Consequently, the total number of applications of the modified Gaussian elimination algorithm is limited to the total number of inequivalent anyons in the topological order, which is independent of the parameters in our algorithm.
    \item Compute the fusion rules of anyons. The process is described in Sec.~\ref{sec: Smith normal form}, which computes the Smith normal form of the anyon relation matrix.
\end{enumerate}
For an $r \times c$ matrix, the Gaussian elimination (GE) algorithm operates in $O(r \cdot c \cdot \min(r, c))$ time. According to Theorem~\ref{thm: modified Gaussian as standard Gaussian}, the modified Gaussian Elimination (MGE) algorithm applied to an $r \times c$ matrix is equivalent to the Gaussian Elimination (GE) algorithm applied to an \((r+c) \times c\) matrix over \(\mathbb{Z}\). Therefore, one might anticipate its running time to be \(O(\log d \cdot (r+c) \cdot c^2)\), where \(\log d\) factor arises from executing the Euclidean algorithm in the principal ideal domain \(\mathbb{Z}\).
However, it is important to note that in the implementation of the MGE algorithm on matrix $A$ in Eq.~\eqref{eq: example A matrix}, rows are inserted one by one. Consequently, the number of rows actively involved in the computation remains $r$, not $r+c$. Therefore, the effective time complexity for the MGE algorithm remains $O(r \cdot c \cdot \min(r, c))$. Additionally, if we account for tracking the relations as depicted in Eq.~\eqref{eq: A'|R matrix}, the actual running time scales with $O(r \cdot (r+c) \cdot \min(r, c))$ due to the necessity of appending an extra $r$ columns to the right side of $A'$. \footnote{For the GE algorithm, appending extra columns does not increase the computation time since the effective number of columns participating in the calculations at each step remains $c$, not $r+c$.} Consequently, the computational complexity for GE and MGE, including tracking of relations, is characterized as follows:
\begin{eqs}
    \mathrm{GE:}&
    \begin{cases}
        O(r^2 c), ~\text{if } r\leq c,\\
        O(r c^2), ~\text{if } r> c,
    \end{cases}\\
    \mathrm{MGE:}&
    \begin{cases}
        O\big( \log d \cdot (r^3 + r^2 c) \big), ~\text{if } r\leq c,\\
        O\big( \log d \cdot (r^2 c + r c^2)\big ), ~\text{if } r> c.
    \end{cases}\\
\label{eq: time complexity of GE and MGE}
\end{eqs}
We have empirically validated the time complexity of executing the GE and MGE algorithms for \( r \times c \) matrices. Our method involved generating random, nearly dense \( r \times c \) matrices—defined as having densities close to 1, where density refers to the ratio of non-zero entries to the total number of entries. We subsequently recorded the time to perform the GE and MGE on these matrices. For this experiment, \( r \) and \( c \) were chosen sufficiently large to ensure that other factors contributing to the time complexity were negligible.

The results, plotted for matrix sizes \( r = 50, 100, 150, \dots, 2000 \) with \( c = 2000 \), are displayed in Fig.~\ref{fig:change_nonsparse_row}. From these plots, we observe that the running time of the MGE follows a cubic function when \( r < c \), and transitions to a quadratic function when \( r > c \).
These findings are consistent with our theoretical analysis and confirm the predicted scaling behavior.

\begin{figure}[htb]
    \centering
    \includegraphics[width=0.45\textwidth]{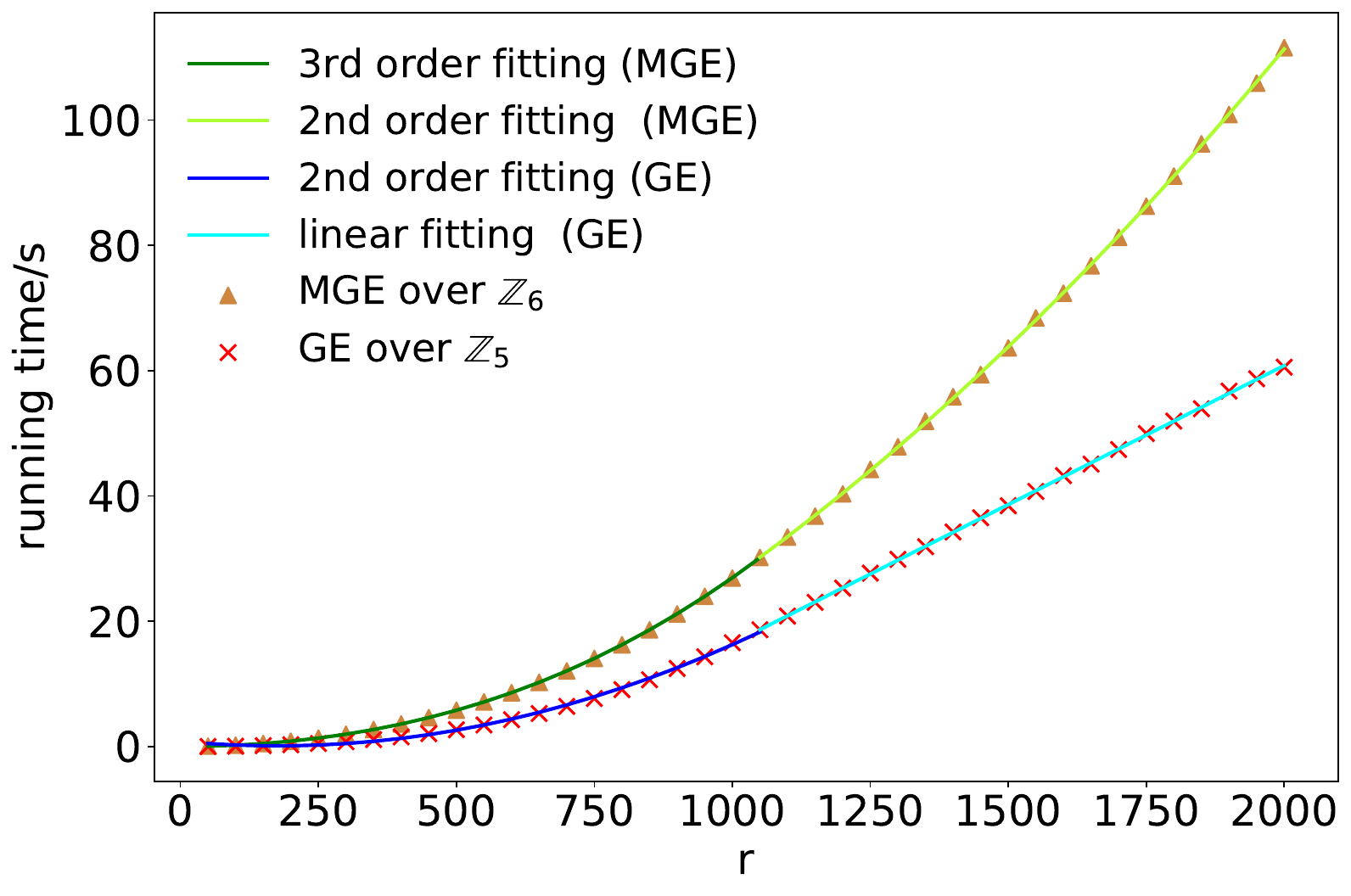}
    \caption{Average running time of Gaussian elimination and modified Gaussian elimination for randomly generated $r \times c$ matrix for $r=50, 100,150,...,2000; c=2000$ on a personal computer. Each data point is averaged over 100 samples of random matrices.}
\label{fig:change_nonsparse_row}
\end{figure}

Typically, computing the Smith normal form is considerably more time-consuming than executing the GE algorithm on matrices of the same size \cite{storjohann1996near}. However, in the described procedures, the Smith normal form is computed for the anyon relation matrix, which is constrained by the number of inequivalent anyons. Consequently, the computational time does not escalate with the truncation range $k$ and the translational duplicate ranges $m, m'$. Consequently, executing the GE or MGE algorithm accounts for the majority of the computational effort. 
Accordingly, our analysis focuses primarily on estimating the time complexity of the GE or MGE, which contributes the most significant terms to the overall time complexity of our algorithm. From our previous analysis, the time complexity for non-prime qudits is predominantly governed by the execution of MGE $N$ times on $\wwide{M}_3$:
\begin{eqs}
    O\Big(N m^2(w+t)^2 (m^2+k^2) \cdot \log d \cdot \min \big((w+t)m^2, tk^2\big)\Big)
\label{eq: overall complexity}
\end{eqs}
As illustrated in Refs.~\cite{haah_module_13, watanabe2023ground}, while the existence of string operators capable of moving anyons is guaranteed, these operators can be long compared to the interaction range. The length of the string operators is only bounded by $l_\mathrm{string}:=d \times 2^r$, where $d$ is the qudit dimension and $r$ is the geometric range of stabilizers.
The parameters \(N\), \(k\), and \(m\) all scale with the length \(l_\mathrm{string}\). Consequently, the worst-case running time of Eq.~\eqref{eq: overall complexity} is bounded by
\begin{equation}
    O(\log d \cdot (w+t)^2 \cdot l_\mathrm{string}^7).
\end{equation}
However, when solving for the string operators in the \(x\)-string, we may select a rectangular truncation region that extends \(l_\mathrm{string}\) in the \(x\)-direction and remains constant in the \(y\)-direction. As a result, the terms \(m^2\) and \(k^2\) are replaced by \(m_x m_y\) and \(k_x k_y\), respectively, which only scale linearly with \(l_\mathrm{string}\). Thus, the worst-case running time is instead bounded by
\begin{equation}
    O(\log d \cdot (w+t)^2 \cdot l_\mathrm{string}^4).
\end{equation}
Fortunately, the worst-case performance scenarios are limited to exotic examples constructed in Refs.~\cite{haah_module_13, watanabe2023ground} and are not typical of the examples usually encountered. 
In practical applications, it is sufficient to select \(N \leq 5\) and \(m, k \leq 10\) to effectively identify all anyons in the examples studied in this paper.
}

\section{Discussion}\label{sec: discussion}

\textcolor{black}{
This work has introduced an algorithm to detect topological order and extract anyon string operators from translation invariant Pauli stabilizer codes. Additionally, the algorithm enables the derivation of braiding statistics and fusion rules from these string operators. Applicable to both prime and non-prime dimensional qudit stabilizer codes, this algorithm has been tested across various translation invariant Pauli stabilizer codes, demonstrating its efficacy. As a result, it serves as an efficient tool for characterizing the topological orders and anyon theories associated with translation invariant Pauli stabilizer codes.}

\textcolor{black}{
Characterizing and constructing various topological orders in qudit systems represents a challenging and enduring problem. Recently, Ref.~\cite{ellison2022pauli} illustrated the possibilities for constructing diverse topological orders using 2d translation invariant Pauli stabilizer codes with non-prime dimensional qudits. However, the analytical and numerical characterization of topological codes with non-prime dimensional qudits was previously lacking. This work presents an algorithm offering a numerical approach to characterizing such topological codes. An intriguing future research direction is exploring Pauli stabilizer codes that exhibit exotic topological orders. For instance, various quantum codes could be constructed using the Quantum Lego \cite{cao2022quantum,cao2023quantum,su2023discovery} formalism; applying our algorithm to characterize the topological properties of these Lego codes will be interesting. Recently, in Ref.~\cite{cao2023weight} it is shown the connections between the underlying tensor network structure of quantum Lego codes and quantum weight enumerator can be used to efficiently predict the properties of the codes, including code distance. Another interesting question is investigating the implications of our polynomial method on the tensor network side.}

\textcolor{black}{
In addition to Pauli stabilizer codes, extensions of the stabilizer formalism, such as the XS \cite{ni2015non} and XP formalisms \cite{webster2022xp, shen2023quantum}, have been developed. These extensions modify the conventional stabilizer approach by integrating roots of Pauli $Z$ into the stabilizer. Since XP stabilizer codes also possess a symplectic representation, extending the polynomial formalism and our algorithm to accommodate these codes is possible. This adaptation is a critical first step toward exploring non-Pauli or even non-Clifford stabilizer codes, which could potentially exhibit interesting non-Abelian anyon statistics that are essential for universal topological quantum computation.} \textcolor{black}{Beyond constructing and characterizing topological order from polynomial formulations, the topological data can be rigorously described by a $G$-crossed braided fusion category \cite{barkeshli2019symmetry}. This framework is crucial for the classification of topological orders, as elucidated in Refs.~\cite{aasen2021torsorial, Classification2022Barkeshli}. Our approach effectively verifies these classifications within Abelian theories using alternative methods.}

\textcolor{black}{
In this work, we demonstrate the applicability of our algorithm across various topological orders by examining extensive cases. Although topological orders are typically discussed concerning closed manifolds, exploring topological orders with gapped boundaries or defects, and their relationship to anyon condensations \cite{kitaev2012models,lan2015gapped} is also crucial. A future extension of this algorithm could involve generating all possible gapped defects and boundaries for any given Pauli stabilizer codes. Furthermore, we contemplate a 3d generalization of our algorithm. Identifying mobile particle excitations in three dimensions is straightforward, similar to our approach in two dimensions. However, challenges arise with the "fracton" phase, where excitations exhibit restricted mobility. We should replace previous string operators with ``fractal operators'' to detect these excitations. Additionally, a new protocol is necessary to detect loop excitations. One potential method involves a dimensional reduction process by compactifying one dimension in 3d, transforming the system into a quasi-2d framework. This adaptation would allow our 2d algorithm to detect both particle and compactified loop excitations effectively.
}

Another potential generalization involves subsystem codes~\cite{Poulin2005subsystem, bombin_Stabilizer_14, Ellison2023paulitopological, liu2023subsystem} and Floquet codes~\cite{hastings2021dynamically, Aasen2022Adiabatic, davydova2023floquet, Ellison2023floquet, Dua2024Floquet}. Subsystem codes offer a relaxation of the requirement that each term in the Hamiltonian must commute; instead, non-commuting terms act as gauge operators. The commutants of all gauge operators then form the stabilizer group.
Furthermore, Floquet codes exploit the temporal sequence of measurements, introducing more structural complexity than subsystem codes. Each measurement cycle induces an instantaneous stabilizer code, such as the toric code in the original example~\cite{hastings2021dynamically}. Thus, applying our algorithm in these scenarios to explore a broader range of subsystem and Floquet codes would yield valuable insights.

Furthermore, an additional extension of our polynomial formalism involves generalizing the $\mathbb{Z}^2$-translational symmetry to more intricate group structures, including non-Abelian groups. Currently, we focus on  Pauli stabilizer codes on two-dimensional lattices, which exhibit a $\mathbb{Z}^2$ symmetry generated by translations in the $x$- and $y$-directions, represented by generators $x$ and $y$ in the Laurent polynomial ring.
We aim to broaden this to incorporate more exotic graphs with any translation group $G$ (potentially non-Abelian), where the translation generators are denoted as $g_1$, $g_2$, $g_3$, etc. Within this framework, we can still employ a "polynomial" ring over these generators. This extension would facilitate the representation of Pauli stabilizer codes over the Cayley graph of $G$, a technique widely used in the construction of quantum low-density parity-check (qLDPC) codes~\cite{Couvreur2011Cayley, breuckmann2021balanced, breuckmann2021quantum, Panteleev2022goodqldpc, Dinur2023Good} recently. We plan to adapt our algorithm to these Cayley graphs, anticipating that our numerical approach will provide deeper insights into these advanced qLDPC codes.

\section*{Acknowledgment}

We thank Nathanan Tantivasadakarn and Victor V. Albert for their helpful discussions and for providing the CSS code example. YC appreciates useful suggestions from Jeongwan Haah at the early stage of this project.
YC thanks Tyler Ellison, Nathanan Tantivasadakarn, and Bowen Yang for teaching the polynomial method, for useful feedback on this manuscript, and for clarifying anyons in stabilizer/subsystem codes. YC thanks Hsin-Po Wang for pointing out the truncation of the Laurent polynomials.
YC thanks Yun-Ting Cheng for technical support. Y.X.\@ thanks the organizers and participants of the YIPQS long-term workshop YITP-T-23-01
``Quantum Information, Quantum Matter and Quantum Gravity (2023),'' held at Yukawa Institute for Theoretical Physics in Kyoto University, where part of this work was completed.
Y.X.\@ was partially supported by ARO Grant No. W911NF-15-1-0397, National Science Foundation QLCI Grant No. OMA2120757, AFOSR-MURI Grant No. FA9550-19-1-0399,
and Department of Energy QSA program. JTI thanks the Joint Quantum Institute at the University of Maryland for support through a JQI fellowship.
Y.X.\@ thanks Yujie Zhang and Shan-Ming Ruan for their hospitality in Tokyo and Kyoto. Y.X. thanks Zhiwei Zhang, Honggeng
Zhang, Xuan Su, Qiongsi Yan, and Bo He for their mental support.

\appendix

\begin{figure*}[t]
    \centering
    \includegraphics[width=0.9\textwidth]{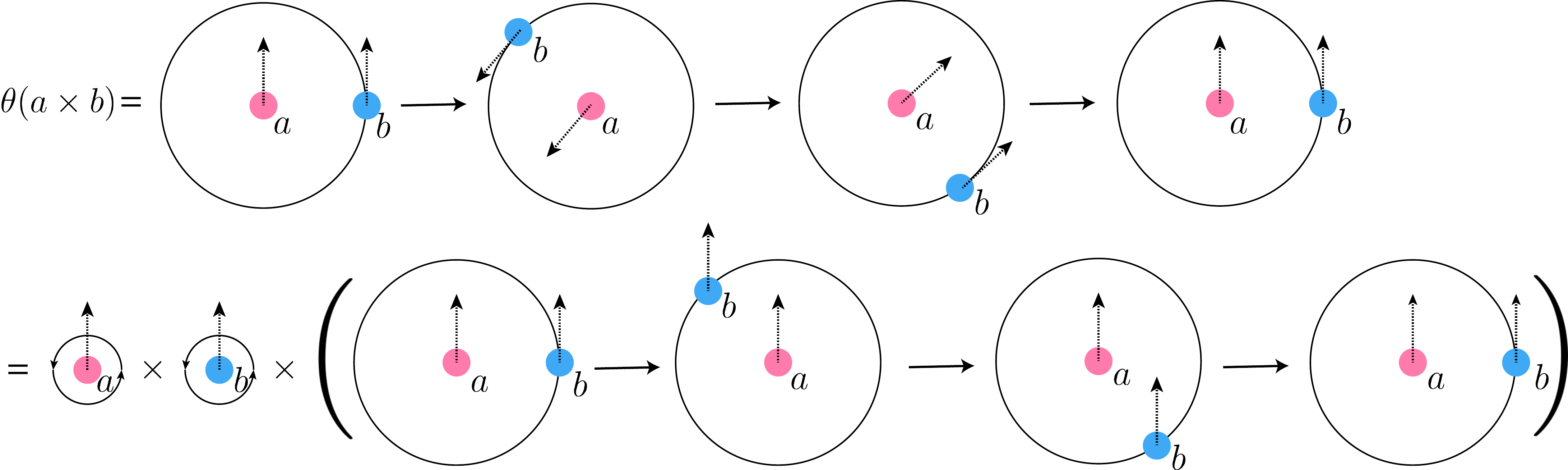}
    \caption{Relation between $\theta(a \times b)$, $\theta(a)$,$\theta(b)$ and the mutual braiding $B(a,b)$. This equality indicates $\theta(a\times b)=\theta(a) \times \theta(b) \times B(a,b)$.}
    \label{fig:relation_topo_spin}
\end{figure*}

\section{Relation between braiding statistics and topological spins}\label{appendix: braiding_statistics}

In this appendix, we review how the topological spins in Abelian anyon theories completely determine the braiding statistics. Specifically, from the topological spin given in Eq.~\eqref{eq: statistics formula}, we derive Eq.~\eqref{eq:braiding statistics}. We again consider the T-junction setup as shown in Fig.~\ref{fig: T junction 3 paths}. Using the T-junction, braiding anyon $a$ counterclockwise around anyon $b$ is seen to be given by
\begin{eqs}
    \label{eq: B theta}
    B_\theta (a, b) =
    &~~(W^b_{3})^\dagger W^b_{2}
    (W^a_{1})^\dagger W^a_{3}
    (W^b_{2})^\dagger W^b_{1}\\
    &\times (W^a_{3})^\dagger W^a_{2}
    (W^b_{1})^\dagger W^b_{3}
    (W^a_{2})^\dagger W^a_{1}.
\end{eqs}
Suppose that $a$ starts at $\bar\gamma_1(0)$ and $b$ at $\bar\gamma_3(0)$. Then this corresponds to moving $a$ and $b$ in the following order:
\begin{enumerate}
    \item $a$: $\bar\gamma_1(0) \to p \to \bar\gamma_2(0)$,
    \item $b$: $\bar\gamma_3(0) \to p \to \bar\gamma_1(0)$,
    \item $a$: $\bar\gamma_2(0) \to p \to \bar\gamma_3(0)$,
    \item $b$: $\bar\gamma_1(0) \to p \to \bar\gamma_2(0)$,
    \item $a$: $\bar\gamma_3(0) \to p \to \bar\gamma_1(0)$,
    \item $b$: $\bar\gamma_2(0) \to p \to \bar\gamma_3(0)$,
\end{enumerate}
thus performing the braid as desired.
Using from Eq.~\eqref{eq: statistics formula} that
\begin{eqs}
    \label{eq: theta a}
    \theta(a) 
    &= (W^a_3)^\dag W^a_2 (W^a_1)^\dag W^a_{3}(W^a_{2})^\dagger W^a_{1},
\end{eqs}
we will show that
\begin{eqs}
    B_\theta (a, b) = \frac{\theta(a\times b)}{\theta(a) \theta(b)}.
\end{eqs}
Ultimately, this follows from the definition $W_i^{a\times b}\coloneqq W_i^b W_i^a$; that is, moving the composite anyon $a\times b$ is equivalent to moving each anyon separately along the same path. 

Since each $W_i^{a}$ is a Pauli, we have that 
\begin{eqs}
    W^a_i W^b_j &= z_{i,j}^{a,b} W^b_j W^a_i \\
    (W^a_i)^\dag (W^b_j)^\dag &= z_{i,j}^{a,b} (W^b_j)^\dag (W^a_i)^\dag \\
    W^a_i (W^b_j)^\dag &= (z_{i,j}^{a,b})^{-1} (W^b_j)^\dag W^a_i ,
\end{eqs}
for some complex numbers $z_{i,j}^{a,b}$.
Starting from Eq.~\eqref{eq: theta a}, we have that
\begin{align}
    \theta&(a\times b)^{-1} \nonumber \\
    &= (W^{a\times b}_1)^\dag W^{a\times b}_2 (W^{a\times b}_3)^\dag W^{a\times b}_{1}(W^{a\times b}_{2})^\dagger W^{a\times b}_{3}  \\
    \begin{split}
        &= (W^a_1)^\dag(W^{b}_1)^\dag W^b_2 W^{a}_2  (W^{a}_3)^\dag (W^b_3)^\dag \\
        &\quad\times W^b_1 W^{a}_{1} (W^a_2)^\dag (W^b_2)^\dag  W^b_3 W^{a}_{3}
    \end{split}\\
    \begin{split}
        &= (z^{a,b}_{2,3})^{-1} z^{a,b}_{3,3} \\
        &\qquad \times (W^a_1)^\dag(W^{b}_1)^\dag W^b_2 (W^b_3)^\dag W^{a}_2  (W^{a}_3)^\dag \\
        &\quad\times  z^{a,b}_{2,2} (z^{a,b}_{1,2})^{-1} \\
        &\qquad \times 
        W^b_1 (W^b_2)^\dag W^{a}_{1} (W^a_2)^\dag W^b_3 W^{a}_{3} 
    \end{split}\\
    \begin{split}
        &= (z^{a,b}_{2,3})^{-1}z^{a,b}_{3,3} z^{a,b}_{1,1} (z^{a,b}_{1,2})^{-1}z^{a,b}_{1,3} \\
        &\qquad \times (W^{b}_1)^\dag W^b_2 (W^b_3)^\dag (W^a_1)^\dag W^{a}_2  (W^{a}_3)^\dag \\
        &\quad\times  z^{a,b}_{2,2} (z^{a,b}_{1,2})^{-1}(z^{a,b}_{2,3})^{-1} z^{a,b}_{1,3} \\
        &\qquad \times 
        W^b_1 (W^b_2)^\dag W^b_3 W^{a}_{1} (W^a_2)^\dag  W^{a}_{3} 
    \end{split}\\
    \begin{split}
        &= (z^{a,b}_{2,3})^{-1} z^{a,b}_{3,3} z^{a,b}_{1,1} (z^{a,b}_{1,2})^{-1}z^{a,b}_{1,3} \\
        &\qquad \times (W^{b}_1)^\dag W^b_2 (W^b_3)^\dag (W^a_1)^\dag W^{a}_2  (W^{a}_3)^\dag \\
        &\quad\times  z^{a,b}_{2,2} (z^{a,b}_{1,2})^{-1}(z^{a,b}_{2,3})^{-1} z^{a,b}_{1,3} \\
        &\qquad \times 
        W^b_1 (W^b_2)^\dag W^b_3 W^{a}_{1} (W^a_2)^\dag  W^{a}_{3} 
    \end{split}\\
    \begin{split}
        &= (z^{a,b}_{2,3})^{-1} z^{a,b}_{3,3} z^{a,b}_{1,1} (z^{a,b}_{1,2})^{-1}z^{a,b}_{1,3} \\
        &\quad\times  z^{a,b}_{2,2} (z^{a,b}_{1,2})^{-1}(z^{a,b}_{2,3})^{-1} z^{a,b}_{1,3} \\
        &\quad\times (z^{a,b}_{3,1})^{-1} z^{a,b}_{2,1}(z^{a,b}_{1,1})^{-1} \\
        &\quad\times z_{3,2}^{a,b} (z^{a,b}_{2,2})^{-1} z^{a,b}_{1,2} \\
        &\quad\times (z^{a,b}_{3,3})^{-1} z^{a,b}_{2,3} (z^{a,b}_{1,3})^{-1} \\
        &\quad \times (W^{b}_1)^\dag W^b_2 (W^b_3)^\dag W^b_1 (W^b_2)^\dag  W^b_3 \\
        &\quad \times (W^a_1)^\dag W^{a}_2  (W^{a}_3)^\dag W^{a}_{1} (W^a_2)^\dag  W^{a}_{3} 
    \end{split}\\
    \begin{split}
        &= z_{3,2}^{a,b} (z^{a,b}_{1,2})^{-1} z^{a,b}_{2,1} (z^{a,b}_{3,1})^{-1} (z^{a,b}_{2,3})^{-1} z^{a,b}_{1,3}   \\
        &\quad \times \theta(b)^{-1} \theta(a)^{-1} .
    \end{split}
\end{align}
Thus,
\begin{equation}
    \label{eq: theta quotient}
    \frac{\theta(a\times b)}{\theta(a) \theta(b)} = (z^{a,b}_{3,2})^{-1} z^{a,b}_{1,2} (z^{a,b}_{2,1})^{-1} z^{a,b}_{3,1} z^{a,b}_{2,3} (z^{a,b}_{1,3})^{-1}.
\end{equation}
Meanwhile, from Eq.~\eqref{eq: B theta},
\begin{align}
    \begin{split}
        B_\theta(a,b)
        &= (W^b_{3})^\dagger W^b_{2}
        (W^a_{1})^\dagger W^a_{3}
        (W^b_{2})^\dagger W^b_{1}\\
        &\quad\times (W^a_{3})^\dagger W^a_{2}
        (W^b_{1})^\dagger W^b_{3}
        (W^a_{2})^\dagger W^a_{1}
    \end{split} \\
    \begin{split}
        &= (z^{a,b}_{3,2})^{-1} z^{a,b}_{1,2} \\
        &\quad \times(W^b_{3})^\dagger
        (W^a_{1})^\dagger W^a_{3} W^b_{1}\\
        &\quad\times (W^a_{3})^\dagger W^a_{2}
        (W^b_{1})^\dagger W^b_{3}
        (W^a_{2})^\dagger W^a_{1}
    \end{split} \\
    \begin{split}
        &= (z^{a,b}_{3,2})^{-1} z^{a,b}_{1,2} (z^{a,b}_{2,1})^{-1} z^{a,b}_{3,1} \\
        &\quad \times(W^b_{3})^\dagger
        (W^a_{1})^\dagger  W^a_{2}
         W^b_{3}
        (W^a_{2})^\dagger W^a_{1}
    \end{split} \\
    &= (z^{a,b}_{3,2})^{-1} z^{a,b}_{1,2} (z^{a,b}_{2,1})^{-1} z^{a,b}_{3,1} z^{a,b}_{2,3} (z^{a,b}_{1,3})^{-1} .
    \label{eq: B in terms of z}
\end{align}
Comparing to Eq.~\eqref{eq: theta quotient}, we see that
\begin{equation}
    B_\theta(a,b) = \frac{\theta(a\times b)}{\theta(a) \theta(b)}.
\end{equation}
Another way to interpret this equation is to treat the topological spin $\theta (a)$ as a $2\pi$ rotation of an anyon $a$, according to the spin-statistics theorem. As depicted in Fig.~\ref{fig:relation_topo_spin}, the $2\pi$ rotation of anyon $a \times b$ can be decomposed into three pieces: the $2\pi$ rotation of anyon $a$, the $2\pi$ rotation of anyon $b$, and the mutual braiding between anyons $a$ and $b$. However, we emphasize that this is a physical picture instead of a rigorous proof derived by Eqs.~\eqref{eq: theta quotient} and \eqref{eq: B in terms of z}.

\section{The Hermite normal form and the Smith normal form}\label{sec: HSF and SNF}

Throughout this work, we have been broadly interested in solving linear systems of equations over a ring $R$. Specifically, suppose we have an $r \times c$ matrix $A$ with entries $A_{ij}\in R$ and a vector $b \in R^r$ represented by a column $b_i \in R$. To \textit{solve} the system defined by $A$ and $b$, we must find a vector $x \in R^c$ such that $Ax = b$, where $=$ denotes equality within $R$. We will denote the linear system by $(A,b)_R$. 

When $R$ is a field $\mathbb K$, such as $\mathbb R$ or a finite field $\mathbb F_{p^n}$, $(A,b)_{\mathbb K}$ can be solved efficiently using Gaussian elimination algorithm. By performing row operations, we effectively transform the system $Ax=b$ into $CAx = C b$, where $C \in \operatorname{GL}(r,\mathbb K)$ with $\operatorname{GL}$ denoting the general linear group. One can efficiently find $C$ such that $CA$ is in reduced echelon form, making $x$ easy to find.

When $R$ is not a field but instead only a principal ideal domain (PID), an analogous procedure finds $C \in \operatorname{GL}(r,R)$ such that $CA$ is the \textit{Hermite normal form} (HNF) of $A$ \cite{martin2012large}.
The HNF of $A$ can be efficiently found.
The most important property of the HNF is that it is upper triangular.
When $R = \mathbb Z$, $C$ is called a \textit{unimodular} matrix. Using the HNF, linear systems $(A,b)_{\mathbb Z}$ over $\mathbb Z$ can be efficiently solved. The ring of modular integers $\mathbb Z_n = \mathbb Z / n \mathbb Z$ is not a PID, but linear systems can still be solved using the HNF by turning the system over $\mathbb Z_n$ into a system over $\mathbb Z$. Specifically, we transform the system $Ax = b \pmod n$ into $A x - ny = b$, where $y \in \mathbb Z^r$ is an additional set of variables. We define the new $(r+c)\times 1$ variable vector $\tilde x = x \oplus y$ and the new $r \times (r+c)$ matrix $\tilde A = \begin{pmatrix}
    A & -n \mathbf 1_{r \times r}
\end{pmatrix}$. The value of $x$ in the solution $\tilde x$ to the system $(\tilde A, b)_{\mathbb Z}$ over $\mathbb Z$ is the solution to the system $(A,b)_{\mathbb Z_n}$ over $\mathbb Z_n$.

In this work, we are most interested in solving linear systems $(A,b)_R$ over the polynomial ring $R = \mathbb Z_n[x,y]$. Using the same procedure above, we can reduce this to solving linear systems $(A,b)_R$ over the polynomial ring $R = \mathbb Z[x,y]$. 
By truncating the polynomials coming from $\mathbb Z[x,y]$ to some maximum degree $d$, we obtain a new (larger) linear system over simply $\mathbb Z$.

One natural question is: can we solve the linear system over $\mathbb Z[x,y]$ \textit{directly}, without truncating the degree of the polynomials?
Unfortunately, $\mathbb Z[x,y]$ is not a PID, and thus the HNF cannot be used. We cannot, in general, get around degree truncation.
Specifically, we can effectively solve linear systems over $\mathbb Z[x,y]$ if and only if we can construct an algorithm that outputs an upper bound on the polynomial degrees that occur during the solving algorithm. 
The $\Rightarrow$ direction is obvious, and the $\Leftarrow$ direction follows from our discussion above, where we reduced a general linear system over $\mathbb Z[x,y]$ to a linear system over $\mathbb Z$ once we picked a degree at which to truncate.
In Eq.~\eqref{eq: anyon equation}, for example, there is no general way of bounding the degree since $v$ can be arbitrary. Hence, the best we can do is pick a degree to truncate.

Finally, we recall the \textit{Smith normal form} (SNF). The HNF of a matrix over a PID is analogous to the reduced echelon form of a matrix over a field.
In a similar way, the SNF of a matrix over a PID is analogous to the diagonalization of a matrix over a field. Specifically, one can efficiently use row and column operations amounting to matrices $C \in \operatorname{GL}(r, R)$ and $D \in \operatorname{GL}(c, R)$ such that $C A D$ is diagonal. $CAD$ is called the SNF of $A$, and it further satisfies that its nonzero diagonal elements $d_1,\dots,d_\ell$ satisfy $d_1 \mid d_2 \mid \dots \mid d_\ell$. $d_1,\dots,d_\ell$ are unique and are called the \textit{elementary divisors} of $A$.

\section{String operators of various examples}\label{appendix: string_operators}

This appendix presents the pictorial description of anyon string operators we solved for the \textcolor{black}{2d honeycomb} color code, the modified color codes, the double semion code, the CSS code from the doubled semion code, and the six-semion code.

As we discussed in Sec.~\ref{sec:extract_string}, we need to solve Eq.~\eqref{eq: anyon equation} for different choices of $n_x,n_y$ and choose the smallest $n_x,n_y$ that maximize the number of anyon solutions. Here, we report the number of independent anyon solutions of the different models with different choices of $n_x, n_y$ in Table.~\ref{tab:number_anyon}.

\begin{table*}[t]
    \centering
    \begin{tabular}{|c|c|c | c| c| c| c| c| c|}
    \hline 
    number of anyons & $n=1$& $n=2$& $n=3$&$n=4$& $n=5$ & $n=6$ & $n=7 $& $n=8$
          \\
          \hline
         \textcolor{black}{2d honeycomb} color code (example \circled{1})~ $x$-direction& 0& 0& 4& 0& 0& 4& 0& 0 \\
         \hline
         \textcolor{black}{2d honeycomb} color code (example \circled{1})~ $y$-direction& 0& 0& 4& 0& 0& 4& 0& 0 \\
         \hline
         Modified color code A (example \circled{3})~ $x$-direction& 0& 0& 0& 0& 8& 0& 0& 0 \\
         \hline
         Modified color code A (example \circled{3})~ $y$-direction& 0& 0& 0& 0& 8& 0& 0& 0 \\
         \hline
         Modified color code B (example \circled{4})~ $x$-direction& 4& 8&8& 12& 4& 12& 4& 12\\
         \hline
         Modified color code B (example \circled{4})~ $y$-direction& 4& 8&8& 12& 4& 12& 4& 12\\
         \hline
         Modified color code C (example \circled{5})~$x$-direction& 2& 4& 2& 8& 2& 4& 2& 8\\
         \hline
         Modified color code C (example \circled{5})~ $y$-direction& 2& 4& 2& 8& 2& 4& 2& 8\\
         \hline
         Modified color code D (example \circled{6})~ $x$-direction& 4& 8& 4& 12& 4& 8& 4& 12\\
         \hline
         Modified color code D (example \circled{6})~ $y$-direction& 4& 8& 4& 12& 4& 8& 4& 12\\
         \hline
         CSS code from double semion code $x$-direction& 4& 4& 4& 4& 4& 4& 4& 4\\
         \hline
         CSS code from double semion code $y$-direction& 4& 4& 4& 4& 4& 4& 4& 4\\
         \hline
         Six-semion code $x$-direction& 2& 2& 2& 2& 2& 2& 2& 2\\
         \hline
         Six-semion code $y$-direction& 2& 2& 2& 2& 2& 2& 2& 2\\
         \hline
         Double semion code& 2& 2& 2& 2& 2& 2& 2& 2\\
         \hline
         Double semion code& 2& 2& 2& 2& 2& 2& 2& 2\\
         \hline
    \end{tabular}
    \caption{The number of independent anyon solutions of Eq.~\eqref{eq: anyon equation} with different choices of $n_x, n_y$ for different examples. We can see that the number of independent anyon solutions has a periodic pattern when we change $n_x, n_y$. Moreover, if $n_x=n_y$, we find that the number of independent anyon solutions along the $x$- and $y$-directions are the same for a given model. The algorithm shows that the modified color code B has 16 independent anyon solutions when $n_x=n_y=12$. We have checked their periodicity up to $n_x, n_y=16$, while only showing $n_x, n_y$ up to 8 in this table.}
    \label{tab:number_anyon}
\end{table*}

\subsection{Color code and modified color codes}\label{appendix:modified_color_code_string}

In this section, we use \textcolor{black}{our algorithm} to investigate five models \textcolor{black}{satisfying topological order conditions shown in Fig.~\ref{fig: modified color codes} except example \circled{2}}, and the resulting string operator is shown in the following figure:

Since all five models are self-dual, we omit whether each qubit represents X or Z Here. The blue dots represent the string operator in the $x$-direction, and the red dots represent the string operator in the $y$-direction. When the string operator in the $x$-direction coincides with the string operator in the $y$-direction, we use purple dots to represent the overlapping qubits. The black dot represents the origin and distinguishes between different stabilizer positions. 

The result is summarized as follows:
\begin{itemize}
    \item For the \textcolor{black}{2d honeycomb} color code (example \circled{1}), we obtain four independent anyon solutions at $(1-x^3)$ and $(1-y^{-3})$. The string operators are shown in Fig.~\ref{fig:string_example1}
    \item The second model does not satisfy the topological condition.
    \item For the modified color code A (example \circled{3}), we can obtain maximally eight independent anyon solutions at $(1-x^5)$ and $(1-y^{-5})$. The string operators are shown in Fig.~\ref{fig:string_example3}
    \item For the modified color code B (example \circled{4}), we can obtain maximally sixteen independent anyon solutions at $(1-x^{12})$ and $(1-y^{-12})$. The string operators are shown in Figs.~\ref{fig:string_example4_1},~~\ref{fig:string_example4_2},~~\ref{fig:string_example4_3},~~\ref{fig:string_example4_4},~~\ref{fig:string_example4_5},~~\ref{fig:string_example4_6},~~\ref{fig:string_example4_7},~~\ref{fig:string_example4_8}.
    \item For the modified color code C (example \circled{5}), we can obtain maximally eight independent anyon solutions at $(1-x^4)$ and $(1-y^{-4})$. The string operators are shown in Fig.~\ref{fig:string_example5}.
    \item For the modified color code D (example \circled{6}), we can obtain maximally twelve independent anyon solutions at $(1-x^4)$ and $(1-y^{-4})$. The string operators are shown in Figs.~\ref{fig:string_example6_1} and \ref{fig:string_example6_2}.
\end{itemize}

\begin{figure*}[htb]
    \centering
    \subfigure[String operator for $v_1$ ]{\includegraphics[width=0.24\textwidth]{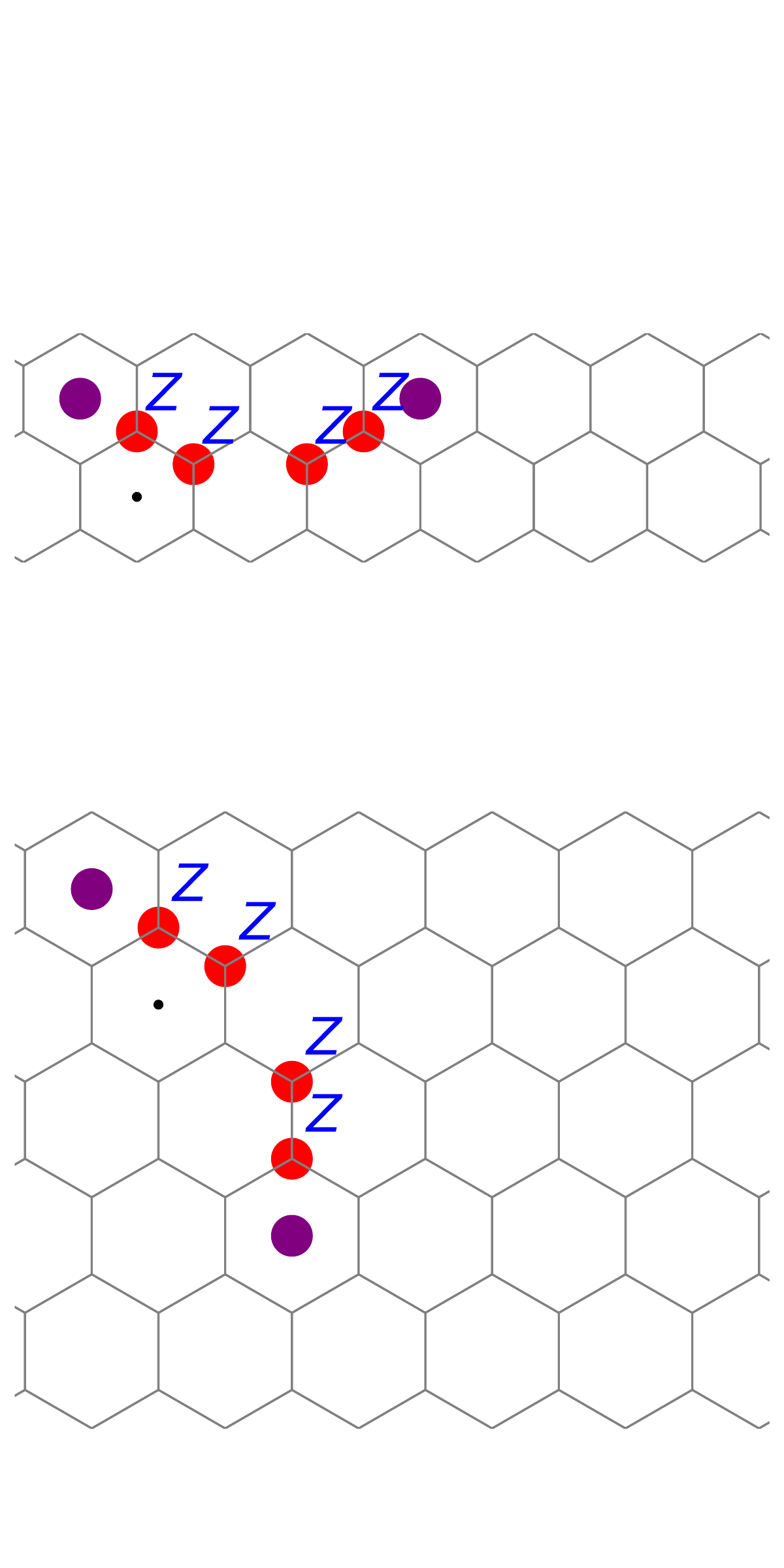}}
    \subfigure[string operator for $v_2$]{\includegraphics[width=0.24\textwidth]{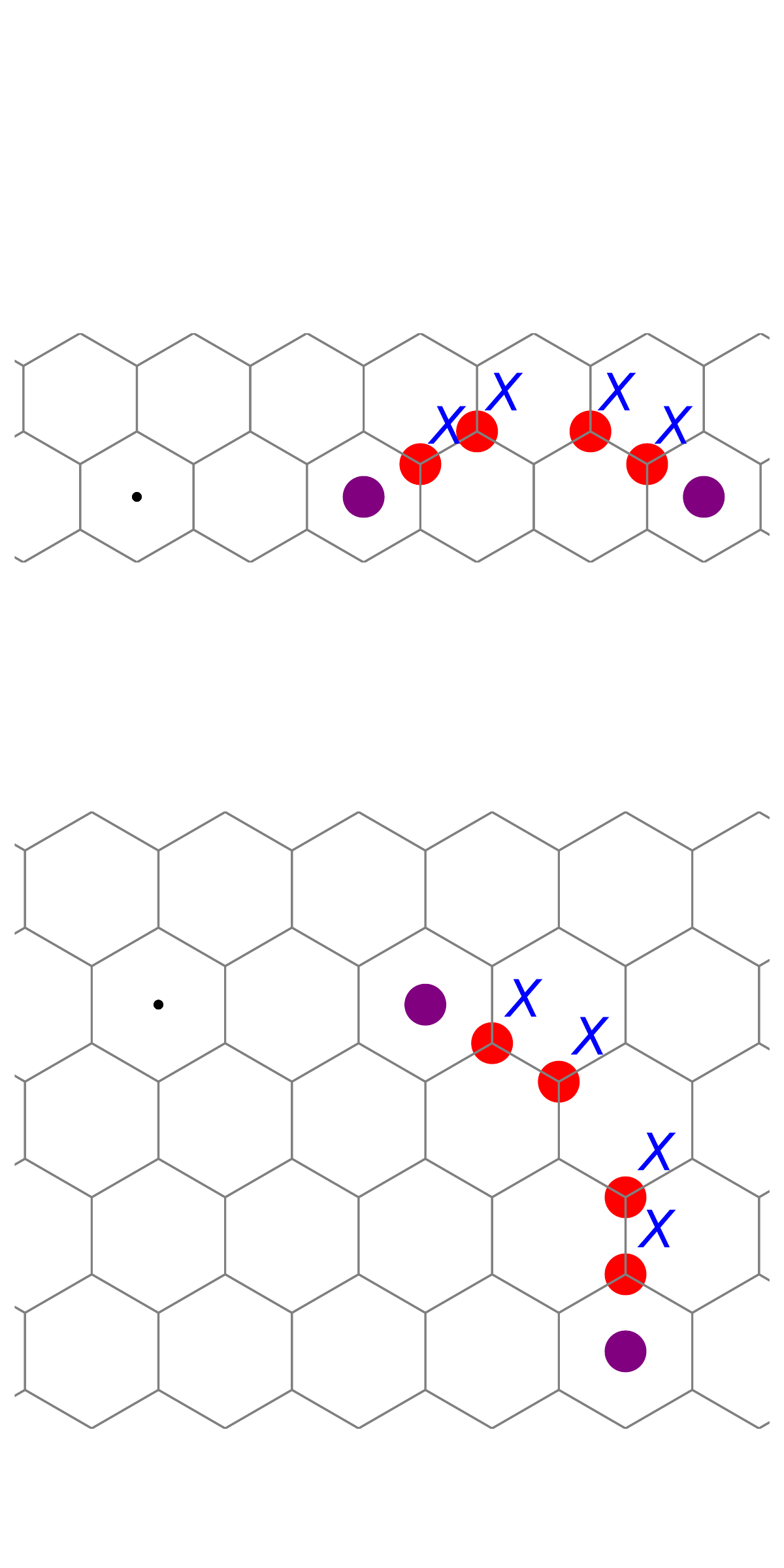}}
    \subfigure[string operator for $v_3$]{\includegraphics[width=0.24\textwidth]{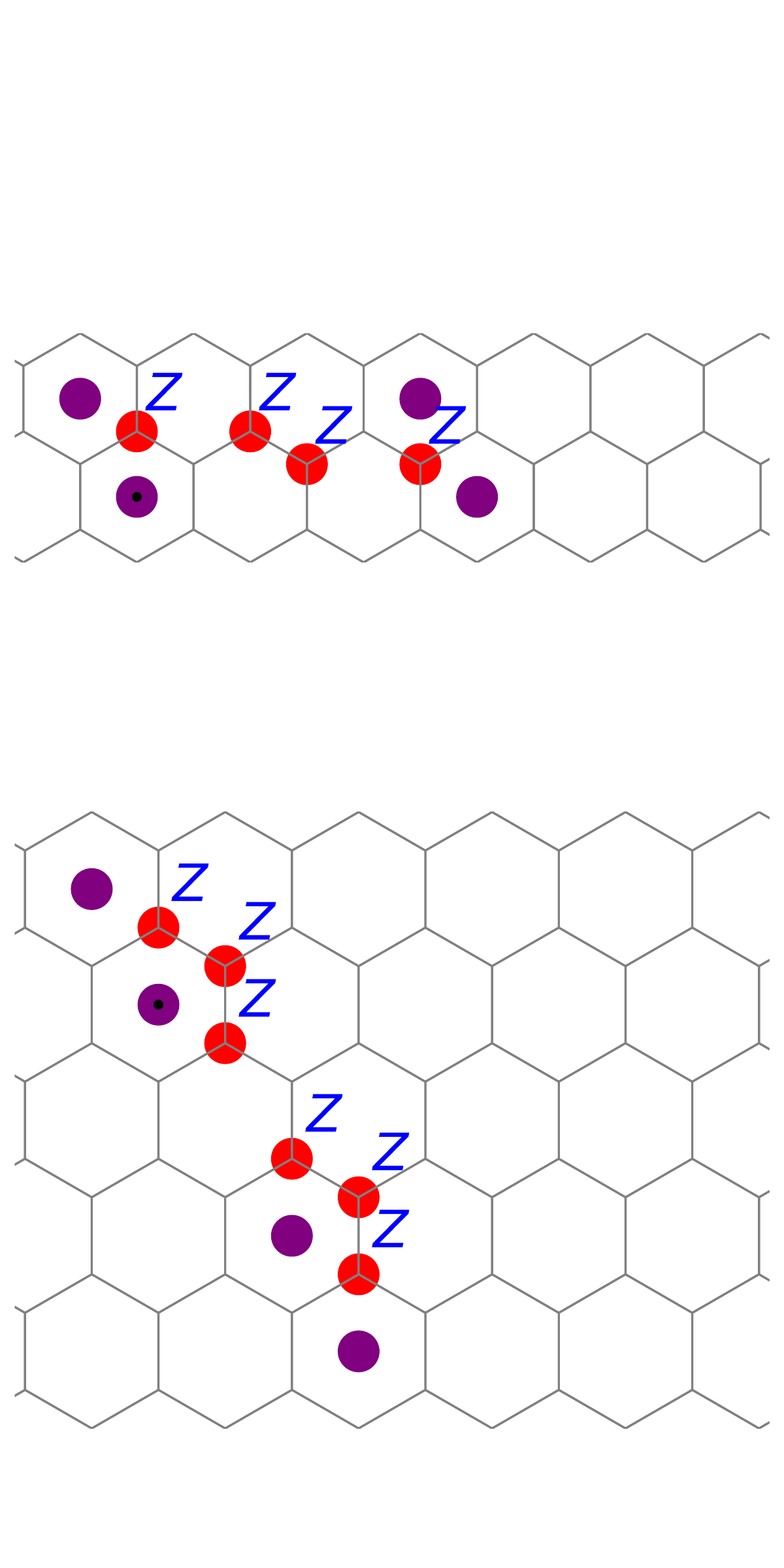}}
    \subfigure[string operator for $v_4$]{\includegraphics[width=0.24\textwidth]{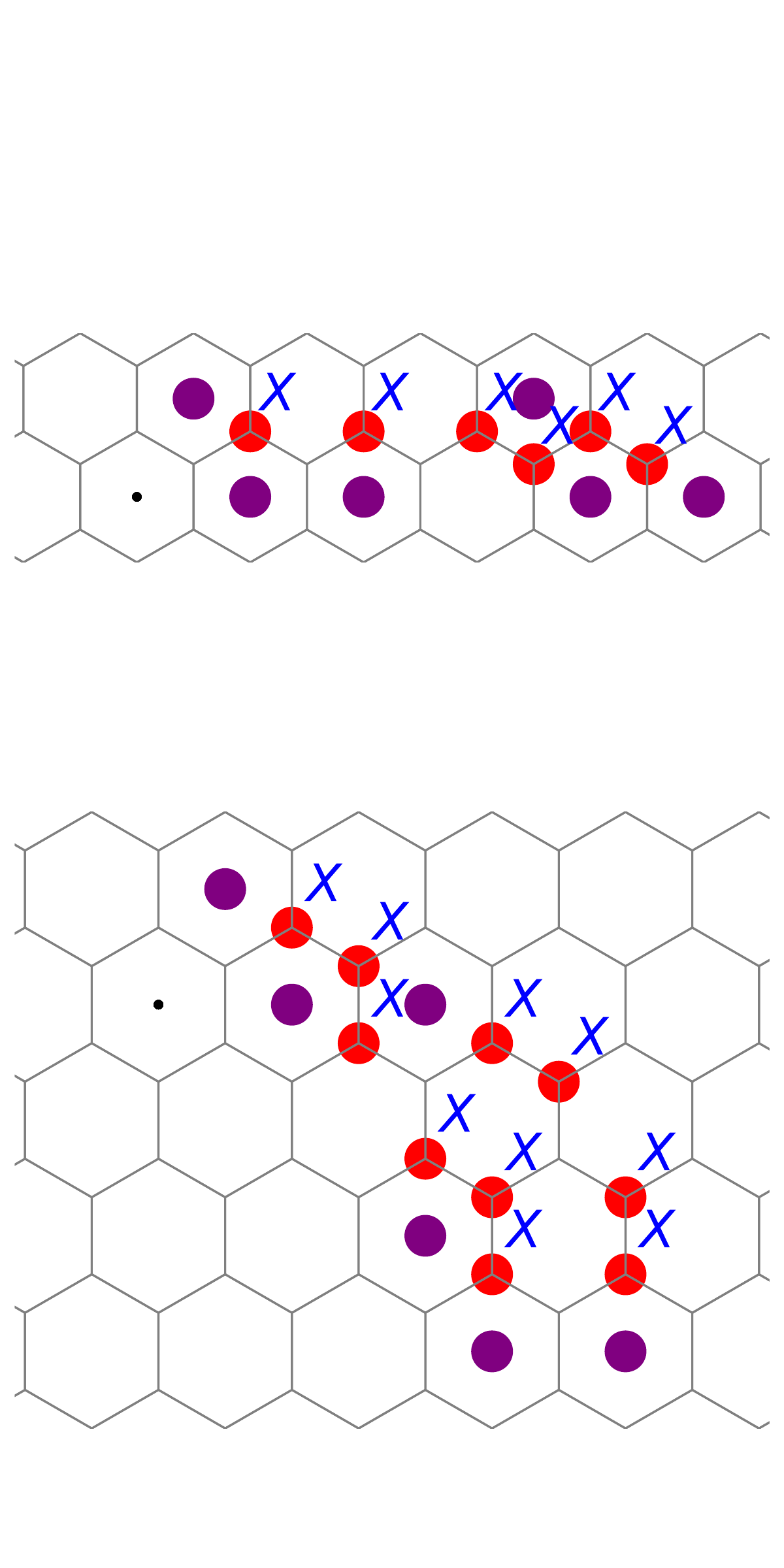}}
    \caption{The string operators in the \textcolor{black}{2d honeycomb} color code (example \textcircled{1}). Subfigures (a)$\sim$ (d) represent the string operators for anyons $v_1,..., v_4$. For each subfigure, the first row represents the anyon string operators that move an anyon along the $x$-direction, and the second row represents the anyon string operators that move an anyon along the $y$-direction.}
\label{fig:string_example1}
\end{figure*}

\begin{figure*}[htb]
    \centering
    \subfigure[string operator for $v_1$]{\includegraphics[width=0.24\textwidth]{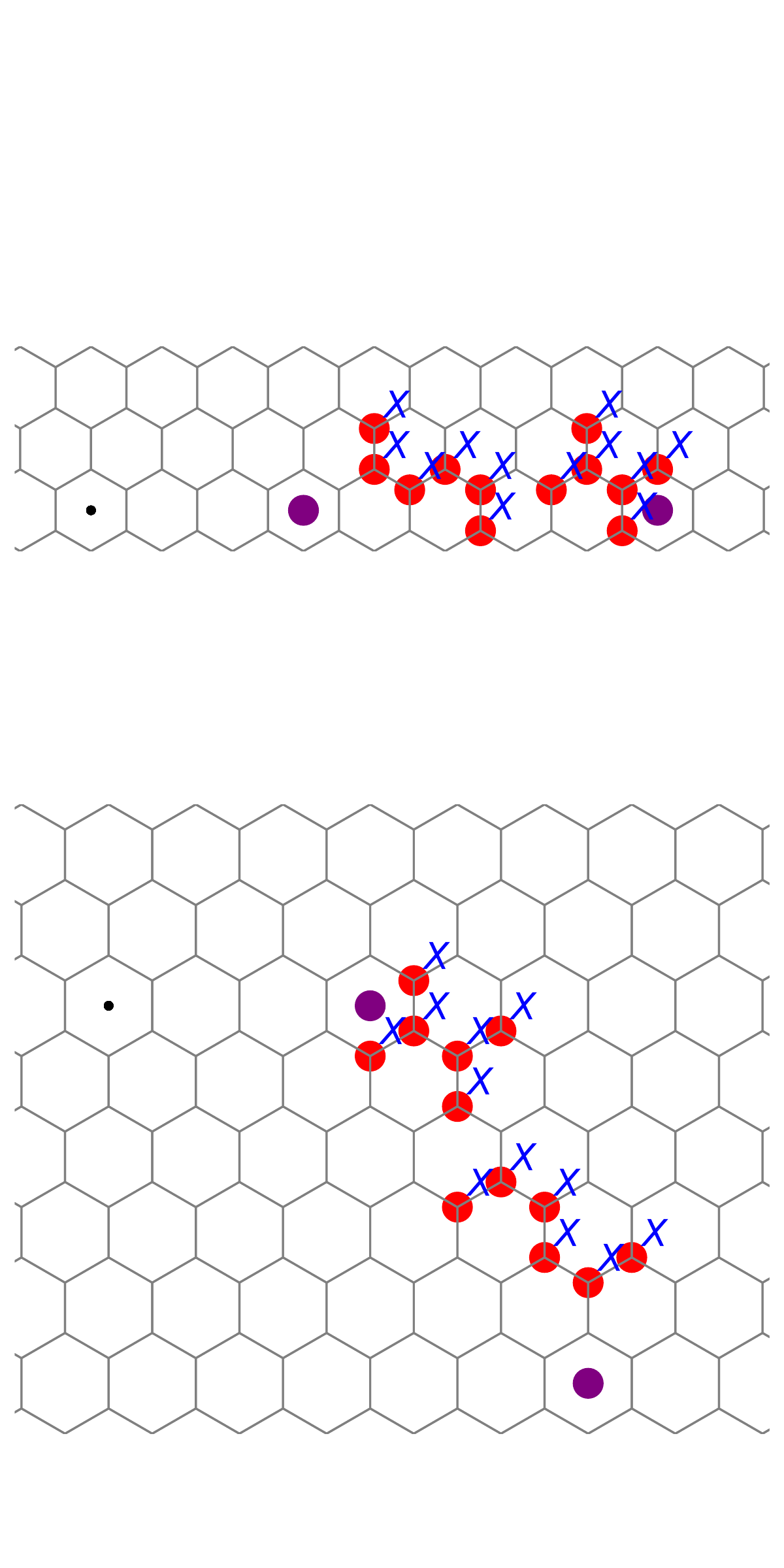}}
    \subfigure[string operator for $v_2$]{\includegraphics[width=0.24\textwidth]{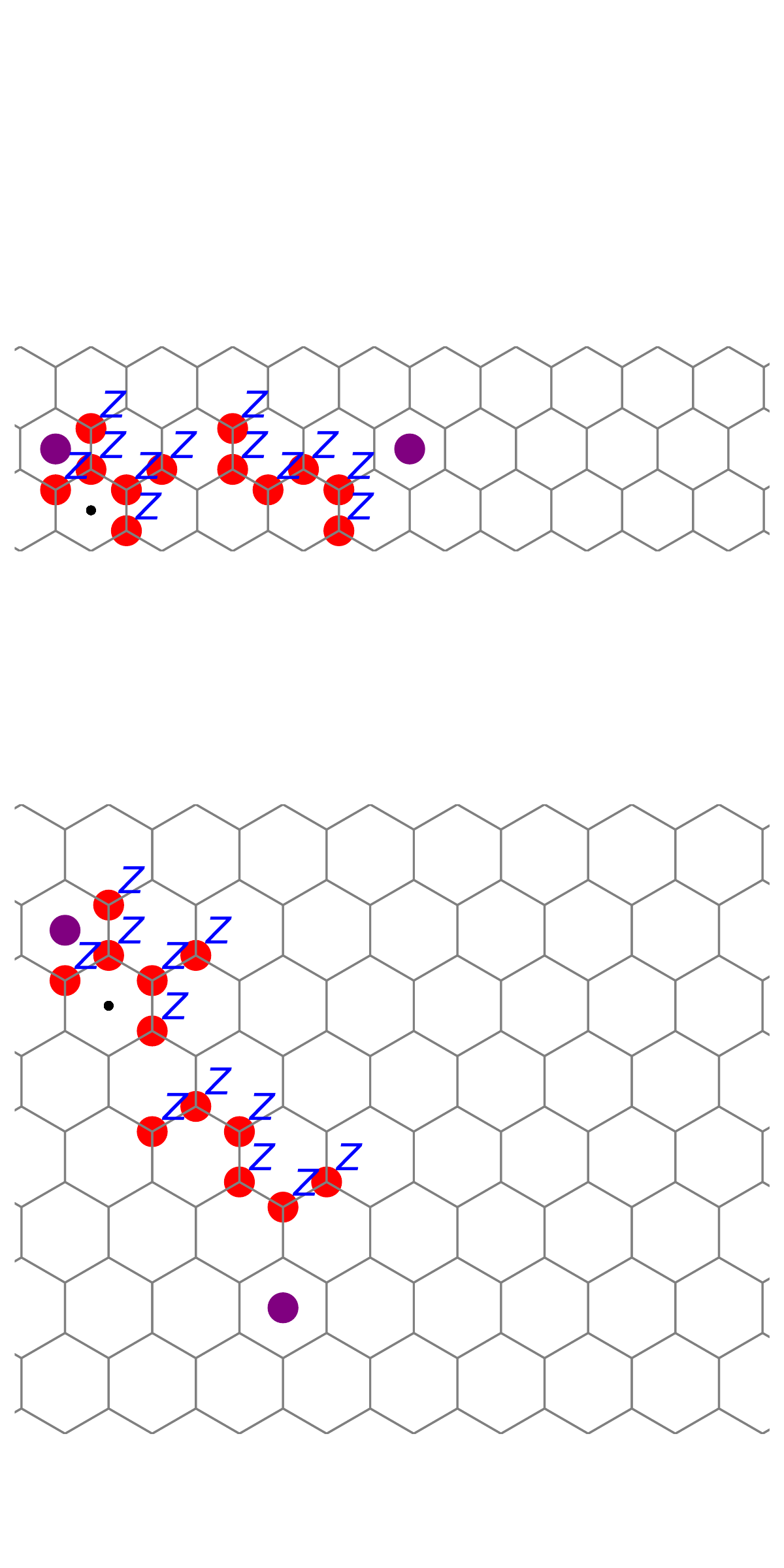}}
    \subfigure[string operator for $v_3$]{\includegraphics[width=0.24\textwidth]{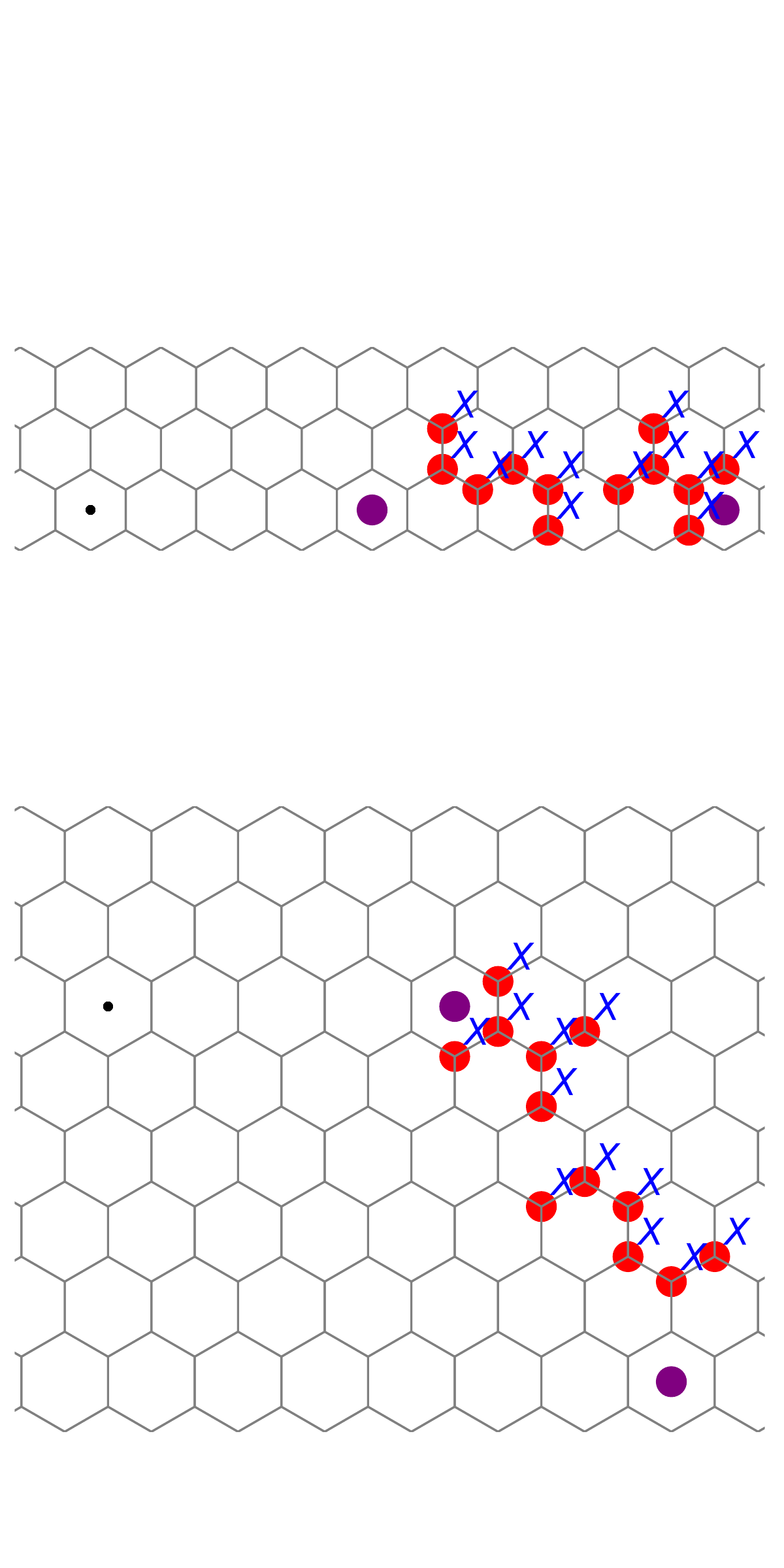}}
    \subfigure[string operator for $v_4$]{\includegraphics[width=0.24\textwidth]{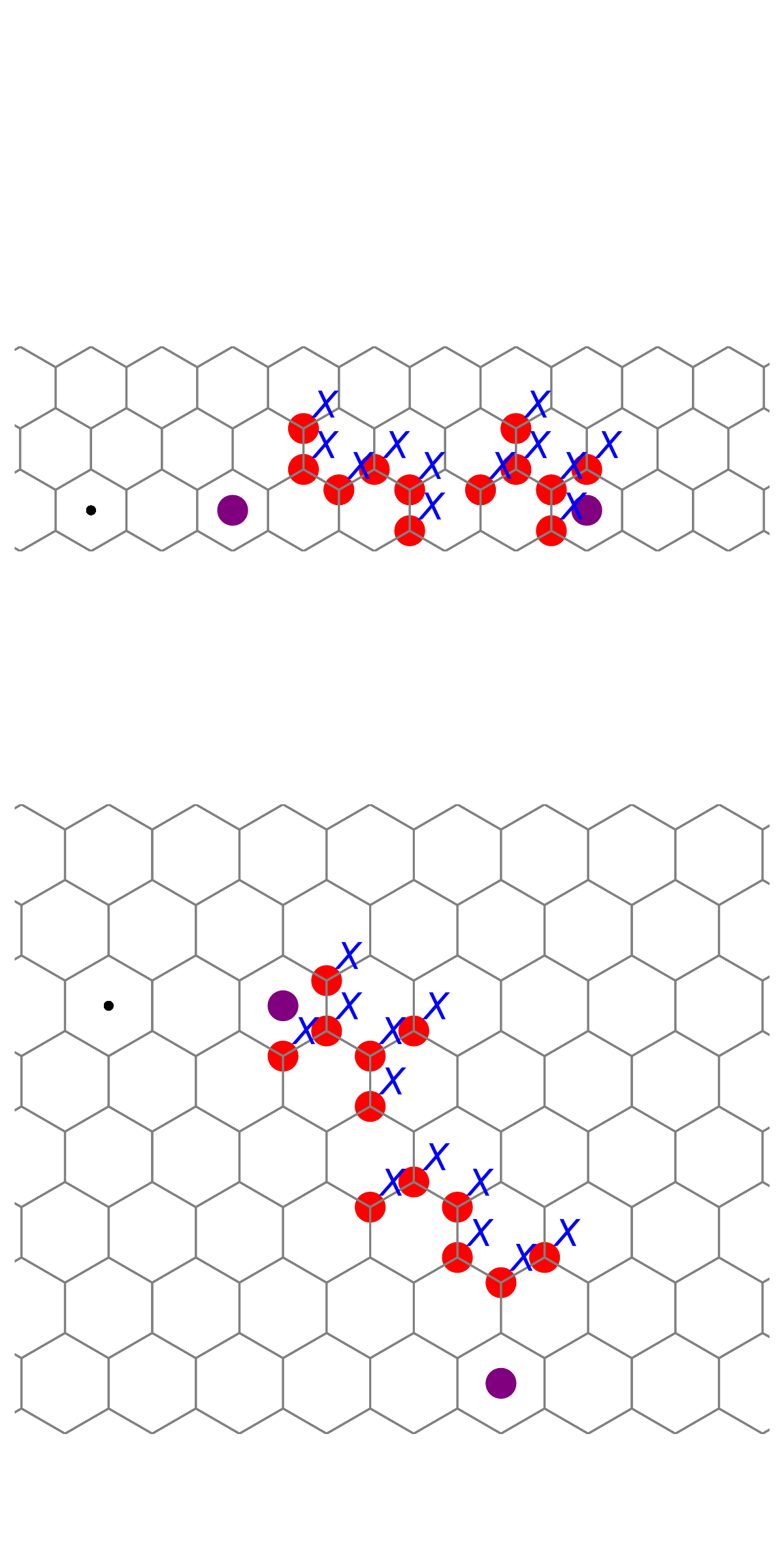}}
    \subfigure[string operator for $v_5$]{\includegraphics[width=0.24\textwidth]{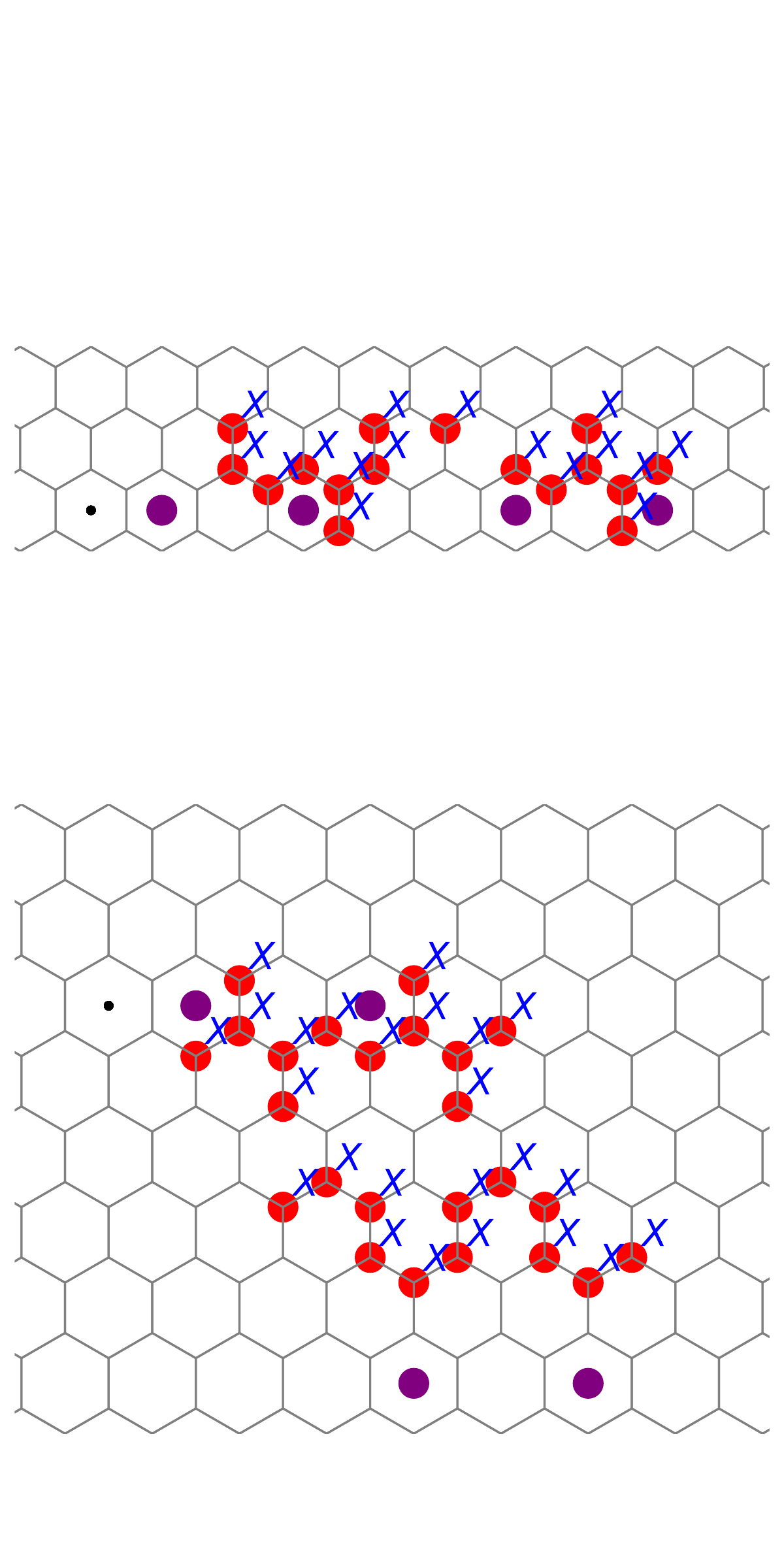}}
    \subfigure[string operator for $v_6$]{\includegraphics[width=0.24\textwidth]{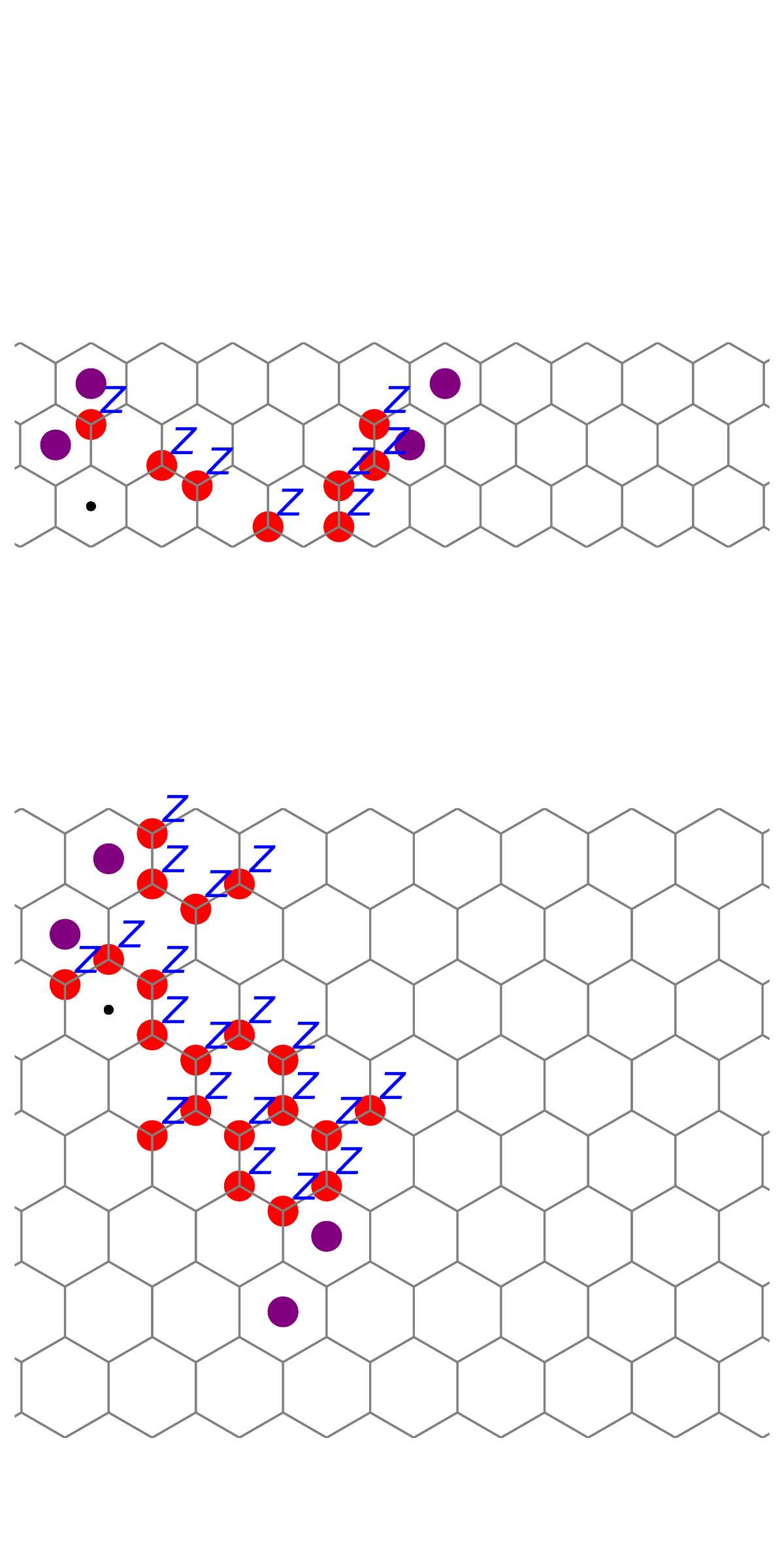}}
    \subfigure[string operator for $v_7$]{\includegraphics[width=0.24\textwidth]{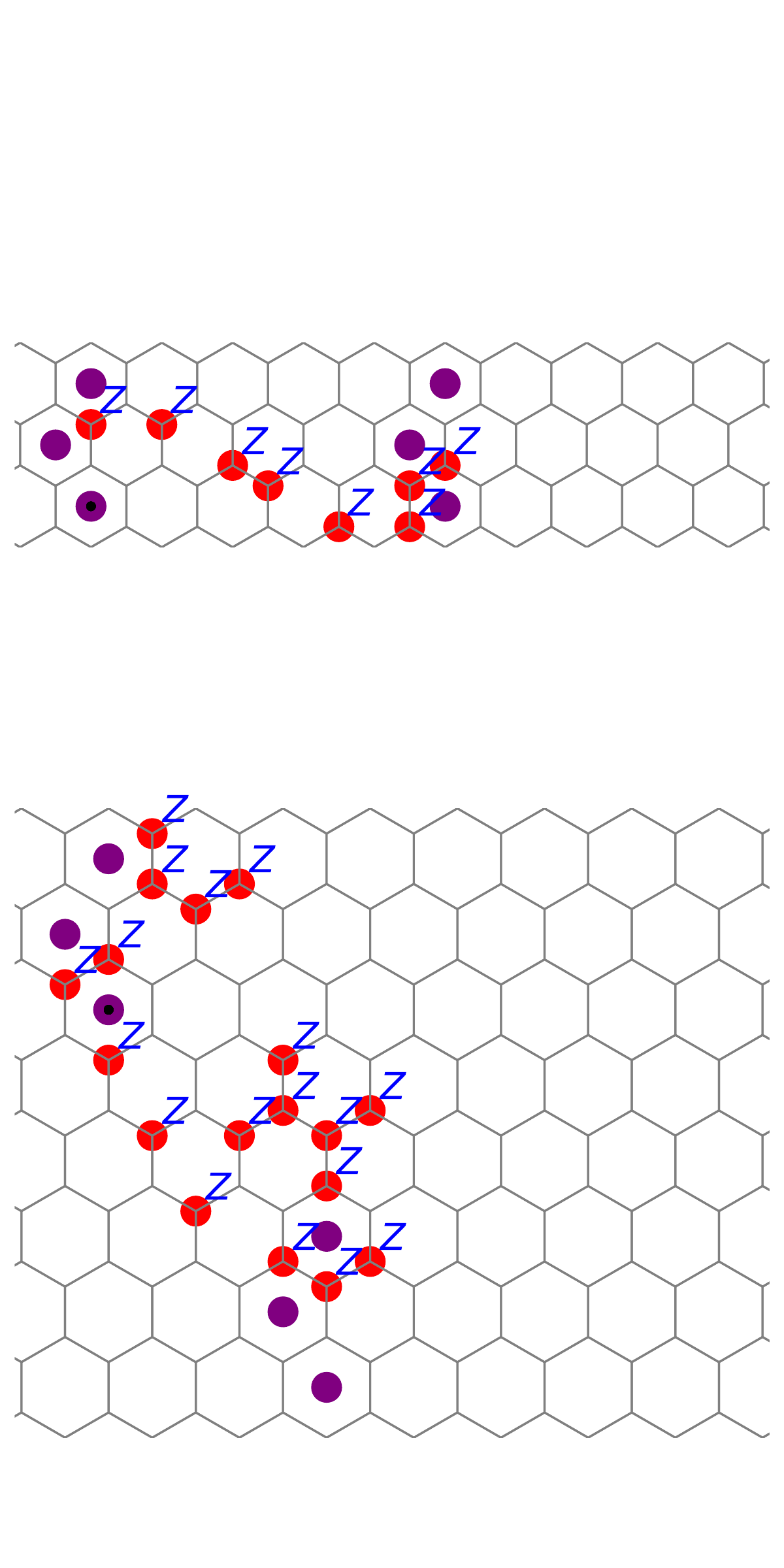}}
    \subfigure[string operator for $v_8$]{\includegraphics[width=0.24\textwidth]{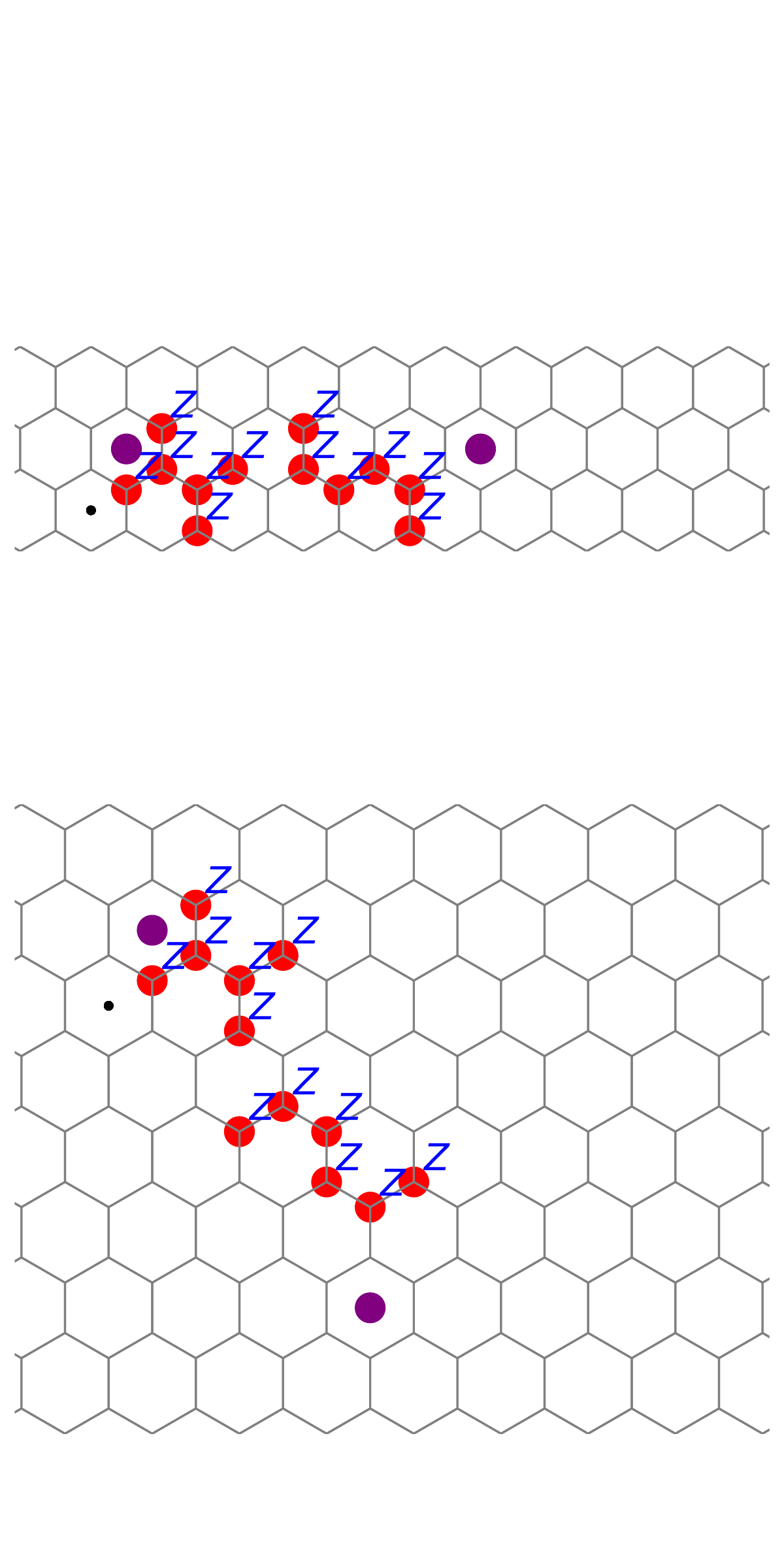}}
    \caption{The string operators in the modified color code A (example \textcircled{3}). Subfigures (a)$\sim$ (h) represent the string operators for anyons $v_1,..., v_8$. For each subfigure, the first row represents the anyon string operators that move an anyon along the $x$-direction, and the second row represents the anyon string operators that move an anyon along the $y$-direction.}
\label{fig:string_example3}
\end{figure*}

\begin{figure*}[htb]
    \centering
    \subfigure[string operator for $v_1$ ]{\includegraphics[width=0.48\textwidth]{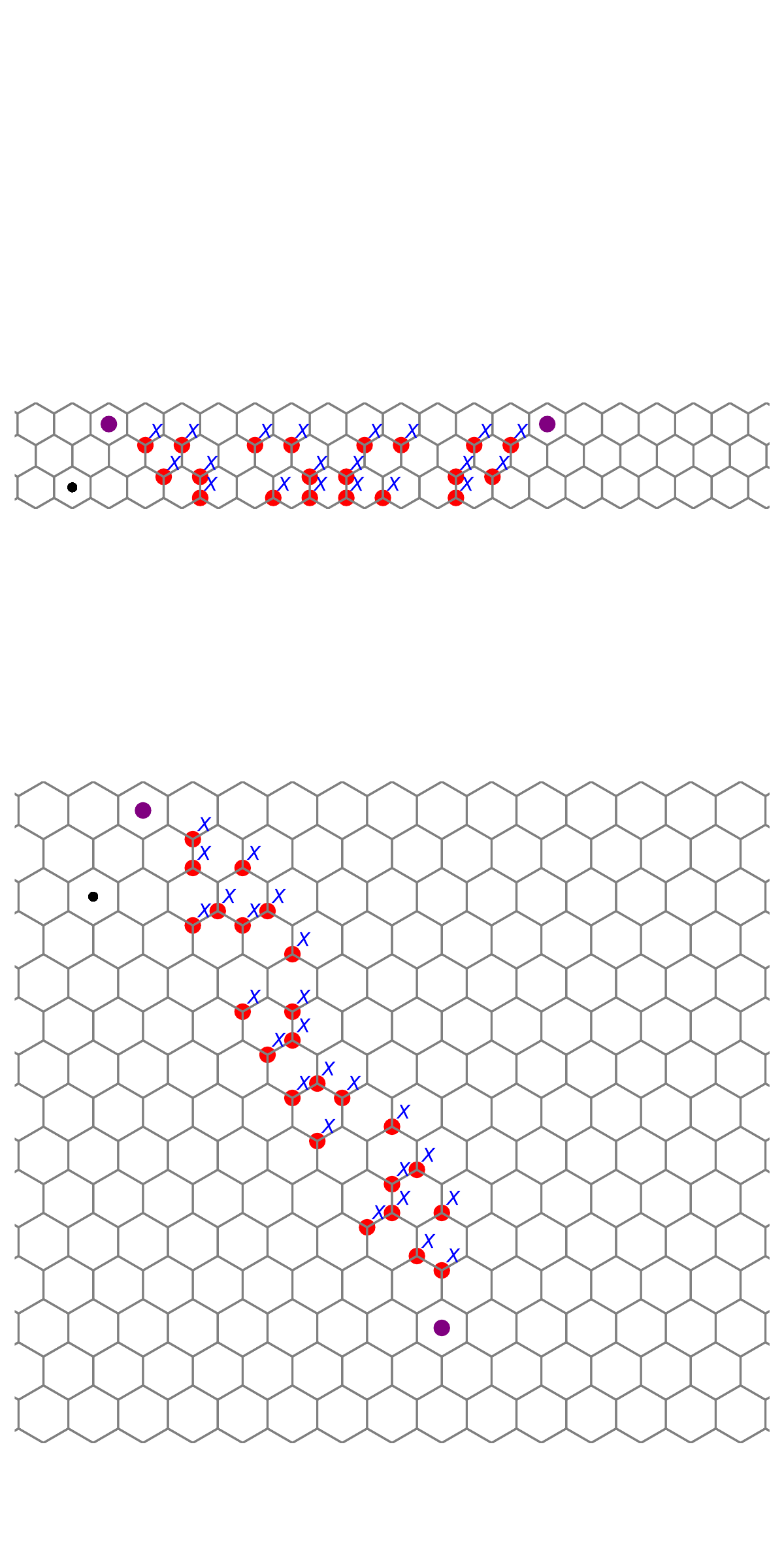}}
    \subfigure[string operator for $v_2$ ]{\includegraphics[width=0.48\textwidth]{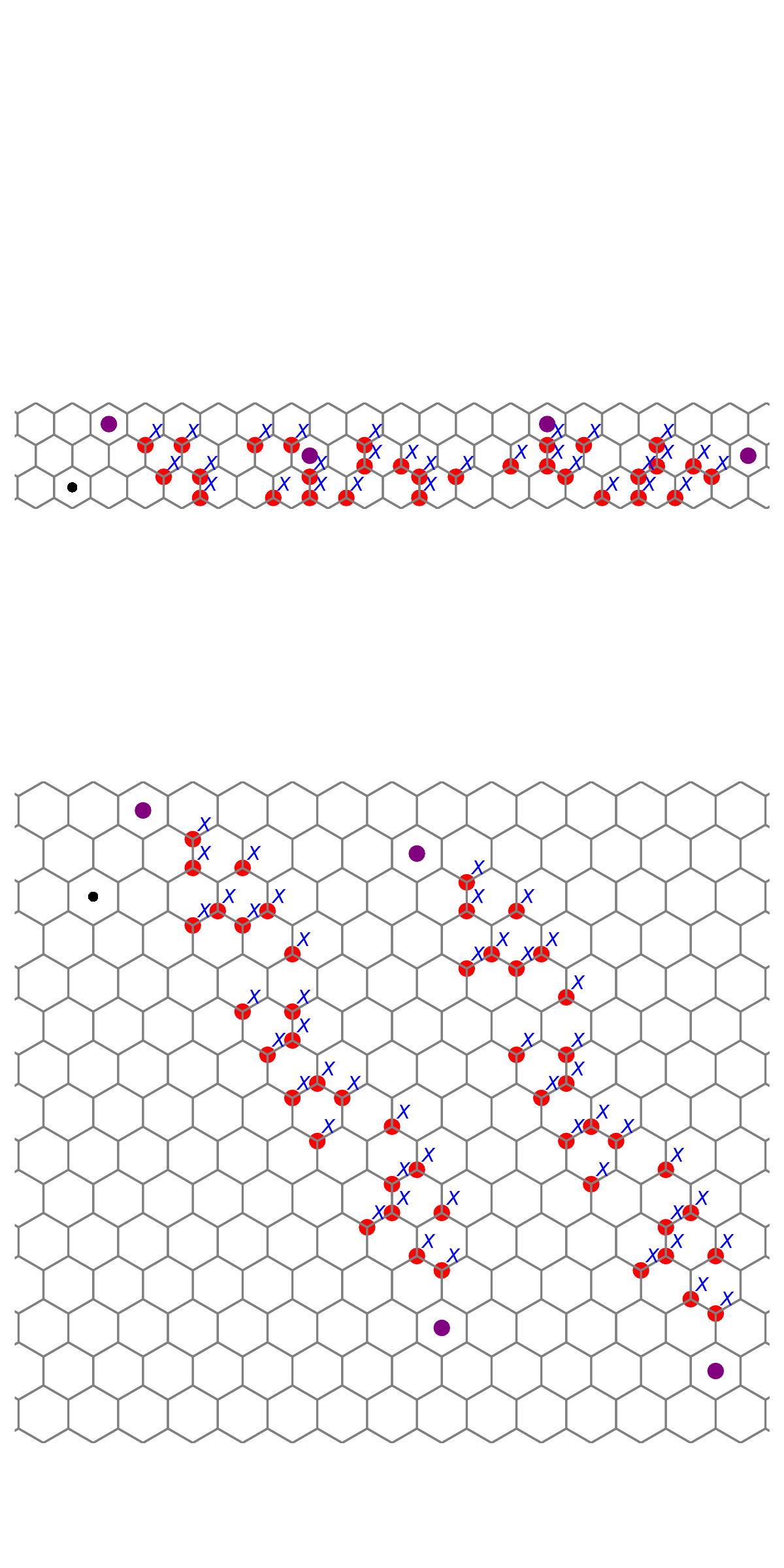}}
    \caption{The string operators in the modified color code B (example \textcircled{4}). Subfigures (a)and (b) represent the string operators for anyons $v_1$ and $v_2$. For each subfigure, the first row represents the anyon string operators that move an anyon along the $x$-direction, and the second row represents the anyon string operators that move an anyon along the $y$-direction.}
\label{fig:string_example4_1}
\end{figure*}

\begin{figure*}[htb]
    \centering
    \subfigure[string operator for $v_3$ ]{\includegraphics[width=0.48\textwidth]{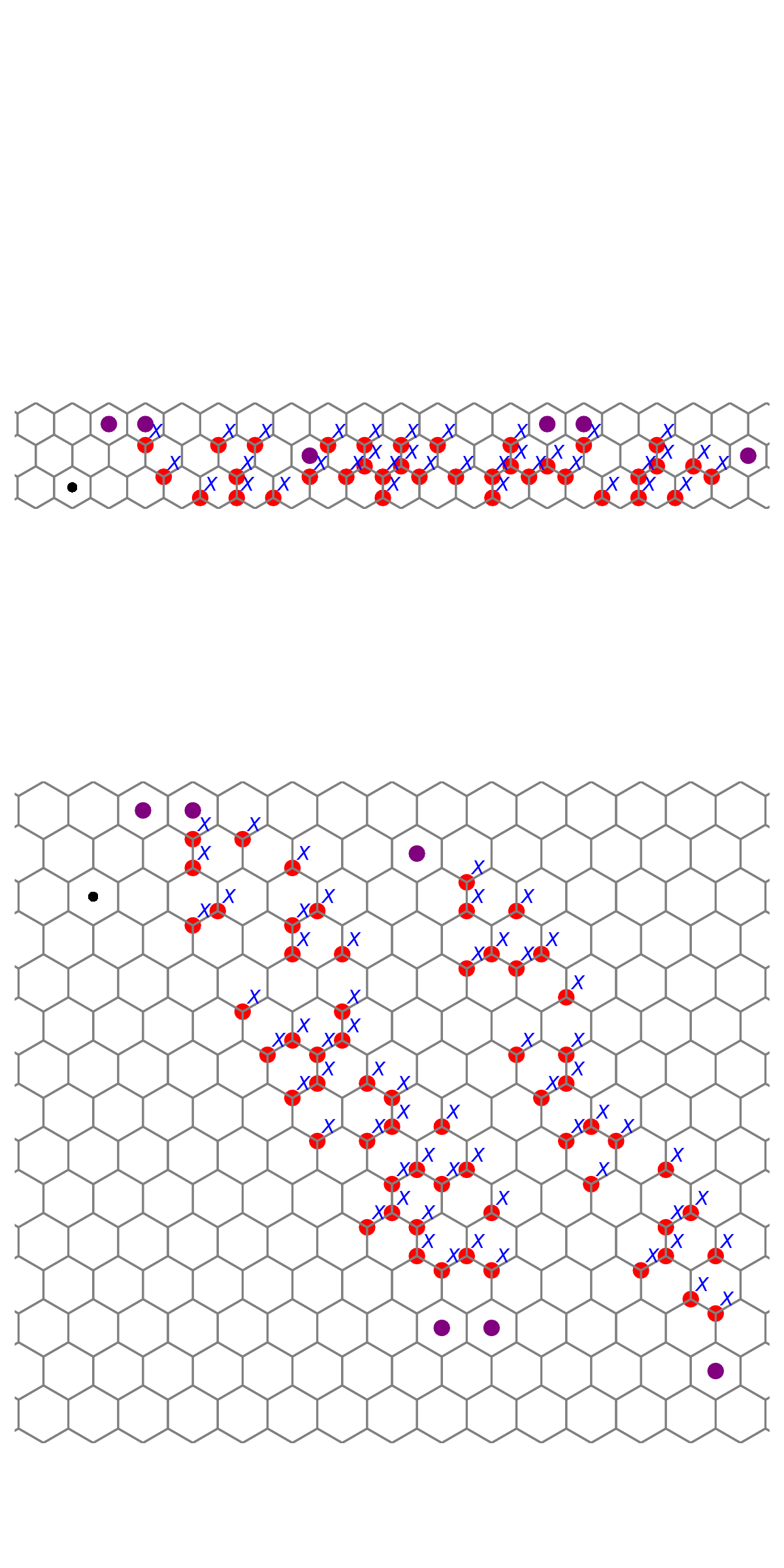}}
    \subfigure[string operator for $v_4$ ]{\includegraphics[width=0.48\textwidth]{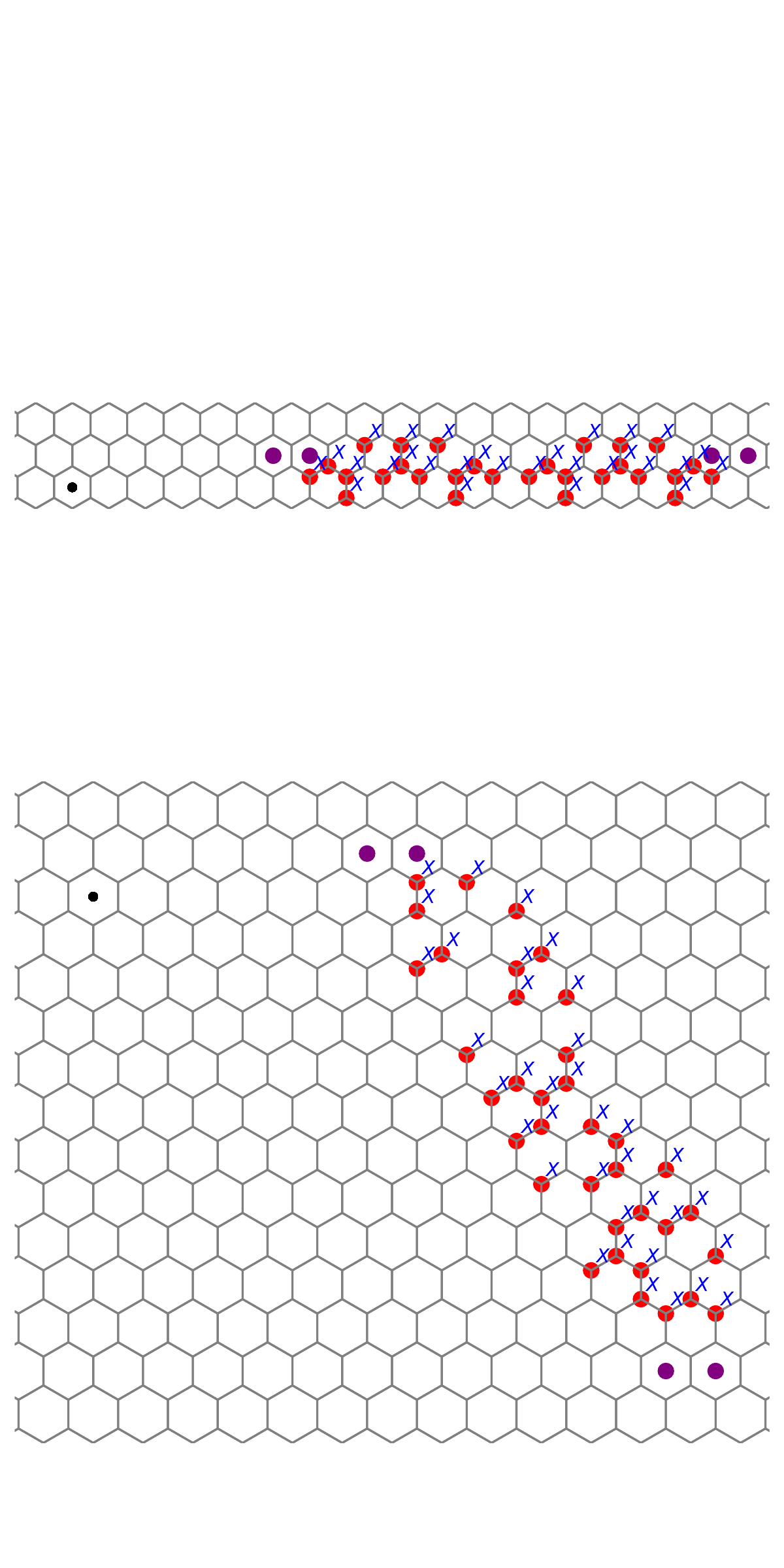}}
    \caption{(Continued) The string operators in the modified color code B (example \textcircled{4}). Subfigures (a) and (b) represent the string operators for anyons $v_3$ and $v_4$. For each subfigure, the first row represents the anyon string operators that move an anyon along the $x$-direction, and the second row represents the anyon string operators that move an anyon along the $y$-direction.}
    \label{fig:string_example4_2}
\end{figure*}

\begin{figure*}[htb]
    \subfigure[string operator for $v_5$ ]{\includegraphics[width=0.48\textwidth]{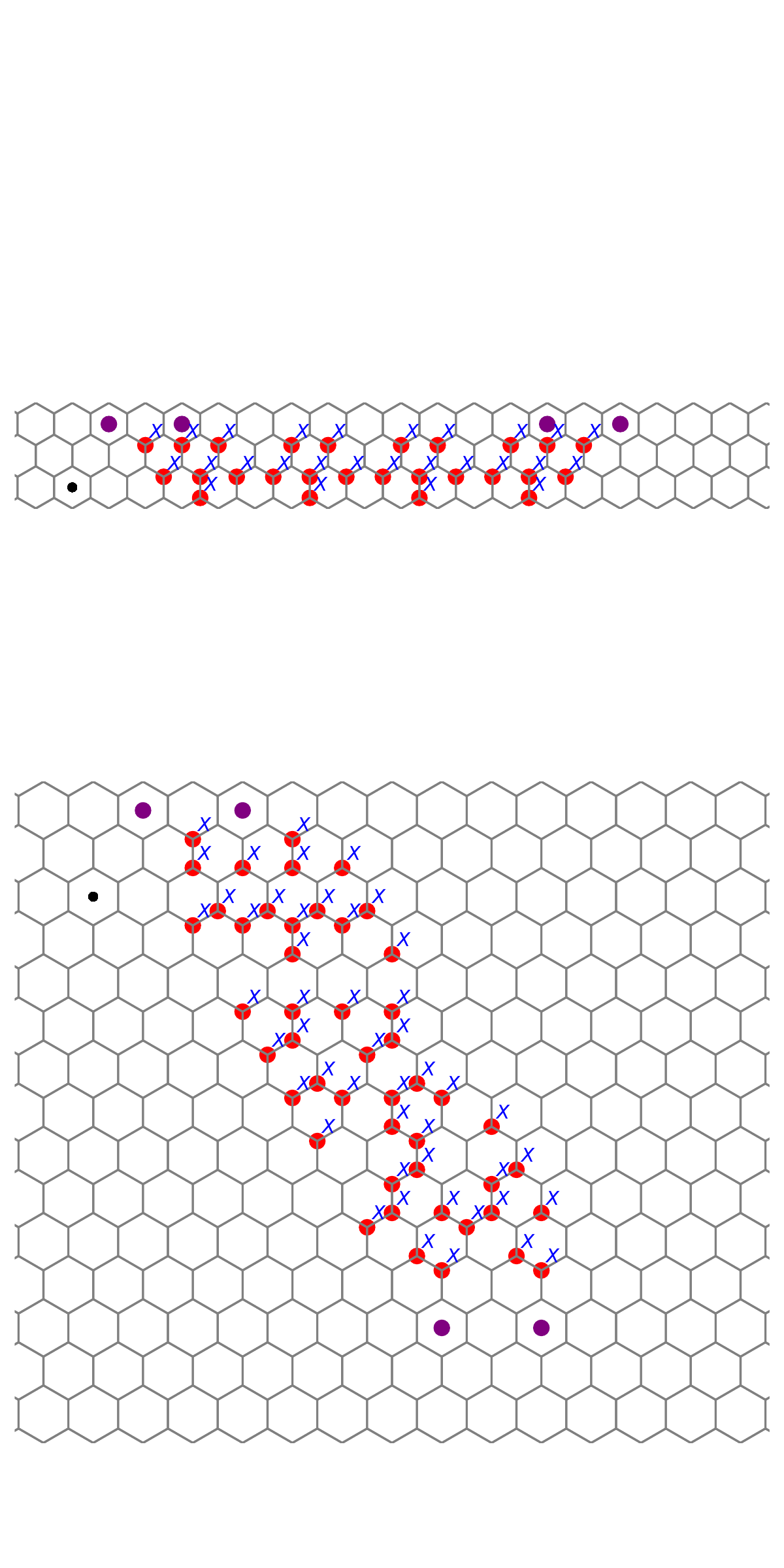}}
    \subfigure[string operator for $v_6$ ]{\includegraphics[width=0.48\textwidth]{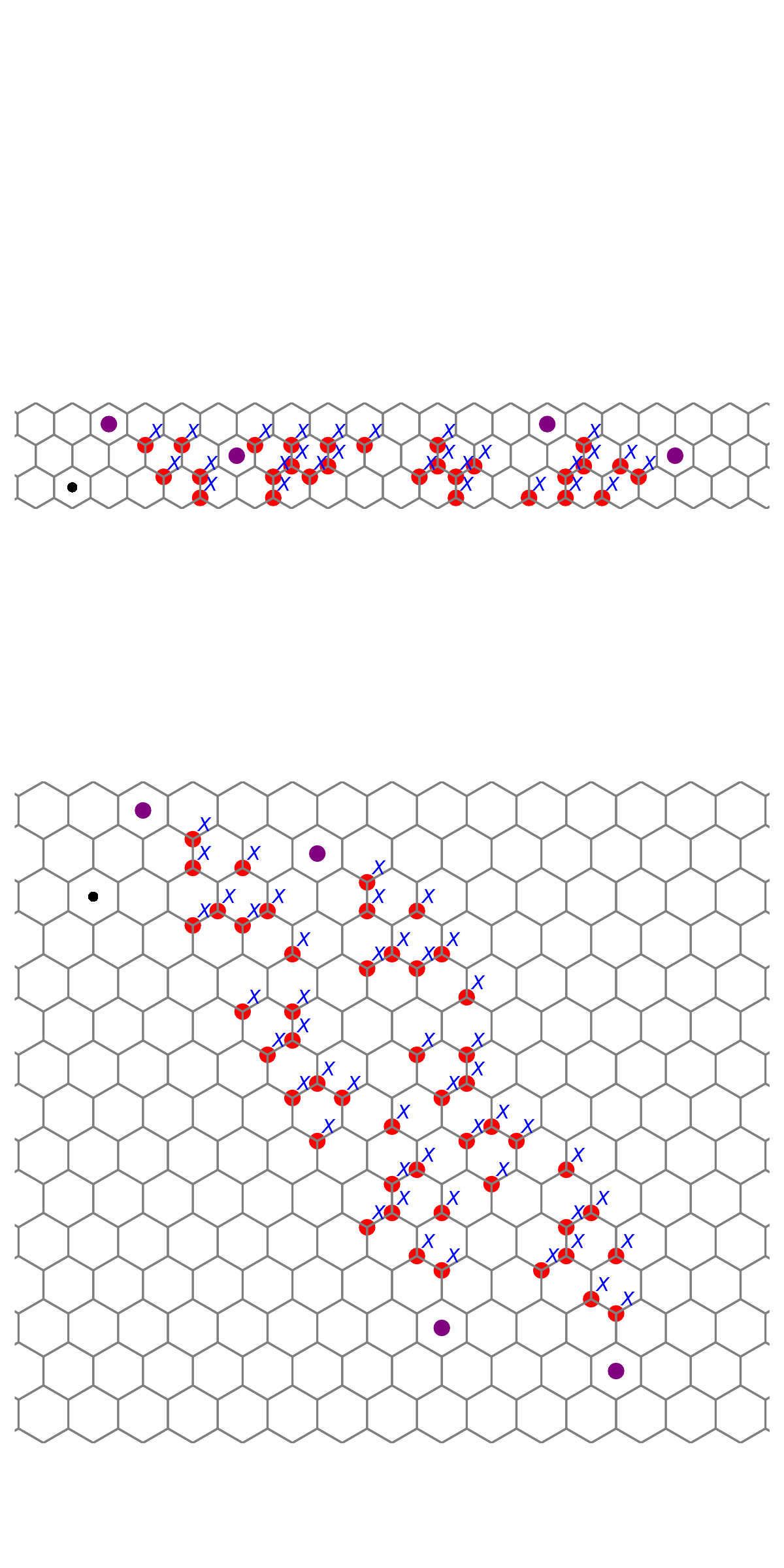}}
    \caption{(Continued) The string operators in the modified color code B (example \textcircled{4}). Subfigure (a) and (b) represent the string operators for anyons $v_5$ and $v_6$. For each subfigure, the first row represents the anyon string operators that move an anyon along the $x$-direction, and the second row represents the anyon string operators that move an anyon along the $y$-direction.}
    \label{fig:string_example4_3}
\end{figure*}

\begin{figure*}[htb]
    \subfigure[string operator for $v_7$ ]{\includegraphics[width=0.48\textwidth]{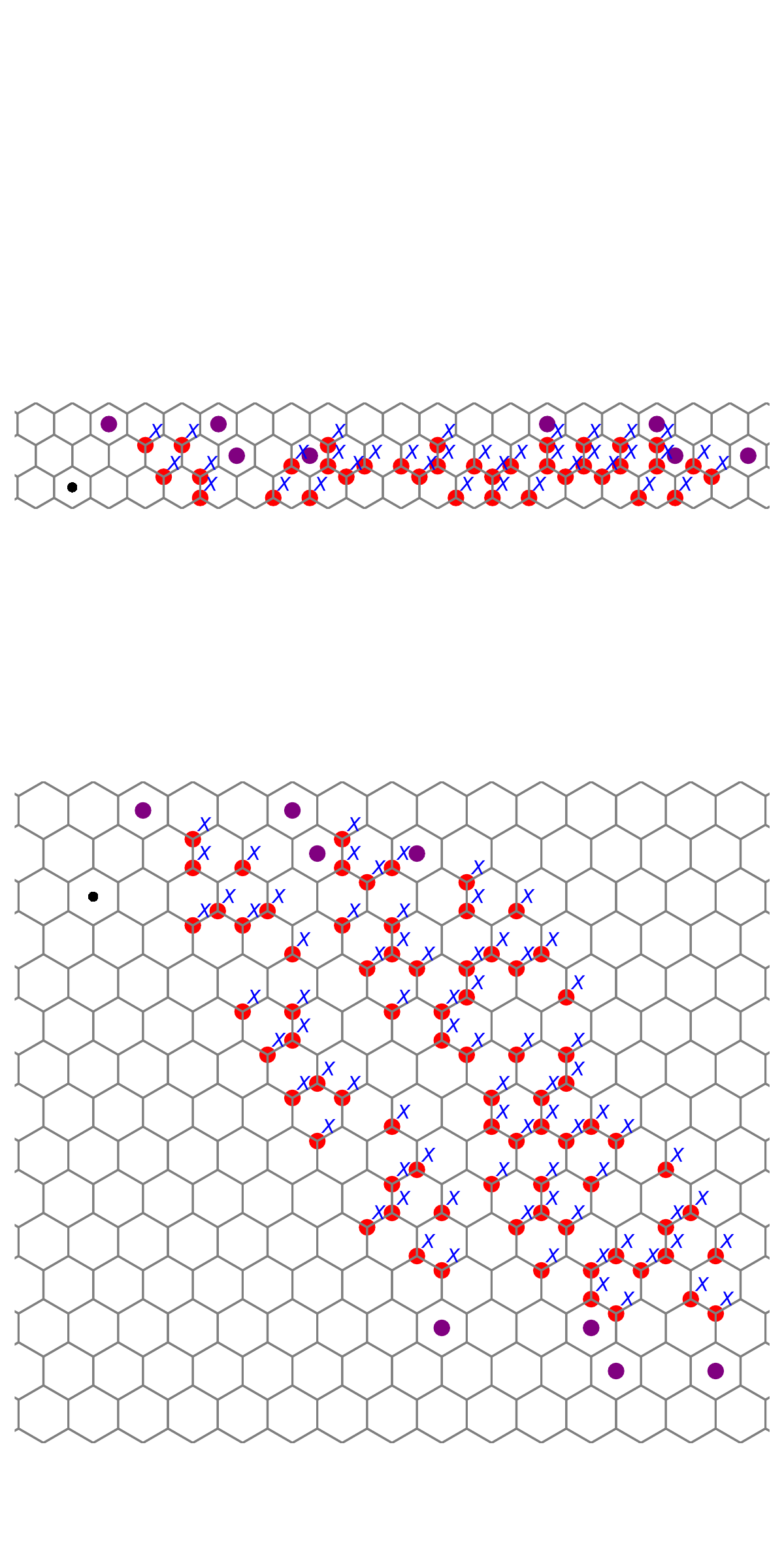}}
    \subfigure[string operator for $v_8$ ]{\includegraphics[width=0.48\textwidth]{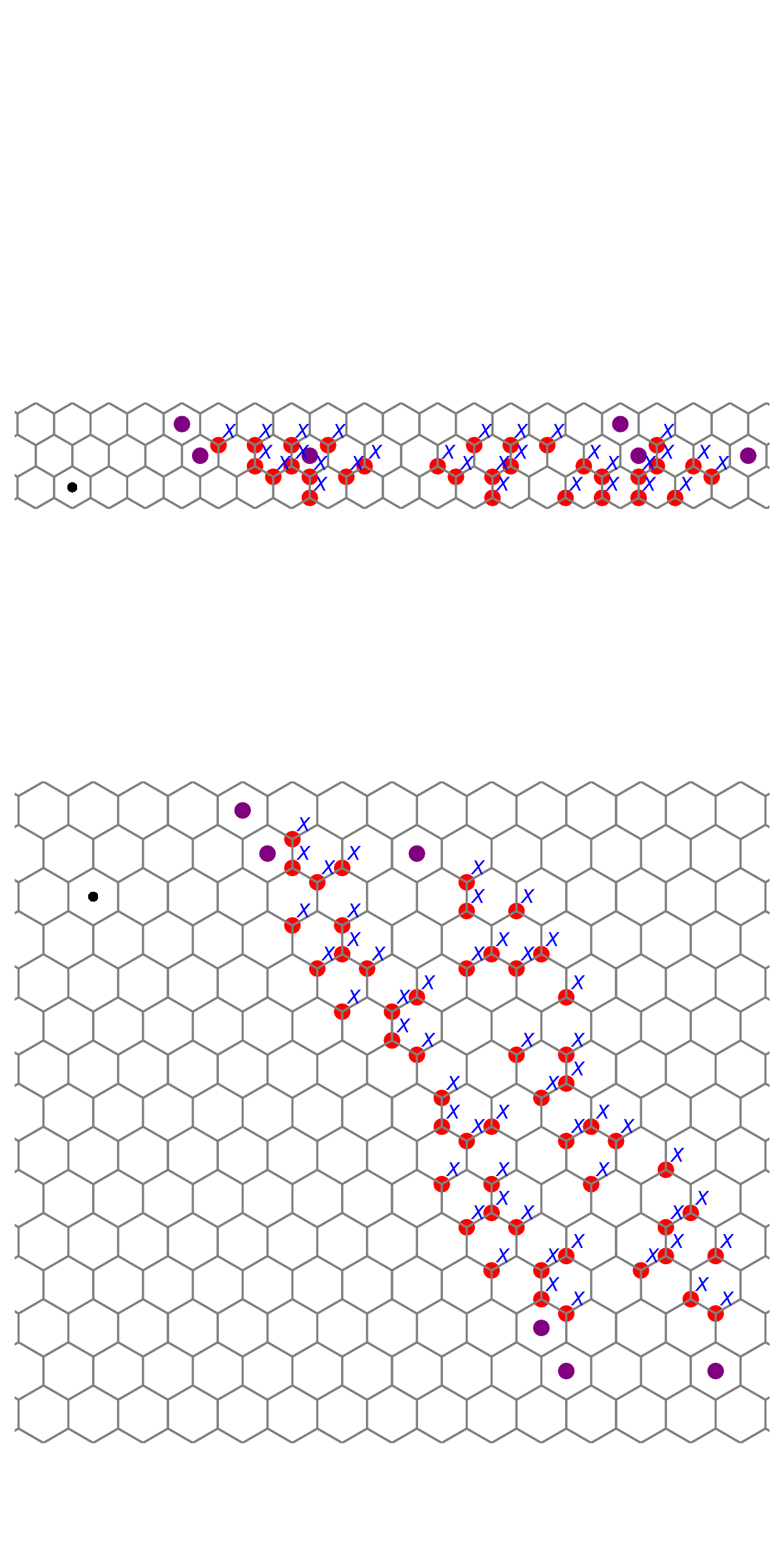}}
    \caption{(Continued) The string operators in the modified color code B (example \textcircled{4}). Subfigures (a) and (b) represent the string operators for anyons $v_7$ and $v_8$. For each subfigure, the first row represents the anyon string operators that move an anyon along the $x$-direction, and the second row represents the anyon string operators that move an anyon along the $y$-direction.}
    \label{fig:string_example4_4}
\end{figure*}

\begin{figure*}[htb]
    \centering
    \subfigure[string operator for $v_9$ ]{\includegraphics[width=0.48\textwidth]{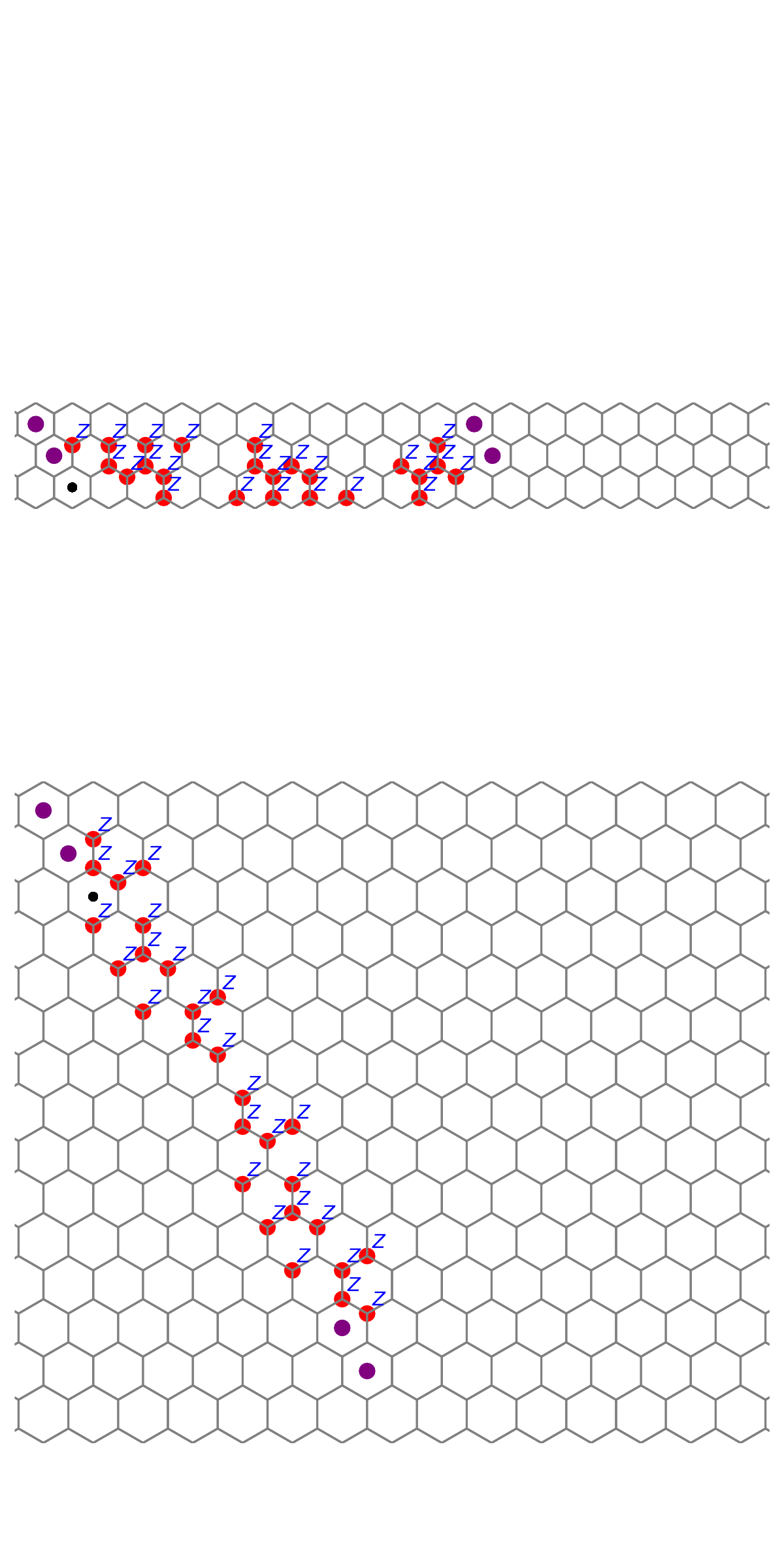}}
    \subfigure[string operator for $v_{10}$ ]{\includegraphics[width=0.48\textwidth]{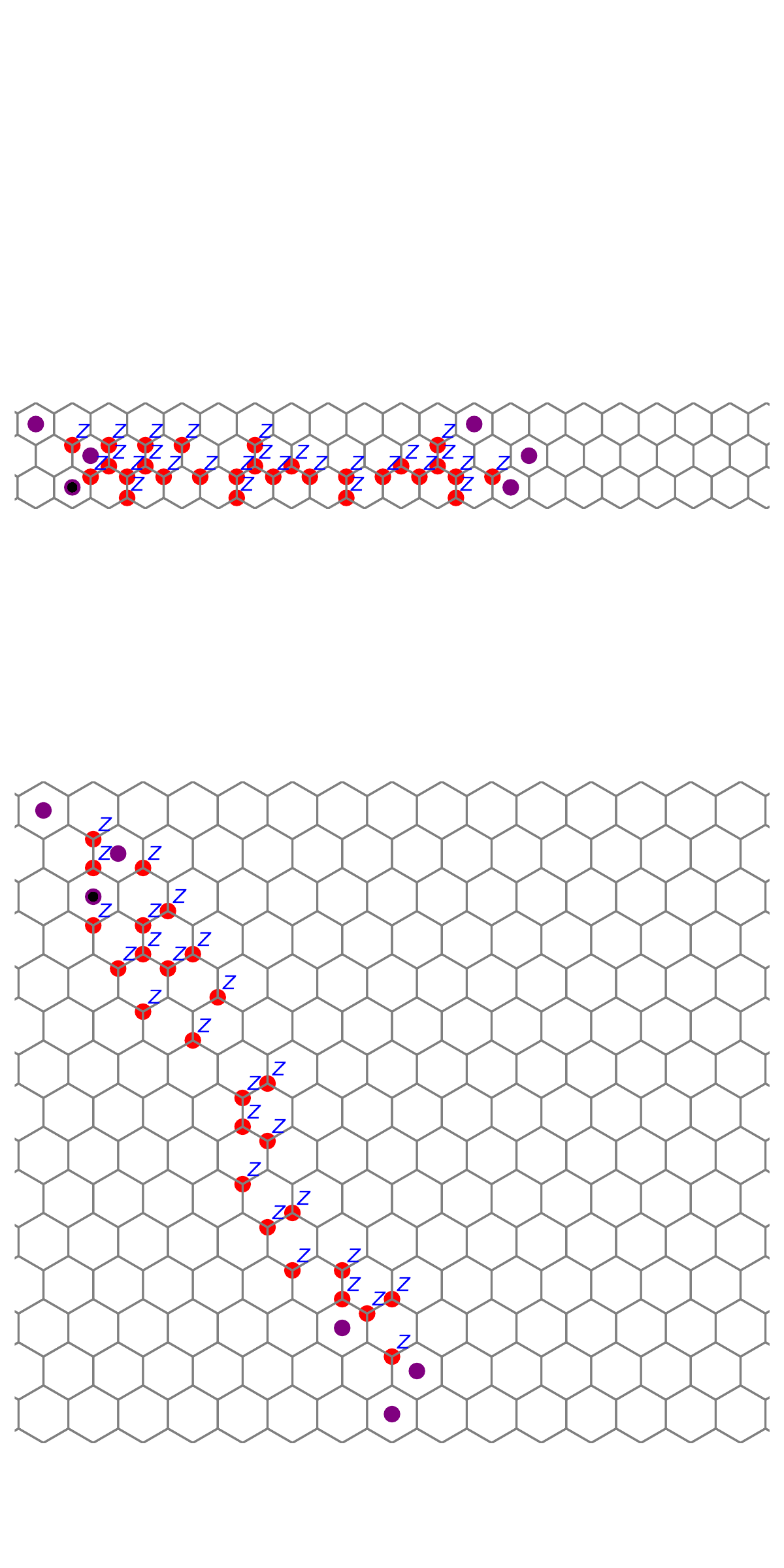}}
    \caption{(Continued) The string operators in of the modified color code B (example \textcircled{4}). Subfigure (a) and (b) represent the string operators for anyons $v_9$ and $v_{10}$. For each subfigure, the first row represents the anyon string operators that move an anyon along the $x$-direction, and the second row represents the anyon string operators that move an anyon along the $y$-direction.}
    \label{fig:string_example4_5}
\end{figure*}

\begin{figure*}
    \centering
    \subfigure[string operator for $v_{11}$ ]{\includegraphics[width=0.48\textwidth]{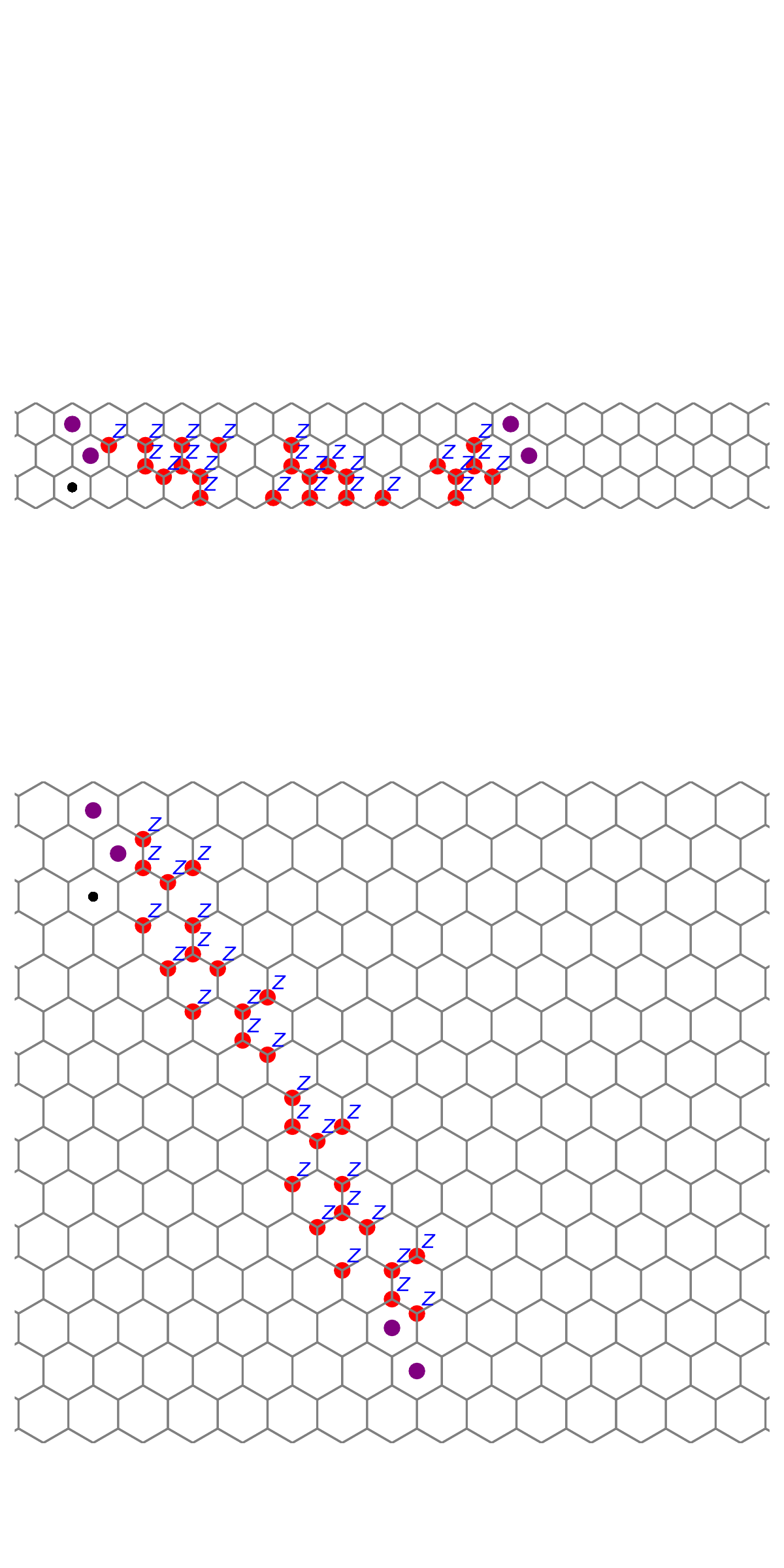}}
    \subfigure[string operator for $v_{12}$ ]{\includegraphics[width=0.48\textwidth]{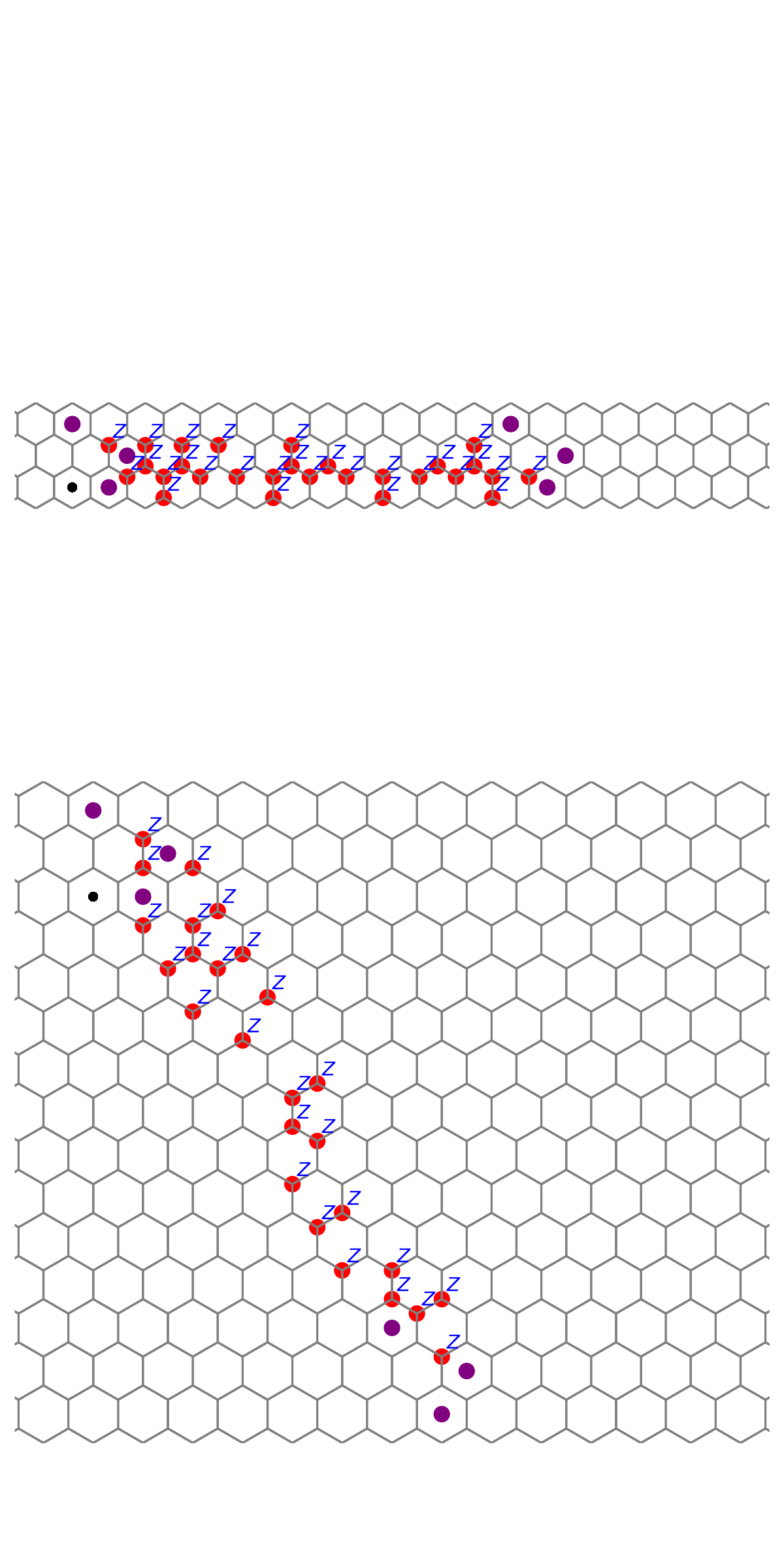}}
    \caption{(Continued) The string operators in the modified color code B (example \textcircled{4}). Subfigure (a) and (b) represent the string operators for anyons $v_{11}$ and $v_{12}$. For each subfigure, the first row represents the anyon string operators that move an anyon along the $x$-direction, and the second row represents the anyon string operators that move an anyon along the $y$-direction.}
    \label{fig:string_example4_6}
    
\end{figure*}

\begin{figure*}
    \centering
    \subfigure[string operator for $v_{13}$ ]{\includegraphics[width=0.48\textwidth]{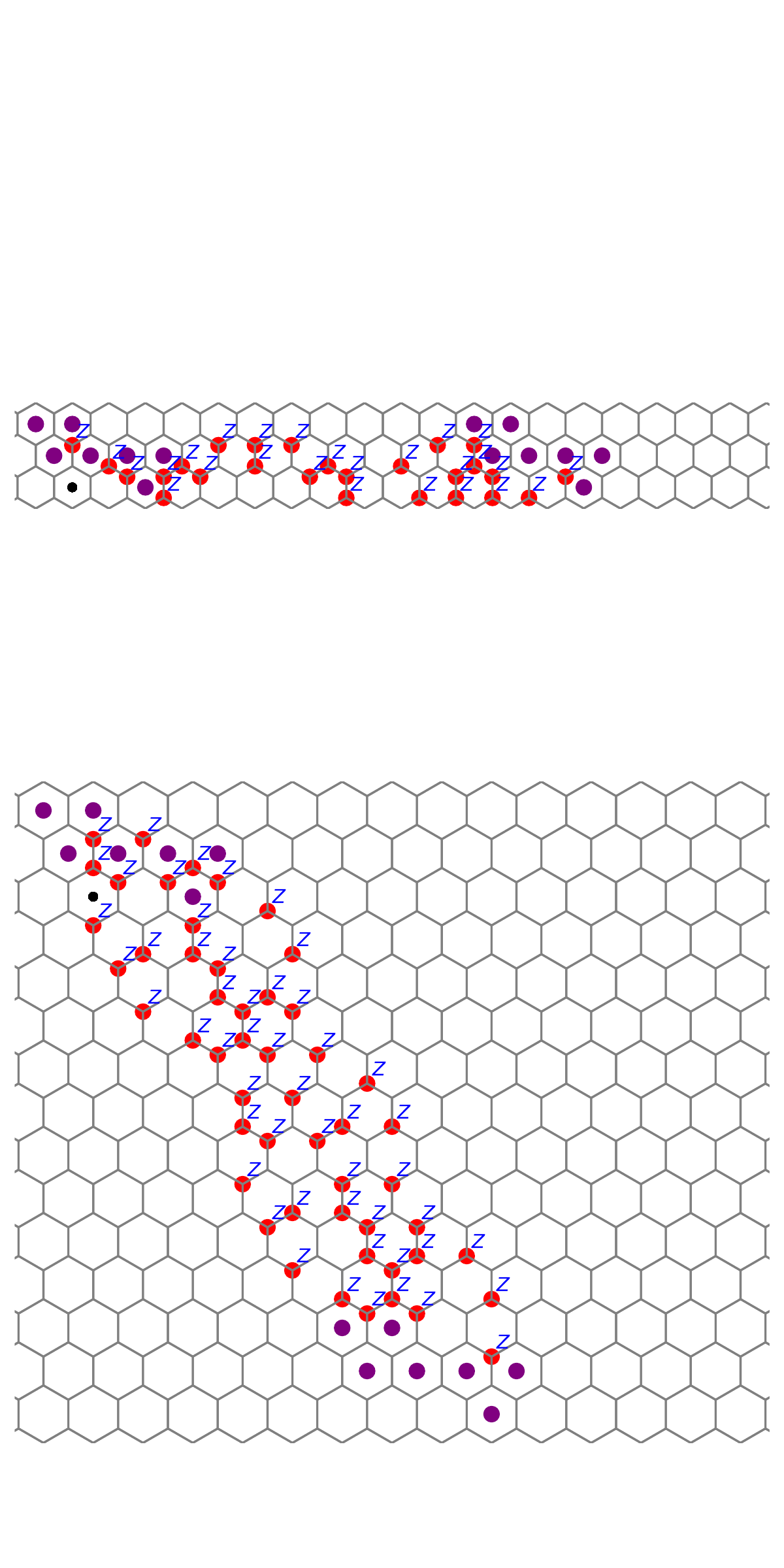}}
    \subfigure[string operator for $v_{14}$ ]{\includegraphics[width=0.48\textwidth]{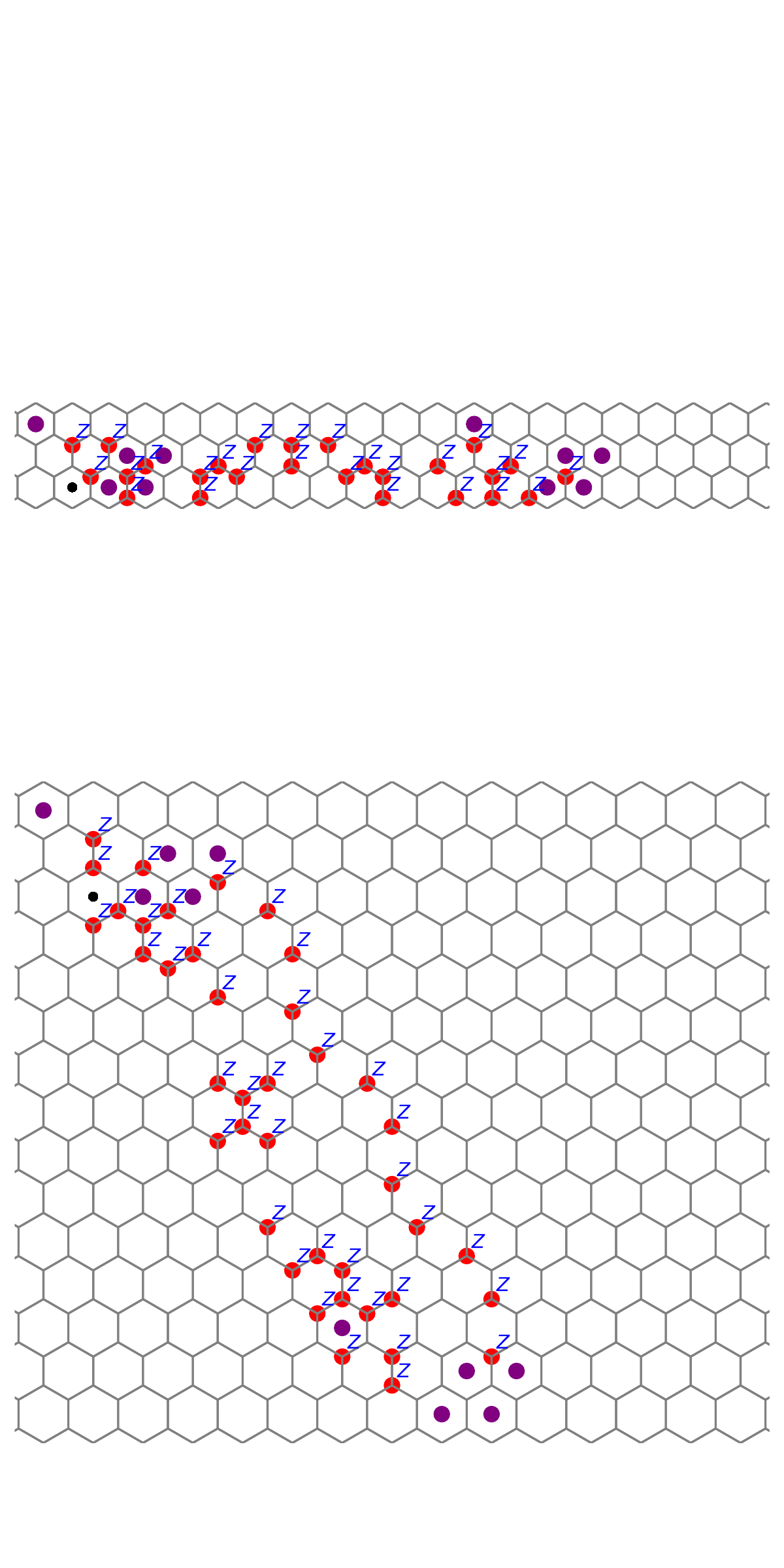}}
    \caption{(Continued) The string operators in the modified color code B (example \textcircled{4}). Subfigures (a) and (b) represent the string operators for anyons $v_{13}$ and $v_{14}$. For each subfigure, the first row represents the anyon string operators that move an anyon along the $x$-direction, and the second row represents the anyon string operators that move an anyon along the $y$-direction.}
    \label{fig:string_example4_7}
\end{figure*}

\begin{figure*}
    \centering
    \subfigure[string operator for $v_{15}$ ]{\includegraphics[width=0.48\textwidth]{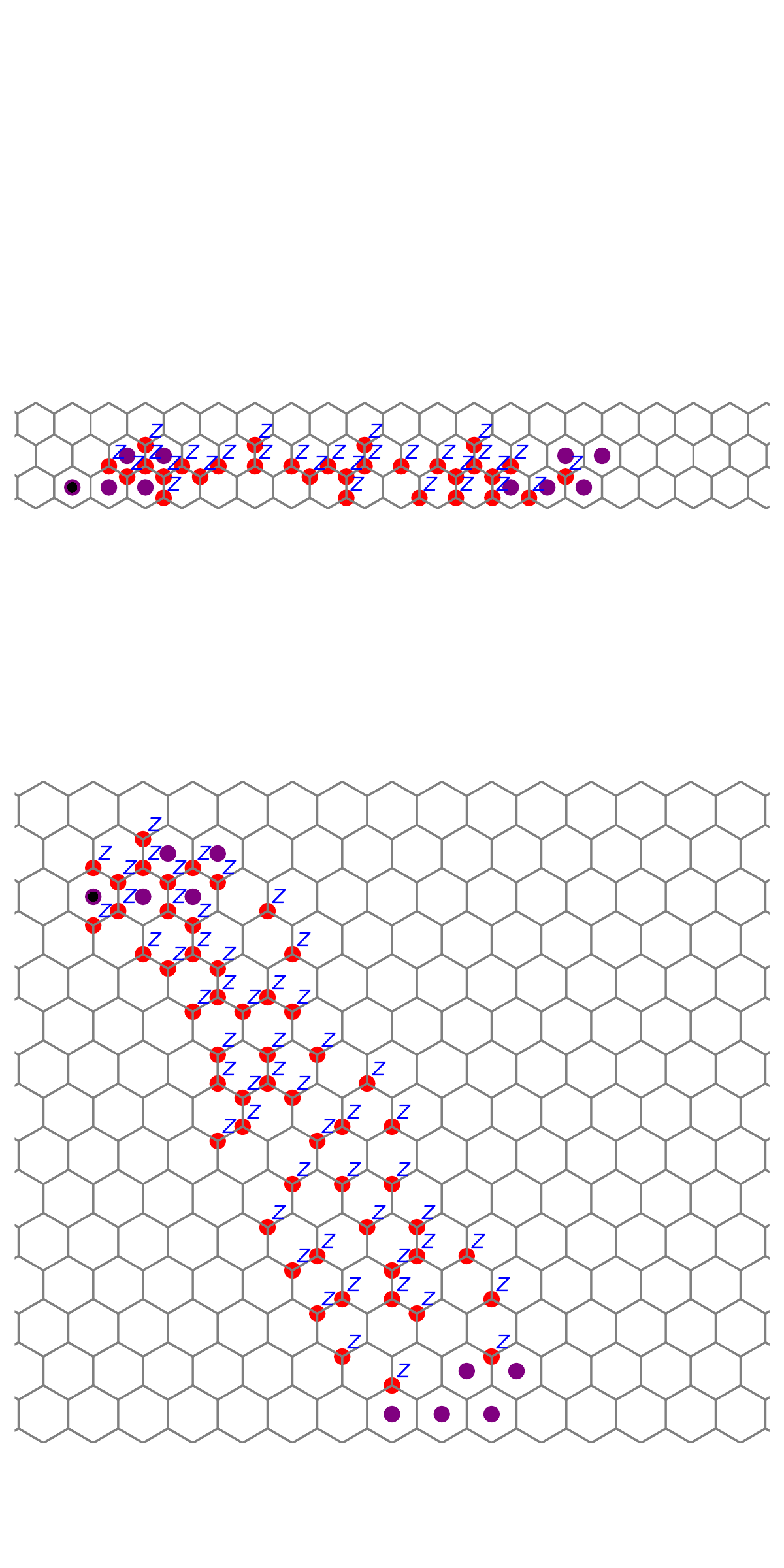}}
    \subfigure[string operator for $v_{16}$ ]{\includegraphics[width=0.48\textwidth]{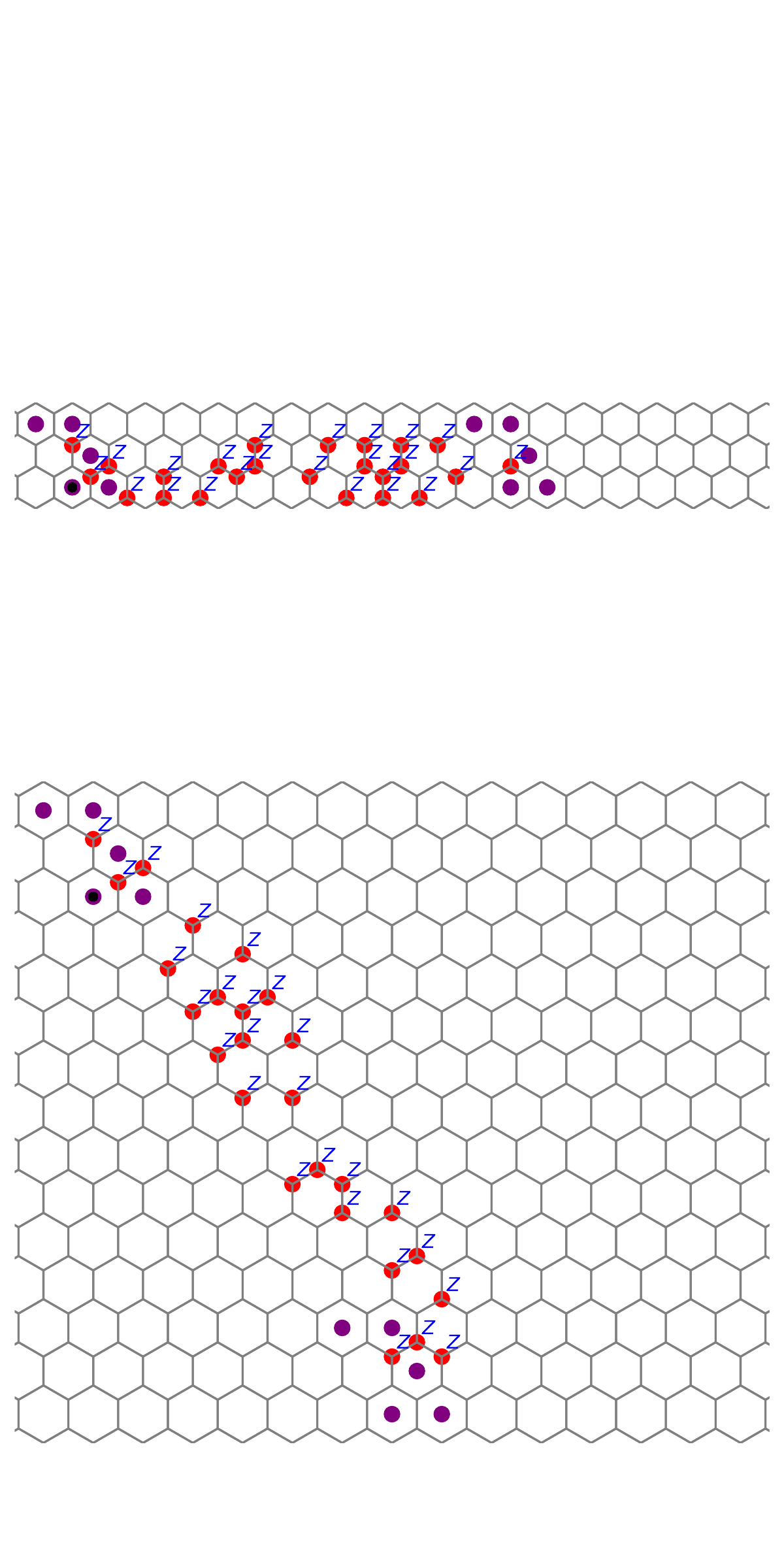}}
    \caption{(Continued) The string operators in the modified color code B (example \textcircled{4}). Subfigure (a) and (b) represent the string operators for anyons $v_{15}$ and $v_{16}$. For each subfigure, the first row represents the anyon string operators that move an anyon along the $x$-direction, and the second row represents the anyon string operators that move an anyon along the $y$-direction.}
    \label{fig:string_example4_8}
\end{figure*}

\begin{figure*}[htb]
    \centering
    \subfigure[string operator for $v_1$ ]{\includegraphics[width=0.24\textwidth]{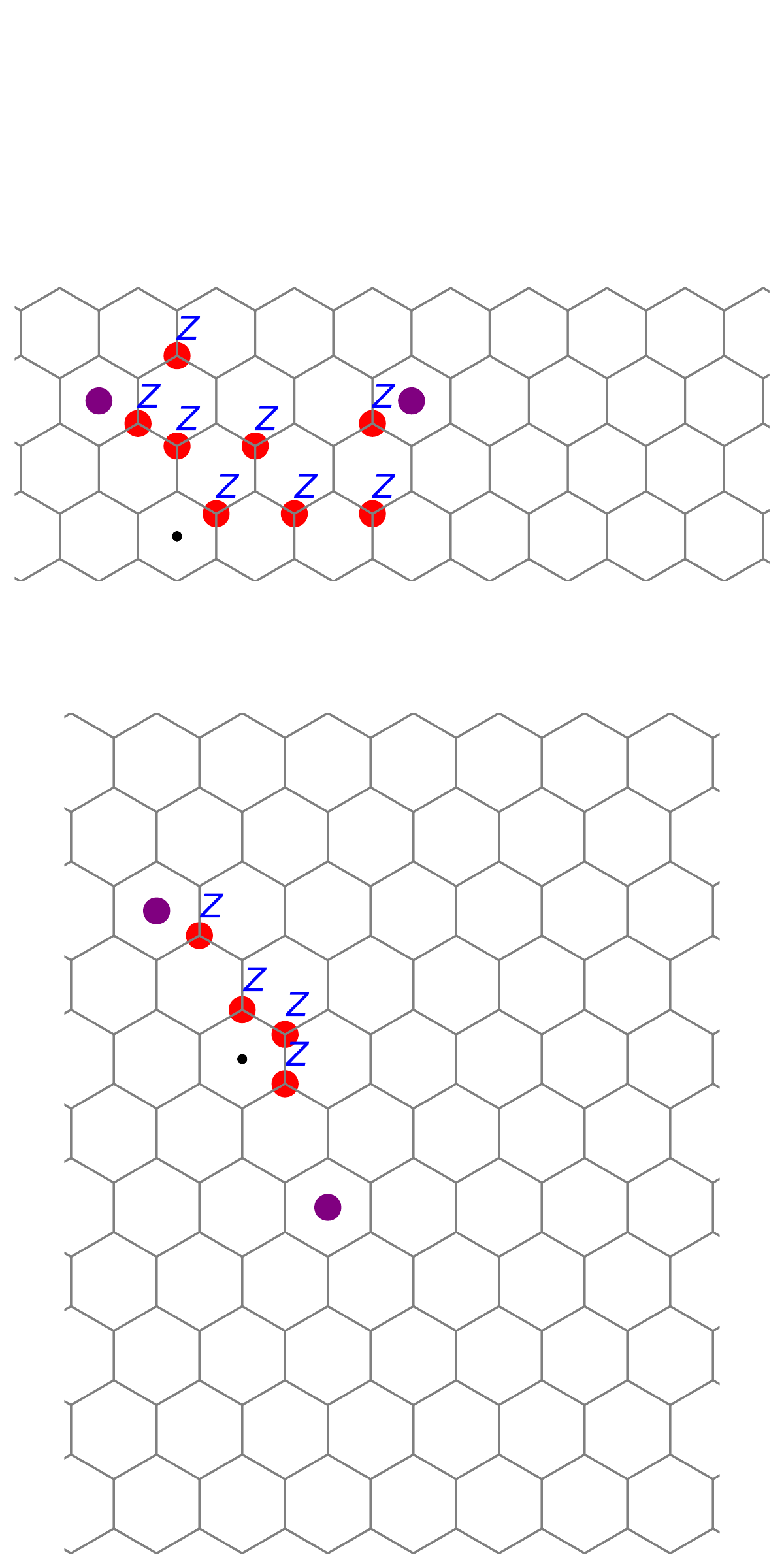}}
    \subfigure[string operator for $v_2$ ]{\includegraphics[width=0.24\textwidth]{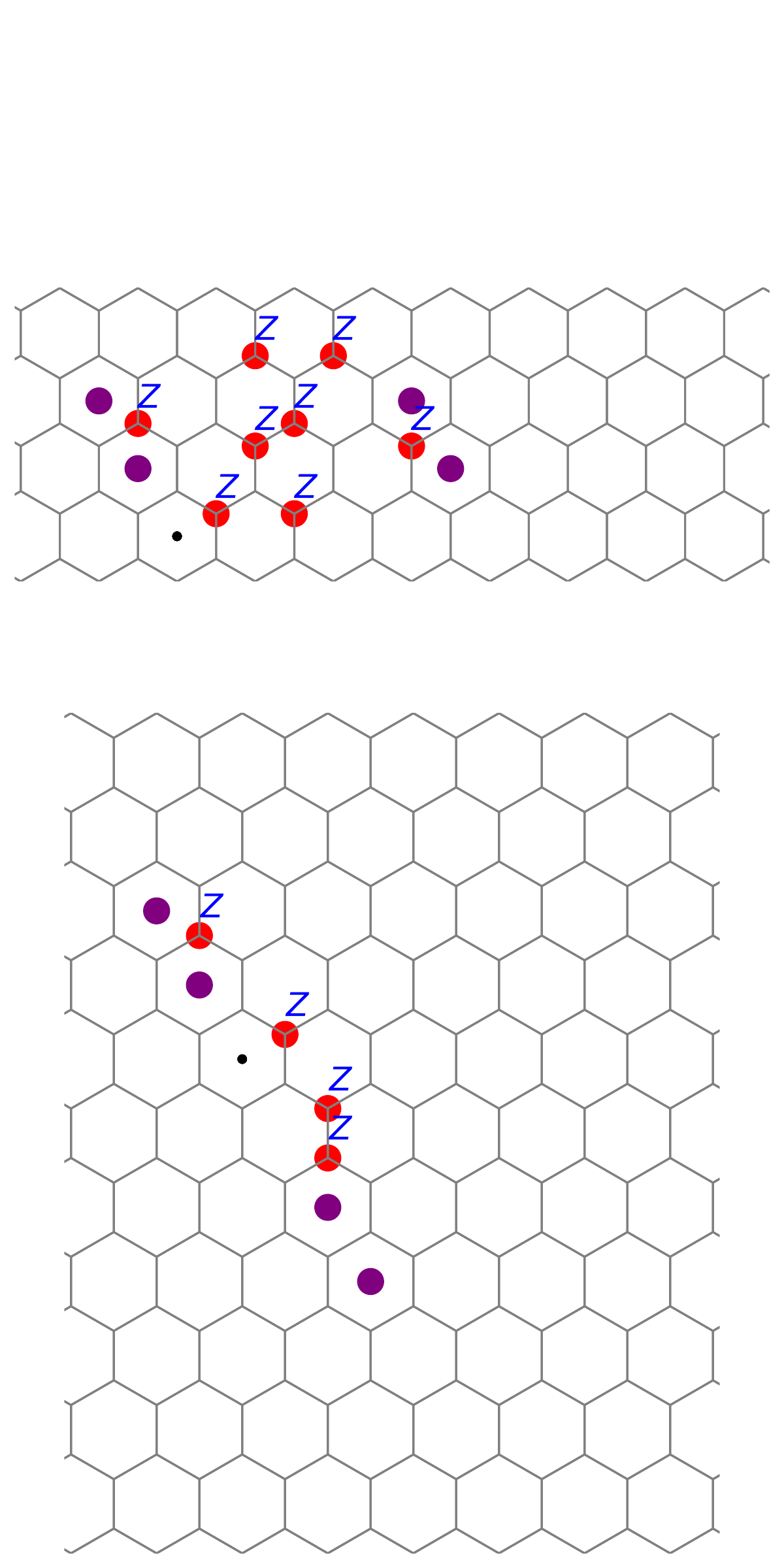}}
    \subfigure[string operator for $v_3$ ]{\includegraphics[width=0.24\textwidth]{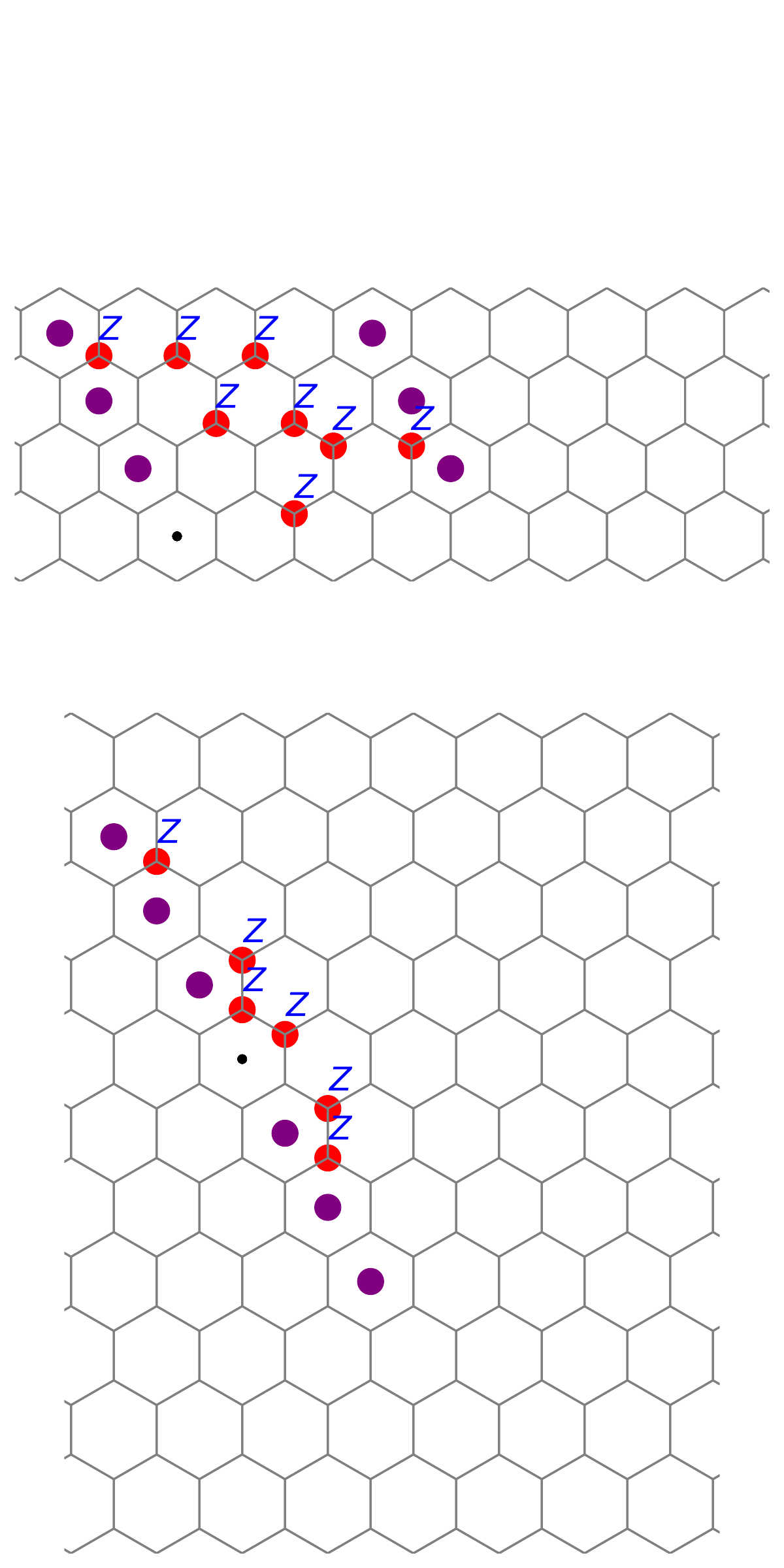}}
    \subfigure[string operator for $v_4$] {\includegraphics[width=0.24\textwidth]{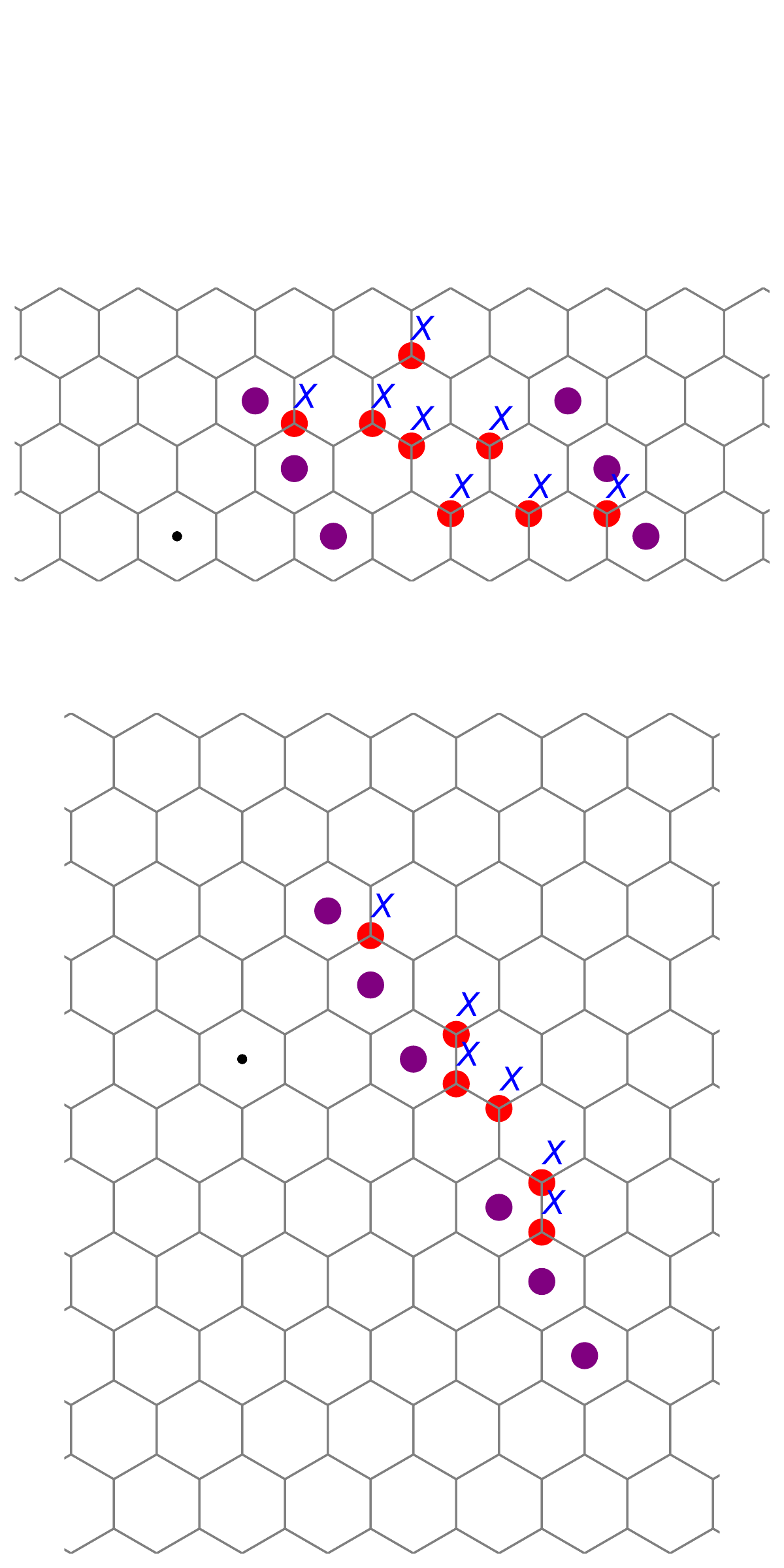}}
    \subfigure[string operator for $v_5$ ]{\includegraphics[width=0.24\textwidth]{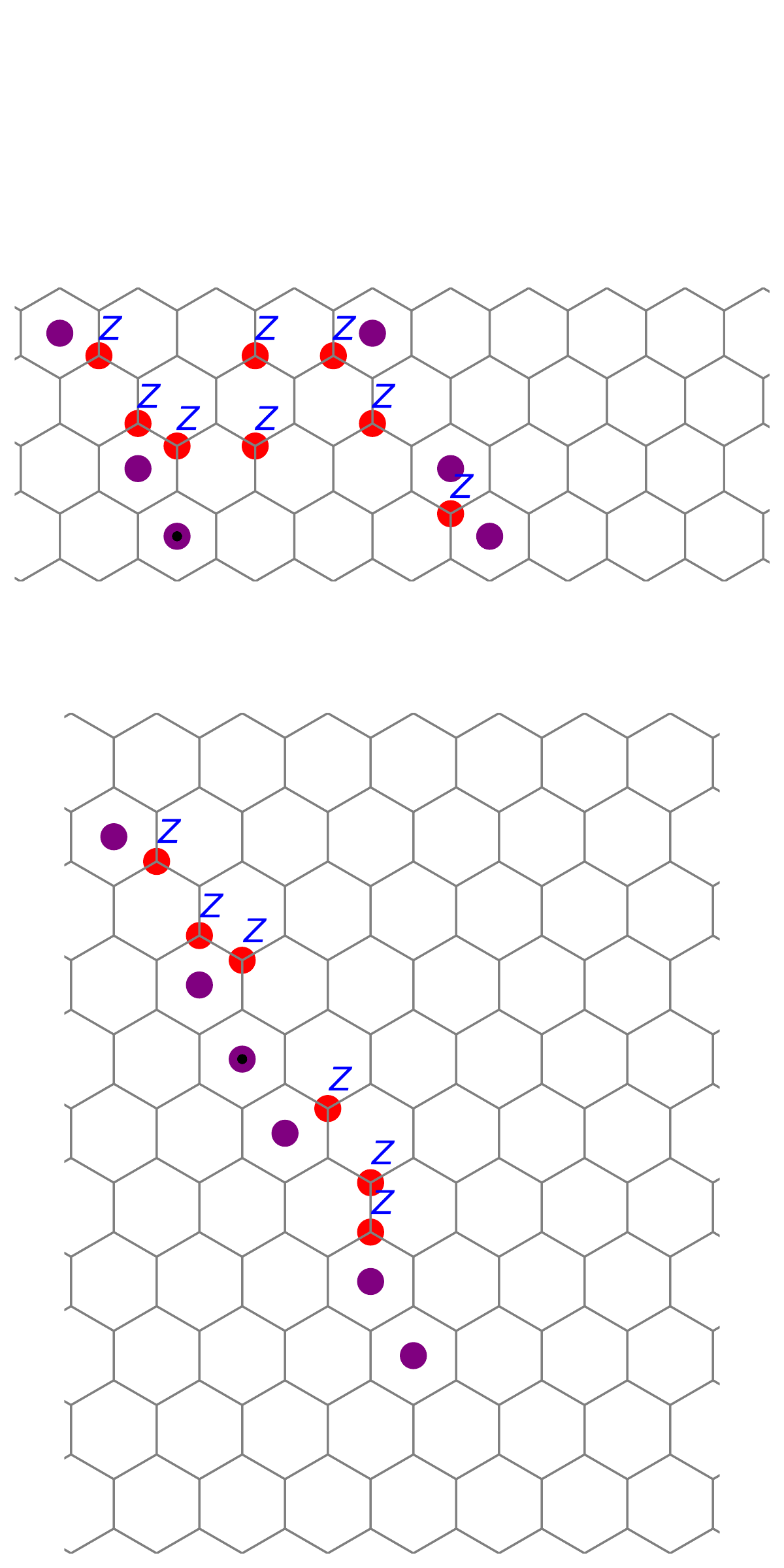}}
    \subfigure[string operator for $v_6$ ]{\includegraphics[width=0.24\textwidth]{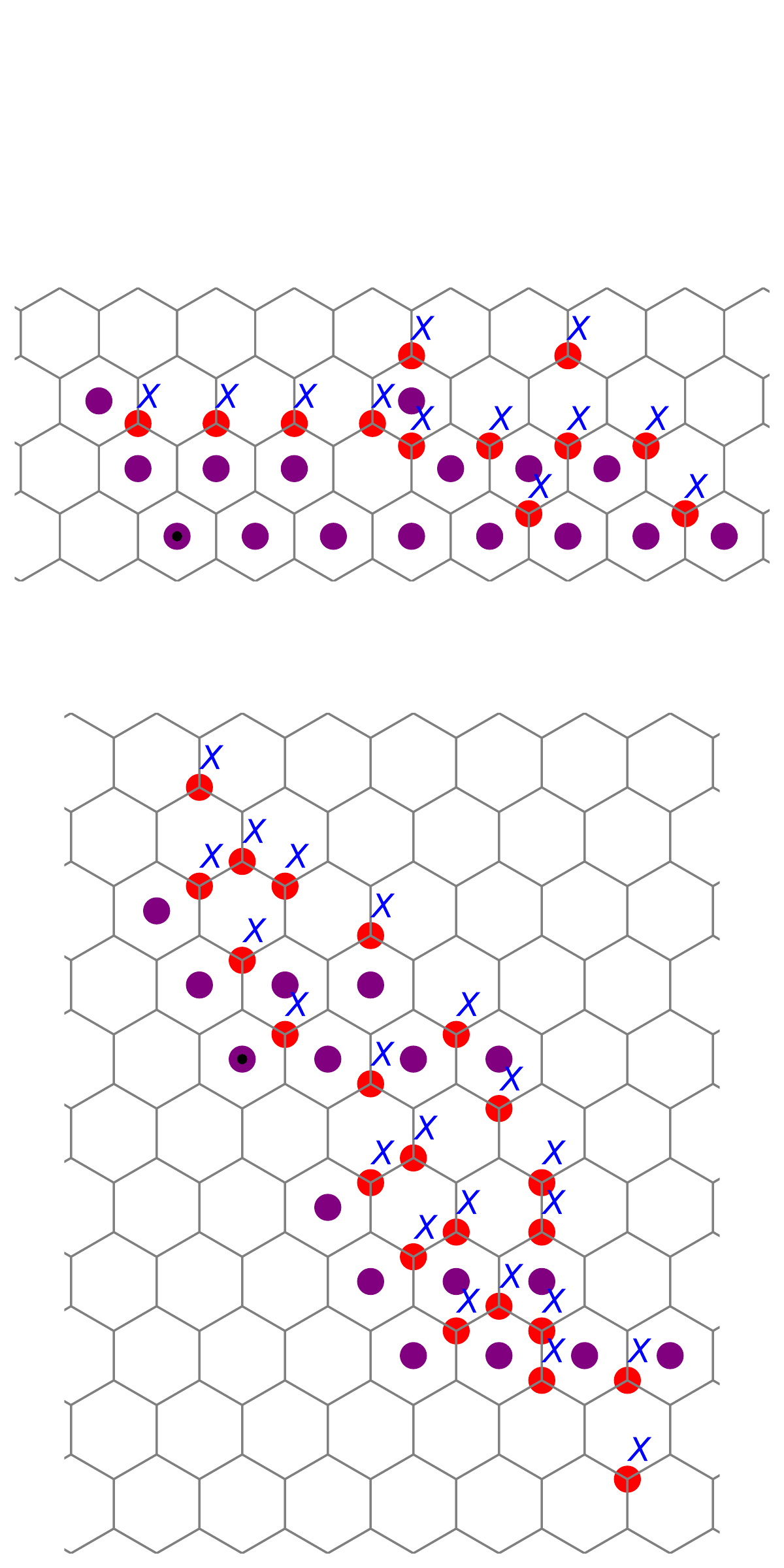}}
    \subfigure[string operator for $v_7$ ]{\includegraphics[width=0.24\textwidth]{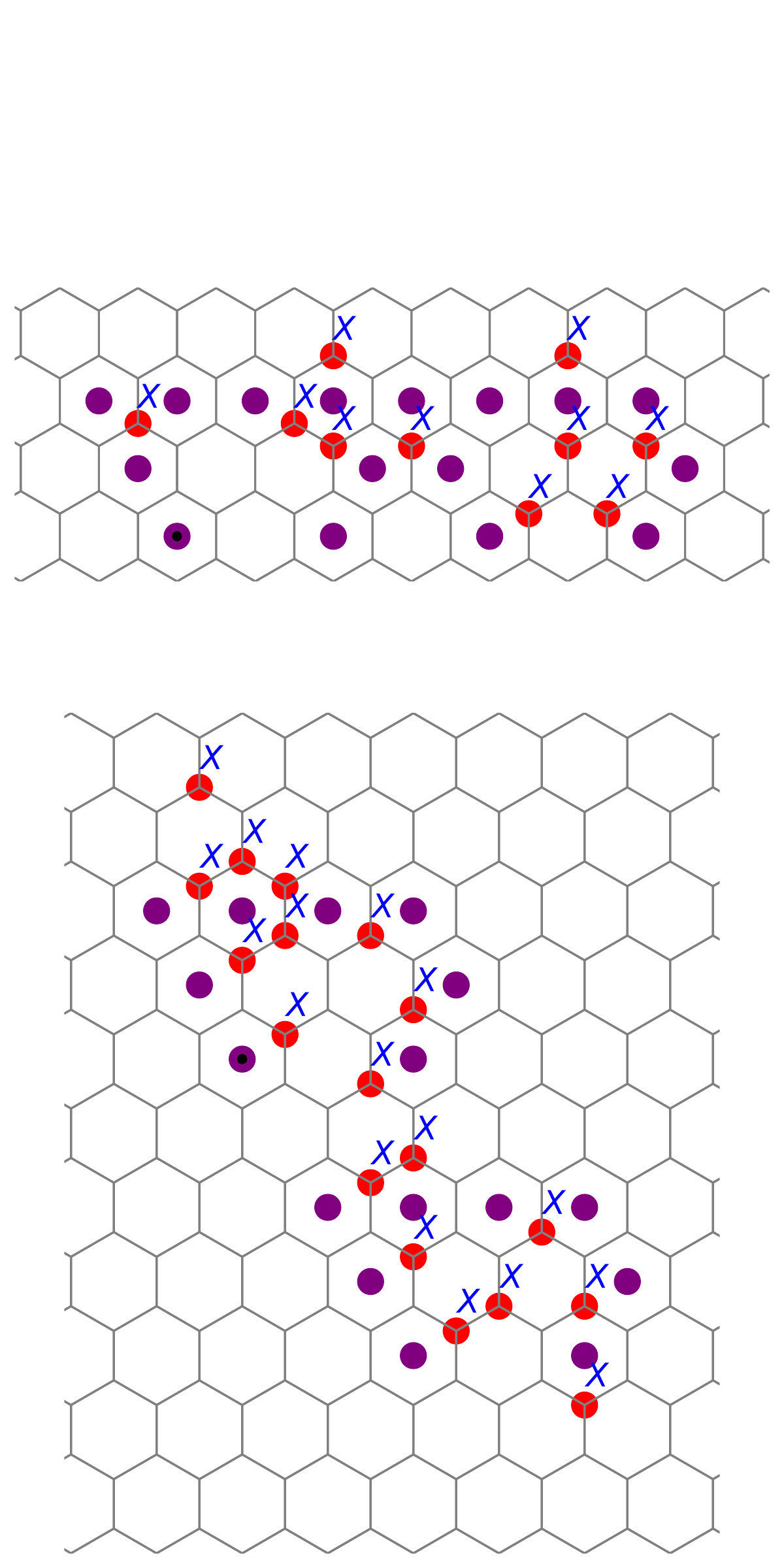}}
    \subfigure[string operator for $v_8$ ]{\includegraphics[width=0.24\textwidth]{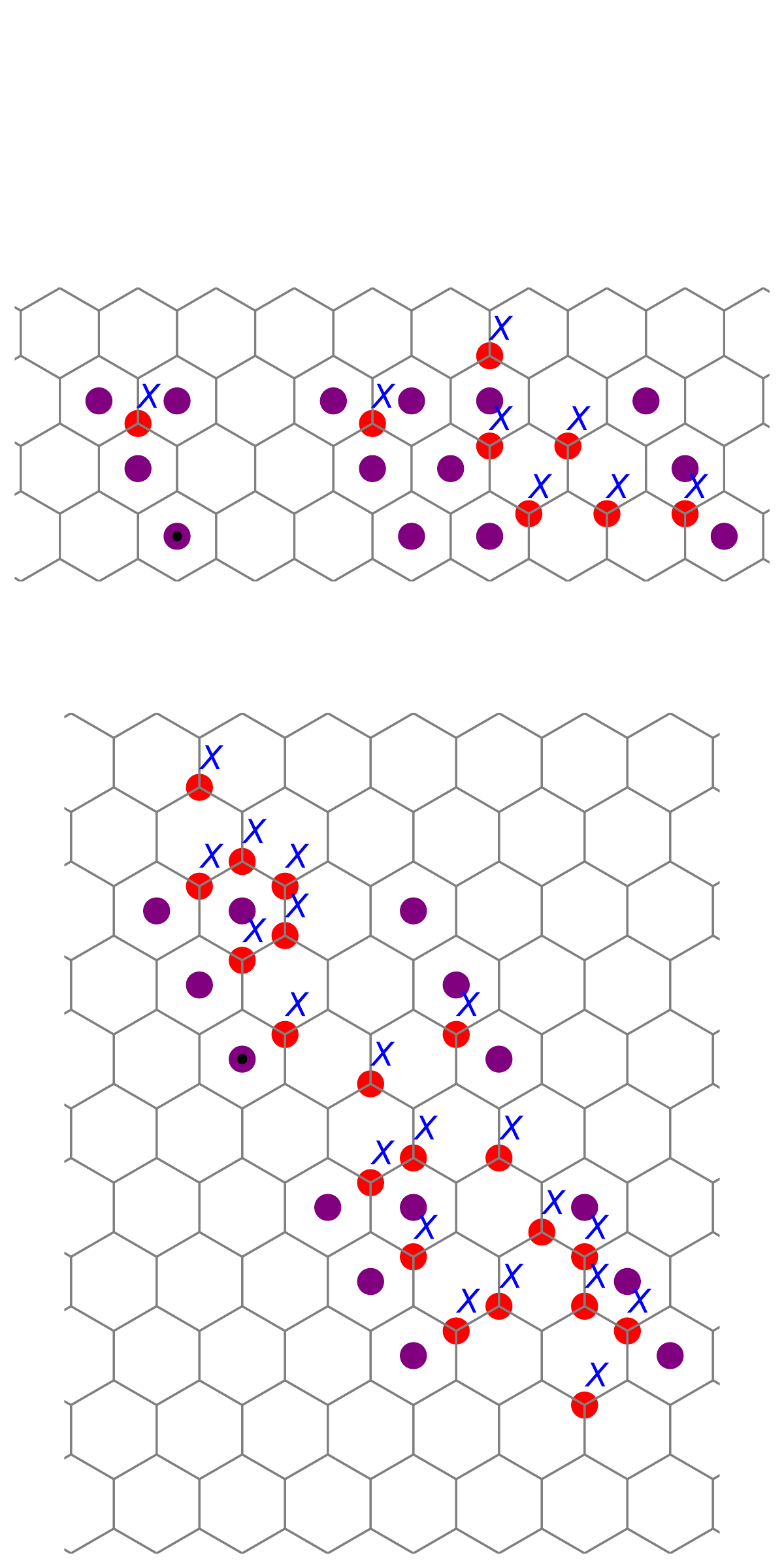}}
    \caption{The string operators in the modified color code C (example \textcircled{5}). Subfigures (a)$\sim$ (h) represent the string operators for anyons $v_1,..., v_8$. For each subfigure, the first row represents the anyon string operators that move an anyon along the $x$-direction, and the second row represents the anyon string operators that move an anyon along the $y$-direction.}
\label{fig:string_example5}
\end{figure*}

\begin{figure*}[htb]
    \centering
    \subfigure[String operator for $v_1$ ]{\includegraphics[width=0.24\textwidth]{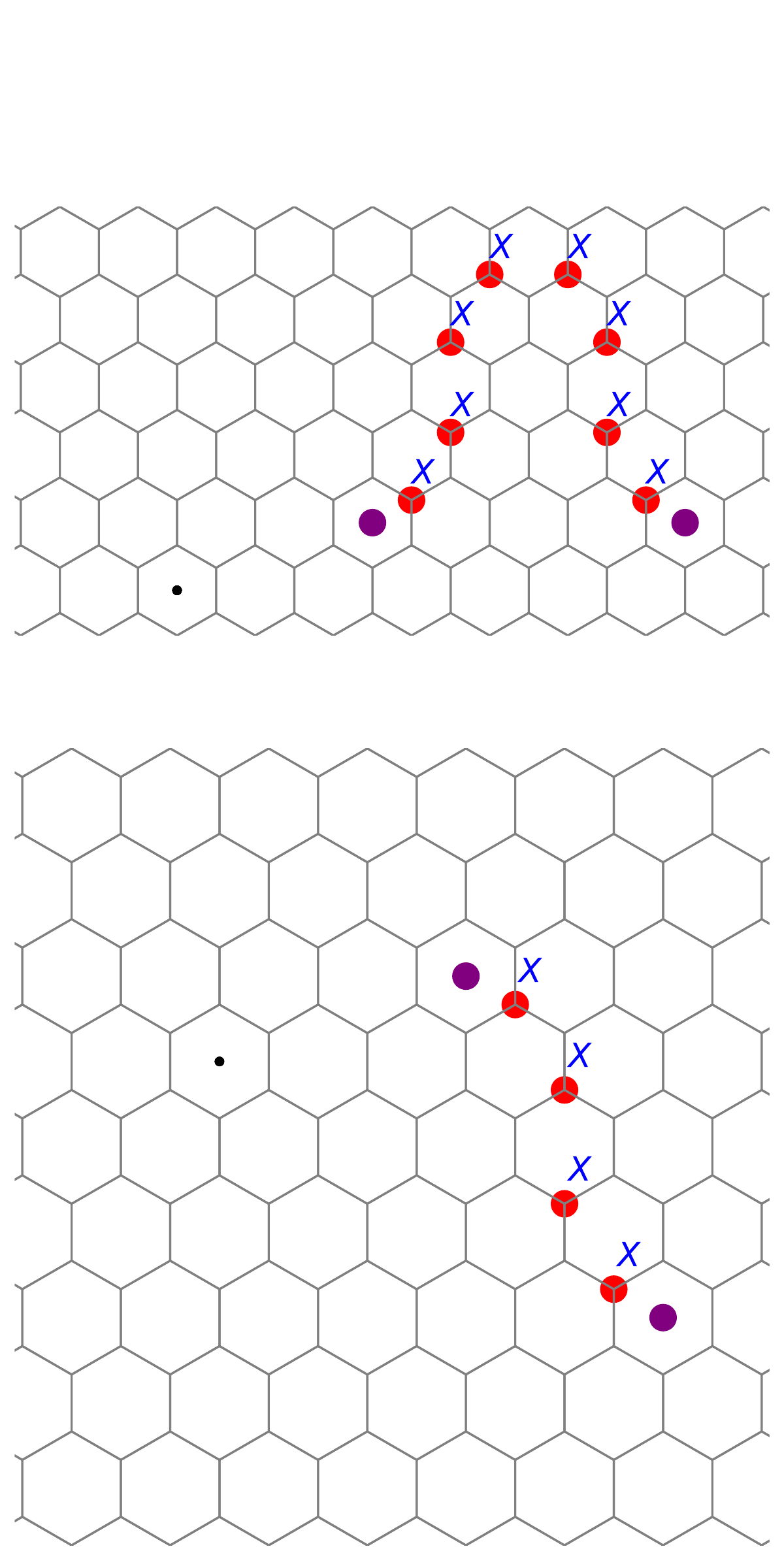}}
    \subfigure[String operator for $v_2$ ]{\includegraphics[width=0.24\textwidth]{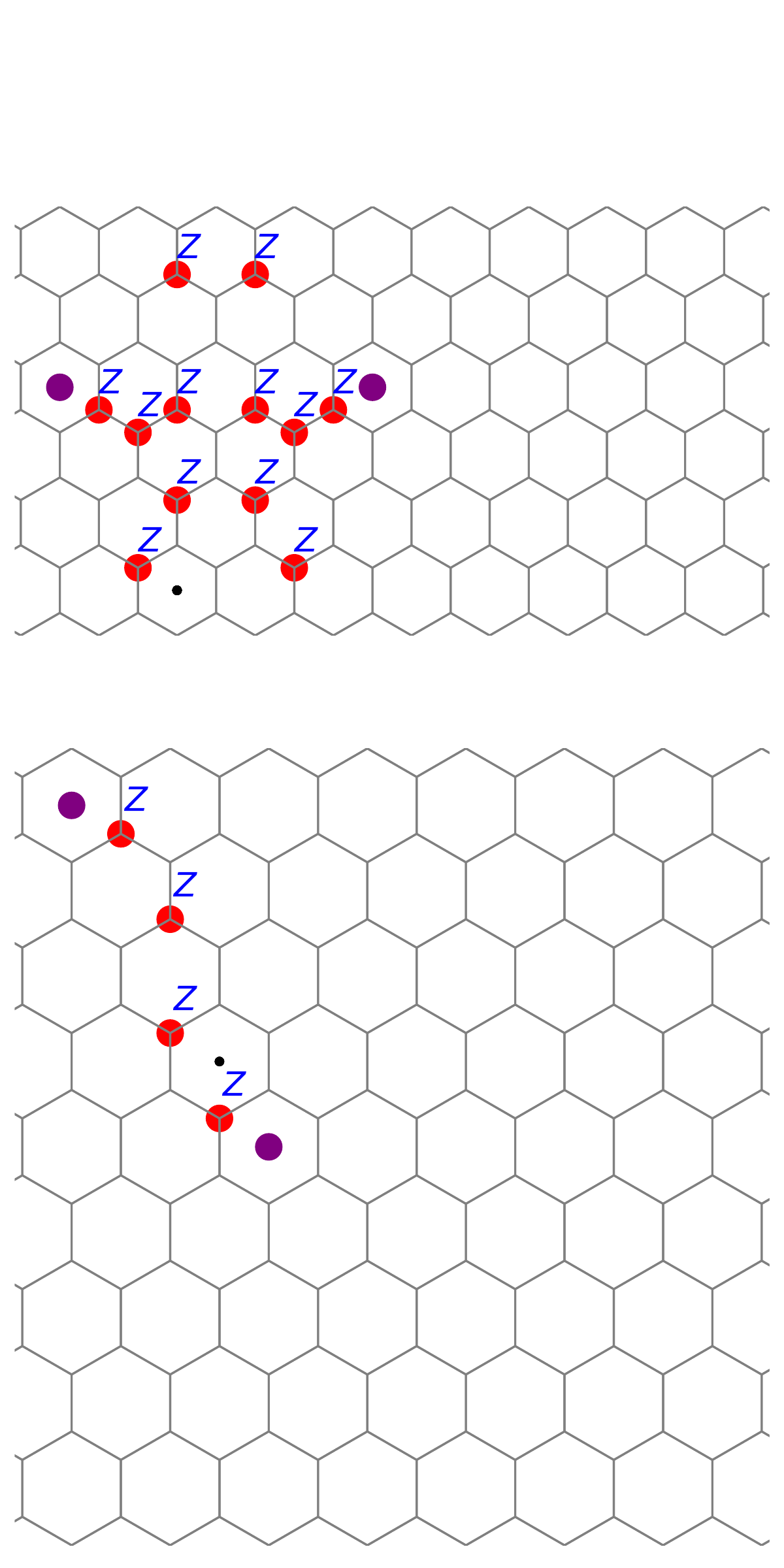}}
    \subfigure[String operator for $v_3$ ]{\includegraphics[width=0.24\textwidth]{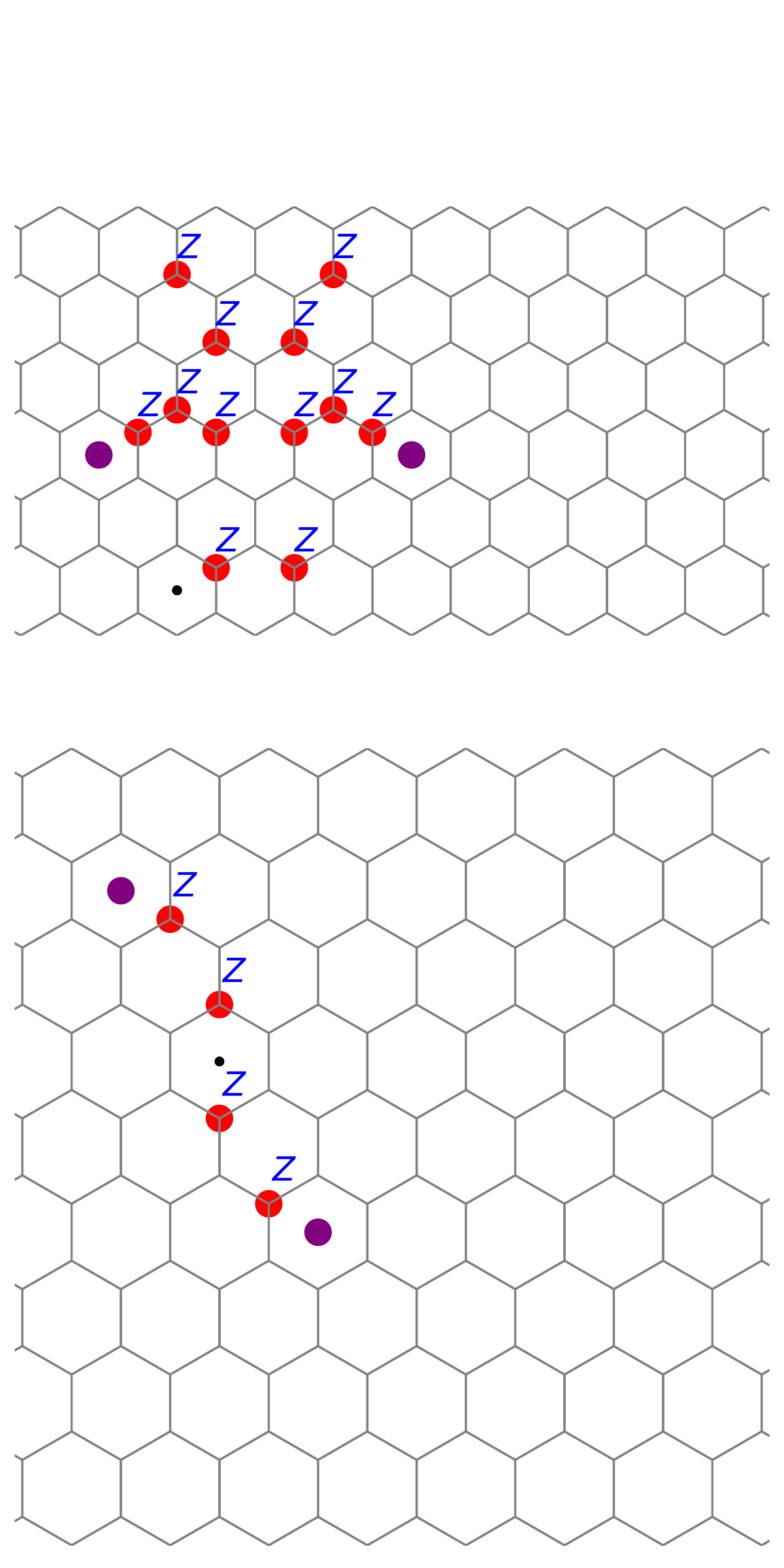}}
    \subfigure[string operator for $v_4$ ]{\includegraphics[width=0.24\textwidth]{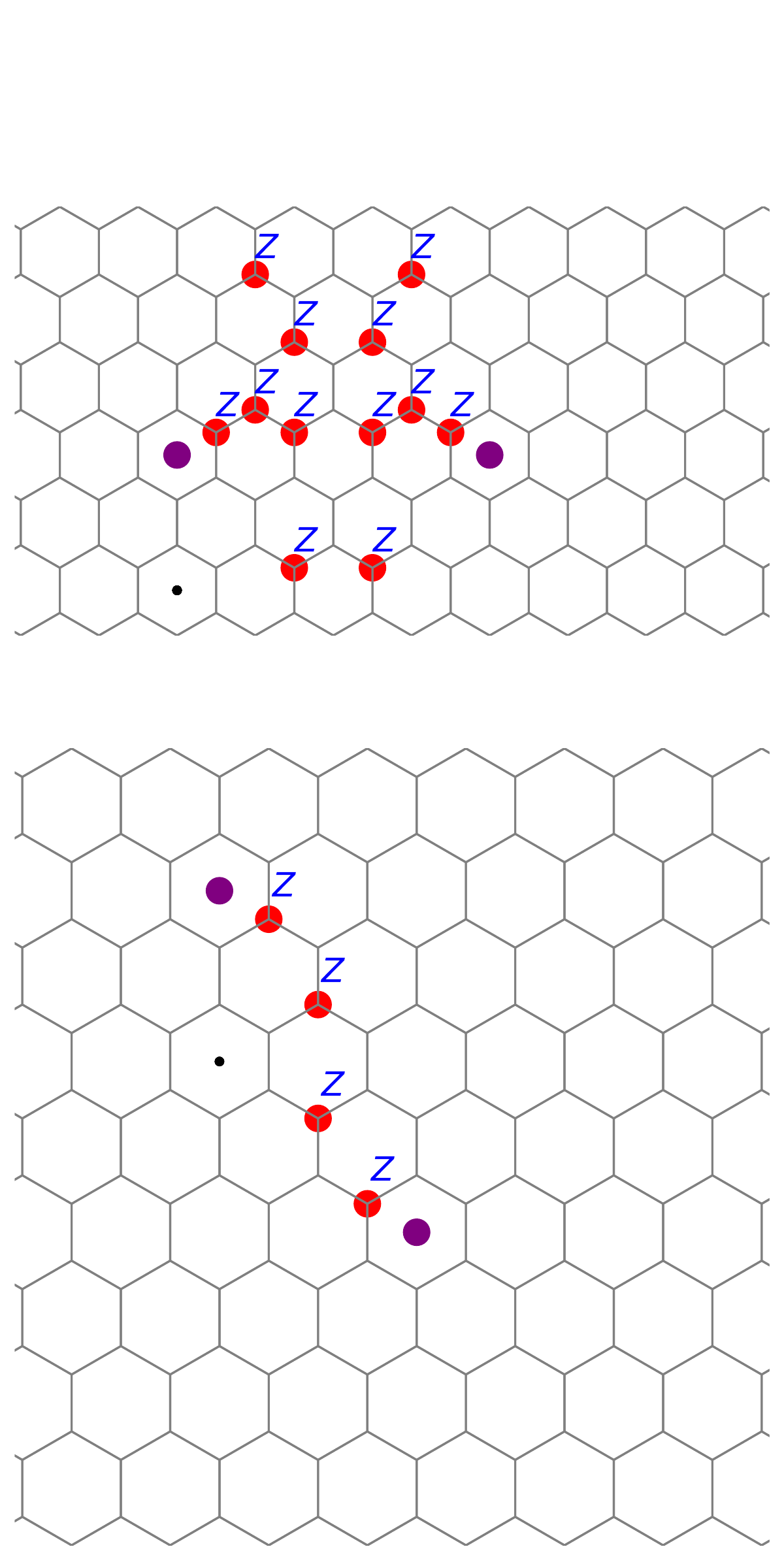}}
    \subfigure[string operator for $v_5$ ]{\includegraphics[width=0.24\textwidth]{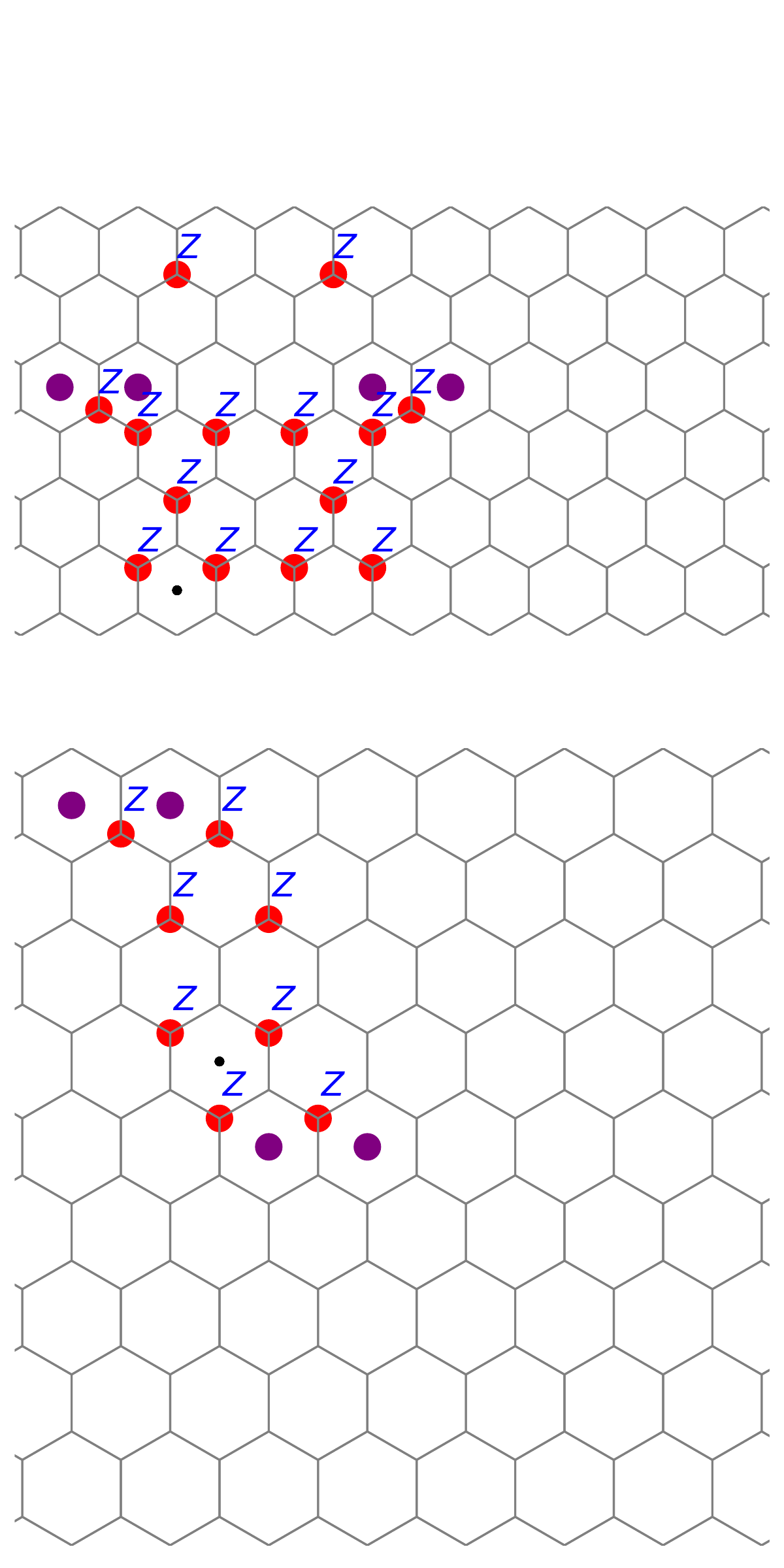}}
    \subfigure[string operator for $v_6$ ]{\includegraphics[width=0.24\textwidth]{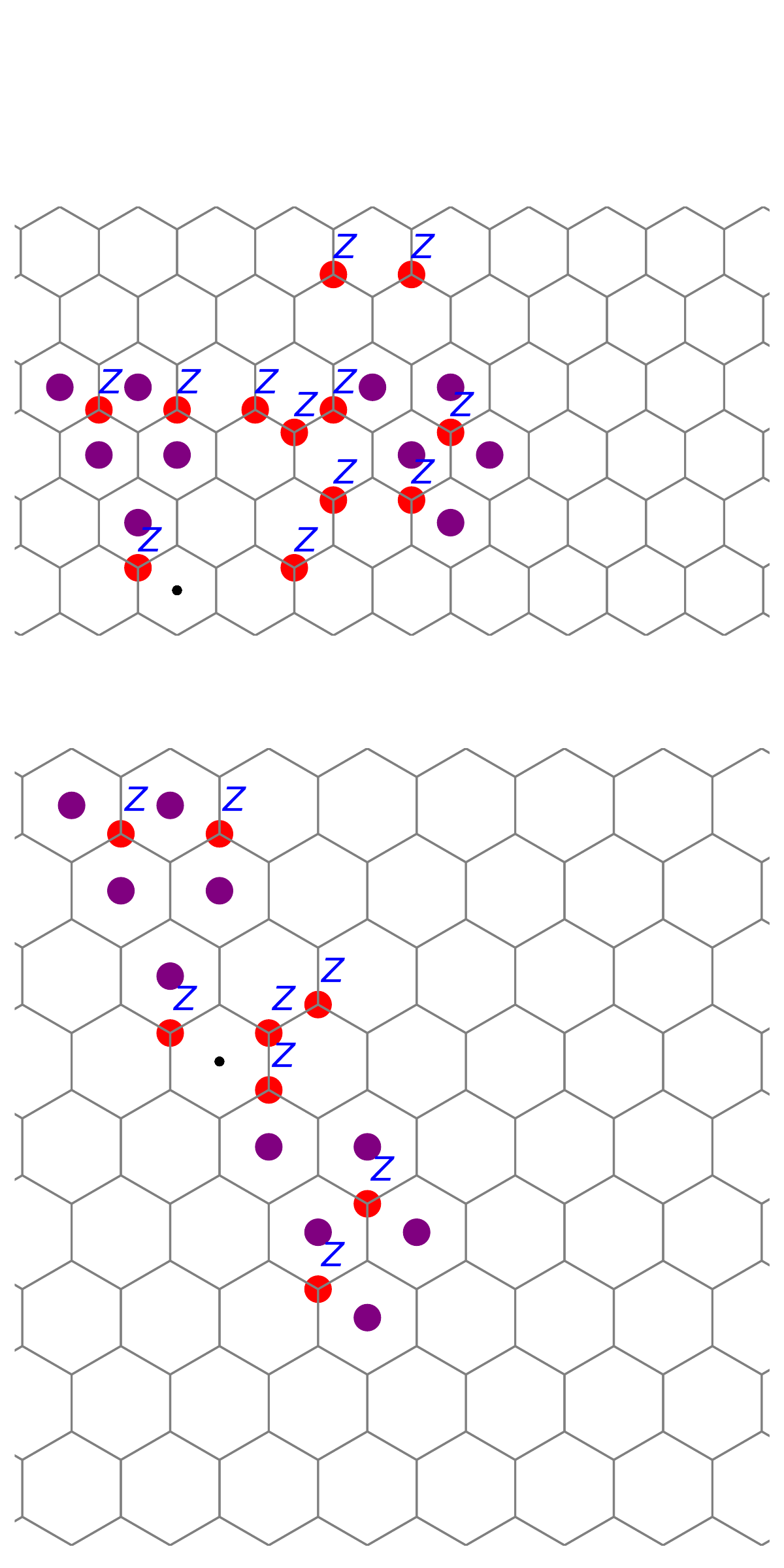}}
    \subfigure[string operator for $v_7$ ]{\includegraphics[width=0.24\textwidth]{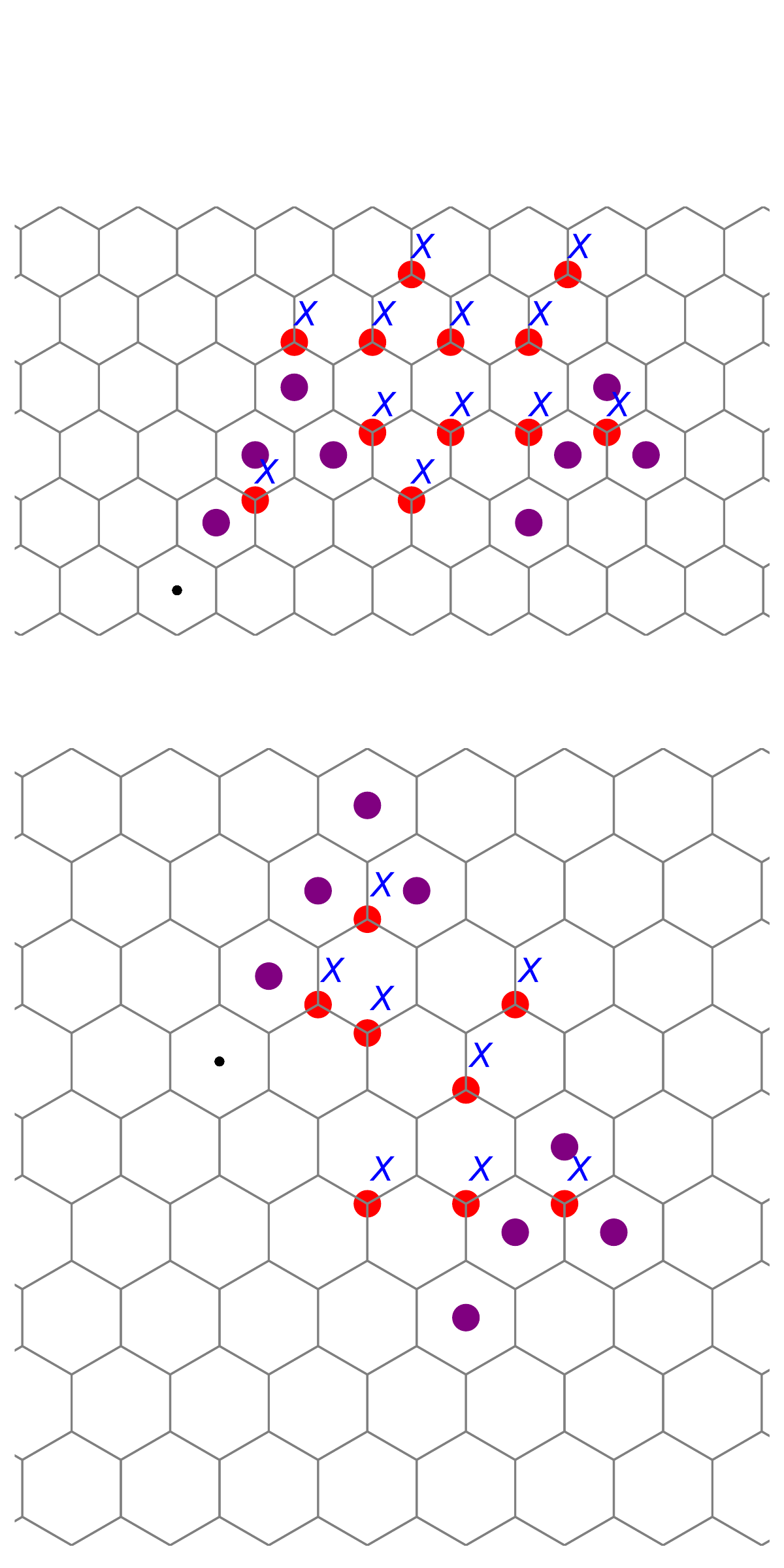}}
    \subfigure[string operator for $v_8$ ]{\includegraphics[width=0.24\textwidth]{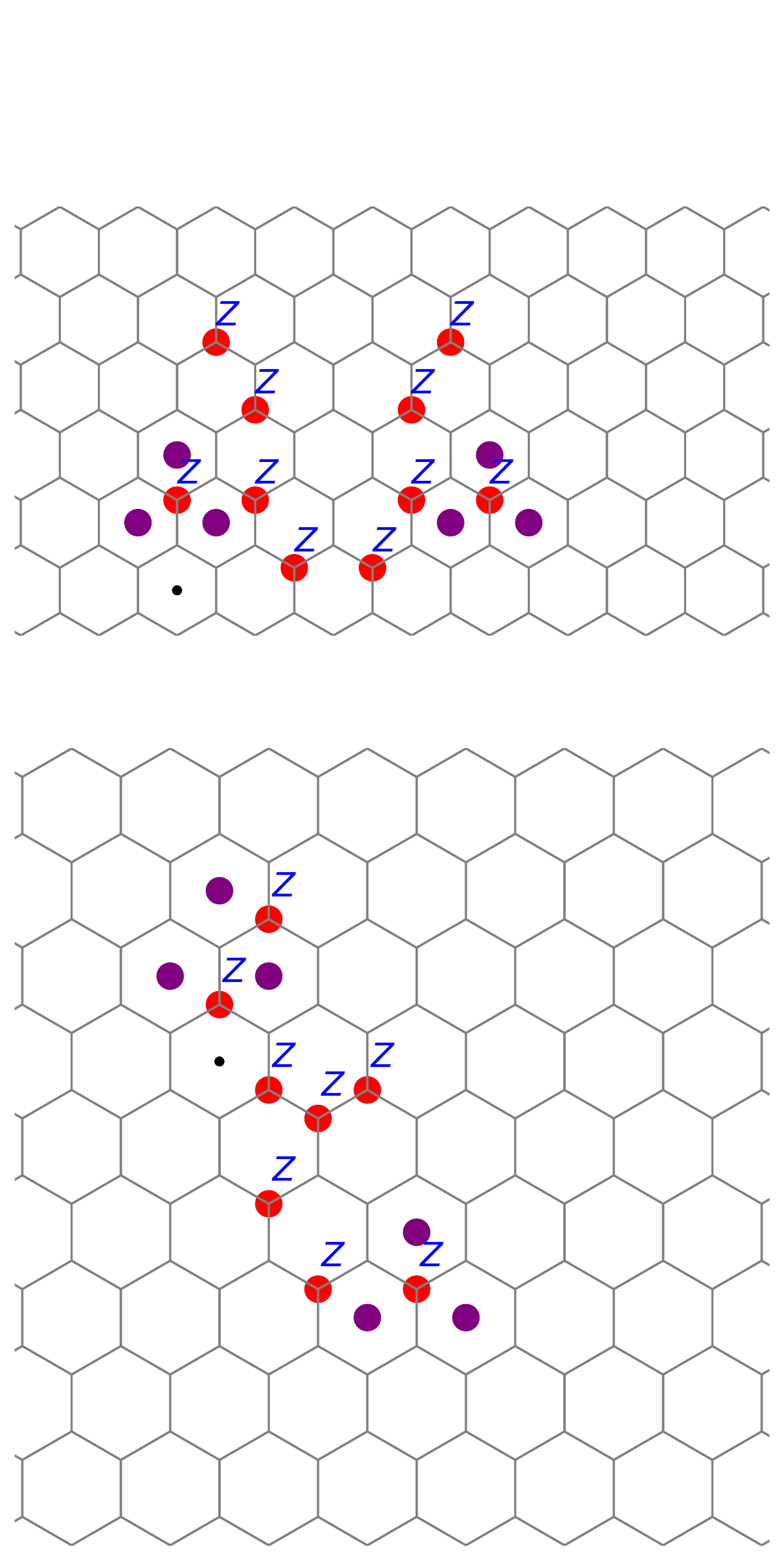}}
    \caption{The string operators in the modified color code D (example \textcircled{6}). Subfigures (a)$\sim$ (h) represent the string operators for anyons $v_1,..., v_8$. For each subfigure, the first row represents the anyon string operators that move an anyon along the $x$-direction, and the second row represents the anyon string operators that move an anyon along the $y$-direction.}
\label{fig:string_example6_1}
\end{figure*}

\begin{figure*}
\centering
    \subfigure[string operator for $v_9$ ]{\includegraphics[width=0.24\textwidth]{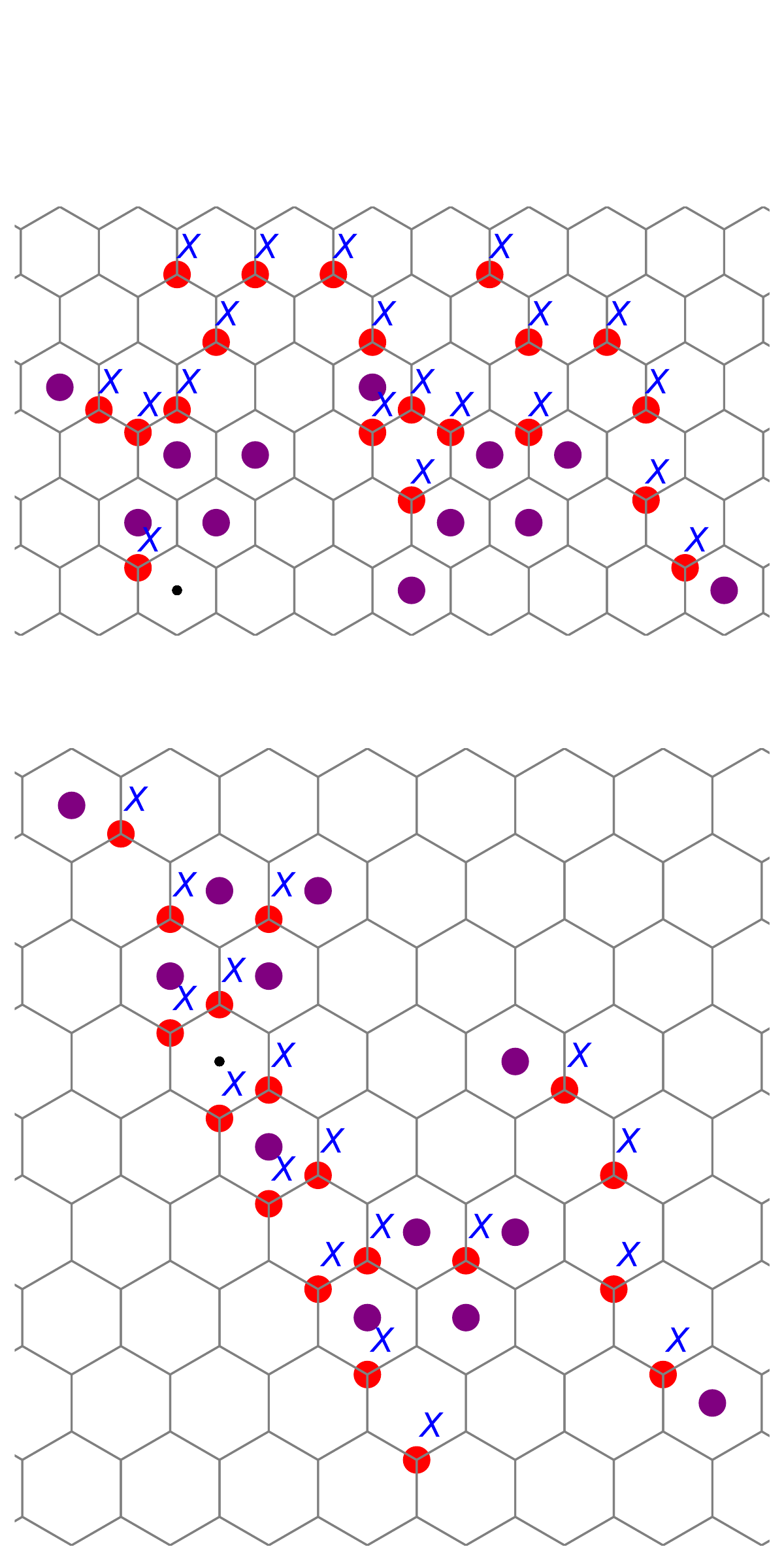}}
    \subfigure[string operator for $v_{10}$ ]{\includegraphics[width=0.24\textwidth]{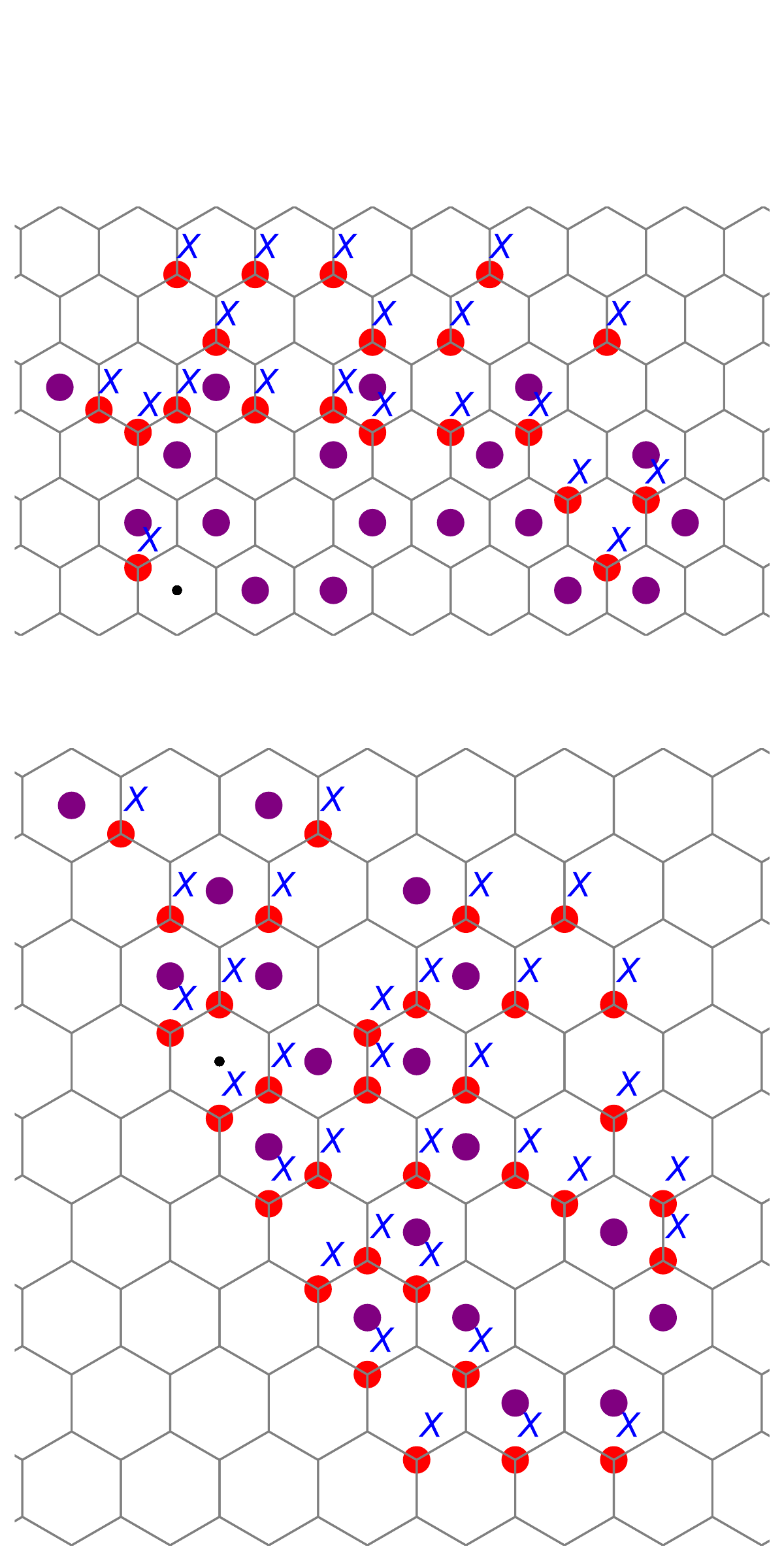}}
    \subfigure[string operator for $v_{11}$ ]{\includegraphics[width=0.24\textwidth]{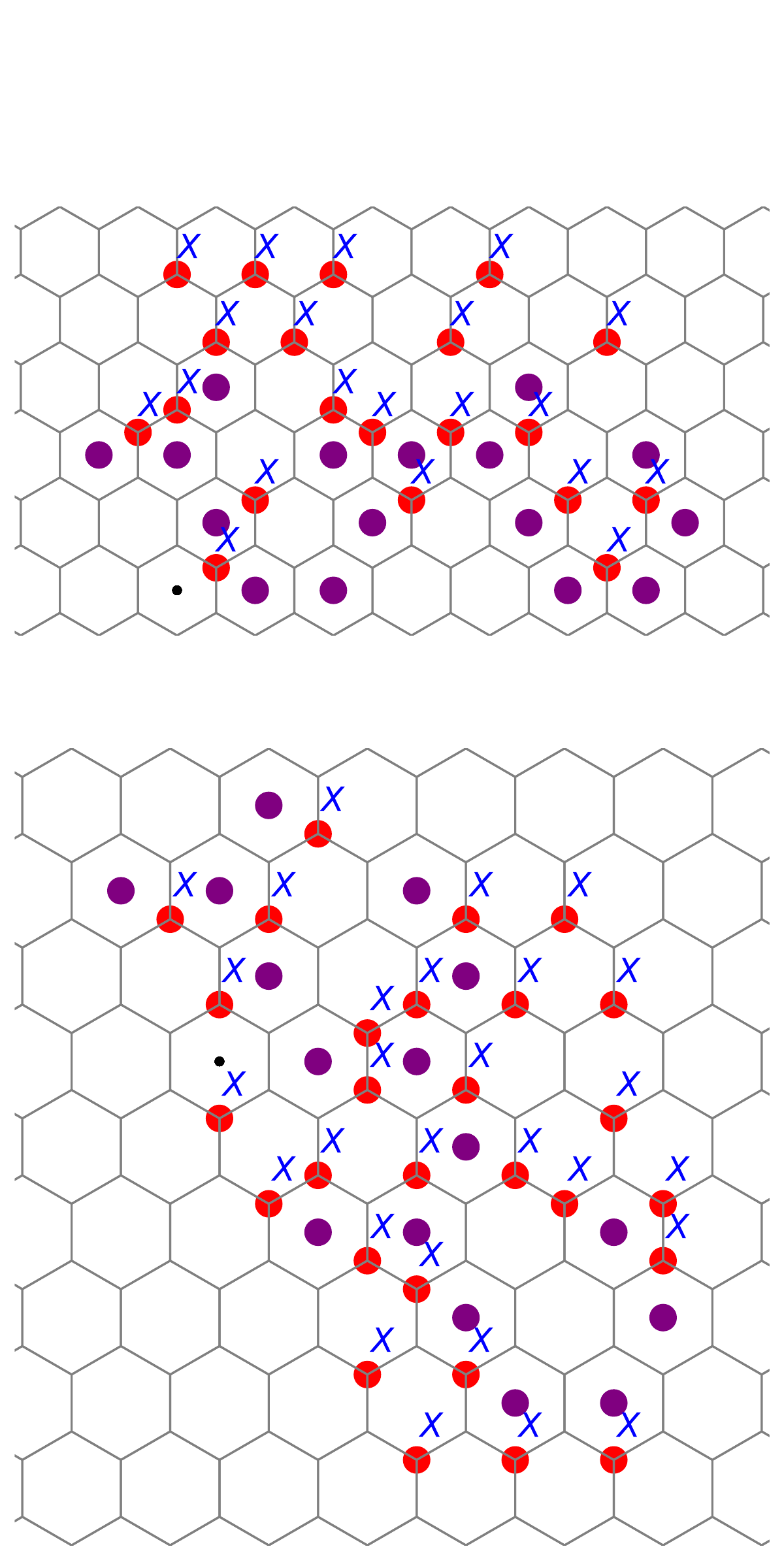}}
    \subfigure[string operator for $v_{12}$ ]{\includegraphics[width=0.24\textwidth]{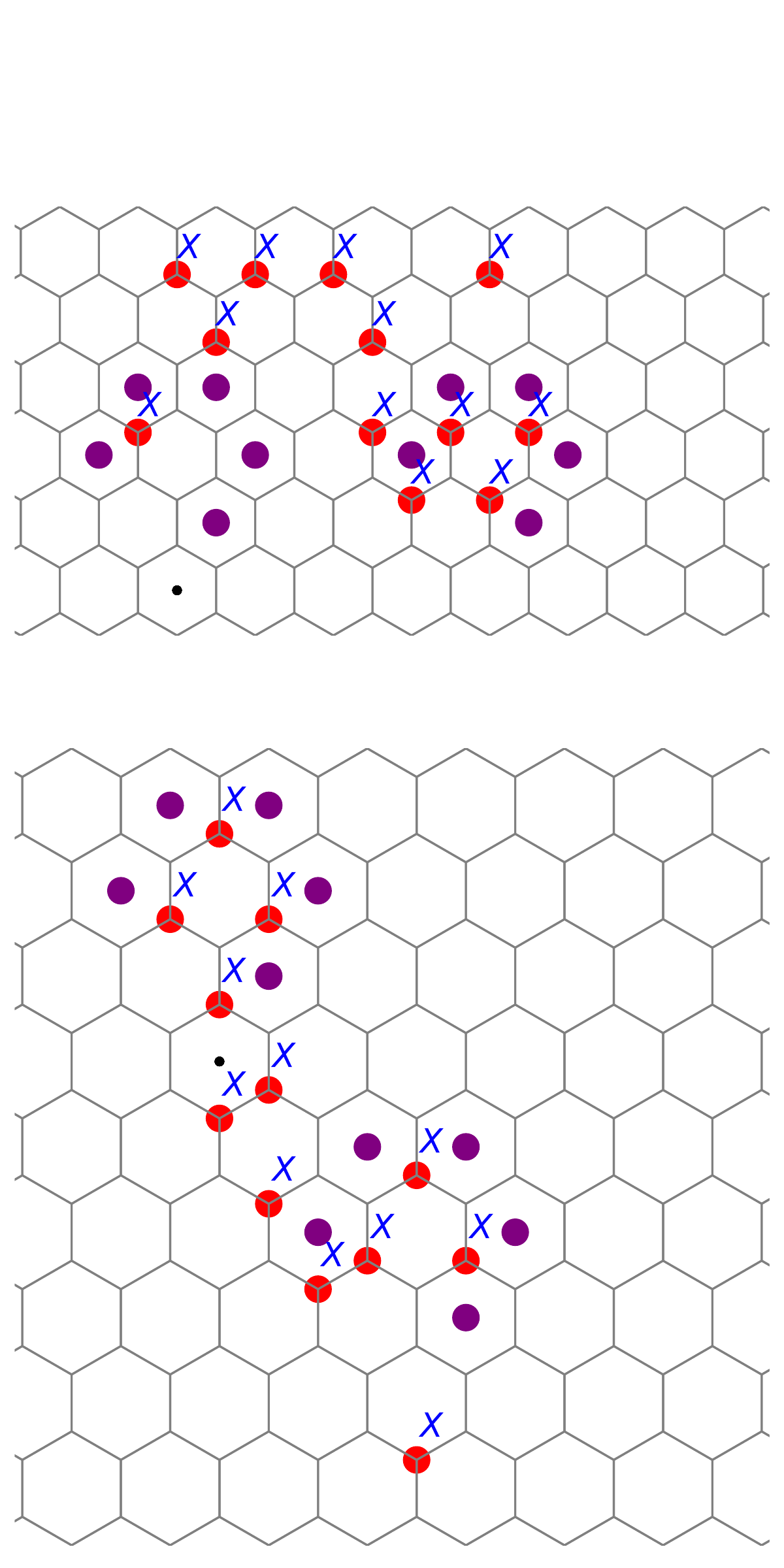}}
    \caption{(Continued) The string operators in the modified color code D (example \textcircled{6}). Subfigures (a)$\sim$ (d) represent the string operators for anyons $v_9,..., v_{12}$. For each subfigure, the first row represents the anyon string operators that move an anyon along the $x$-direction, and the second row represents the anyon string operators that move an anyon along the $y$-direction.}
\label{fig:string_example6_2}
\end{figure*}

\subsection{CSS code induced from double semion code}

For the CSS code induced from double semion code, we obtain four basis anyons at $(1-x^1)$ and $(1-y^{-1})$ (also for greater $n_x, n_y$). The string operators for $n_x, n_y=5$ are shown in Fig.~\ref{fig:string_css_double_semion1} and \ref{fig:string_css_double_semion2}.
\begin{figure*}[htb]
    \centering
    \subfigure[The string operator of anyon $v_1$ in the CSS code induced from double semion code.]{\includegraphics[width=0.48\textwidth]{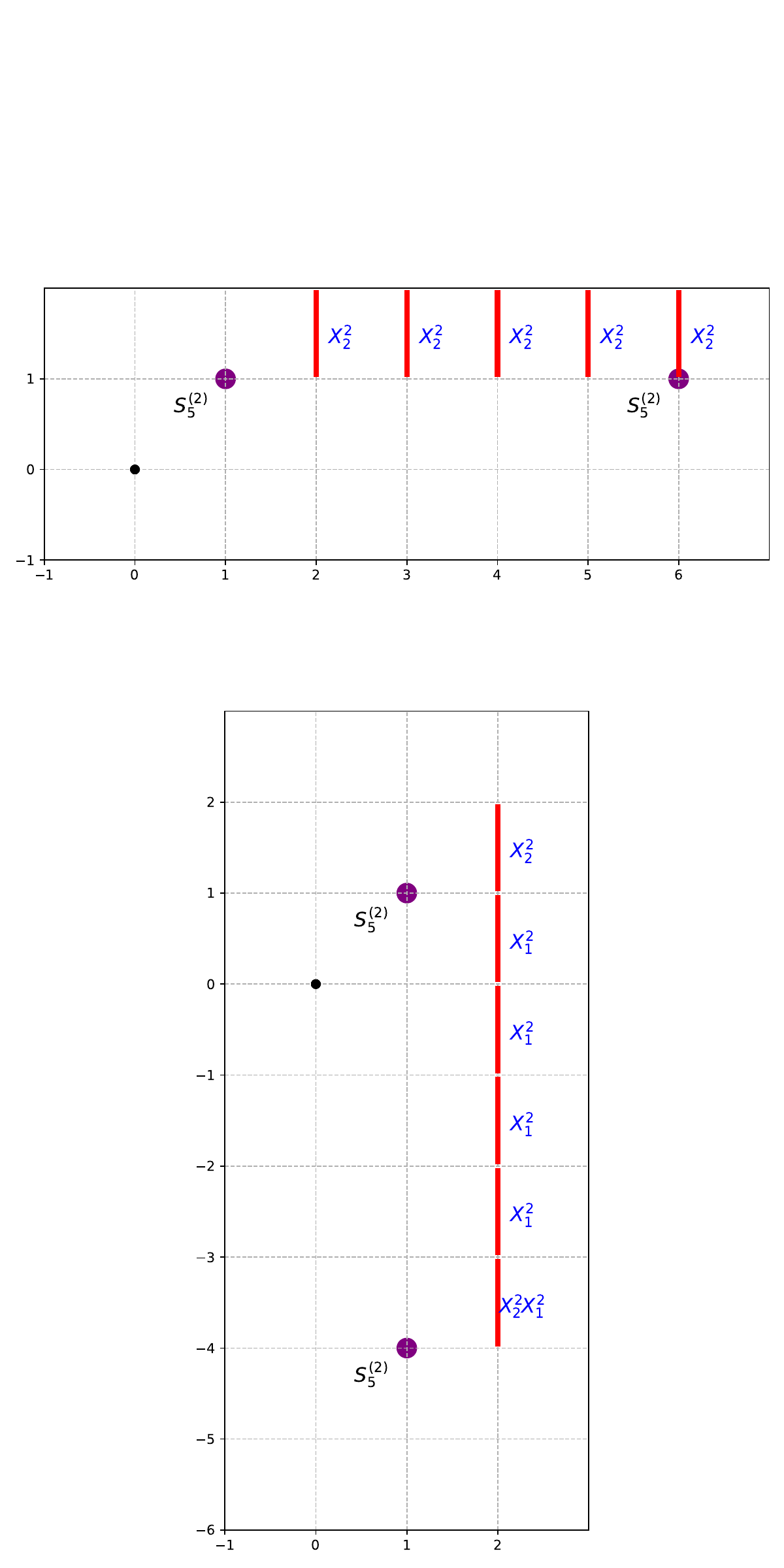}}
    \subfigure[The string operator of anyon $v_2$ in the CSS code induced from double semion code.]{\includegraphics[width=0.48\textwidth]{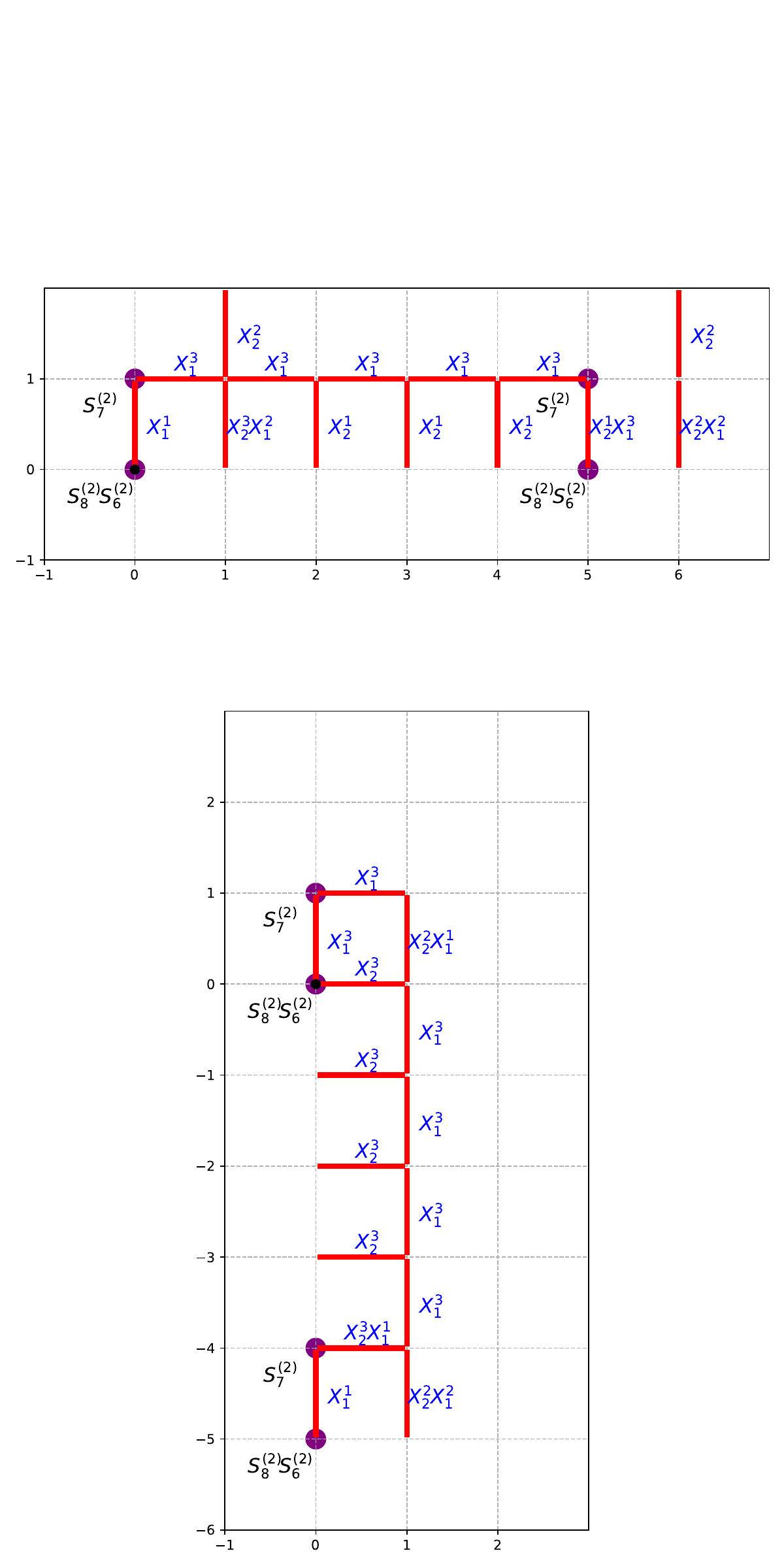}}
    \caption{The string operators of anyons $v_1$ and $v_2$ in the CSS code induced from double semion codes. For each subfigure, the first row represents the anyon string operators that move an anyon along the $x$-direction, and the second row represents the anyon string operators that move an anyon along the $y$-direction. Purple dots label the syndrome patterns at the endpoints of each string. Here, we denote the violation of plaquette terms $S_2$ and $S_6$ by purple dots on their bottom-left vertex. For $S_3, S_7$, we denote the violation of them by purple dots on the top vertex above the vertical edge. For $S_4, S_8$, we denote the violation of them by purple dots at the vertex below the vertical edge. Since this code is defined over $\mathbb{Z}_4$ qudits, we put superscript on $S$ to indicate the violation of the stabilizer with a phase, for example, $S_5^{(2)}$ indicates that violations of $S_5$ up to a phase $\omega^2=-1$ where $\omega=e^{i2\pi/4}$. We use this convention in all the following figures.}
    \label{fig:string_css_double_semion1}
\end{figure*}

\begin{figure*}
    \centering
    \subfigure[The string operator of anyon $v_3$ in the CSS code induced from double semion code.]{\includegraphics[width=0.48\textwidth]{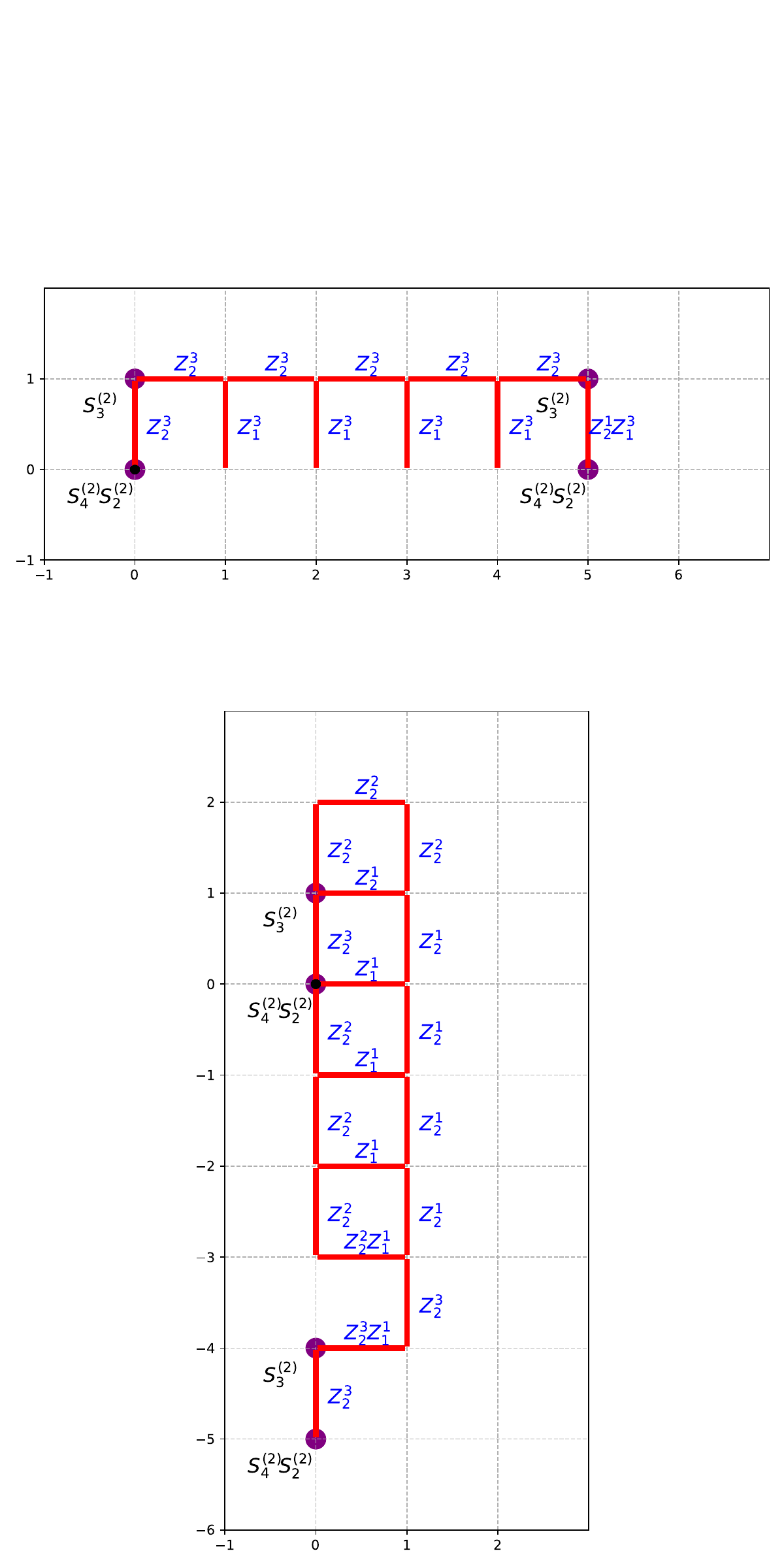}}
    \subfigure[The string operator of anyon $v_4$ in the CSS code induced from double semion code.]{\includegraphics[width=0.48\textwidth]{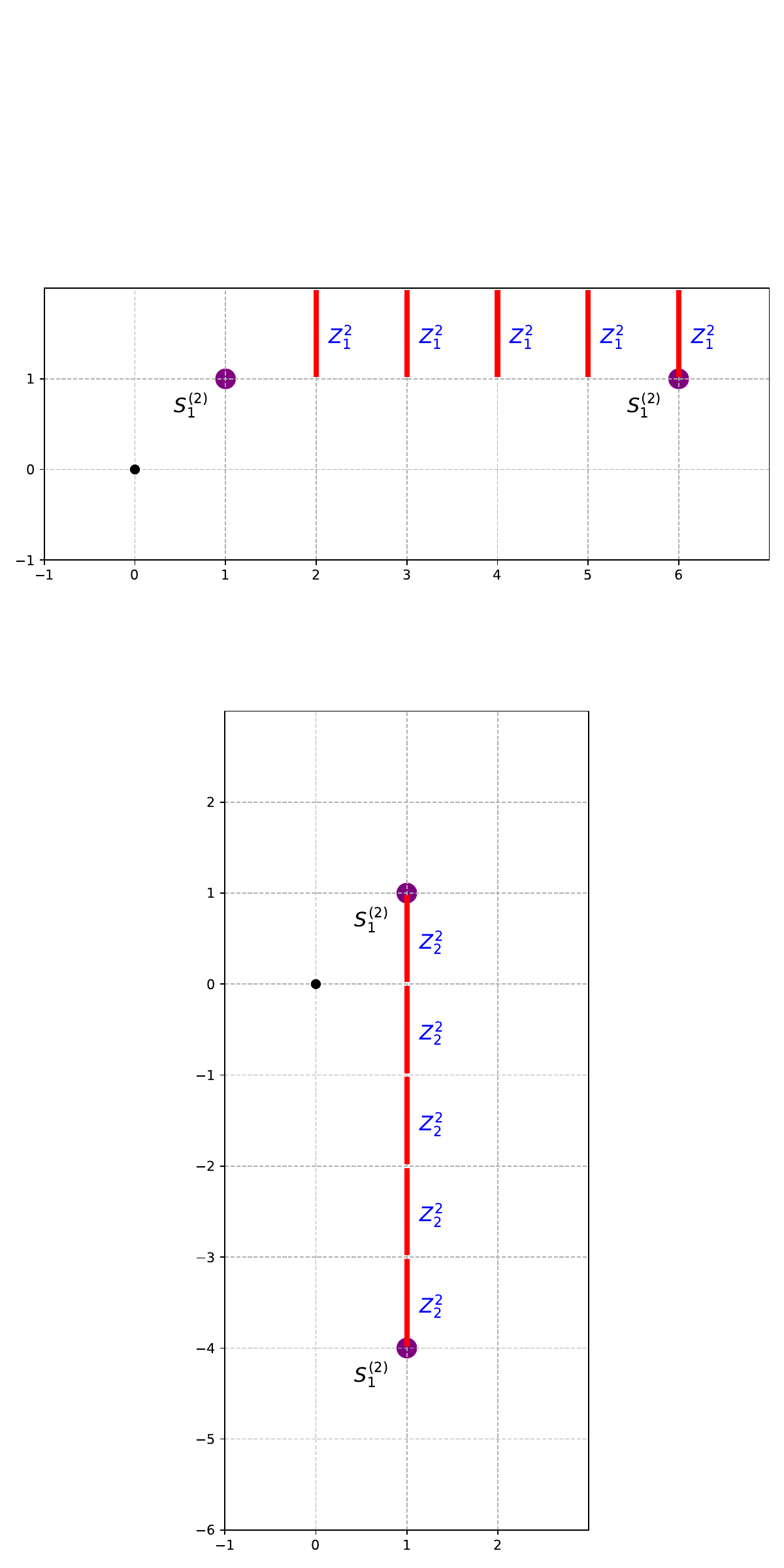}}
    \caption{(Continued) The string operators of anyons $v_3$ and $v_4$ in the CSS code induced from double semion codes. For each subfigure, the first row represents the anyon string operators that move an anyon along the $x$-direction, and the second row represents the anyon string operators that move an anyon along the $y$-direction. Purple dots label the syndrome patterns at the endpoints of each string. }
    \label{fig:string_css_double_semion2}
\end{figure*}

\subsection{Double semion code}

For the six-semion code, we obtain two basis anyons at $(1-x^1)$ and $(1-y^{-1})$ (also for greater $n_x, n_y$). The string operators for $n_x, n_y=5$ are shown in Fig.~\ref{fig:string_double_semion}.

\begin{figure*}[htb]
    \centering
    \subfigure[The string operator of anyon $v_1$ in the double semion code.]{\includegraphics[width=0.48\textwidth]{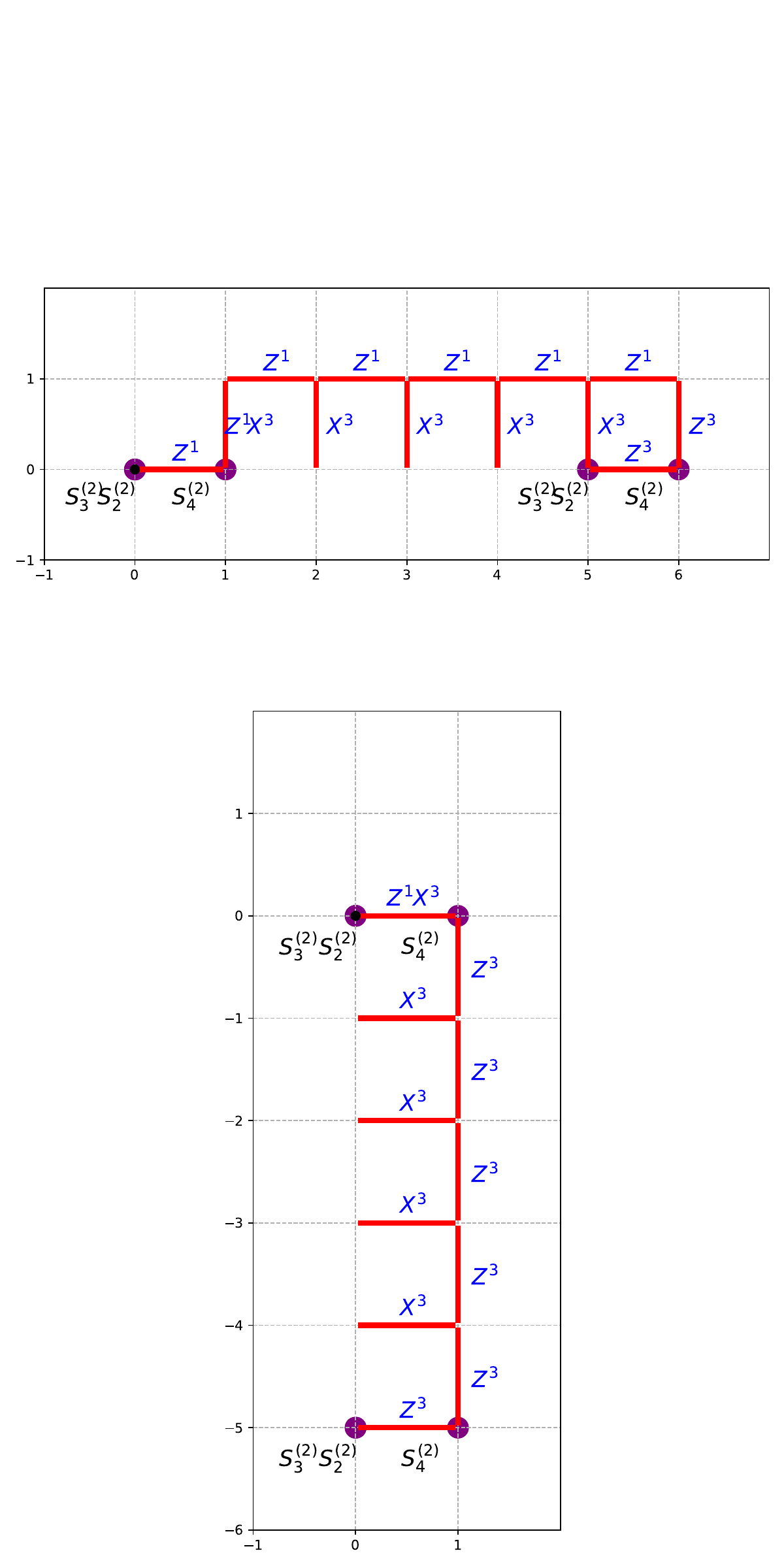}}
    \subfigure[The string operator of anyon $v_2$ for the double semion code.]{\includegraphics[width=0.48\textwidth]{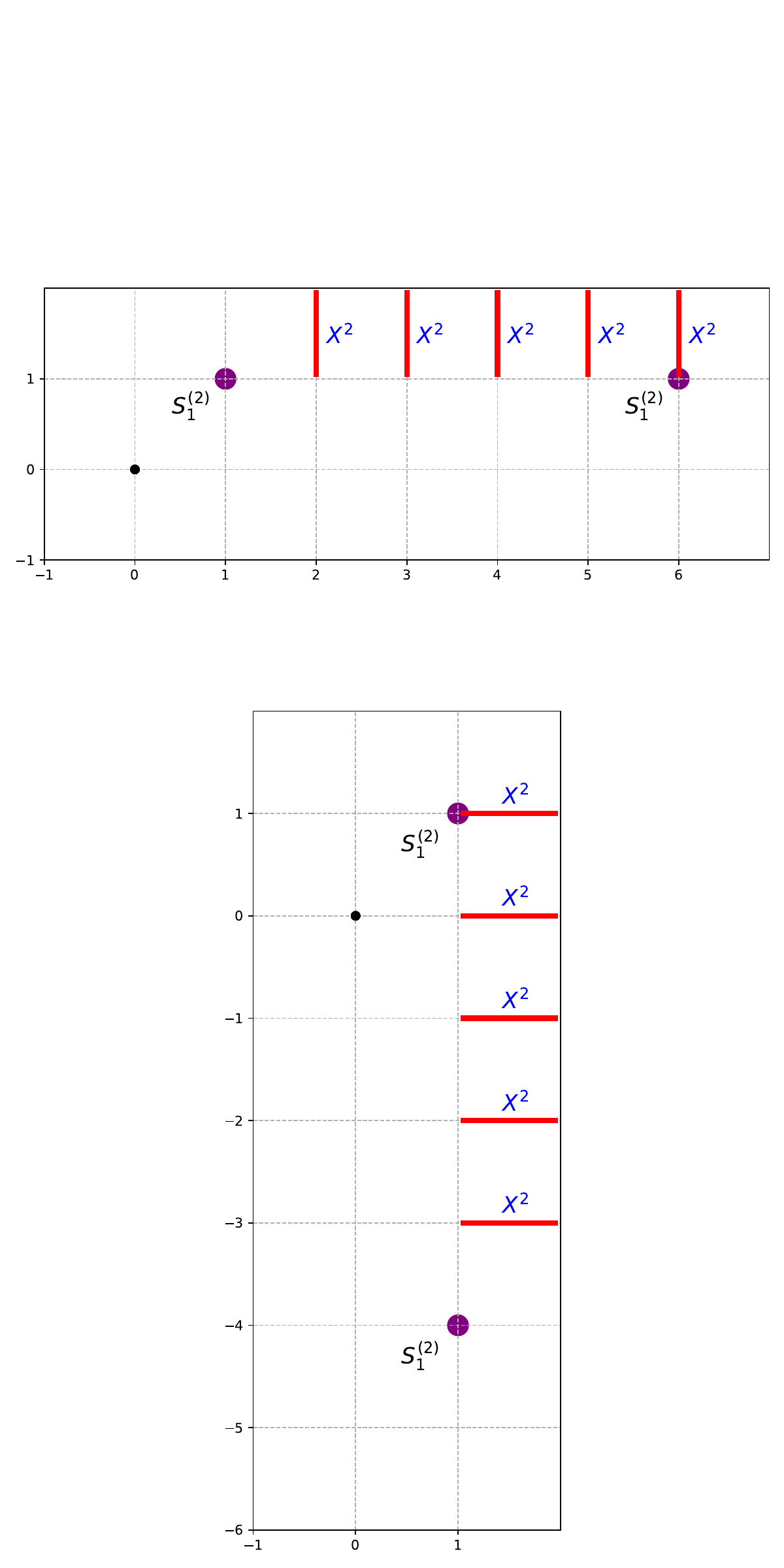}}
    \caption{The string operators of anyons $v_1$ and $v_2$ in the double semion code. For each subfigure, the first row represents the anyon string operators that move an anyon along the $x$-direction, and the second row represents the anyon string operators that move an anyon along the $y$-direction. Purple dots label the syndrome patterns at the endpoints of each string. }
    \label{fig:string_double_semion}
\end{figure*}

\subsection{Six-semion code}

For the six-semion code, we obtain {\color{black}two} basis anyons at $(1-x^1)$ and $(1-y^{-1})$ (also for greater $n_x, n_y$). The string operators for $n_x, n_y=5$ are shown in Fig.~\ref{fig:string_six_semion}. 
\begin{figure*}[htb]
    \centering
    \subfigure[The string operator of anyon $v_1$ in the six-semion code.]{\includegraphics[width=0.48\textwidth]{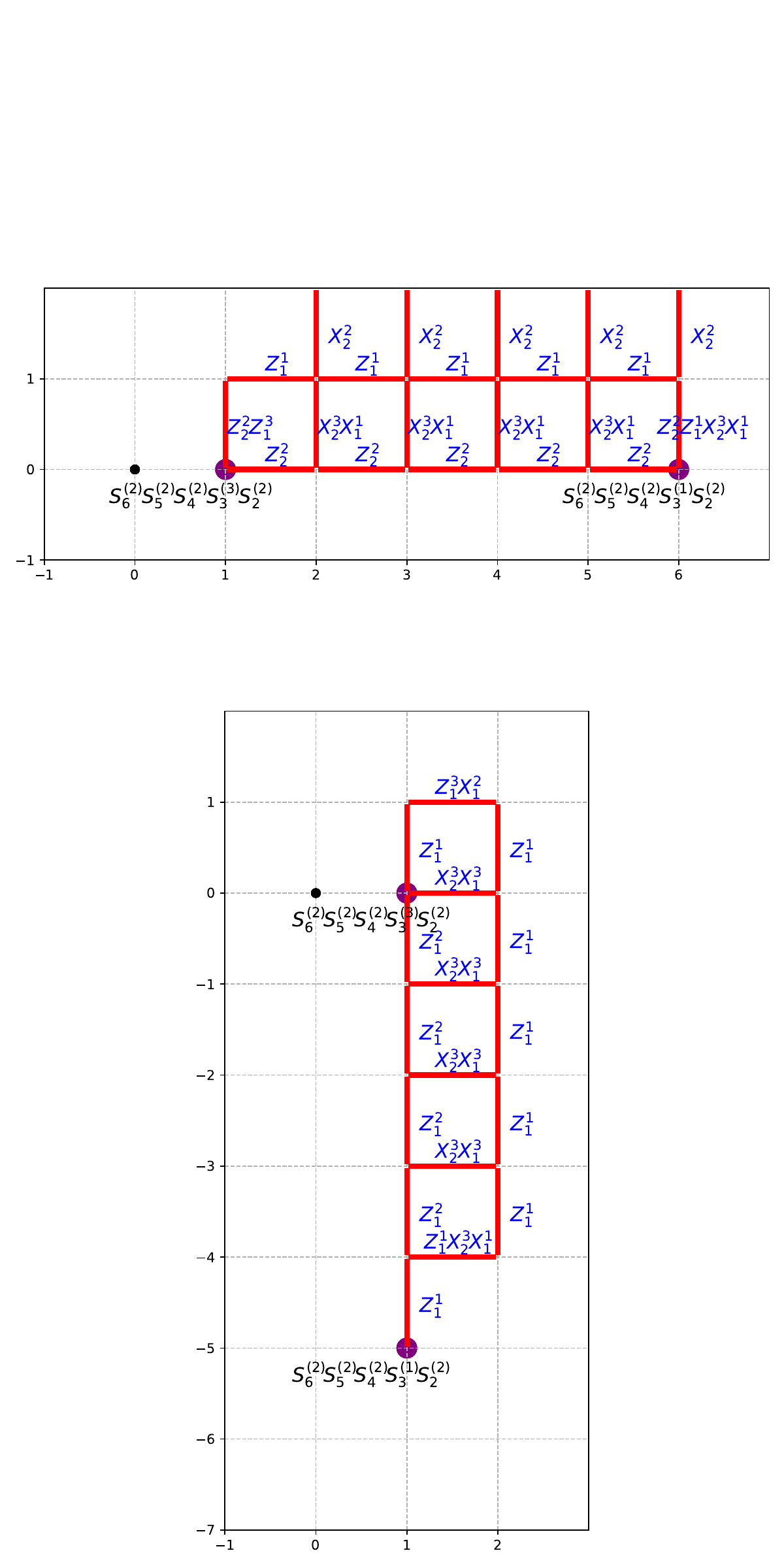}}
    \subfigure[The string operator of anyon $v_2$ in the six-semion code.]{\includegraphics[width=0.48\textwidth]{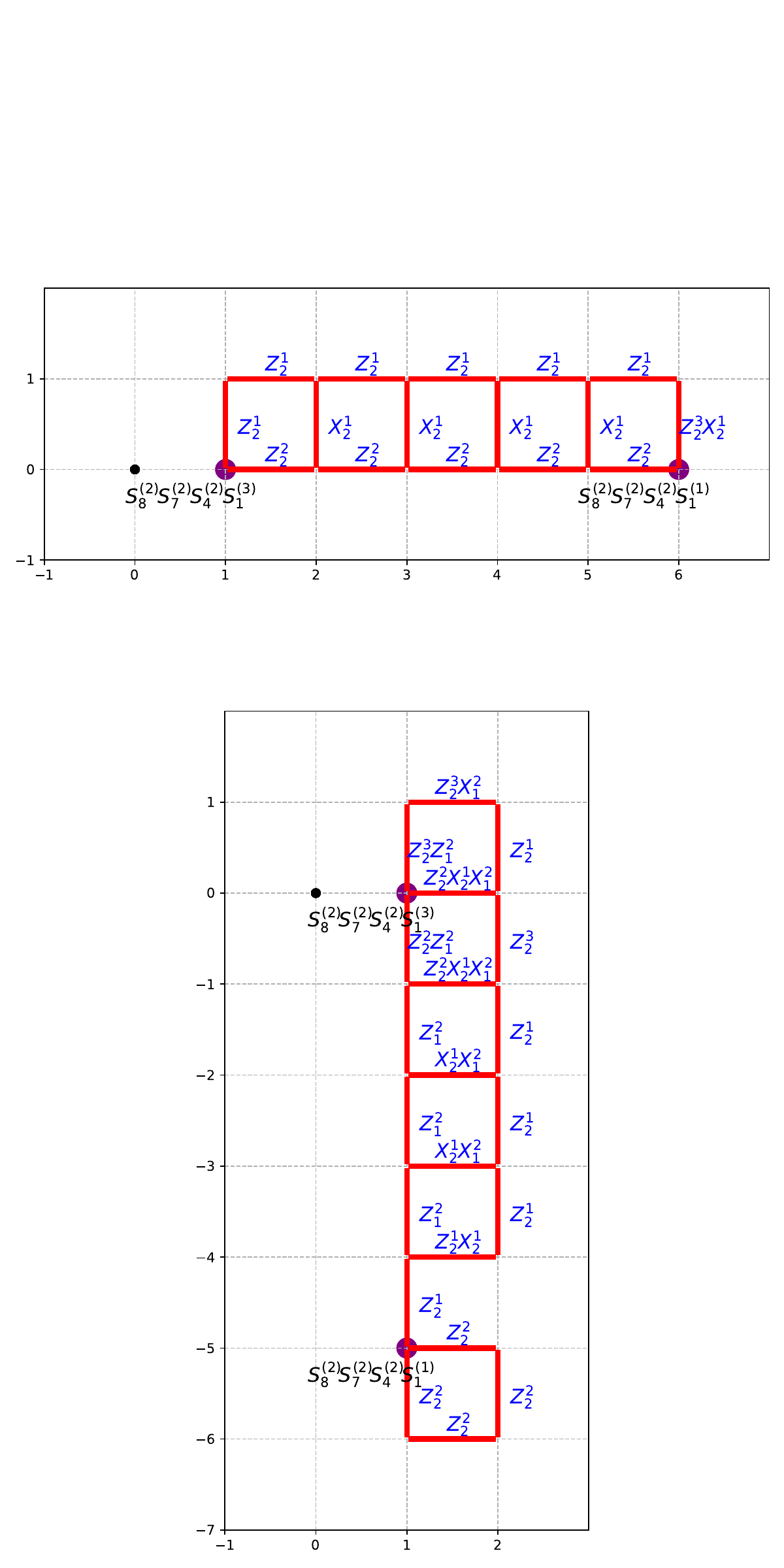}}
    \caption{The string operators of anyons $v_1$ and $v_2$ in the six-semion code. For each subfigure, the first row represents the anyon string operators that move an anyon along the $x$-direction, and the second row represents the anyon string operators that move an anyon along the $y$-direction. The syndrome pattern of endpoints of each string is labeled by purple dots. }
    \label{fig:string_six_semion}
\end{figure*}

\clearpage
\section{Braiding statistics for modified color codes}\label{appendix: braiding_tables}

This appendix demonstrates the complete braiding statistics of modified color codes (examples \circled{3} , \circled{4}, \circled{5}, \circled{6}) in Tables.~\ref{tab:Model 3}, \ref{tab:Model 4}, \ref{tab:Model 5}, \ref{tab:Model 6}. The anyons $v_i$ refers to the anyon string operators defined in the previous Appendix~\ref{appendix: string_operators}.

\begin{table}[h]
    \centering
    \begin{tabular}{|c|c|c|c|c|c|c|c|c|}
    \hline
         &$v_1$& $v_2$& $v_3$& $v_4$ & $v_5$& $v_6$& $v_7$&$v_8$\\
         \hline
         $v_1$& 1& 1& 1&1 & 1& 1& 1&\textcolor{red}{-1}\\
         \hline
         $v_2$& 1& 1& 1&\textcolor{red}{-1}& 1& 1& 1&1 
\\
         \hline
         $v_3$& 1& 1& 1& 1& 1& \textcolor{red}{-1}& 1&1\\
         \hline
         $v_4$& 1& \textcolor{red}{-1}& 1&1 & 1& 1& 1&1 \\
         \hline
 $v_5$& 1& 1& 1& 1 
& 1& 1& \textcolor{red}{-1}&1 
\\\hline
 $v_6$& 1& 1& \textcolor{red}{-1}& 1 
& 1& 1& 1&1 
\\\hline
 $v_7$& 1& 1& 1& 1& \textcolor{red}{-1}& 1& 1&1\\\hline
 $v_8$& \textcolor{red}{-1}& 1& 1& 1 & 1& 1& 1&1 \\\hline
    \end{tabular}
    \caption{Topological spins and braiding statistics of anyons for the modified color code A (example \textcircled{3}). We can see this table is formed by four copies of decoupled toric code, where $\{v_1, v_8\}, \{ v_2, v_4\}, \{v_3, v_6\}, \{v_5,v_7\}$ correspond to the $\{e,m\}$ pairs in the four copies of toric codes.}
    \label{tab:Model 3}
\end{table}

\begin{table}[h]
    \centering
    \resizebox{\columnwidth}{!}{%
    \begin{tabular}{|c|c|c|c|c|c|c|c|c|c|c|c|c|c|c|c|c|}
    \hline
        &$v_1$& $v_2$& $v_3$& $v_4$ & $v_5$& $v_6$& $v_7$&$v_8$ & $v_9$& $v_{10}$& $v_{11}$&$v_{12}$ & $v_{13}$& $v_{14}$& $v_{15}$&$v_{16}$\\
         \hline
        $v_1$& 1& 1& 1&1& 1& 1& 1&1 & \textcolor{red}{-1}& 1& 1&1  & 1& 1& 1&1 \\
         \hline
        $v_2$& 1& 1& 1&1& 1& 1& 1&1  & 1& 1& \textcolor{red}{-1}&1 & 1& 1& 1&1 \\
         \hline
        $v_3$& 1& 1& 1& 1& 1& 1& 1&1 & 1& 1& 1&1  & \textcolor{red}{-1}& 1& 1&1 \\
        \hline
        $v_4$& 1& 1& 1&1 & 1& 1& 1&1  & 1& 1& 1&\textcolor{red}{-1}& 1& 1& 1&1 \\
         \hline
        $v_5$& 1& 1& 1& 1 & 1& 1& 1&1& 1 & 1& 1&1 & 1& \textcolor{red}{-1}& 1&1 \\\hline
        $v_6$& 1& 1& 1& 1 & 1& 1& 1&1& 1& \textcolor{red}{-1}& 1&1 & 1& 1& 1&1 \\\hline
        $v_7$& 1& 1& 1& 1& 1& 1& 1&1 & 1& 1& 1&1  & 1& 1& \textcolor{red}{-1}&1\\\hline
        $v_8$& 1& 1& 1& 1 & 1& 1& 1&1& 1& 1& 1 &1& 1& 1& 1&\textcolor{red}{-1}\\\hline
        $v_9$ & \textcolor{red}{-1}& 1& 1& 1& 1& 1 & 1& 1 & 1& 1& 1&1  & 1& 1& 1&1\\\hline
        $v_{10}$& 1& 1& 1& 1 & 1& \textcolor{red}{-1}& 1& 1& 1 & 1& 1&1  & 1& 1& 1&1\\\hline
        $v_{11}$& 1& \textcolor{red}{-1}& 1& 1& 1 & 1& 1& 1& 1& 1 & 1&1  & 1& 1& 1&1 \\\hline
        $v_{12}$& 1& 1& 1& \textcolor{red}{-1}& 1& 1& 1 & 1& 1& 1& 1&1  & 1& 1& 1&1\\\hline
        $v_{13}$& 1& 1& \textcolor{red}{-1}& 1 & 1& 1& 1& 1  & 1& 1& 1& 1  & 1& 1& 1&1 \\\hline
        $v_{14}$& 1& 1& 1& 1 
        & \textcolor{red}{-1}& 1& 1& 1& 1 & 1& 1& 1& 1& 1 & 1&1\\\hline
        $v_{15}$& 1& 1& 1& 1& 1 & 1 & \textcolor{red}{-1}& 1& 1 
        & 1& 1& 1& 1& 1 & 1&1\\\hline
        $v_{16}$& 1& 1& 1& 1 
        & 1& 1& 1& \textcolor{red}{-1}& 1 & 1& 1& 1& 1& 1& 1&1\\\hline
    \end{tabular}}
    \caption{Topological spins and braiding statistics of anyons for the modified color code B (example \textcircled{4}). We can see this table is formed by four copies of decoupled toric code, where $\{v_1, v_9\}, \{ v_2, v_{11}\}, \{v_3, v_{13}\}$, $\{v_4,v_{12}\}, \{v_5,v_{14}\}, \{v_6,v_{10}\}, \{v_7,v_{15}\}, \{v_8,v_{16}\}$ correspond to the $\{e,m\}$ pairs in the eight copies of toric codes.}
    \label{tab:Model 4}
\end{table}

\begin{table}[h]
    \centering
    \begin{tabular}{|c|c|c|c|c|c|c|c|c|}
    \hline
         &$v_1$& $v_2$& $v_3$& $v_4$ & $v_5$& $v_6$& $v_7$&$v_8$\\
         \hline
         $v_1$& 1& 1& 1&\textcolor{red}{-1}& 1& 1& 1&1\\
         \hline
         $v_2$& 1& 1& 1&1& 1& \textcolor{red}{-1}& 1&1 \\
         \hline
         $v_3$& 1& 1& 1& 1& 1& 1& \textcolor{red}{-1}&1\\
         \hline
         $v_4$& \textcolor{red}{-1}& 1& 1&1 & 1& 1& 1&1 \\
         \hline
 $v_5$& 1& 1& 1& 1 
& 1& 1& 1&\textcolor{red}{-1}\\\hline
 $v_6$& 1& \textcolor{red}{-1}& 1& 1 
& 1& 1& 1&1 
\\\hline
 $v_7$& 1& 1& \textcolor{red}{-1}& 1& 1& 1& 1&1\\\hline
 $v_8$& 1& 1& 1& 1 & \textcolor{red}{-1}& 1& 1&1 \\\hline
    \end{tabular}
    \caption{Topological spins and braiding statistics of anyons for the modified color code C (example \textcircled{5}). We can see this table is formed by four copies of decoupled toric code, where $\{v_1, v_4\}, \{ v_2, v_6\}, \{v_3, v_7\}, \{v_5,v_8\}$ correspond to the $\{e,m\}$ pairs in the four copies of toric codes.}
    \label{tab:Model 5}
\end{table}

\begin{table}[h]
    \centering
    \begin{tabular}{|c|c|c|c|c|c|c|c|c|c|c|c|c|}
    \hline
         &$v_1$& $v_2$& $v_3$& $v_4$ & $v_5$& $v_6$& $v_7$&$v_8$ & $v_9$& $v_{10}$& $v_{11}$&$v_{12}$\\
         \hline
         $v_1$& 1& 1& 1&\textcolor{red}{-1}& 1& 1& 1&1 & 1& 1& 1&1 \\
         \hline
         $v_2$& 1& 1& 1&1& 1& 1& 1&1 
 & \textcolor{red}{-1}& 1& 1&1\\
         \hline
         $v_3$& 1& 1& 1& 1& 1& 1& \textcolor{red}{-1}&1 & 1& 1& 1&1 \\
         \hline
         $v_4$& \textcolor{red}{-1}& 1& 1&1 & 1& 1& 1&1  & 1& 1& 1&1 \\
         \hline
 $v_5$& 1& 1& 1& 1 
& 1& 1& 1&1& 1 & \textcolor{red}{-1}& 1&1\\\hline
 $v_6$& 1& 1& 1& 1 
& 1& 1& 1&1& 1& 1 & \textcolor{red}{-1}&1\\\hline
 $v_7$& 1& 1& \textcolor{red}{-1}& 1& 1& 1& 1&1 & 1& 1& 1&1 \\\hline
 $v_8$& 1& 1& 1& 1 & 1& 1& 1&1& 1& 1& 1 &\textcolor{red}{-1}\\\hline
 $v_9$
& 1& \textcolor{red}{-1}& 1& 1& 1& 1 & 1& 1 & 1& 1& 1&1 \\\hline
 $v_{10}$
& 1& 1& 1& 1 & \textcolor{red}{-1}& 1& 1& 1& 1 & 1& 1&1 \\\hline
 $v_{11}$
& 1& 1& 1& 1& 1 & \textcolor{red}{-1}& 1& 1& 1& 1 & 1&1 \\\hline
 $v_{12}$& 1& 1& 1& 1& 1& 1& 1 & \textcolor{red}{-1}& 1& 1& 1&1 \\\hline
    \end{tabular}
    \caption{Topological spins and braiding statistics of anyons for the modified color code D (example \textcircled{6}). We can see this table is formed by four copies of decoupled toric code, where $\{v_1, v_4\}, \{ v_2, v_9\}, \{v_3, v_7\}, \{v_5,v_{10}\},  \{v_6,v_{11}\},  \{v_8,v_{12}\}$ correspond to the $\{e,m\}$ pairs in the six copies of toric codes.}
    \label{tab:Model 6}
\end{table}

\section{Algorithm pseudocode for extracting topological orders from generalized Pauli stabilizer codes}\label{appendix: pseudocode}

{\color{black}
This appendix presents the pseudocode for the algorithm described in Sections~\ref{sec: computational_method} and~\ref{sec: algorithm}.
\begin{breakablealgorithm}
    \caption{Checking topological order condition }
    \begin{algorithmic}[1] 
        \Require Stabilizer polynomials $\mathcal{S}$, truncation range $k$, translation range $m$, searching range $m'${, $\ZZ_d$ $d$}
        \Ensure Whether the stabilizers satisfy the TO condition
        \State Calculate the error syndrome $\epsilon(\mathcal{P})$ for single-Pauli operator $\forall \mathcal{P} \in \{\mX_1, \mZ_1,...,\mX_w,\mZ_w\}$ from stabilizer polynomials. Get the syndrome matrix 
        \begin{eqs}
        E\leftarrow \begin{bmatrix}
            \epsilon(\mX_1)\\
            \epsilon(\mX_2)\\
            \vdots\\
            \epsilon(\mX_w)\\
            \epsilon(\mZ_1)\\
            \vdots\\
            \epsilon(\mZ_w)
        \end{bmatrix},
        \end{eqs}
        which is a $(2w)\times [t(2k+1)^2]$ matrix for system with $t$ stabilizers and $w$ physical qudits per unit cell.
        \State Apply the translation duplicate map $\text{TD}_m$ with range $m$ to each row of $E$, such that 
        \begin{eqs}
        \wwide{M}_1 \leftarrow \wwide{\text{TD}_m(E)}=\begin{bmatrix}
            \wwide{\text{TD}_m(\epsilon(\mX_1))}\\
            \wwide{\text{TD}_m(\epsilon(\mX_2))}\\
            \vdots\\
            \wwide{\text{TD}_m(\epsilon(\mZ_w))}
        \end{bmatrix}.\end{eqs} $\wwide{M}_1$ is a $[2w(2m+1)^2]\times [t(2k+1)^2]$ matrix 
        \State Perform modified Gaussian elimination (for non-prime dimensional qudits) or Gaussian elimination (for prime dimensional qudits) of on $\wwide{M}_1$, get $\text{(M)GE}(\wwide{M}_1)$ and a $[2w(2m+1)^2]\times [(2k+1)^2]$ relation matrix $R_1$.
        \State Obtain local operator set $O$ from rows of $R_1$ that correspond to zero rows in (modified) Gaussian elimination
        
        \State Apply translation duplication map $\text{TD}_{m'}$ with range $m'<k$ to stabilizer matrix $\mathcal{S}$, 
        \begin{eqs}
            \wwide{M}_2 \leftarrow \begin{bmatrix}
            \wwide{\text{TD}_{m'}(\mS_1)}\\
            \wwide{\text{TD}_{m'}(\mS_2)}\\
            \vdots\\
            \wwide{\text{TD}_{m'}(\mS_t)}
        \end{bmatrix}.\end{eqs} $\wwide{M}_2$ is a $[t(2m'+1)^2]\times [(2k+1)^2]$ matrix.

        \State Perform modified Gaussian elimination (for non-prime dimensional qudits) or Gaussian elimination (for prime dimensional qudits) of on $\wwide{M}_2$, get $\text{(M)GE}(\wwide{M}_2)$.

        


        \If{$d$ is prime}
            \State $r_1\leftarrow \text{rank}(GE(\wwide{M}_2))$.
             \State $r_2$ $\leftarrow \text{GE}(\text{concate}(O,\text{GE}(\wwide{M}_2)))$
             \If{$r_1\neq r_2$}
                 \State \Return "Does not satisfy the TO condition"
             \EndIf
         \EndIf
         \If{$d$ is nonprime}
            \For{$O_i$ in local operator set $O$}
                 \If{$O_i$ is not in the row span of $\text{MGE}(\wwide{M}_2)$.}
                     \State \Return "Does not satisfy the TO condition"
                \EndIf
             \EndFor
         \EndIf
         \State \Return "Satisfy the TO condition"
        
    \end{algorithmic}
\end{breakablealgorithm}

\begin{breakablealgorithm}
    \caption{Extracting string operators }
    \begin{algorithmic}[1] 

    \Require Syndrome matrix $E$, $\wwide{M}_1=\wwide{\text{TD}_m(E)}$, search range $N_x$, $N_y$ for $x$- and $y$-direction.
    \Ensure Basis anyon matrix $\wwide{V}$ and string operators $\wwide{P_x}, \wwide{P_y}$
    \For{$n_x = 1, 2, ..., N_x$}
        \State Define matrix $\wwide{M}_3^X$ as 
        \begin{eqs}
            \wwide{M}_3^X(n_x) \leftarrow \begin{bmatrix}
            \wwide{\mathrm{TD}_m( \eps(\mX_1))} \\
            \vdots\\
            \wwide{\mathrm{TD}_m( \eps(\mZ_w))} \\
            \wwide{\mathrm{TD}_m( (1-x^{n_x})\mathbf{1}_1)} \\
            \wwide{\mathrm{TD}_m(  (1-x^{n_x})\mathbf{1}_2)} \\
            \vdots \\
            \wwide{\mathrm{TD}_m(  (1-x^{n_x})\mathbf{1}_q)}
        \end{bmatrix},
        \end{eqs}
        where $\mathbf{1}_i$ represents a $t$-dimensional one-hot vector with the $i$-th entry equals to 1. $\wwide{M}_3$ is a $[(2w+t)(2m+1)^2] \times [t(2k+1)^2]$ matrix.
        \State Calculate $\text{(M)GE}(\wwide{M}_3^X (n_x))$, obtain a anyon matrix 
        \begin{eqs}
            \wwide{V_X(n_x)}\leftarrow \begin{bmatrix}
                \wwide{v_1(n_x)}\\
                \vdots\\
                \wwide{v_{\alpha} (n_x)}
            \end{bmatrix}
        \end{eqs} which is a $\alpha \times (t(2k+1)^2)$ matrix.\footnote{Assume we get $\alpha$ different anyon solutions here.}
        Their string operators along the $x$-direction form the string operator matrix
        \begin{eqs}
            \wwide{P_x(n_x)}\leftarrow\begin{bmatrix}
                \wwide{P_x^{v_1(n_x)}}\\
                \vdots\\
                \wwide{P_x^{v_\alpha(n_x)}}
            \end{bmatrix}
        \end{eqs}
        which is a $\alpha \times (2w(2k+1)^2)$ matrix obtained from the relation matrix $R_X(n_x)$.
        \State $\wwide{V_X(n_x)}, \wwide{P_x(n_x)}\leftarrow \text{Find\_Independent\_Anyons} (\wwide{V_X(n_x)}, \wwide{P_x(n_x)})$
        
    \EndFor
    \State 
    \begin{eqs}
        &n_x^*\leftarrow \min(\argmax_{n_x}(\text{Shape}(\wwide{V_X(n_x)},0))) \\
        &\wwide{M}_3^X\leftarrow \wwide{M}_3^X(n_x^*)\\
        &\wwide{V_X} \leftarrow \wwide{V_X(n_x^*)}\\
        &\wwide{P_x} \leftarrow \wwide{P_x(n_x^*)}
    \end{eqs}
    
    \State
    \For{$n_y = 1, 2, ..., N_y$}
        \State Define matrix $\wwide{M}_3^Y$ as 
        \begin{eqs}
             \wwide{M}_3^Y \leftarrow \begin{bmatrix}
             \wwide{\mathrm{TD}_m( \eps(\mX_1))} \\
             \vdots\\
             \wwide{\mathrm{TD}_m( \eps(\mZ_w))} \\
             \wwide{\mathrm{TD}_m( (1-y^{n_y})\mathbf{1}_1)} \\
             \wwide{\mathrm{TD}_m(  (1-y^{n_y})\mathbf{1}_2)} \\
             \vdots \\
             \wwide{\mathrm{TD}_m(  (1-y^{n_y})\mathbf{1}_q)}
        \end{bmatrix},
        \end{eqs}
         $\wwide{M}_3$ is a $[(2w+t)(2m+1)^2] \times [t(2k+1)^2]$ matrix.
        \State Calculate $\text{(M)GE}(\wwide{M}_3^X (n_y))$, obtain a anyon matrix $\wwide{V_Y(n_y)}=\{\wwide{v_1(n_y)}, \wwide{v_2(n_y)},...\}$ and their string operators along the $y$-direction $\wwide{P_y(n_y)}=\{ \wwide{P_y^{v_1(n_y)}}, \wwide{P_y^{v_2(n_y)}},... \}$ from the relation matrix $R_Y(n_y)$.
        \State $\wwide{V_Y(n_y)}, \wwide{P_y(n_y)}\leftarrow \text{Find\_Independent\_Anyons} (\wwide{V_Y(n_y)}, \wwide{P_y(n_y)})$
    \EndFor
     \State 
    \begin{eqs}
        &n_y^*\leftarrow \min(\argmax_{n_y}(\text{Shape}(\wwide{V_Y(n_x)},0))\\
        &\wwide{M}_3^Y\leftarrow \wwide{M}_3^Y(n_y^*)
    \end{eqs}
    
    {
    \State
    \State $\wwide{V}, \wwide{P_x} \leftarrow \text{Find\_Independent\_Anyons}(\wwide{V_X},\wwide{P_x})$ \Comment{Extract the basis anyon matrix $\wwide{V}$ and its string operators $\wwide{P_x}$ from $\wwide{V_X}$ and $\wwide{P_x}$.}
 
    \State Reset $\wwide{M}_1 \leftarrow \wwide{\text{TD}_m(E)}$
    \For{$v_i$ is a rows of anyon matrix $\wwide{V}$}
    \State Expand $\wwide{(1-y^{-n_y})}v_i$ in the span of rows of $\wwide{M}_1$, obtain the corresponding string operator $\wwide{P_y^{v_i}}$
    \State $\wwide{P_y} \leftarrow \text{concate}(\wwide{P_y^{v_i}} , \wwide{P_y})$
    \EndFor
    \State \Return Basis anyon matrix $\wwide{V}$ and their string operators along the $x-$ and $y-$ direction $\wwide{P_x}, \wwide{P_y}$
    
    }
    \State
    \State\Function{Find\_Independent\_Anyons}{$\widetilde{V}',\widetilde{P}'$}
    \Comment{$\widetilde{V}’$ is the input anyon set, $\widetilde{P}'$ is the corresponding string operators.}

    \State Construct $M_v$ using the convention in Eq.~\eqref{eq: anyon relation example}
        \State Calculate the Smith normal form of $M_v $ as  $PAQ=M_v$
        \State $index \leftarrow \arg_i{A(i,i)\neq \pm 1}$
        \State The relation matrix between basis anyon matrix $\widetilde{V}$ and the input anyon matrix $\widetilde{V}'$ is $Q(index,:)$ such that $\widetilde{V}=\widetilde{V}'Q(:,index) $. \Comment{Find the basis anyon matrix.}
        \State $\widetilde{P}\leftarrow  \widetilde{P}' Q(:,index) $
    \State \Return $\widetilde{V}, \widetilde{P}$
  \EndFunction
    \end{algorithmic}
\end{breakablealgorithm}

\begin{breakablealgorithm}
    \caption{Extracting topological spins } 
    \begin{algorithmic}[1]
        \Require String operator matrices $\wwide{P_x}, \wwide{P_y}$ for basis anyon matrix $\wwide{V}$ in the $x$ and $y$ direction, the length extension $q$ of string operator 
        \Ensure Topological spin array $T$ and braiding matrix $B$

        \State Define Topological spin array $T=\{\}$
        \For {$\wwide{v_i}$ is a row of anyon matrix $\wwide{V}$}
            \State Convert $\wwide{P_x^{v_i}}$ and $\wwide{P_y^{v_i}}$ to the corresponding polynomials $P_x^{v_i}$ and $P_y^{v_i}$ 
            \State Extend the length of string operators to its $q$ times 
            \begin{eqs}
                    U_1^{v_i} &\leftarrow(x^{-q n_x}+ x^{-(q-1) n_x}+\cdots +x^{-n_x} )P_x^v,\\
        U_2^{v_i}  &\leftarrow (1+y^{-n_y}+y^{-2n_y}+\cdots+y^{-q n_y})P_y^v,\\
        U_3^{v_i}  &\leftarrow -(1+x^{n_x}+x^{2 n_x}+\cdots +x^{q n_x})P_x^v, 
                \end{eqs}
            \State 
                \begin{eqs}
                \theta(v_i)\leftarrow \frac{2 \pi i}{d} [U_1, U_2]+ [U_2, U_3] + [U_3, U_1] 
                \end{eqs}
            \State $T\leftarrow \text{concate}(\theta(v_i), T)$
        \EndFor
        
        \State
        \State Define braiding matrix $B=\{\}$
        \For {$\wwide{v_i}$ is a row of anyon matrix $\wwide{V}$}
            \For{$\wwide{v_j}$ is a row of anyon matrix $\wwide{V}$, and $\wwide{v_j}\neq \wwide{v_i}$ }

                \State Convert $\wwide{P_x^{v_i}}, \wwide{P_y^{v_i}}, \wwide{P_x^{v_j}}, \wwide{P_y^{v_j}}$ to the corresponding polynomials $P_x^{v_i}, P_y^{v_i}, P_x^{v_j}, P_y^{v_j}$.
                \State Combine the string operator 
                \begin{eqs}
                    &P_x^{v_i\times v_j}\leftarrow P_x^{v_i}P_x^{v_j}\\
                    &P_y^{v_i\times v_j}\leftarrow P_y^{v_i}P_y^{v_j}
                \end{eqs}
                \State Extend the length of string operators to its $q$ times 
                \begin{eqs}
        U_1^{v_i\times v_j}  &\leftarrow(x^{-q n_x}+ x^{-(q-1) n_x}+\cdots +x^{-n_x} )P_x^{v_i\times v_j},\\
        U_2^{v_i\times v_j}  &\leftarrow (1+y^{-n_y}+y^{-2n_y}+\cdots+y^{-q n_y})P_y^{v_i\times v_j},\\
        U_3^{v_i\times v_j}  &\leftarrow -(1+x^{n_x}+x^{2 n_x}+\cdots +x^{q n_x})P_x^{v_i\times v_j}, 
                \end{eqs}
                \State
                \begin{eqs}
                \theta(v_i\times v_j)\leftarrow \frac{2 \pi i}{d} [U_1, U_2]+ [U_2, U_3] + [U_3, U_1] 
                \end{eqs}
                \State $B(v_i, v_j)\leftarrow \frac{\theta(v_i \times v_j)}{\theta(v_i) \theta(v_j)}$
            
            \EndFor
        \EndFor
        \State Take the converged values of $T$ and $B$. 

        \State \Return Topological spin array $T$ and braiding matrix $B$
    \end{algorithmic}
\end{breakablealgorithm}

\begin{breakablealgorithm}
    \caption{Rearranging the basis anyons} 
    \begin{algorithmic}[1]
        \Require Basis anyon set $V$, topological spin array $T$ and braiding matrix $B$
        \Ensure Rearranged anyon set $V_{\text{dec}}$ formed by decoupled $\{e,m\}$ pairs
        \State $V_{\text{dec}} \leftarrow \{\}$
        \For{Find a nontrivial boson $b\in V$}
            \State Find a anyon $c\in (V-\{b\})$ such that $\textcolor{black}{B(b,c)=\exp (\frac{2\pi i}{p})}$
            \If{$c$ is not a boson}
                \State $c \leftarrow c \times b^\#$ \Comment{make $c$ to be a boson}
                \State $V_{\text{dec}}\leftarrow\{ \{b,c\} \text{pairs} \}\cup V_{\text{dec}}$
            \EndIf
            \For{$v_i \in (V-\{b,c\})$}
                \If{$B(b,v_i) \neq 1$ or $B(c,v_i) \neq 1$}
                    \State $v_i \gets v_i\times b^{\#}\times c^{\#}$  \Comment{decouple $b,c$ with the rest of anyons}
                    \State $V\leftarrow V-\{b,c\}$
                \EndIf
            \EndFor
        \EndFor

        \State \Return Rearranged anyon set $V_{\text{dec}}$ formed by decoupled $\{e,m\}$ pairs
    \end{algorithmic}
\end{breakablealgorithm}
We have completed the pseudocode for our main algorithm. For a detailed theoretical background, please refer to Sections~\ref{sec: laurent_polynomial}, \ref{sec: computational_method}, and \ref{sec: algorithm}.
}

\bibliography{bibliography.bib}

\end{document}